\documentclass[a4paper,notitlepage,12pt]{article}
\usepackage[sectionbib]{natbib}
\usepackage{amsmath,mathtools}
\usepackage{amssymb,amsthm} 
\usepackage{geometry}
\usepackage{setspace}
\usepackage{enumitem}
\usepackage{dsfont}
\usepackage{algorithm,algpseudocode,float}
\usepackage{lipsum}
\usepackage{graphicx}
\usepackage{multirow}
\usepackage{booktabs}
\usepackage{caption}
\usepackage{xr}
\usepackage{longtable,tabularx}
\usepackage{threeparttable}
\usepackage{adjustbox}
\usepackage{siunitx}
\usepackage{array}
\usepackage{makecell}
\usepackage{url}
\usepackage{dcolumn}
\newcolumntype{d}{D{.}{.}{-1}} % decimal-aligned column

\allowdisplaybreaks[1]     % 允许公式跨页显示
\usepackage{array}
\usepackage{xr}

\makeatletter
\newcommand*{\addFileDependency}[1]{% 
  \typeout{(#1)}
  \@addtofilelist{#1}
  \IfFileExists{#1}{}{\typeout{No file #1.}}
}
\makeatother

\newcommand{\bm}{\mathbf}
\newcommand{\bbm}{\boldsymbol}

\newcommand{\op}{\mathrm{op}}
\newcommand{\F}{\mathrm{F}}
\newcommand{\tr}{\mathrm{tr}}

\DeclareMathOperator*{\argmin}{arg\,min}

\makeatletter
\newenvironment{breakablealgorithm}
{% \begin{breakablealgorithm}
		\begin{center}
			\refstepcounter{algorithm}% New algorithm
			\hrule height.8pt depth0pt \kern2pt% \@fs@pre for \@fs@ruled
			\renewcommand{\caption}[2][\relax]{% Make a new \caption
				{\raggedright\textbf{\ALG@name~\thealgorithm} ##2\par}%
				\ifx\relax##1\relax % #1 is \relax
				\addcontentsline{loa}{algorithm}{\protect\numberline{\thealgorithm}##2}%
				\else % #1 is not \relax
				\addcontentsline{loa}{algorithm}{\protect\numberline{\thealgorithm}##1}%
				\fi
				\kern2pt\hrule\kern2pt
			}
		}{% \end{breakablealgorithm}
		\kern2pt\hrule\relax% \@fs@post for \@fs@ruled
	\end{center}
}
\makeatother

\usepackage[
  colorlinks,
  linkcolor=blue,
  citecolor=blue,
  filecolor=magenta,
  urlcolor=cyan,
  pageanchor=true   % 关键：确保写出有效页锚
]{hyperref}

\makeatletter
\renewcommand*{\l@section}[2]{%
  \addpenalty{-\@highpenalty}%
  \addvspace{0.7em}%
  \@dottedtocline{1}{1em}{2em}{\bfseries #1}{#2}}
\renewcommand*{\l@subsection}[2]{%
  \@dottedtocline{2}{3em}{2em}{#1}{#2}}
\makeatother

\newtheorem{assumption}{Assumption}
\newtheorem{definition}{Definition}

\newtheorem{theorem}{Theorem}
\newtheorem{lemma}{Lemma}

\newtheorem{proposition}{Proposition}

\title{\vspace{-2cm}Private Federated Learning for High-dimensional \\Time Series}
\author{Kejun Chen and Qianqian Zhu\thanks{Address for correspondence: Qianqian Zhu, School of Statistics and Data Science, Institute of Big Data Research, Shanghai University of Finance and Economics, Shanghai, China. Email: zhu.qianqian@mail.shufe.edu.cn}  \\ \vspace{-0.3cm}\textit{Shanghai University of Finance and Economics}}

\begin{document}

\setlength{\parindent}{16pt}

\maketitle

\begin{abstract}
	In the era of big data, leveraging information from multiple clients while preserving data privacy has emerged as a critical challenge in modern statistical modeling and forecasting. This paper introduces a privacy-preserving federated learning framework for high-dimensional vector autoregressive models, where each client's dynamics are characterized by a common low-rank structure augmented with sparse client-specific deviations. %The proposed approach makes two key contributions. First, 
	We develop a two-stage estimation procedure that integrates differentially private representation learning for the shared component with local personalization for client-specific adjustments, enabling effective information pooling under selective privacy constraints. %Second, 
	Non-asymptotic error bounds are established for both the single-client and federated estimators to characterize the inherent privacy-utility trade-off, and consistency of a ridge-type rank selection criterion is proved. %providing theoretical guarantees that characterize the inherent privacy--utility trade-off.  
	Simulation studies demonstrate that federation substantially improves estimation accuracy when local sample sizes are limited. Two empirical applications to analyzing electricity-economy linkages across U.S. states and conducting multi-task macroeconomic forecasting across countries, highlight the superior predictive accuracy of the proposed method over existing single-client benchmarks.  
	%The practical value of the framework is further demonstrated through two empirical applications: a selective federated analysis of electricity-economy linkages across U.S. states and a multi-task macroeconomic forecasting study across countries. 
%This paper develops a privacy-preserving federated learning framework for vector autoregressive (VAR) models in high-dimensional settings. We study a structured multi-client setting in which each client shares a common structure with client-specific departures. As a benchmark, we study a single-client estimator based solely on local data To estimate these components under selective privacy protection, we propose a two-stage procedure. Stage~I performs differentially private representation learning for the common low-rank component. Stage~II carries out local personalization at each client. We establish nonasymptotic error bounds and provide a robust rank selection method with consistency guarantees. Simulation studies illustrate the privacy--utility trade-off and show that federation can substantially improve performance when local sample sizes are limited. Two empirical applications, including a selective federated analysis of electricity--economy linkages and a multi-task macroeconomic forecasting study, demonstrate the practical utility of the proposed approach.
\end{abstract}


\newpage
	\begingroup
	\setcounter{tocdepth}{2}
	\setstretch{1.0}
	\small
	\tableofcontents
	\endgroup
	\vspace{0.5em}

\vspace{1cm}

\linespread{1.55}
\selectfont{}

\section{Introduction}

Recent advances in data collection and statistical processing have substantially improved data quality, making increasingly rich and high-dimensional datasets central to policy design and decision-making by governments, regulators, and local administrators \citep{phillips2020evidence}. 
At the subnational level, economic variables are particularly valuable as they provide a granular view of local economic conditions and enable targeted policy interventions. 
In practice, however, the sample size available for a small region is often limited. 
When the number of economic variables collected for such a region is large relative to its sample size, relying solely on local data may be insufficient for reliable inference or effective policy design. 
A natural remedy is to pool information across economically similar regions.
However, many informative variables for small regions, such as local purchasing managers' indexes, local electricity consumption, employment, and energy-use measures, are often not publicly available or directly shareable due to confidentiality or disclosure-risk concerns \citep{matthews2011confidentiality,young2022using}. 
%Highly disaggregated statistics can reveal distinctive features of the local economy and are therefore frequently subject to confidentiality restrictions or disclosure-risk considerations \citep{matthews2011confidentiality,young2022using}. As a result, while such local data would be valuable for cross-region information pooling, they are often unavailable for direct sharing in practice.
This practical need for cross-region information sharing for high-dimensional forecasting under the constraints of data privacy motivates the framework developed in this paper.

Federated learning provides a natural framework for decentralized and privacy-constrained settings, where data remain stored at local clients, and a global model is learned via iterative client-server communication instead of raw data centralization \citep{mcmahan2017communication,zhang2021survey}.
In standard federated learning, each client transmits local parameters or gradients to the server for aggregation, but these transmissions lack formal privacy protection, as sensitive information about local datasets may still be inferred from the transmitted parameters or gradients \citep{kairouz2021advances}. 
%While this approach enables multiple data holders to jointly leverage distributed information while avoiding raw-data centralization, keeping data local does not guarantee rigorous privacy protection \citep{kairouz2021advances}. Sensitive information about local datasets may still be inferred from the transmitted parameters or gradients. 
To address this limitation, a growing body of research integrates differential privacy (DP) into federated learning \citep{geyer2017differentially}. %privacy-preserving federated learning by incorporating differential privacy (DP) into the federated learning process \citep{geyer2017differentially}.
Classical DP provides a rigorous framework for protecting individual records \citep{dwork2006calibrating,dwork2014algorithmic}, which requires that the output of a randomized algorithm does not change substantially when any single record is modified.  
%Classical DP is formulated for datasets as collections of individual records \citep{dwork2006calibrating,dwork2014algorithmic}. Intuitively, DP requires that the output of a randomized algorithm should not change substantially when any single individual's record is modified, so that the server cannot reliably infer whether a particular individual is included in the dataset or recover sensitive information.
In practice, DP is implemented by injecting carefully calibrated random noise into the information transmitted by each client, ensuring that the contribution of any individual user is formally protected while still allowing the model to effectively utilize each client's data information \citep{geyer2017differentially}.
%formally protecting individual contributions while still enabling effective modeling using each client's data information \citep{geyer2017differentially}.

However, most existing theoretical work on privacy-preserving federated learning assumes independent observations. Examples include federated linear regression \citep{cai2021cost,liu2023near,li2024federated}, federated ReLU regression \citep{ding2025nearly}, and federated covariance estimation \citep{li2024federatedpca}. By contrast, federated methods and guarantees for multivariate time series remain scarce. Moreover, extending DP to dependent time series introduces additional challenges. 
%Despite this progress, most existing theoretical treatments of privacy-preserving federated learning are still developed for independent observations. For example, a substantial literature studies federated linear regression \citep{cai2021cost,liu2023near,li2024federated}, federated ReLU regression \citep{ding2025nearly}, and federated covariance estimation \citep{li2024federatedpca}. By contrast, federated methods and guarantees for multivariate time-series forecasting where observations are temporally dependent remain comparatively limited.
%Moreover, while classical DP protects against changes in a single record, extending DP to dependent time series introduces additional subtleties. 
Even when dependent or sequential data are considered, the privacy notion is typically still defined via neighboring datasets that differ only at a single time point, without explicitly accounting for how such a perturbation propagates through the underlying dynamic system and influences subsequent observations \citep{cardoso2022differentially,huang2023shuffled}.
This gap motivates the need for a privacy-preserving federated framework specifically designed for high-dimensional time series forecasting. 

%With federated learning and DP in place, the focus shifts to multivariate time series forecasting, where temporal dependence and dynamic propagation pose additional challenges.
Vector autoregressive (VAR) models provide a natural and widely used framework for multivariate time series forecasting \citep{lutkepohl2005new,zivot2006vector,hamilton2020time}. 
However, when the number of variables is large relative to the available sample size, the forecasting problem becomes inherently high-dimensional. 
To address this issue, recent work has developed high-dimensional VAR models that can capture complex dynamic dependencies even with limited observations. 
One prominent line of work assumes sparsity in the transition matrices, assuming that only a small fraction of lagged cross-variable effects are nonzero; see, e.g., \citet{basu2015regularized} and \citet{kock2015oracle}. A complementary line exploits low-rank structure to capture low-dimensional common cross-equation dependence in large VAR systems; see, e.g., \citet{reinsel1998multiple} and \citet{wang2022high}. 
In many forecasting applications, however, it is often more realistic to model VAR transition matrices with a low-rank plus sparse decomposition, where the low-rank part captures common low-dimensional dependence, while the sparse part accommodates idiosyncratic effects that cannot be well represented by a low-rank structure alone; see, e.g., \citet{basu2019low} and \citet{bai2023multiple}.
Intuitively, the low-rank component represents common latent forces that induce broad comovement across series, while the sparse component captures a limited number of direct dynamic interactions. 
Ignoring either component can lead to model misspecification.

Motivated by the challenges outlined above, this paper develops a novel privacy-preserving federated framework for high-dimensional time series forecasting based on VAR models. This framework jointly captures cross-client commonality and client-specific heterogeneity in decentralized settings where sensitive local data cannot be directly shared and privacy protection is required. 
Specifically, for each client, we model the local time series using a VAR specification in which the transition matrices admit a low-rank-plus-sparse decomposition. 
To facilitate effective information sharing across clients while preserving heterogeneity, we assume the low-rank component is common across all clients to capture shared latent dynamics, whereas the sparse components are client-specific to accommodate individual deviations. 
We propose a two-stage estimation method with formal privacy guarantees, and provide a single client learning as a baseline. Theoretical comparisons reveal that the proposed federated approach outperforms single client learning when the privacy cost is moderate, client-specific sparse components are small, and the pooled sample is sufficiently large. 
We demonstrate the empirical effectiveness of the proposed framework through simulation studies and two real applications. The numerical results show that our method can improve forecasting accuracy in data-scarce local environments while maintaining formal privacy protection.
Our main contributions are threefold.
\begin{enumerate}
    \item[(i)] We introduce a privacy-preserving federated framework for high-dimensional VAR forecasting that jointly models cross-client commonality and client-specific heterogeneity. Notably, when privacy protection is not required, our framework naturally reduces to a personalized multi-task learning approach, providing a useful tool for improving prediction accuracy in multi-task settings. 

    \item[(ii)] We propose a regularized single-client estimator as a baseline, and develop a privacy-preserving two-stage federated procedure: Stage I performs private representation learning to estimate the common structure under federated constraints, while Stage II refines client-specific dynamics via local personalization without sharing raw data.

    \item[(iii)] Non-asymptotic properties are established for the single client learning and federated methods. We quantify the estimation errors of both approaches, and explicitly characterize the impact of privacy noise, cross-client information aggregation, and high dimensionality. Our analysis clarifies when federated learning improves over purely local estimation and how privacy protection affects statistical accuracy.
\end{enumerate}

The rest of this paper is organized as follows. 
Section~\ref{sec:methodology} introduces the proposed methodology, including the model setting, the single client learning procedure, and the federated learning method across multiple clients.
Section~\ref{sec:implement} discusses implementation issues, including the computational algorithm, rank selection, and the selection of tuning parameters and initialization. 
Section~\ref{sec:theoretical_analysis} presents the theoretical analysis, including error bounds for the single-client and federated estimators, as well as the consistency of rank selection procedure. 
Section~\ref{sec:Simulation} reports simulation results for both the single-client and federated settings. 
Section~\ref{sec:RealData} presents two empirical applications, including a selective federated learning analysis of electricity-economy linkages and an application to macroeconomic forecasting in the multi-task learning setting.
Additional technical details and proofs are provided in the Supplementary Material.

% \begin{itemize}[leftmargin=1.6em]
%   \item \textbf{Motivation (multi-client economic forecasting with selective privacy).} Federated economic forecasting where multiple related clients hold raw observations that cannot be pooled, and only a subset of coordinates is privacy-sensitive, requiring \emph{selective} privacy protection with cross-client information sharing.

%   \item \textbf{Related literature.} (i) high-dimensional VAR estimation; (ii) multi-task learning that leverages shared structure across heterogeneous clients under centralized access; (iii) federated and personalized federated learning under communication constraints and client heterogeneity; and (iv) differential privacy.

%   \item \textbf{Gap (time dependence + selective DP + VAR structure).} (i) Existing personalized federated learning theory is largely developed for independent data and does not directly cover time series settings; (ii) most DP analyses used in federated learning do not explicitly account for time-series dependence.

%   \item \textbf{Main contributions.} (i) a structured federated VAR model capturing commonality and heterogeneity; (ii) a privacy-preserving two-stage estimation procedure (private representation learning + local personalization); (iii) nonasymptotic theory; (iv) our framework also applies to the non-private personalized multi-task learning setting as a special case; and (iv) empirical results on economic forecasting.

%   \item \textbf{Organization.}
% \end{itemize}

%\subsection{Notation}
%The following notation are used throughout the paper. 
Throughout the paper, vectors and matrices are denoted by boldface lowercase and uppercase letters, respectively (e.g., $\bbm a$, $\bbm \xi$, $\bm A$, $\bm \Delta$). For a vector $\bbm{a}$, its Euclidean norm is denoted by $\|\bbm{a}\|_2$. For a matrix $\bm A \in \mathbb R^{m \times n}$, we denote its transpose, Frobenius norm, and nuclear norm by $\bm A^\top$, $\|\bm A\|_\F$, and $\|\bm A\|_*$, respectively. For $q \in [1,\infty)$, the $\ell_q$ norm of $\bm A$ is defined as $\|\bm A\|_q =( \sum_{i=1}^m \sum_{j=1}^n |A_{ij}|^q)^{1/q}$; the same formula for $q \in (0,1)$ defines the entrywise $\ell_q$ quasi-norm. When $q=0$, $\|\bm A\|_0=\#\{(i,j): A_{ij} \neq 0\}$ is the number of nonzero entries of $\bm A$.
Moreover, let $\mathrm{SVD}_r(\bm A) = \sum_{i=1}^r \sigma_i \bbm u_i \bbm v_i^\top$ denote the best rank-$r$ approximation of $\bm A$ obtained from its singular value decomposition, where $\sigma_i$ is the $i$-th singular value, and $\bbm u_i$ and $\bbm v_i$ are the corresponding left and right singular vectors.
For a rank-$r$ matrix $\bm A \in \mathbb R^{n_1 \times n_2}$ with singular value decomposition $\bm A = \bm U \bm\Sigma \bm V^\top$, its tangent space is defined as $\mathcal T_r(\bm A) = \bigl\{\bm \Xi = \bm U \bm R^\top + \bm L \bm V^\top: \bm R \in \mathbb R^{n_2 \times r}, \bm L \in \mathbb R^{n_1 \times r}\bigr\}$.
%Denote the tangent space of a rank-$r$ matrix $\bm A \in \mathbb R^{n_1 \times n_2}$ with singular value decomposition $\bm A = \bm U \bm\Sigma \bm V^\top$ as $\mathcal T_r(\bm A) = \bigl\{\bm \Xi = \bm U \bm R^\top + \bm L \bm V^\top: \bm R \in \mathbb R^{n_2 \times r}, \bm L \in \mathbb R^{n_1 \times r}\bigr\}$.
The projection of $\bm B \in \mathbb R^{n_1 \times n_2}$ onto the tangent space $\mathcal T_r(\bm A)$ is given by $\mathcal P_{\mathcal T_r(\bm A)}(\bm B) = \bm U \bm U^\top \bm B + \bm B \bm V \bm V^\top - \bm U \bm U^\top \bm B \bm V \bm V^\top$.
For sequences $\{x_n\}$ and $\{y_n\}$, we write $x_n\gtrsim y_n$ if there exists a constant $C>0$ such that $x_n\geq C y_n$ for all $n$, and $x_n\asymp y_n$ if both $x_n\gtrsim y_n$ and $y_n\gtrsim x_n$ hold.
The dataset in Section \ref{sec:RealData} and computer programs for the analysis are available at \url{https://github.com/CKKQ/FL-VAR}. 

\section{Methodology}\label{sec:methodology}

\subsection{Model Setting}\label{sec:Setting}

%Suppose there are $K$ clients, indexed by $k \in [K] := \{1,\ldots,K\}$. For the $d$-dimensional time series $\{\bbm y_{k,t}\}$ at client $k$, we consider a vector autoregressive (VAR) model as follows
Consider a federated learning scenario with $K$ clients, indexed by $k \in [K] := {1,\ldots,K}$. For each client $k$, we observe a $d$-dimensional time series $\{\bbm y_{k,t}\}$ of length $T_k$, which follows a vector autoregressive (VAR) model of order $p$:
\begin{equation}
  \bbm y_{k,t} = \bm A_{k,1} \bbm y_{k,t-1} + \cdots + \bm A_{k,p} \bbm y_{k,t-p} + \bbm \epsilon_{k,t},  \quad t \in [T_k], \; k \in [K],
  \label{eq:VARpk}
\end{equation}
where $\bm A_{k,j} \in \mathbb R^{d\times d}$ with $1\leq j\leq p$ are the client-specific transition matrices, $\{\bbm \epsilon_{k,t}\}$ are independent and identically distributed ($i.i.d.$) innovations with zero mean and finite covariance matrix $\bbm \Sigma_{\epsilon,k}\in \mathbb{R}^{d\times d}$, and $\{\bbm \epsilon_{k,t}\}$ are mutually independent across clients. 
While the sample size $T_k$ can vary from client to client, all clients share a common lag order $p$ for model parsimony. In practice, individual clients with a smaller true order $p_k \le p$ are accommodated by setting the corresponding higher-lag matrices $\bm A_{k,j}$ for $j > p_k$ to zero.

Let $\bbm x_{k,t} = [\bbm y_{k,t-1}^\top,\ldots,\bbm y_{k,t-p}^\top]^\top \in \mathbb R^{pd}$ be the lagged variable vector and $\bm A_k = [\bm A_{k,1},\ldots,\bm A_{k,p}] \in \mathbb R^{d \times pd}$ be the stacked transition matrix. Then model \eqref{eq:VARpk} can be written as
\begin{equation}
  \bbm y_{k,t} = \bm A_k \bbm x_{k,t} + \bbm \epsilon_{k,t}, \quad k \in [K], \; t \in [T_k].
  \label{eq:VARpk-compact}
\end{equation}
To handle the high-dimensionality and enable information sharing across clients, we assume that each client-specific coefficient matrix $\bm A_k$ admits a low-dimensional structure as follows:  
\begin{equation}
  \bm A_k = \bm{A}_0 + \bm\Delta_k, \quad k \in [K],
  \label{eq:decomposition}
\end{equation}
where $\bm A_0$ is an exactly low-rank matrix that captures the common temporal dynamics across all clients, and $\bm\Delta_k$ is a client-specific and weakly sparse deviation matrix.
This structure allows for effective information sharing across clients via $\bm A_0$ while accommodating individual client characteristics through $\bm\Delta_k$.
It is important to note that $\bm A_0$ itself is not required to be uniquely identifiable; see also \cite{xu2025multitask} and \cite{huang2025optimal}. 

In the following, we consider two approaches for estimating the client-specific VAR models introduced in \eqref{eq:VARpk-compact} under the structured decomposition \eqref{eq:decomposition}. The first relies solely on local data to estimate $\bm A_k$, serving as a baseline. The second, which is our primary focus, is a federated method that leverages data from all clients while adhering to privacy constraints.

\subsection{Single Client Learning}\label{sec:single-client-estimation}

We first consider a baseline approach where each client's VAR model, as specified in \eqref{eq:VARpk-compact} under the decomposition \eqref{eq:decomposition}, is estimated using only its locally available data. For client $k$, the single-client estimator of $\bm A_k$ is defined as $\widetilde{\bm A}_k = \widetilde{\bm A}_{0,k} + \widetilde{\bbm \Delta}_k$, where $(\widetilde{\bm A}_{0,k}, \widetilde{\bbm \Delta}_k)$ is obtained by solving the following regularized optimization problem:
\begin{equation}
	(\widetilde{\bm A}_{0,k}, \widetilde{\bbm \Delta}_k)
	\in \argmin_{\bm A_{0,k}, \bbm \Delta_k} \frac{1}{T_k} \sum_{t=1}^{T_k} 
	\left\|\bbm y_{k,t} -(\bm A_{0,k} + \bbm \Delta_k)\bbm x_{k,t}\right\|_2^2
	+ \lambda_k \|\bm A_{0,k}\|_* + \omega_k \|\bbm \Delta_k\|_1,
	\label{eq:local-estimator}
\end{equation}
subject to $\|\bbm \Delta_{k}\|_\op \leq \zeta$. The subscript $k$ in $\widetilde{\bm A}_{0,k}$ emphasizes that this estimate of the common structure is derived solely from client $k$'s own data. 
%Here, the subscript $k$ in $\widetilde{\bm A}_{0,k}$ indicates that this common component is estimated locally from client $k$'s data. $\lambda_k$, $\omega_k$ and $\zeta$ are tuning parameters: 
The tuning parameters $\lambda_k$, $\omega_k$ and $\zeta$ control distinct aspects of regularization. 
The nuclear-norm penalty $\|\bm A_{0,k}\|_*$ induces low-rankness in the shared component $\bm A_0$, while the entrywise $\ell_1$-norm penalty $\|\bbm \Delta_k\|_1$ promotes sparsity in the client-specific deviation $\bbm \Delta_k$. The operator-norm constraint $\|\bbm\Delta_k\|_{\op}\leq \zeta$ imposes a weak identifiability control on the deviation component. This constraint reduces the ambiguity between the sparse deviation and the low-rank common component by preventing the deviation term from capturing an excessive amount of the shared low-rank signal. Its role is therefore analogous to the weak-identifiability control provided by spikiness-type conditions in low-rank-plus-sparse decompositions; see, e.g., \citet{agarwal2012noisy}.

The optimization problem in \eqref{eq:local-estimator} is convex and can be solved using standard first-order proximal gradient methods. In our implementation, we employ an alternating direction method of multipliers (ADMM) to calculate $\widetilde{\bm A}_k$ for simplicity; further details are provided in Section \ref{sec:algorithm} of the Supplementary materials. 
%In particular, we use a proximal gradient (ISTA) algorithm as in Section \ref{sec:algorithm} for implementation.

As established in Proposition~\ref{thm:local-model-error-upper-bound}, the estimation error of the single-client estimator $\widetilde{\bm A}_k$ is fundamentally limited by the local sample size $T_k$. Consequently, in practical settings where $T_k$ is relatively small, this approach may yield unsatisfactory estimation and prediction accuracy. To address this limitation and improve performance for each client, we next propose a federated learning method that enables information sharing across clients with data-privacy protection. %while respecting data-privacy constraints.
%the error bounds of the single-client estimator $\widetilde{\bm A}_k$ depend on the local sample size $T_k$. %, which can be relatively small in practice. This may lead to limited accuracy of estimation and prediction when $T_k$ is relatively small in practice. To improve the estimation and prediction accuracy of each client, in the following we propose a federated learning method that pools information across clients with data-privacy protection.    

\subsection{Federated Learning via Multiple Clients}\label{sec:estimation} 

Building on the limitations of local learning, we now develop a privacy-preserving federated procedure that leverages data across all clients. Our approach employs a two-stage estimation strategy including differentially private representation learning for the shared low-rank component, and a personalized refinement step for each client.
%For client-specific VAR models at \eqref{eq:VARpk-compact} under the decomposition in \eqref{eq:decomposition}, our goal is to exploit similarity among clients to improve their estimation and prediction, while respecting data-privacy constraints and allowing for heterogeneous model parameters across clients. For this purpose, we propose a two-stage estimation strategy:
\begin{itemize}
    \item \textbf{Stage I: Differentially Private Representation Learning.} To estimate the low-rank common coefficient matrix $\bm A_0$ under privacy constraints, we propose a differentially private federated procedure based on noisy aggregated gradients; see Algorithm~\ref{algorithmTL}. 
    %We run the differentially private federated procedure in Algorithm~\ref{algorithmTL} to learn the low-rank common coefficient matrix $\widehat{\bm A}_0$ based on noisy aggregated gradients.
    \item \textbf{Stage II: Personalized Refinement.} %To obtain entrywise sparse deviations $\widehat{\bm \Delta}_k$'s and client-specific coefficient matrices $\widehat{\bm A}_k$'s, we solve the regularized local optimization problems at each client based on the estimated matrix $\widehat{\bm A}_0$; see Algorithm~\ref{algorithmTL-stage2}.
    Given the estimated common matrix $\widehat{\bm A}_0$, each client then solves a regularized local optimization problem to obtain an entrywise sparse deviation $\widehat{\bm \Delta}_k$ and, consequently, its client-specific coefficient matrix $\widehat{\bm A}_k$; see Algorithm~\ref{algorithmTL-stage2}.
\end{itemize}

%Next, we introduce the two-stage estimation method in detail. Stage~I estimates the common coefficient matrix $\bm A_0$ by starting from the pooled constrained least-squares loss
%\[
%  \min \limits_{\mathrm{rank}(\bm A_0) \leq r}
%  \sum_{k=1}^K \frac{T_k}{T}\sum_{t=1}^{T_k}\frac{1}{T_k}
%  \bigl\|\bbm y_{k,t} - \bm A_0 \bbm x_{k,t}\bigr\|_2^2,
%\]
%where $T = \sum_{k=1}^K T_k$. The weight $T_k/T$, proportional to the sample size of client $k$, ensures that the objective function corresponds to the empirical risk over the pooled dataset. Alternative weighting schemes are also possible without affecting the core methodology.
%Here, other weighting schemes can also be adopted without affecting the subsequent development. Note that the above objective function does not explicitly encode differential privacy. 
We now describe each stage in detail. Stage~I aims to estimate the common low-rank matrix $\bm A_0$ using a privacy-preserving gradient-based procedure by applying the Gaussian mechanism to the projected local gradients.  
%To incorporate differential privacy, Stage~I implements a privacy-preserving gradient-based procedure by applying the Gaussian mechanism to the projected local gradients. 
To isolate this common structure $\bm A_0$, we treat the client-specific deviations $\bbm\Delta_k$ in decomposition \eqref{eq:decomposition} as nuisance parameters. This allows us to rewrite the original model \eqref{eq:VARpk-compact} as $\bbm y_{k,t} = \bm A_0 \bbm x_{k,t} + \bbm e_{k,t}$ with $\bbm e_{k,t}= \bbm\Delta_k\bbm x_{k,t}+\bbm \epsilon_{k,t}$, where the composite error term $\bbm e_{k,t}$ absorbs both the client-specific deviation and the innovation. Based on the reformulated model, for each client $k$ we define the local least-squares loss $\ell_k(\bm A)=T_k^{-1}\sum_{t=1}^{T_k}\|\bbm y_{k,t} - \bm A \bbm x_{k,t}\|_2^2$ and its gradient $\bm G_k(\bm A)=2T_k^{-1}\sum_{t=1}^{T_k}(\bm A \bbm x_{k,t} - \bbm y_{k,t})\bbm x_{k,t}^\top$. Denote by $\bm A_0^{(n)}$ the estimate of $\bm A_0$ at iteration $n\geq 0$, and let $\xi_k^{(n)}$ be the sensitivity of the local gradient $\bm G_k(\bm A_0^{(n)})$. 

At each iteration $n$, client $k$ computes its local gradient $\bm G_k(\bm A_0^{(n)})$ and adds an $i.i.d.$ Gaussian noise matrix $\bm W_k^{(n)}$ to ensure differential privacy. The noisy gradient is then projected onto the tangent space $\mathcal T_r(\bm A_0^{(n)})$ as follows:
%Consider the update of $\bm A_0^{(n)}$ at the $n$-th iteration for illustration. Let $\xi_k^{(n)}$ be the sensitivity of the local gradient $\bm G_k(\bm A_0^{(n)})$ for the $k$-th client.  
%Denote $\mathcal D_k=\{\bm y_{k,t}\}_{t=1}^{T_k}$, and let $\mathcal D_k\sim \mathcal D_k'$ be the two local datasets differ in a single time index.
%For differential privacy, we add an $i.i.d.$ Gaussian noise matrix $\bm W_k^{(n)}$ to the local gradient $\bm G_k(\bm A_0^{(n)})$, and project the noisy gradient onto the tangent space $\mathcal T_r(\bm A_0^{(n)})$, 
%with the noise variance $\sigma_n^2$ calibrated to ensure $(\epsilon,\delta)$-differential privacy. The server then aggregates these privatized gradients and performs the noisy projected update in \eqref{eq:stage2-update}.
\begin{equation}
  \bm Z_k^{(n)}
  = \mathcal P_{\mathcal T_r(\bm A_0^{(n)})}\left(\bm G_k(\bm A_0^{(n)}) + \bm W_k^{(n)}\right),
  \qquad
  [\bm W_k^{(n)}]_{ij} \overset{i.i.d.}{\sim} \mathcal{N}(0,(\sigma_k^{(n)})^2),
  \label{eq:stage2-priv-grad}
\end{equation}
where the noise variance $(\sigma_k^{(n)})^2$ is calibrated according to the desired privacy parameters $(\epsilon,\delta)$ and the sensitivity $\xi_k^{(n)}$ of the local gradient. The rank $r$ of the common matrix $\bm A_0$ is a key structural parameter, and its selection is discussed in Section \ref{sec:rank-selection}. 
The server collects the privatized gradients $\bm Z_k^{(n)}$ from all clients, and performs a weighted aggregation using sample-size weights $T_k/T$, where $T = \sum_{k=1}^K T_k$ is the total sample size. A gradient descent step is then taken, followed by a projection onto the rank-$r$ manifold via singular value decomposition:
\begin{equation}
  \bm A_0^{(n+1)}
  = \mathrm{SVD}_{r}\left(\bm A_0^{(n)} - \rho \sum_{k =1}^K \frac{T_k}{T}\bm Z_k^{(n)}\right),
  \label{eq:stage2-update}
\end{equation}
with step size $\rho>0$. The operator $\mathrm{SVD}_{r}(\bm A)$ guarantees that the updated iterate remains low-rank, leading to the best rank-$r$ approximation of $\bm A$ and enabling efficient Riemannian-style updates. The weight $T_k/T$, proportional to the sample size of client $k$, ensures that clients with more data contribute proportionally more to the global update. 
We stop updating $\bm A_0$ after $N_g$ iterations, and set $\widehat{\bm A}_0 = \bm A_0^{(N_g)}$ as the estimated common matrix of $\bm A_0$. The complete description of this stage is detailed in Algorithm~\ref{algorithmTL}.

%{\color{blue}Recall that $\bm A_k=\bm A_0+\bm\Delta_k$, where $\bm\Delta_k$ captures client-specific heterogeneity and is assumed to be sparse. Stage~I focus on learning a shared low-rank representation $\bm A_0$ by treating $\{\bm\Delta_k\}_{k=1}^K$ as nuisance deviations: although the client-specific gradients are computed from data generated under $\bm A_k$, their aggregation provides an informative descent direction for the common component when the deviations are not too large in aggregate. The impact of $\bm\Delta_k$ is not ignored; rather, it is explicitly quantified in our theory through a heterogeneity term (see Theorem~\ref{thm:common-A0-error-bounds}).}

Given the estimated common low-rank matrix $\widehat{\bm A}_0$ from Stage I, Stage~II aims to recover client-specific deviations $\bm\Delta_k$ and the corresponding personalized coefficient matrices $\bm A_k$.
For each client $k\in[K]$, we define the personalized estimator of $\bbm\Delta_k$ as the minimizer of a penalized local empirical loss as follows:
\begin{equation} \label{eq:stage2-opt}
  \widehat{\bbm\Delta}_k^{\text{opt}}
  \in \argmin_{\bm\Delta_k \in \mathbb R^{d \times pd}}
  \left\{\frac{1}{T_k} \sum_{t=1}^{T_k}\left\|\bbm{y}_{k,t} - (\widehat{\bm A}_0 + \bbm \Delta_k)\bbm{x}_{k,t}\right\|_2^2
  + \varpi_k \|\bbm\Delta_k\|_1\right\}, \quad k \in [K],
\end{equation}
where $\varpi_k>0$ is a client-specific regularization parameter. The $\ell_1$-norm penalty promotes entrywise sparsity in $\bm\Delta_k$, aligning with the assumption that each client's deviation from the common structure is weakly sparse. Note that the optimization problem in \eqref{eq:stage2-opt} is convex and amenable to efficient first-order methods. We adopt a FISTA-based algorithm to solve it and denote the algorithmic solution by $\widehat{\bbm \Delta}_k$; see Algorithm~\ref{algorithmTL-stage2} for details. 
Then the personalized estimator of matrix $\bm A_k$ is given by 
\begin{align*}
  \widehat{\bm A}_k = \widehat{\bm A}_0 + \widehat{\bbm\Delta}_k, \quad k\in[K].
\end{align*}

\section{Implementary Issues}\label{sec:implement}

\subsection{Algorithm}\label{sec:algorithm}

This section presents two algorithms for the two-stage privacy-preserving federated method. 
Practical considerations, including initialization strategies and the selection of rank $r$ and regularization parameters, are discussed in Sections~\ref{sec:rank-selection} and~\ref{sec:tunning_param}.

Algorithm~\ref{algorithmTL} performs differentially private representation learning to estimate the shared low-rank matrix $\bm A_0$ by aggregating noisy projected gradients from all clients.

\begin{breakablealgorithm}
\caption{Differentially Private Representation Learning at Stage I}\label{algorithmTL}
\begin{algorithmic}[1]
  \State \textbf{Input:} Local data $\{\bbm y_{k,t}\}_{t=1}^{T_k}$ for $k\in[K]$, number of global iterations $N_g$, step size $\rho$, rank $r$, standard deviation of the Gaussian noise $\{\sigma_k^{(n)}\}$, and initial global parameter $\bm A_0^{(0)}$.
  \State \textbf{Preprocessing:} For each client $k\in[K]$, construct $\bbm x_{k,t}$
  and compute $T=\sum_{k=1}^K T_k$.
  \For{$n=0,\ldots,N_g-1$}
    \For{each client $k\in[K]$ in parallel}
      \State Compute the local gradient $\bm G_k(\bm A_0^{(n)})$.
      \State Form the privatized message $\bm Z_k^{(n)}$ according to \eqref{eq:stage2-priv-grad} and send $\bm Z_k^{(n)}$ to the server.
    \EndFor
    \State \textbf{Server update:} Update $\bm A_0^{(n+1)}$ according to \eqref{eq:stage2-update}.
  \EndFor
  \State Set $\widehat{\bm A}_0=\bm A_0^{(N_g)}$.
  \State \textbf{Output:} $\widehat{\bm A}_0$.
\end{algorithmic}
\end{breakablealgorithm}

Given the estimated common structure $\widehat{\bm A}_0$ from Stage I, Algorithm~\ref{algorithmTL-stage2} refines the estimate locally for each client via a FISTA-based algorithm, yielding sparse deviations $\widehat{\bbm\Delta}_k$ and the final client-specific matrices $\widehat{\bm A}_k$. For each client $k\in[K]$, define the empirical loss function conditioned on $\widehat{\bm A}_0$ as $\mathcal L_k(\bm\Delta_k) =T_k^{-1}\sum_{t=1}^{T_k}\|\bbm y_{k,t}-(\widehat{\bm A}_0+\bm\Delta_k)\bbm x_{k,t}\|_2^2$, and let $\nabla\mathcal L_k(\bm\Delta_k)$ denote its gradient with respect to $\bm\Delta_k$.
Following the standard FISTA algorithm \citep{beck2009fast}, each iteration of the optimization proceeds in two steps. First, we perform a proximal gradient update. Starting from the extrapolated point $\widetilde{\bbm\Delta}_k^{(n)}$, we take a gradient step with respect to $\mathcal L_k(\bm\Delta_k)$ and apply the soft-thresholding operator to enforce entrywise sparsity:
\begin{equation}
  \bm\Delta_k^{(n+1)} = \mathcal S_{\eta\varpi_k}\left(\widetilde{\bbm\Delta}_k^{(n)}-\eta\,\nabla\mathcal L_k(\widetilde{\bbm\Delta}_k^{(n)})\right),
  \label{eq:stage2-prox-update}
\end{equation}
where $[\mathcal S_{\tau}(\mathbf Z)]_{ij} = \operatorname{sign}(Z_{ij})\max(|Z_{ij}|-\tau,0)$. Second, we update the momentum parameter $q_{n+1}$ and the extrapolated iterate $\widetilde{\bbm\Delta}_k^{(n+1)}$ according to the standard FISTA extrapolation rule:
\begin{equation}
  q_{n+1}=\frac{1+\sqrt{1+4q_n^2}}{2}
  \quad\text{and}\quad
  \widetilde{\bbm\Delta}_k^{(n+1)} = \bm\Delta_k^{(n+1)} + \frac{q_n-1}{q_{n+1}}\left(\bm\Delta_k^{(n+1)}-\bm\Delta_k^{(n)}\right).
  \label{eq:stage2-extrapolation}
\end{equation}
These extrapolation steps accelerate convergence while maintaining the simplicity of first-order methods; see \citet{beck2009fast} for further discussions on the momentum parameter and extrapolated iterate in FISTA. The full procedure is detailed in Algorithm~\ref{algorithmTL-stage2}.

\begin{breakablealgorithm}
\caption{Personalized Refinement at Stage II}\label{algorithmTL-stage2}
	\begin{algorithmic}[1]
		\State \textbf{Input:} Global coefficient matrix $\widehat{\bm A}_0$ from Stage I, local data $\{\bbm y_{k,t}\}_{t=1}^{T_k}$ for each client $k \in [K]$, number of local iterations $N_l$, step size $\eta$, and regularization parameter $\varpi_k$.
		\For{each client $k \in [K]$ in parallel}
		  \State \textbf{Preprocessing:} For client $k$, construct $\bbm x_{k,t}$.
		  \State Initialize the local deviation by $\bbm\Delta_k^{(0)} = \bm 0$, $q_0 = 1$, and set $\widetilde{\bbm\Delta}_k^{(0)} = \bbm\Delta_k^{(0)}$.
		  \For{$n = 0,\ldots, N_l-1$}
		    \State Update $\bm\Delta_k^{(n+1)}$ according to \eqref{eq:stage2-prox-update}.
		    \State Update $q_{n+1}$ and $\widetilde{\bbm\Delta}_k^{(n+1)}$ according to \eqref{eq:stage2-extrapolation}.
		  \EndFor
      \State Set the personalized estimator for client $k$ as $\widehat{\bm A}_k = \widehat{\bm A}_0 + \widehat{\bbm\Delta}_k$ with $\widehat{\bbm\Delta}_k = \bm\Delta_k^{(N_l)}$.
		\EndFor
		\State \textbf{Output:} Personalized coefficient matrices $\{\widehat{\bm A}_k\}_{k \in [K]}$.
	\end{algorithmic}
\end{breakablealgorithm}

\subsection{Rank Selection}\label{sec:rank-selection}

The estimation of the common low-rank matrix $\bm A_0$ in Stage I depends critically on the choice of its rank $r$. We estimate $r$ using a two-step procedure based on the ridge-type ratio criterion \citep{xia2015consistently}. First, for each client $k$, we obtain a candidate rank $\widehat{r}_k$ by minimizing the ratio of successive singular values of the local estimator $\widetilde{\bm A}_{0,k}$:
\begin{align}\label{eq:Ridge-type estimator}
  \widehat{r}_k = \argmin_{1 \leq r \leq \bar{r}-1} \frac{\widetilde{\sigma}_{k,r+1} + c(d,T_k)}{\widetilde{\sigma}_{k,r} + c(d,T_k)},
\end{align}
where $\widetilde{\sigma}_{k,r}$ denotes the $r$-th singular value of $\widetilde{\bm A}_{0,k}$ defined in \eqref{eq:local-estimator}, and the term $c(d,T_k)$ is a ridge-type penalty that stabilizes the ratio.
Based on the consistency conditions for ridge-type ratio estimators established in Theorem~\ref{thm:rank-selection-consistency}, we recommend setting $c(d,T_k)=10^{-2}\sqrt{pd/T_k}$. 
Next, to achieve a robust global estimate, we aggregate these client-specific ranks by taking their mode 
$$\widehat{r} = \operatorname{mode}(\{\widehat{r}_1, \dots, \widehat{r}_K\}).$$
The mode is chosen for its robustness against potentially noisy or uninformative individual estimates. This estimated rank $\widehat{r}$ is then employed in the federated learning of Stage I.

\subsection{Tuning Parameter Selection and Initialization}\label{sec:tunning_param}

All tuning parameters in the proposed method, including $\lambda_k$, $\omega_k$ and $\zeta$ in \eqref{eq:local-estimator} for the single-client estimator and $\varpi_k$ in \eqref{eq:stage2-opt} for Stage~II of the federated learning procedure, are selected via cross-validation based on a rolling forecasting procedure. Specifically, for each client, we choose the optimal parameters by minimizing the one-step-ahead prediction error over a prespecified grid of candidate values.

For Stage~I in Algorithm~\ref{algorithmTL}, the initial value $\bm A_0^{(0)}$ is taken as the single-client estimator $\widetilde{\bm A}_{0,k^{'}}$ from the client with the largest sample size $T_{k^{'}}$, which satisfies the condition on $\bm A_0^{(0)}$ in Theorem~\ref{thm:common-A0-error-bounds} when $T_{k^{'}}$ is large enough. This choice is motivated by the error bound in Proposition~\ref{thm:local-model-error-upper-bound}, which suggests that clients with more data provide a more reliable initial estimate of the common structure $\bm A_0$.
The number of global iterations is set to $N_g=\lceil 10\log(T)\rceil$ with $T = \sum_{k=1}^K T_k$, in accordance with the logarithmic iteration requirement in Theorem~\ref{thm:common-A0-error-bounds}. 
The step size $\rho$ is selected by cross-validation, which naturally satisfies the condition on $\rho$ in Theorem~\ref{thm:common-A0-error-bounds}. 
For differential privacy, the standard deviation of the Gaussian noise is calibrated as $\sigma_k^{(n)}=\kappa\sqrt{2\log(1.25/\delta)}/\varepsilon$, where $\kappa$ is chosen adaptively based on the scale of the data, following the Gaussian mechanism (see, e.g., \citet{dwork2014algorithmic}). 

For Stage~II in Algorithm~\ref{algorithmTL-stage2}, the local deviation for each client is initialized as $\bm\Delta_k^{(0)}=\bm 0$, which aligns with the conditions required in Theorem~\ref{thm:personalized-error-bounds}. 
The number of local iterations is set to $N_l=20$, as we observe empirically that the algorithm typically converges within $15$ iterations. This choice balances computational efficiency with solution accuracy. 
The step size is set to $\eta=(2T_k^{-1}\|\sum_{t=1}^{T_k}\bbm x_{k,t}\bbm x_{k,t}^\top\|_{\op})^{-1}$ as required in Theorem~\ref{thm:personalized-error-bounds}, which corresponds to the inverse of twice the Lipschitz constant and ensures convergence of the FISTA procedure.

\section{Theoretical Analysis}\label{sec:theoretical_analysis}

We begin by formalizing the key structural assumptions underlying the decomposition in \eqref{eq:decomposition}. Throughout this section, we use the superscript $*$ to denote the true underlying parameters. 
The shared matrix $\bm A_0^*$ is assumed to be exactly low rank with $\mathrm{rank}(\bm A_0^*) = r$ for some integer $r>0$. For each client $k \in [K]$, the deviation matrix $\bbm\Delta_k^*$ is assumed to be entrywise weakly sparse and lie in an $\ell_q$-ball for some $q \in [0,1)$. More precisely, there exists a constant $s_q>0$ such that for all $k\in[K]$,
\[
\bbm\Delta_k^*\in \mathbb B_q(s_q) = \Bigl\{\bbm\Delta\in\mathbb R^{d\times pd}: \sum_{i=1}^{d}\sum_{j=1}^{pd}\bigl|\Delta_{ij}\bigr|^{q}\leq s_q\Bigr\}.
\]
This $\ell_q$-ball condition provides a convenient unified control over the magnitude of the deviations. In particular, it implies a uniform bound on their Frobenius norms: for any $q\in[0,1)$ and $\bbm\Delta_k^*\in\mathbb B_q(s_q)$, we have $\|\bm\Delta_k^*\|_{\F} \leq \phi(s_q)$ for all $k\in[K]$, where $\phi(s_q) = \max_{1\leq k\leq K}\|\bbm\Delta_k^*\|_{\infty}^{1-\frac q2}s_q^{1/2}$ is a uniform upper bound on the Frobenius norms $\|\bm\Delta_k^*\|_{\F}$.

To address the inherent non-identifiability between the low-rank structure $\bm A_0^*$ and the sparse deviations $\{\bm\Delta_k^*\}$, we impose a weak identifiability condition on $\{\bm\Delta_k^*\}$. Specifically, we assume that $\max_{k\in[K]}\|\bbm\Delta_k^*\|_{\op}\leq \zeta$, for some constants $\zeta>0$. This condition restricts the size of the client-specific deviations in directions that may overlap with the low-rank common component, in a spirit similar to the weak identifiability role of spikiness-type conditions in low-rank-plus-sparse decompositions; see, e.g., \citet{agarwal2012noisy}.

\subsection{Single Client Error Bounds}\label{sec:single-client-error-bounds}

In this section, we establish theoretical error bounds for the single-client estimator $\widetilde{\bm A}_k$ defined in \eqref{eq:local-estimator} for each client. These results serve as baselines for the proposed federated learning. We begin by introducing two key assumptions on the data-generating process.

\begin{assumption}[Stationarity]
	\label{assump:stationarity}
	For each client $k \in [K]$, the matrix-valued autoregressive polynomial $\mathcal A_k(z) = \bm I_d - \sum_{j=1}^p \bm A_{k,j}^* z^j$ satisfies $\det \mathcal A_k(z) \neq 0$ for all $|z| \leq 1$.
\end{assumption}

\begin{assumption}[Sub-Gaussian Innovations]\label{assump:subg}
	For each client $k\in[K]$, let $\bbm{\epsilon}_{k,t}=\bm{\Sigma}_{\epsilon,k}^{1/2}\bbm{\xi}_{k,t}$, where $\{\bbm{\xi}_{k,t}\}$ are $i.i.d.$ across both $t$ and $k$, with $\mathbb{E}(\bbm{\xi}_{k,t})=\bm{0}$ and $\textup{var}(\bbm{\xi}_{k,t})=\bm{I}_d$, and $\bm{\Sigma}_{\epsilon,k}=\textup{var}(\bbm{\epsilon}_{k,t})$ is a positive definite matrix. Moreover, the entries $(\xi_{k,t,i})_{1\leq i\leq d}$ of $\bbm{\xi}_{k,t}$ are mutually independent and $\sigma^2$-sub-Gaussian, i.e., $\mathbb{E}(e^{\mu\xi_{k,t,i}})\leq e^{\mu^2\sigma^2/2}$, for any $\mu\in\mathbb{R}$, $k\in[K]$, $t\in[T_k]$, and $i\in[d]$.
\end{assumption}

Assumption~\ref{assump:stationarity} ensures that the VAR process for each client is stable and stationary, a standard condition in time series analysis; see, e.g., \citet{lutkepohl2005new} for details. Assumption~\ref{assump:subg} imposes a standard sub-Gaussian tail condition on the innovation process, which is commonly adopted in high-dimensional time series analysis to facilitate concentration inequalities. 
The independence of innovations across clients simplifies the analysis, while this condition can be relaxed to allow for $\alpha$-mixing dependence across clients using techniques similar to those in \citet{merlevede2011bernstein}.

Under Assumptions~\ref{assump:stationarity} and \ref{assump:subg}, we introduce several constants that characterize the VAR process for each client $k\in[K]$. Define the spectral quantities $\mu_{\min}(\mathcal A_k) = \min_{|z|=1} \lambda_{\min}(\mathcal A_k^\dagger(z)\mathcal A_k(z))$ and $\mu_{\max}(\mathcal A_k) = \max_{|z|=1} \lambda_{\max}(\mathcal A_k^\dagger(z)\mathcal A_k(z))$, where $\dagger$ denotes the conjugate transpose. It can be shown that $\mu_{\min}(\mathcal A_k)>0$ under Assumption~\ref{assump:stationarity}; see, e.g., Proposition~2.2 in \cite{basu2015regularized}. This positivity ensures that the process has a well-defined spectral density.
Based on these quantities and the innovation covariance $\bm \Sigma_{\epsilon,k}$, we define the following positive constants
\[
C_{\epsilon,\mathcal A}^{(k)} = \frac{\lambda_{\max}(\bm \Sigma_{\epsilon,k})}{\mu_{\min}(\mathcal A_k)}, \quad
c_{\epsilon,\mathcal A}^{(k)} = \frac{\lambda_{\min}(\bm \Sigma_{\epsilon,k})}{\mu_{\max}(\mathcal A_k)}, \quad
C_{\mathrm{RSC}}^{(k)} = \frac{1}{2}c_{\epsilon,\mathcal A}^{(k)}, \quad\text{and}\quad
\kappa_{\epsilon,\mathcal A}^{(k)} = \frac{C_{\epsilon,\mathcal A}^{(k)}}{c_{\epsilon,\mathcal A}^{(k)}}.
\]
The constant $C_{\mathrm{RSC}}^{(k)}$ will play a key role in establishing the restricted strong convexity property of the loss function. Importantly, we allow $C_{\epsilon,\mathcal A}^{(k)}$ and $C_{\mathrm{RSC}}^{(k)}$ to grow with the dimension $d$.

\begin{proposition}[Single-client Learning]\label{thm:local-model-error-upper-bound}
	Suppose that Assumptions~\ref{assump:stationarity} and \ref{assump:subg} hold. If the local sample size satisfies $T_k \gtrsim (C_{\epsilon,\mathcal A}^{(k)}/C_{\mathrm{RSC}}^{(k)}\vee 1)^2p^2d \vee p\log(pd)$, and the tuning parameters are chosen as
	\[
	\lambda_k \asymp \sigma^2 C_{\epsilon,\mathcal A}^{(k)}\sqrt{\frac{pd}{T_k}} \vee C_{\mathrm{RSC}}^{(k)}\zeta
	\ \ \text{and} \ \
	\omega_k \asymp \sigma^2 C_{\epsilon,\mathcal A}^{(k)}\sqrt{\frac{\log(pd)}{T_k}},
	\quad \forall k\in[K],
	\]
	then with probability at least $1 - C\exp(-Cpd) - C\exp(-C\log(pd))$,
	the following bound holds simultaneously for all $k\in[K]$:
	\[
	\bigl\|\widetilde{\bm A}_{0,k}-\bm A_0^*\bigr\|_{\F} + \bigl\|\widetilde{\bbm\Delta}_k-\bbm\Delta_k^*\bigr\|_{\F} \lesssim \underbrace{\sqrt{r}\frac{\lambda_k}{C_{\mathrm{RSC}}^{(k)}} + \sqrt{s_q}\left(\frac{\omega_k}{C_{\mathrm{RSC}}^{(k)}}\right)^{1-q/2}}_{\mathsf{Error}_{\mathrm{loc}}^{(k)}}.
	\]
\end{proposition}

Proposition~\ref{thm:local-model-error-upper-bound} provides a non-asymptotic error bound for the local estimator $(\widetilde{\bm A}_{0,k}, \widetilde{\bbm \Delta}_k)$. Although the displayed bound contains two grouped terms, the choice $\lambda_k\asymp \sigma^2 C_{\epsilon,\mathcal A}^{(k)}\sqrt{pd/T_k}\vee C_{\mathrm{RSC}}^{(k)}\zeta$ shows that it can be interpreted as consisting of three components. The first component, $\sqrt r\,\sigma^2 C_{\epsilon,\mathcal A}^{(k)}C_{\mathrm{RSC}}^{(k)-1}\sqrt{pd/T_k}$, is the usual stochastic error for recovering the low-rank component under nuclear-norm regularization. The second component, $\sqrt r\,\zeta$, comes from the weak-identifiability control imposed by the operator-norm bound $\|\bbm\Delta_k^*\|_{\op}\leq \zeta$; it accounts for the residual ambiguity between the low-rank component and the client-specific deviation. The third component, $\sqrt{s_q}(\omega_k/C_{\mathrm{RSC}}^{(k)})^{1-q/2}$, corresponds to the estimation error for the entrywise weakly sparse deviation $\bbm\Delta_k^*$; in the special case of strict sparsity with $q=0$, it reduces to the standard upper bound for $\ell_1$-regularized estimation. Overall, the bound depends critically on the local sample size $T_k$. When $T_k$ is small, the bound becomes large, indicating that reliable estimation of both $\bm A_0^*$ and $\bbm\Delta_k^*$ is difficult using only local data. This limitation is also corroborated by the simulation results in Section~\ref{sec:simu-single-client-est}, motivating the need for federated information sharing.

Since the optimization problem in \eqref{eq:local-estimator} is convex, standard convergence results for the ADMM guarantee that the optimization error can be made negligible relative to the statistical error by running a sufficient number of iterations. For brevity, we therefore omit a detailed analysis of the optimization error here.

\subsection{Federated Error Bounds}\label{sec:fed-error-bounds}

We begin this section by formalizing privacy protection in our federated time-series setting. In many practical applications, only certain coordinates of a multivariate time series contain sensitive information, while others may be non-sensitive. To accommodate this, we introduce a sensitive index set $\mathcal I \subset [d]$, indicating that only the coordinates of $\bbm y_{k,t}$ indexed by $\mathcal I$ require protection. 
In our parametric VAR models, variations in the series $\bbm y_{k,t}$ are driven by the underlying innovations $\bbm\epsilon_{k,t}$. Consequently, we define privacy at the innovation level, which provides a natural and tractable framework for theoretical analysis.

Before stating the formal definition of differential privacy, we first specify the notion of neighboring datasets under this selective protection requirement. 
For each client $k\in[K]$, let $\mathcal D_k=\{\{\bbm y_{k,t}\}_{t=1-p}^{0},\{\bbm\epsilon_{k,t}\}_{t=1}^{T_k}\}$ denote its local dataset, comprising both the initial observations and the innovation sequence. We say that two datasets $\mathcal D_k$ and $\mathcal D_k'$ are neighboring with respect to the sensitive index set $\mathcal I$, denoted $\mathcal D_k \overset{\mathcal I}{\sim} \mathcal D_k'$, if they differ only in the sensitive coordinates of the innovation at a single time point.
More precisely, this requires the existence of some $t_0\in[T_k]$ such that: (i) $\bbm\epsilon_{k,t_0,\mathcal I}\neq \bbm\epsilon'_{k,t_0,\mathcal I}$; (ii) $\bbm\epsilon_{k,t,\mathcal I}=\bbm\epsilon'_{k,t,\mathcal I}$ for all $t\neq t_0$; and (iii) $\bbm\epsilon_{k,t,\mathcal I^{c}}=\bbm\epsilon'_{k,t,\mathcal I^{c}}$ for all $t\in[T_k]$. Given this concept, we introduce the following definitions.

\begin{definition}[Selective federated $(\varepsilon,\delta)$-DP]\label{def:federated-selective-dp}
  Fix a sensitive index set $\mathcal I\subset[d]$ and consider $K$ clients, where client $k$ holds a local dataset $\mathcal D_k$.
  Let $\mathcal O_k$ denote the output space of client $k$.
  A collection of randomized local mechanisms $\{M_k:\mathcal D_k \to\mathcal O_k\}_{k=1}^K$ is said to satisfy selective federated $(\varepsilon,\delta)$-differential privacy on $\mathcal I$ if, for every client $k\in[K]$, all neighboring datasets $\mathcal D_k\overset{\mathcal I}{\sim}\mathcal D_k'$, and all measurable sets $S_k\subseteq\mathcal O_k$,
  \[
    \mathbb P\bigl(M_k(\mathcal D_k)\in S_k\bigr)
    \leq
    \mathrm e^{\varepsilon}\mathbb P\bigl(M_k(\mathcal D_k')\in S_k\bigr)+\delta.
  \]
\end{definition}

The parameters $(\varepsilon,\delta)$ quantify the strength of the privacy guarantee (see, e.g., \citet{dwork2006calibrating,dwork2014algorithmic}).
The guarantee is conditional on any fixed pair of neighboring datasets $\mathcal D_k\overset{\mathcal I}{\sim}\mathcal D_k'$; hence the probability statement in Definition~\ref{def:federated-selective-dp} is independent of the sampling distribution of the data and only reflects the randomness added by the mechanism.
A smaller $\varepsilon$ corresponds to stronger privacy, as it forces the output distributions under neighboring datasets to be closer.
The parameter $\delta$ allows a small probability of failure of the pure DP guarantee, and setting $\delta=0$ recovers the standard $\varepsilon$-DP.
In general, smaller values of $(\varepsilon,\delta)$ provide stronger privacy protection for the local mechanism. %$M_k$.

\begin{definition}[Sensitivity]\label{def:grad-sensitivity}
Fix a sensitive index set $\mathcal I\subset[d]$. For client $k\in[K]$ and iteration $n$, let $\bm G_k(\bm A_0^{(n)};\mathcal D_k)$ denote the local gradient evaluated at $\bm A_0^{(n)}$ using dataset $\mathcal D_k$. 
We define the sensitivity $\xi_k^{(n)}$ as any deterministic upper bound of $\sup_{\mathcal D_k \overset{\mathcal I}{\sim} \mathcal D_k'}\|\bm G_k(\bm A_0^{(n)};\mathcal D_k) -\bm G_k(\bm A_0^{(n)};\mathcal D_k')\|_{\F}$.
\end{definition}

The sensitivity $\xi_k^{(n)}$ quantifies the maximum possible change in the gradient message transmitted from client $k$ to the server when the underlying dataset is replaced by a neighboring one. 
This quantity plays a crucial role in calibrating the Gaussian mechanism: a larger sensitivity $\xi_k^{(n)}$ requires injecting more noise into the transmitted gradient to achieve the privacy guarantee, which in turn increases the privacy-utility trade-off in Stage~I.

To state the federated error bound concisely, we first introduce the following global constants
\[
\Lambda_{\epsilon,\max} = \max_{k\in[K]}\lambda_{\max}(\bbm\Sigma_{\epsilon,k}),\quad
c_{\epsilon,\mathcal A}^{\min} = \min_{k\in[K]}c_{\epsilon,\mathcal A}^{(k)},\quad
C_{\epsilon,\mathcal A}^{\max} = \max_{k\in[K]}C_{\epsilon,\mathcal A}^{(k)},\quad \text{and} \quad
\kappa_{\epsilon,\mathcal A} = \frac{C_{\epsilon,\mathcal A}^{\max}}{c_{\epsilon,\mathcal A}^{\min}}.
\]

\begin{theorem}[Private Representation Learning]\label{thm:common-A0-error-bounds}
Assume that Assumptions~\ref{assump:stationarity} and \ref{assump:subg} hold. Suppose Algorithm~\ref{algorithmTL} satisfies: the initialization $\|\bm A_0^{(0)}-\bm A_0^*\|_{\F}\leq c_{\epsilon,\mathcal A}^{\min}(3C_{\epsilon,\mathcal A}^{\max}+c_{\epsilon,\mathcal A}^{\min})/20(6C_{\epsilon,\mathcal A}^{\max}+c_{\epsilon,\mathcal A}^{\min})^2\sigma_r(\bm A_0^*)$, the iteration number $N_g\asymp \allowbreak \log(T)$, the step size $\rho = 2/(c_{\epsilon,\mathcal A}^{\min}+3C_{\epsilon,\mathcal A}^{\max})$, the non-identifiability-induced deviation satisfies $K C_{\epsilon,\mathcal A}^{\max}\zeta \leq c_0 \sigma_r(\bm A_0^*)$ for some small $c_0>0$, and the standard deviation of Gaussian noise $\sigma_k^{(n)} \asymp \xi_k^{(n)}\log T \allowbreak \log(1.25\log T/\delta)/\varepsilon$ with sensitivities $\xi_k^{(n)} \asymp \sqrt{C_{\epsilon,\mathcal A}^{\max}r}(\sigma_r(\bm A_0^*)+\zeta)\sigma^2\Lambda_{\epsilon,\max}p \allowbreak \sqrt{pd+\log T}\sqrt{|\mathcal I| +\log T}/T_k$ for $n\in[N_g]$. Then Algorithm~\ref{algorithmTL} satisfies selective federated privacy $(\varepsilon,\delta)$-DP on the sensitive index set $\mathcal I$ with probability at least $1-\exp(-C \log T)$. Moreover, if the total sample size satisfies $T\gtrsim \{(C_{\epsilon,\mathcal A}^{\max}\sigma^2)^2 \vee \sigma^2(\sigma_r(\bm A_0^*)+\zeta)p\Lambda_{\epsilon,\max}/(\varepsilon\sqrt{C_{\epsilon,\mathcal A}^{\max}})\sqrt{K|\mathcal I|\log(1.25/\delta)}\}rpd$, then with probability at least $1 - \exp(-Cpd) - K\exp(-C\log T)$, the following error bound holds:
\begin{align*}
\|\widehat{\bm A}_0-\bm A_0^*\|_{\F}
\lesssim \kappa_{\epsilon,\mathcal A}\left(
\mathsf{Error}_{\mathrm{stat}}
+\mathsf{Error}_{\mathrm{h}}
+\mathsf{Error}_{\mathrm{DP}}\right),
\end{align*}
where $\mathsf{Error}_{\mathrm{stat}} = \sigma^2\sqrt{rpd/T}$, $\mathsf{Error}_{\mathrm{h}} = \sqrt{r}K\zeta$, and
\begin{align*}
\mathsf{Error}_{\mathrm{DP}} = \frac{\sigma^2r\bigl(\sigma_r(\bm A_0^*)+\zeta\bigr)p\Lambda_{\epsilon,\max}}{\sqrt{C_{\epsilon,\mathcal A}^{\max}}}\frac{\log T}{\varepsilon}\sqrt{\log\Bigl(\frac{1.25\log T}{\delta}\Bigr)} \frac{\sqrt{K(|\mathcal I|+\log T)}(pd+\log T)}{T}.
\end{align*}
\end{theorem}

Theorem~\ref{thm:common-A0-error-bounds} gives a high-probability sensitivity-calibrated selective federated privacy guarantee and a non-asymptotic error bound for $\widehat{\bm A}_0$. The high-probability statement only concerns whether the realized, given dataset generated from the VAR process satisfies the prescribed sensitivity upper bound. That is, the probability is taken over the data-generating process to control the event that the data-dependent sensitivity $\xi_k^{(n)}$ may exceed its deterministic upper bound; it is not a failure probability of the randomized Gaussian mechanism. Conditional on any fixed neighboring datasets satisfying this sensitivity bound, the selective federated $(\varepsilon,\delta)$-DP inequality is a mechanism-level statement whose probabilities are taken only over the added Gaussian noise. The estimation error decomposes into three components:
\begin{enumerate}
	\item[(i)] Statistical estimation error ($\mathsf{Error}_{\mathrm{stat}}$): This component corresponds to the standard error rate for estimating a rank-$r$ matrix from pooled data. 
	\item[(ii)] Client heterogeneity error ($\mathsf{Error}_{\mathrm{h}}$): This term reflects the intrinsic ambiguity in separating the shared low-rank component $\bm A_0^*$ from the client-specific deviations $\{\bbm\Delta_k^*\}_{k=1}^K$. $\|\bbm\Delta_k^*\|_{\op}\leq\zeta$ provides a weak-identifiability control on this ambiguity, leading to the bound $\mathsf{Error}_{\mathrm{h}}=\sqrt r K\zeta$. This $K$-dependence is a conservative bound for the common-component estimation step and is refined after the personalized Stage~II step, where the corresponding heterogeneity term no longer carries the factor $K$ under the per-client sample size condition in Theorem~\ref{thm:personalized-error-bounds}.
	\item[(iii)] DP Gaussian-noise error ($\mathsf{Error}_{\mathrm{DP}}$): This term quantifies the additional error incurred by injecting Gaussian noise to enforce selective federated $(\varepsilon,\delta)$-DP on the sensitive index set $\mathcal I$. As expected, it grows with stronger privacy (smaller $\varepsilon$ and/or $\delta$) and with the number of protected coordinates $|\mathcal I|$, capturing the fundamental privacy--utility trade-off.
\end{enumerate}

\begin{theorem}[Personalized Refinement]\label{thm:personalized-error-bounds}
Assume that the conditions of Theorem~\ref{thm:common-A0-error-bounds} hold.
Suppose Algorithm~\ref{algorithmTL-stage2} is initialized with $\bm\Delta_k^{(0)}=\bm 0$ for each client $k$. 
Let the step size be $\eta=(2T_k^{-1}\bigl\|\sum_{t=1}^{T_k}\bbm x_{k,t}\bbm x_{k,t}^\top\bigr\|_{\op})^{-1}$, and choose the number of local iterations as $N_l \asymp \sqrt{\kappa_{\epsilon,\mathcal A}^{(k)}}\phi(s_q)/\mathsf{Error}_{\Delta}^{(k)}$ with $\mathsf{Error}_{\Delta}^{(k)} = \sqrt{\kappa_{\epsilon,\mathcal A}^{(k)}}\|\widehat{\bm A}_0-\bm A_0^*\|_{\F} + \sqrt{s_q}(\varpi_k/C_{\mathrm{RSC}}^{(k)})^{1-q/2}$.
If the per-client sample size satisfies $T_k \gtrsim \bigl(C_{\epsilon,\mathcal A}^{(k)}/C_{\mathrm{RSC}}^{(k)}\vee 1\bigr)^2 p^2 d \vee p\log(pd)$, and the regularization parameter is chosen as $\varpi_k \asymp \sigma^2 C_{\epsilon,\mathcal A}^{(k)}\sqrt{\log(pd)/T_k}$,
then, with probability at least $1 - C\exp(-Cpd) - C\exp(-C\log(pd))$, we have $\|\widehat{\bbm\Delta}_k-\bbm\Delta_k^*\|_{\F} \lesssim \mathsf{Error}_{\Delta}^{(k)}$ for all $k\in[K]$.
Moreover, with probability at least $1-K\exp(-Cpd)-K\exp(-C\log T)$, we have
\[
	\|\widehat{\bm A}_0-\bm A_0^*\|_{\F} + \|\widehat{\bm\Delta}_k-\bm\Delta_k^*\|_{\F}
	\lesssim
	\sqrt{\kappa_{\epsilon,\mathcal A}^{(k)}}\kappa_{\epsilon,\mathcal A}\left(\mathsf{Error}_{\mathrm{stat}} + \mathsf{Error}_{\mathrm{DP}} + \underbrace{C_{\epsilon,\mathcal A}^{\max}\sqrt{r}\zeta}_{\mathsf{Error}_{\mathrm{h}}'}\right)
	+ \underbrace{\sqrt{s_q}\left(\frac{\varpi_k}{C_{\mathrm{RSC}}^{(k)}}\right)^{1-q/2}}_{\mathsf{Error}_{\mathrm{p}}^{(k)}}.
\]
\end{theorem}

Theorem~\ref{thm:personalized-error-bounds} provides a non-asymptotic error bound for the personalized estimator $\widehat{\bm A}_k=\widehat{\bm A}_0+\widehat{\bbm\Delta}_k$.
The terms $\mathsf{Error}_{\mathrm{stat}}$ and $\mathsf{Error}_{\mathrm{DP}}$ are inherited from Stage~I and respectively represent the pooled statistical error and the additional Gaussian-noise error for estimating the common component $\bm A_0^*$. The refined heterogeneity term $\mathsf{Error}_{\mathrm h}'$ comes from the weak identifiability of the decomposition $\bm A_k^*=\bm A_0^*+\bbm\Delta_k^*$; compared with $\mathsf{Error}_{\mathrm h}$ in Theorem~\ref{thm:common-A0-error-bounds}, it no longer carries the factor $K$. The personalization term $\mathsf{Error}_{\mathrm p}^{(k)}$ is the local sparse-deviation estimation error in Stage~II. The number of local FISTA iterations $N_l$ is chosen so that the optimization error is dominated by the corresponding statistical error.

We next compare the personalized federated bound with the single-client bound in Proposition~\ref{thm:local-model-error-upper-bound} for a fixed client $k\in[K]$. Suppose the relevant condition numbers are uniformly bounded, i.e., $\kappa_{\epsilon,\mathcal A}^{(k)}\leq \kappa_{\epsilon,\mathcal A}\leq C$ for some constant $C>0$. Under this simplification, the personalization term $\mathsf{Error}_{\mathrm p}^{(k)}$ has the same order as the sparse-deviation part of the single-client bound, while the low-rank component benefits from the pooled sample size $T=\sum_{k=1}^K T_k$ through $\mathsf{Error}_{\mathrm{stat}}=\sigma^2\sqrt{rpd/T}$.
Thus, the federated estimator improves over the single-client estimator when the gain from pooling dominates the additional costs from privacy noise. A sufficient condition is that $\mathsf{Error}_{\mathrm{DP}} \lesssim \sigma^2\sqrt{rpd/T_k}$. In particular, it is enough to require the pooled sample size $T$ satisfies
\begin{align*}
T
\gtrsim
\frac{\bigl(\sigma_r(\bm A_0^*)+\zeta\bigr)p\Lambda_{\epsilon,\max}}
{\varepsilon\sqrt{C_{\epsilon,\mathcal A}^{\max}}}
\sqrt{rT_kpdK|\mathcal I|\log(1.25/\delta)}.
\end{align*}

\subsection{Rank Selection Consistency}\label{sec:rank-select-con}
This section establishes the consistency of the rank selection criterion proposed in Section~\ref{sec:rank-selection}.

\begin{theorem}[Rank selection consistency]\label{thm:rank-selection-consistency}
Let $r^*$ be the true rank of $\bm A_0^*$, with $r^* \leq \overline r$. Under the conditions of Proposition~\ref{thm:local-model-error-upper-bound}, if for each client $k\in[K]$, the ridge-type penalty $c(d,T_k)$ satisfies
\[
  \sqrt{\overline{r}}\frac{\lambda_k}{C_{\mathrm{RSC}}^{(k)}}
  + \sqrt{s_q}\left(\frac{\omega_k}{C_{\mathrm{RSC}}^{(k)}}\right)^{1-q/2}
  \ll c(d,T_k)
  \ll \sigma_{r^*}(\bm A_0^*) \min_{1\leq j\leq r^*-1}\frac{\sigma_{j+1}(\bm A_0^*)}{\sigma_{j}(\bm A_0^*)},
\]
then $\mathbb{P}(\widehat{r}_k=r^*)\to 1$ as $d,T_k\to\infty$ for each $k\in[K]$.
\end{theorem}

Theorem~\ref{thm:rank-selection-consistency} establishes that the proposed ridge-type ratio criterion consistently recovers the true rank $r^*$ of the common low-rank matrix $\bm A_0^*$. The two-sided condition on $c(d,T_k)$ in Theorem~\ref{thm:rank-selection-consistency} has a clear interpretation.
The lower bound enforces that the stochastic perturbation of the local singular values (as controlled by Proposition~\ref{thm:local-model-error-upper-bound}) is dominated by $c(d,T_k)$, so that the ridge-type ratio $(\widetilde{\sigma}_{k,r+1}+c(d,T_k))/(\widetilde{\sigma}_{k,r}+c(d,T_k))$ is not driven by noise-induced tail singular values.
The upper bound requires $c(d,T_k)$ to remain sufficiently small relative to the population singular-value separation among the leading $r^*$ components of $\bm A_0^*$.
Intuitively, this condition may be violated if the smallest nonzero singular value $\sigma_{r^*}(\bm A_0^*)$ is too small, or if there is an abrupt drop in magnitude between two adjacent singular values.
In either case, the ridge-type ratio becomes less informative for identifying the true rank $r^*$.

\section{Simulation}\label{sec:Simulation}

This section conducts six simulation experiments to illustrate the finite sample performance of the proposed methods. The first two experiments in Section~\ref{sec:simu-single-client-est}, evaluate the proposed single-client estimator and the consistency of the associated rank selection procedure. The remaining four experiments in Section~\ref{sec:simu-fed-est} assess the performance of the federated method.
For each setting, the experiment is repeated over 1{,}000 replications. 
All tuning parameters are chosen following the guidelines in Section~\ref{sec:tunning_param} unless otherwise specified. 
%In all figures, the shaded bands represent the 5th to 95th percentile range across these replications. We state this convention here and omit repeating it thereafter.

The data are generated from the VAR model in~\eqref{eq:VARpk} with dimension $d = 50$ and lag order $p = 1$. The innovations $\bbm\epsilon_{k,t}$ are drawn as $d$-dimensional standard Gaussian random vectors, independent across both time $t$ and clients $k$. Unless otherwise specified, we set the sample size per client to $T_k=400$ for all $k\in[K]$, the privacy budget to $(\varepsilon,\delta)=(2,0.1)$, and the number of clients to $K=5$. 
The client-specific transition matrices are constructed as $\bm A_k=\bm A_0+\bbm\Delta_k$. We generate the common component $\bm A_0$ by drawing its entries independently from a standard Gaussian distribution and then enforcing $\operatorname{rank}(\bm A_0)=2$ unless specified otherwise. Each deviation matrix $\bbm\Delta_k$ is initially drawn from a standard Gaussian distribution, then normalized to satisfy the weak sparsity constraint $\|\bbm\Delta_k\|_{0.1}^{0.1}\leq 10$. We then rescale each $\bbm\Delta_k$ so that the Frobenius-norm ratio satisfies $\|\bm A_0\|_{\F}:\|\bbm\Delta_k\|_{\F}=5:1$. Finally, we rescale $\bm A_k$ to ensure that the resulting VAR processes are stationary, as required by Assumption~\ref{assump:stationarity}.

%Considering the privacy analysis in Theorem~\ref{thm:common-A0-error-bounds} characterizes the order of the Gaussian noise level via the sensitivities $\xi_k^{(n)}$, while the associated leading constant depends on unknown problem-specific quantities and is therefore difficult to calibrate accurately.
The theoretical analysis in Theorem~\ref{thm:common-A0-error-bounds} provides an order-wise characterization of the Gaussian noise level $\sigma_k^{(n)}$ in terms of the sensitivities $\xi_k^{(n)}$. However, the leading constants involve unknown problem-specific quantities and are therefore difficult to calibrate exactly in practice. To isolate and illustrate the effect of the privacy parameters $(\varepsilon,\delta)$ in our simulation, we simply set $\sigma_k^{(n)}=\sqrt{2\log(1.25/\delta)}/\varepsilon$ for all $k\in[K]$ and $n\in[N_g]$ to evaluate the privacy-utility trade-off.
This choice can be viewed as a conservative benchmark, as it does not leverage the potentially smaller sensitivity that arises when clients are similar and sample sizes $T_k$ are sufficiently large. In such favorable settings, the actual noise required for a given $(\varepsilon, \delta)$ could be much smaller, leading to improved empirical performance relative to the results reported here. 
%When the client-specific parameters are sufficiently similar and $T_k$ satisfies a mild sample-size requirement, the noise scale is much smaller under the same $(\varepsilon,\delta)$, thereby improving empirical performance relative to the above choice of $\sigma_k^{(n)}$.

\subsection{Single Client Estimation and Rank Selection}\label{sec:simu-single-client-est}

We conduct two experiments to examine the estimation accuracy and rank selection consistency for a single client using only its local data. 
The single-client estimator in \eqref{eq:local-estimator} is obtained via ADMM algorithm; see Section~\ref{sec:ADMM} of the Supplementary Material for implementation details. 
%For notational simplicity, we write $T$ for the sample size of this client and suppress the subscript $k$.

Figure~\ref{fig:single_client_errors} reports the average Frobenius-norm estimation errors of the single-client estimators $\widetilde{\bm A}_{0,k}, \widetilde{\bbm \Delta}_k$ and $\widetilde{\bm A}_{k}=\widetilde{\bm A}_{0,k} + \widetilde{\bbm \Delta}_k$, computed over 1,000 replications, with the 5th--95th percentile bands shown as shaded regions. 
All errors decrease as $T_k$ increases, confirming that larger local datasets lead to more accurate estimation. For small $T_k$, the total error $\|\widetilde{\bm A}_{k}-\bm A_k^*\|_{\F}$ closely tracks the deviation error $\|\widetilde{\bbm \Delta}_k-\bbm \Delta_k^*\|_{\F}$, indicating that the total estimation error is dominated by the difficulty of recovering the client-specific sparse deviation $\bm\Delta_k^*$ when data are scarce. As $T$ becomes larger, the deviation error $\|\widetilde{\bbm \Delta}_k-\bbm \Delta_k^*\|_{\F}$ decreases faster, while the total error $\|\widetilde{\bm A}_k-\bm A_k^*\|_{\F}$ remains slightly larger due to the estimation error in the common component $\bm A_0^*$.
Finally, the 5th--95th percentile bands become narrower as $T_k$ increases, indicating improved statistical stability of the single-client estimator with larger sample sizes.

Table~\ref{tab:rank-correct-rate} reports the correct selection rates of the proposed ridge-type rank estimator in~\eqref{eq:Ridge-type estimator} across different sample sizes $T_k$ and true ranks $r^*$. For any given $r^*$, the correct selection rate improves as $T_k$ increases, confirming the rank-selection consistency established in Theorem \ref{thm:rank-selection-consistency}. At a given sample size, correctly identifying a larger true rank is generally more challenging, leading to lower selection accuracy when $T_k$ is small. However, the selector remains consistent, as evidenced by the rapid improvement in selection rates as $T_k$ grows across all rank values.
%Moreover, for a fixed sample size, identifying a larger true rank is generally more challenging and thus yields a lower correct rate at smaller $T$; nevertheless, the selector remains consistent, as its accuracy improves rapidly with $T_k$ for all values of $r^*$.

\begin{figure}[H]
    \centering
    \includegraphics[width=0.8\textwidth]{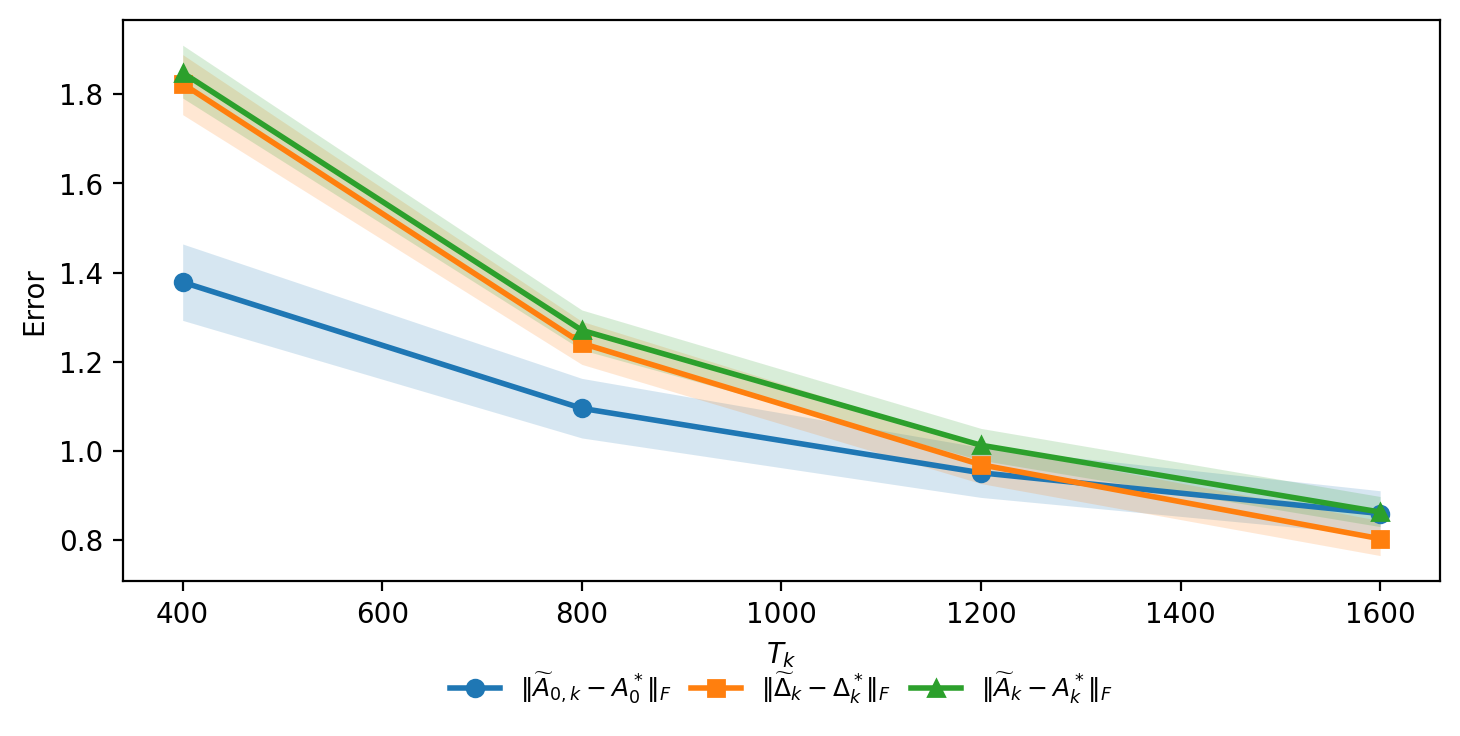}
    \caption{Average Frobenius-norm estimation errors for $\bm A_0^*, \bbm \Delta_k^*$ and $\bm A_k^*$ versus the sample size $T_k$ under the single client setting. Shaded bands indicate the 5th--95th percentiles range across $1,000$ replications.}
    \label{fig:single_client_errors}
\end{figure}

\begin{table}[H]
\centering
\caption{Correct selection rate (\%) of the ridge-type rank estimator.}
\label{tab:rank-correct-rate}
\setlength{\tabcolsep}{14pt}
\renewcommand{\arraystretch}{0.8}
\begin{tabular}{lcccc}
\toprule
$r^* / T_k$ & 400 & 800 & 1200 & 1600 \\
\midrule
1 & 68.20 & 99.40 & 100.00 & 100.00 \\
2 & 42.90 & 86.30 & 98.40 & 100.00 \\
3 & 41.80 & 85.50 & 99.00 & 100.00 \\
\bottomrule
\end{tabular}
\end{table}

\subsection{Federated Estimation}\label{sec:simu-fed-est}

We conduct four experiments to assess the finite-sample performance of the proposed federated method from four perspectives:  
(i) the impact of the privacy budget $(\varepsilon,\delta)$ on the number of iterations required in Stage I; (ii) the estimation errors under varying privacy budgets $(\varepsilon,\delta)$; (iii) the estimation errors under varying numbers of clients $K$; and (iv) the effect of client-specific sample sizes $T_k$ on estimation accuracy.  
In all experiments, the common component $\widehat{\bm A}_0$ is obtained by Algorithm~\ref{algorithmTL}, the client-specific deviations $\widehat{\bbm \Delta}_k$ are calculated using Algorithm~\ref{algorithmTL-stage2}, and the personalized coefficient matrices are given by $\widehat{\mathbf{A}}_{k}=\widehat{\bm A}_0+\widehat{\bbm \Delta}_k$.

Figure~\ref{fig:stage1_A0_fro_err_vs_iter__by_eps_delta} illustrates the average Frobenius-norm estimation error of $\bm A_0^{(n)}$ with respect to the iteration number $n$ for different privacy budgets, computed over 1,000 replications, with the 5th–95th percentile bands shown as shaded regions.  
Smaller values of $\varepsilon$ and $\delta$ (i.e., more stringent privacy) lead to higher estimation errors and more volatile convergence behavior. This is consistent with Theorem~\ref{thm:common-A0-error-bounds}, where the privacy-induced term $\mathsf{Error}_{\mathrm{DP}}$ dominates and increases sharply as the privacy budget tightens. For larger $(\varepsilon,\delta)$ (i.e., weaker privacy), the federated estimator achieves lower errors and smoother convergence than the single-client initialization.
%In contrast, when the privacy budget is less restrictive (larger $(\varepsilon,\delta)$), the federated procedure yields more accurate and stable estimates of $\bm A_0$ than the single-client initialization.

\begin{figure}[H]
	\centering
	\includegraphics[width=1\textwidth]{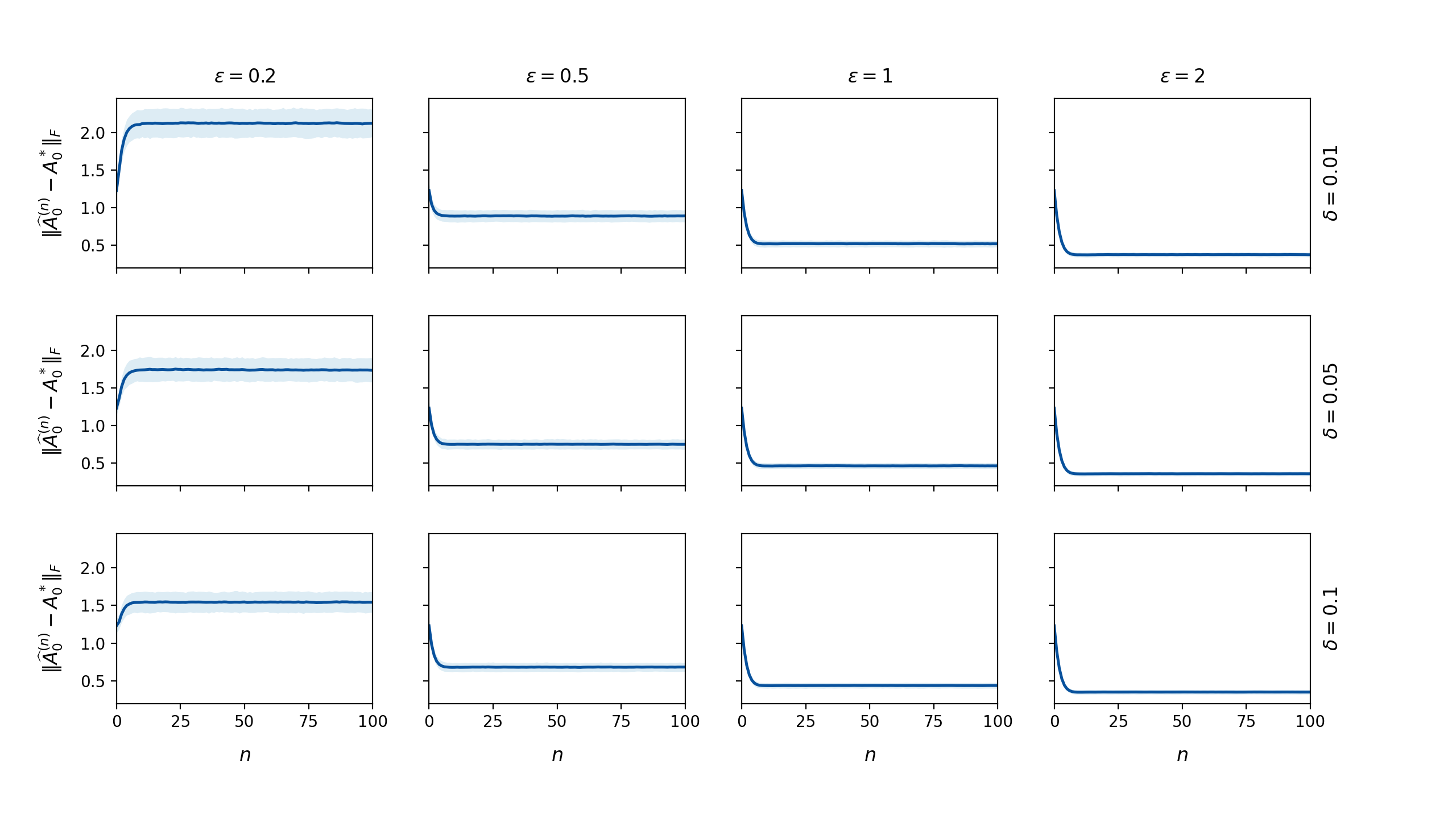}
	\caption{Average Frobenius-norm estimation errors of $\bm A_0^{(n)}$ versus the number of iterations $n$ under different $(\varepsilon,\delta)$ configurations for $T_k=400$ and $K=5$. Shaded bands indicate the 5th--95th percentiles range across $1,000$ replications.}
	\label{fig:stage1_A0_fro_err_vs_iter__by_eps_delta}
	\end{figure}

Figure~\ref{fig:sweep_privacy_eps_delta_heat_mosaic} displays the average estimation errors for $\widehat{\bm A}_0$, $\widehat{\bm\Delta}_k$, and $\widehat{\bm A}_k$ under different privacy budgets $(\varepsilon,\delta)$, and the benefit of the federated method over the single-client baseline in \eqref{eq:local-estimator}, based on 1,000 replications. 
The benefit for a generic parameter matrix $\{\bbm \Xi_k^*\}_{k=1}^K$ is defined as
\begin{align}\label{eq:benifit}
\mathrm{Benefit}(\bbm \Xi_k^*) = \frac{1}{K}\sum_{k=1}^{K}\|\widetilde{\bbm \Xi}_{k}-\bbm \Xi_k^*\|_\F - \frac{1}{K}\sum_{k=1}^{K}\|\widehat{\bbm \Xi}_{k}-\bbm \Xi_k^*\|_\F,
\end{align}
where $\{\widetilde{\bbm \Xi}_{k},\widehat{\bbm \Xi}_{k}, \bbm \Xi_k^*\}=\{\widetilde{\bm A}_{0,k},\widehat{\bm A}_{0}, \bm A_{0}^*\}$, $\{\widetilde{\bbm \Delta}_{k},\widehat{\bbm \Delta}_{k}, \bbm \Delta_{k}^*\}$ or $\{\widetilde{\bm A}_{k},\widehat{\bm A}_{k}, \bm A_{k}^*\}$ for all $k\in[K]$. 
%Here $\widetilde{\bbm \Xi}_{k}$ and $\widehat{\bbm \Xi}_{k}$ denote the single-client and federated estimators, respectively. In particular, $\{\widetilde{\bbm \Xi}_{k},\widehat{\bbm \Xi}_{k}\}\in\{\{\widetilde{\bm A}_{0,k},\widehat{\bm A}_{0}\},\{\widetilde{\bbm \Delta}_{k},\widehat{\bbm \Delta}_{k}\},\{\widetilde{\bm A}_{k},\widehat{\bm A}_{k}\}\}$. For the common component $\bm A_0^*$, we set $\bbm \Xi_k^*=\bm A_0^*$, $\widetilde{\bbm \Xi}_k=\widetilde{\bm A}_{0,k}$, and $\widehat{\bbm \Xi}_k=\widehat{\bm A}_0$ for all $k\in[K]$.
It can be seen that all the estimation errors decrease as $\varepsilon$ and/or $\delta$ increase (i.e., the privacy constraint is relaxed). 
The benefit heatmaps further validate the expected privacy-utility trade-off.
%reveal that the advantage of federation is sensitive to the privacy budget. 
Under stringent privacy (e.g., small $\varepsilon$), the federated method may fail to outperform the single-client baseline, particularly for the common component $\bm A_0^*$. However, once the privacy budget is moderately relaxed, the federated procedure delivers clear improvements over local learning across most $(\varepsilon,\delta)$ configurations.
%In summary, Figure~\ref{fig:sweep_privacy_eps_delta_heat_mosaic} demonstrates the expected privacy-utility tradeoff: stricter privacy leads to larger estimation errors, while a moderately relaxed privacy budget allows federation to deliver clear improvements over the single-client baseline.

\begin{figure}[H]
	\centering
	\includegraphics[width=1\textwidth]{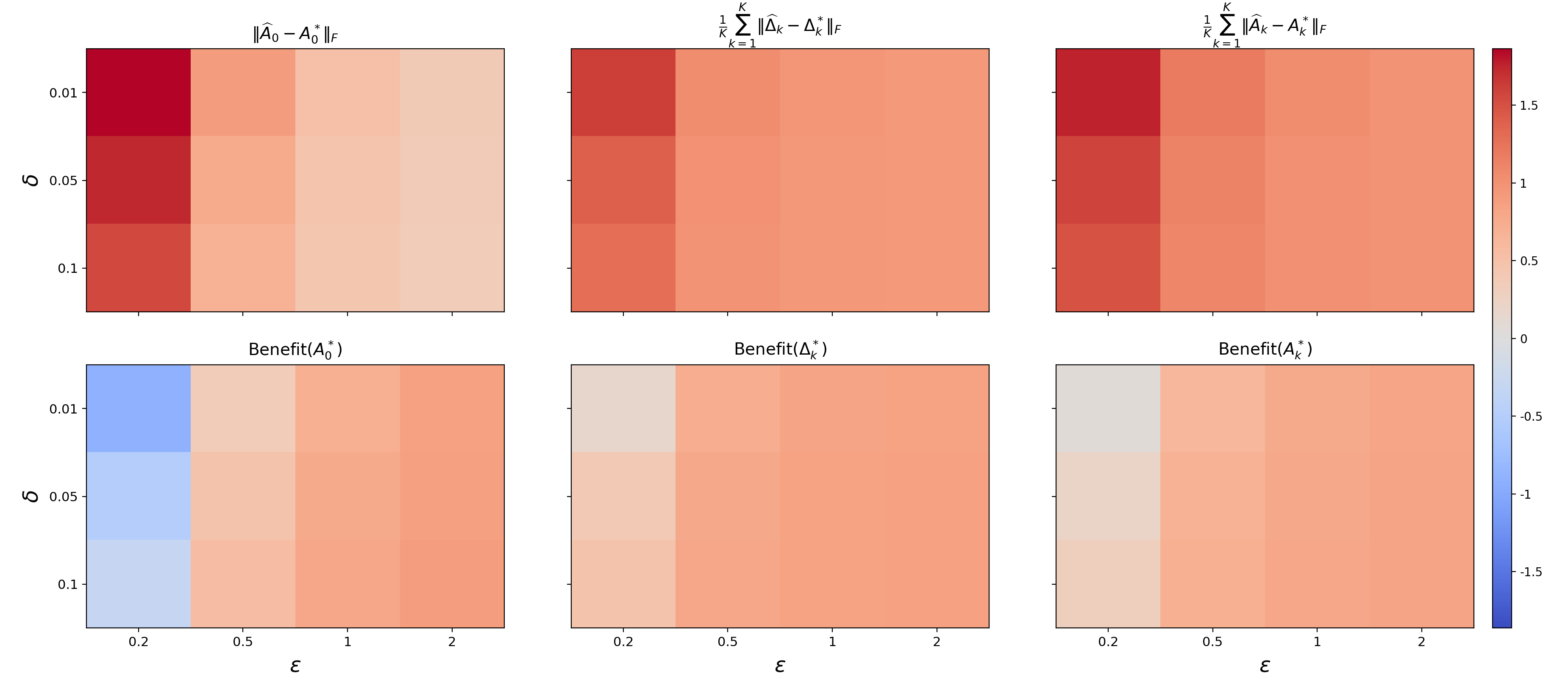}
	\caption{Heatmaps of average Frobenius-norm estimation errors (upper panel) for $\bm A_0^*$, $\bbm \Delta_k^*$ and $\mathbf{A}_k^*$, and their benefit (lower panel) over single-client estimators, under different $(\varepsilon,\delta)$ configurations for $T_k=400$ and $K=5$.}
	\label{fig:sweep_privacy_eps_delta_heat_mosaic}
	\end{figure}

Figure~\ref{fig:sweep_K_errors} shows the average Frobenius-norm estimation errors versus the number of clients $K$, with $T_k=400$ for all $k$ and $(\varepsilon,\delta)=(2,0.1)$, averaged over 1,000 replications. In Figure~\ref{fig:sweep_K_errors}, the label \(\widehat{\bbm \Delta}_k\) represents the average error $K^{-1}\sum_{k=1}^K\|\widehat{\bbm\Delta}_k-\bbm\Delta_k^*\|_{\F}$, and the other labels are defined analogously.
We can see that increasing $K$ leads to a noticeable reduction in the federated estimation error for the shared structure $\bm A_0^*$, as a larger $K$ provides more cross-client information for estimating $\bm A_0^*$. In contrast, the errors for the deviation $\bbm\Delta_k^*$ and the overall client-specific matrix $\bm A_k^*$, remain nearly unchanged as $K$ grows. This is because the deviation component $\bbm\Delta_k^*$ is client-specific and must be estimated from the local data of client $k$ alone. 
%This pattern is consistent with the decomposition $\bm A_k^*=\bm A_0^*+\bbm\Delta_k^*$. Since the local sample size $T_k$ is fixed at 400 for each client, increasing $K$ mainly contributes additional information for estimating the shared component $\bm A_0^*$ through cross-client aggregation, thus leading to a clear reduction in $\|\widehat{\bm A}_0-\bm A_0^*\|_{\F}$. By contrast, the deviation component $\bbm\Delta_k^*$ is client-specific and must be estimated from the local data of client $k$ alone. 
With $T_k$ held fixed, adding more clients does not reduce the statistical difficulty of estimating each $\bbm\Delta_k^*$, so the average error remains stable.
Moreover, as observed in Figure~\ref{fig:single_client_errors}, when $T_k$ is relatively small, the estimation error of $\widehat{\bm A}_k$ is dominated by the error in estimating $\bbm\Delta_k^*$, which is consistent with the pattern seen here.

\begin{figure}[H]
	\centering
	\includegraphics[width=0.8\textwidth]{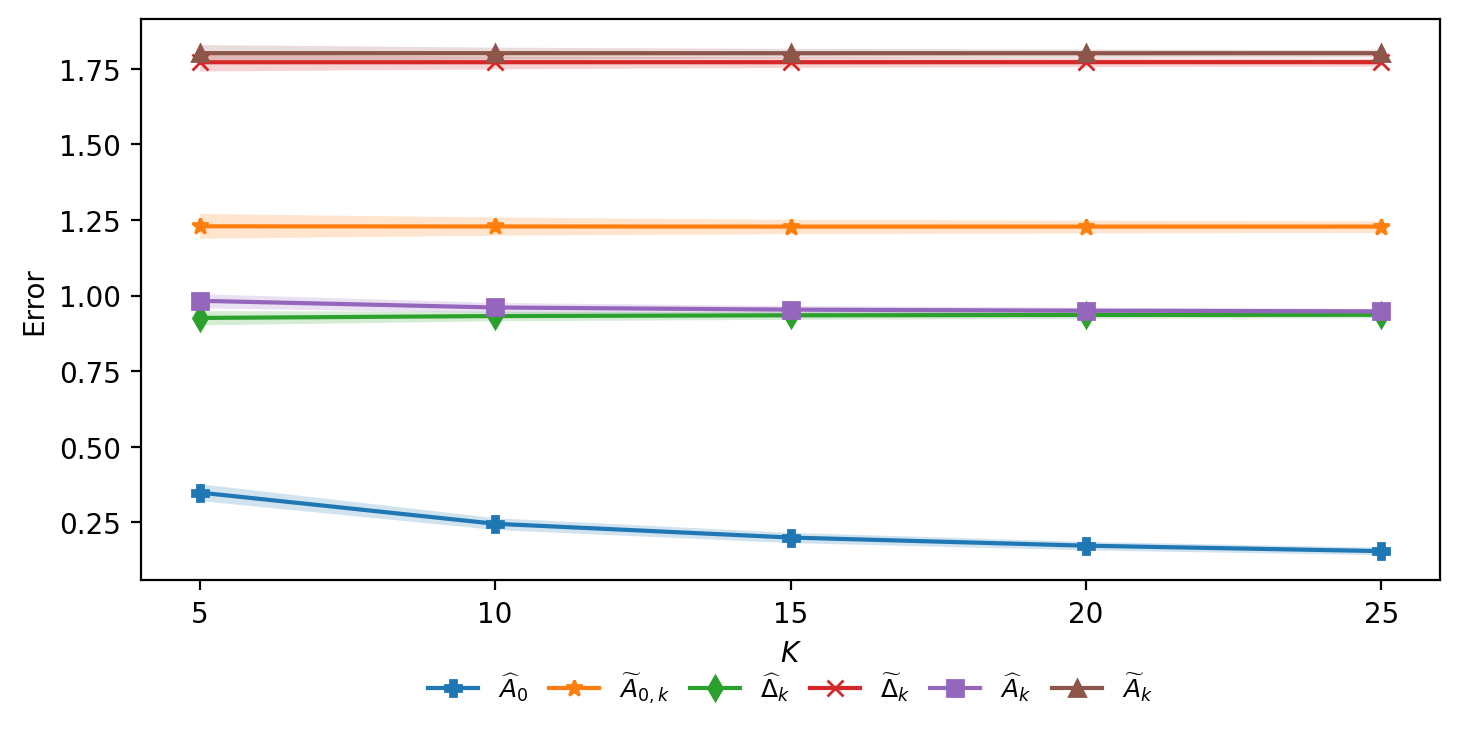}
	\caption{Average Frobenius-norm estimation errors for $\bm A_0^*$, $\bbm \Delta_k^*$ and $\mathbf{A}_k^*$ versus client number $K$ for $T_k=400$ and $(\varepsilon,\delta)=(2,0.1)$. Shaded bands indicate the 5th--95th percentiles range across $1,000$ replications.}
	\label{fig:sweep_K_errors}
\end{figure}

Figure~\ref{fig:teff_sweep_errors} compares the average Frobenius-norm estimation errors of the federated and single-client methods versus the local sample size $T_k$, with $K=5$ and $(\varepsilon,\delta)=(2,0.1)$, averaged over 1,000 replications. All clients share the same $T_k$ for simplicity. In each panel, ``Federated'' and ``Local'' denote the average estimation errors of the corresponding estimators for the coefficient matrix indicated in the subtitle, while ``Benefit'' quantifies the improvement of the federated method over the local baseline as defined in \eqref{eq:benifit}. %These averages are computed in the same way as in Figure~\ref{fig:sweep_K_errors}. “Benefit” denotes the improvement of the federated estimator over the local baseline, as defined in \eqref{eq:benifit}.
%Figure~\ref{fig:teff_sweep_errors} shows that all three components benefit from federated estimation, but the magnitude and shape of the benefit differ across panels. The shared component $\bm A_0^*$ benefits directly from cross-client aggregation in Stage~I, whereas the client-specific components $\bbm\Delta_k^*$ and hence the full matrices $\bm A_k^*$ benefit indirectly through a more accurate estimate of $\bm A_0^*$.

Figure~\ref{fig:teff_sweep_errors} shows that all three components benefit from federated estimation, though the benefit patterns differ across panels. 
For the shared component $\bm A_0^*$ (left panel), the benefit decreases as $T_k$ increases. This aligns with the error bounds in Proposition~\ref{thm:local-model-error-upper-bound} and Theorem~\ref{thm:common-A0-error-bounds}. With $K=5$ fixed, the federated estimator pools $T=\sum_{k=1}^K T_k = 5T_k$ samples, starting from a much lower error level. While its error continues to decrease with $T_k$, the reduction is less pronounced than that of the single-client estimator, whose error is governed solely by $T_k$. As a result, the benefit shrinks as $T_k$ grows.
In contrast, the benefit for the client-specific deviation $\bbm\Delta_k^*$ (middle panel) and the full transition matrix $\bm A_k^*$ (right panel) initially increases with $T_k$ and then gradually plateaus. When $T_k$ is small, $\widehat{\bm A}_0$ remains noisy, limiting the improvement passed to $\widehat{\bbm\Delta}_k$. As $T$ increases, $\widehat{\bm A}_0$ becomes more accurate and stable, enhancing the estimation of $\bbm\Delta_k^*$. Once $\widehat{\bm A}_0$ is estimated accurately enough, the dominant constraint for estimating $\bbm\Delta_k^*$ shifts to the local sample size $k$ itself, so further improvements in Stage I yield only marginal gains in Stage II, stabilizing the benefit. 
Finally, as observed in Figure~\ref{fig:single_client_errors}, the estimation error of $\widehat{\bm A}_k$ is largely driven by that of $\widehat{\bbm\Delta}_k$ when $T_k$ is small, explaining why the error and benefit patterns for $\bm A_k^*$ closely mirror those for $\bbm\Delta_k^*$.

\begin{figure}[H]
	\centering
	\includegraphics[width=1\textwidth]{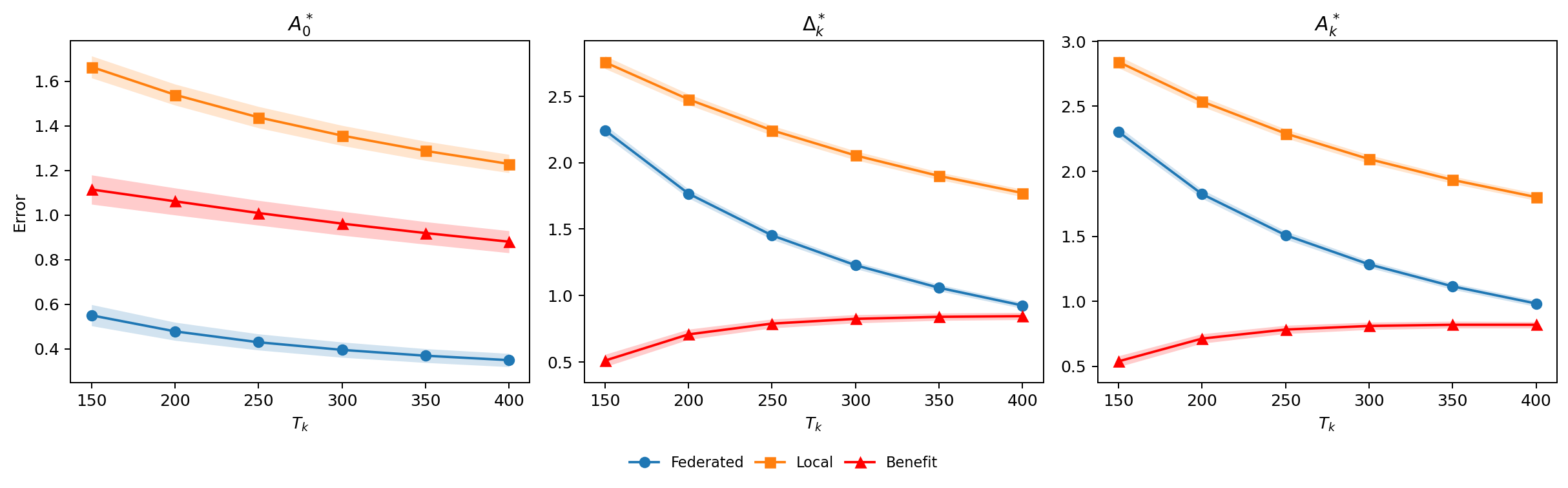}
	\caption{Average Frobenius-norm estimation errors of federated and single-client methods, and the benefit of federation, versus sample size $T_k$ for $K=5$ and $(\varepsilon,\delta)=(2,0.1)$. Shaded bands indicate the 5th--95th percentiles range across $1,000$ replications.}
	\label{fig:teff_sweep_errors}
\end{figure}

\section{Empirical Analysis}\label{sec:RealData}

We illustrate the practical utility of our federated VAR framework through two empirical studies. Section~\ref{sec:empirical_design_monthly_electricity}, investigates the relationship between electricity consumption and economic indicators across U.S. states, a setting motivated by privacy-sensitive scenarios where raw local data cannot be directly shared. %This setting is inherently privacy-sensitive, as raw local data cannot be directly shared across clients. 
Section~\ref{sec:macro-forecasting} then considers macroeconomic forecasting across multiple countries. Although privacy protection is not the primary concern here, this application illustrates the broader utility of our framework as a personalized multi-task learning approach that operates without privacy constraints.
%In this case, privacy protection is not the main concern. Instead, the motivation is to demonstrate the usefulness of our framework in a more general setting where it reduces to a personalized multi-task learning approach without privacy constraints.

The two datasets are constructed from publicly available sources. The first combines data from the U.S. Energy Information Administration (EIA) (\url{https://www.eia.gov/electricity/data/eia861m/}) and Federal Reserve Economic Data (FRED) (\url{https://fred.stlouisfed.org/}), while the second is compiled from the OECD Economic Outlook database (\url{https://www.oecd.org/economic-outlook/}). For each client, we align all variables to a common sample period, allowing sample sizes and time spans to vary across clients; see Table~\ref{tab:sample_range_all} in the Supplementary Material. %Table~\ref{tab:sample_range_all} in the Supplementary Material reports the sample period for each application. 
All data are then preprocessed by possible seasonal adjustment and transformations to remove non-stationarity, followed by standardization to zero mean and unit variance. Variable definitions, category labels, and preprocessing details are provided in Tables~\ref{tab:data_preprocessing_1} and \ref{tab:data_preprocessing_2} of the Supplementary Material.

For each application, we compare the forecasting performance of our federated VAR method with several single-client benchmarks: the combined nuclear-norm and $\ell_1$-regularized VAR (Nuc+$\ell_1$) introduced in Section~\ref{sec:single-client-estimation}, the nuclear-norm-regularized VAR (Nuc), the $\ell_1$-regularized VAR ($\ell_1$), and the unregularized least-squares VAR (LS). 
For our federated method, we evaluate the privacy-utility trade-off by adding Gaussian noise with standard deviation $\sigma_k^{(n)} = \sqrt{2\log(1.25/\delta)}/(10\varepsilon)$ for all $k\in[K]$ and $n\in[N_g]$, and report results under three privacy configurations: $(\varepsilon,\delta)=(0.2,0.05)$, $(0.1,0.01)$, and a non-private case (NDP) without noise.
Following \cite{koop2013forecasting}, all methods use a VAR(4) specification in both applications. 
Hyperparameters and initialization for our methods follow the procedure described in Section~\ref{sec:tunning_param}, while tuning parameters for the single-client benchmarks are chosen by standard cross-validation.

To evaluate the forecasting performance, we use the root mean squared forecast error (RMSFE) for each variable, computed via an expanding window approach. Specifically, for each forecast origin, the model is re-estimated using all available data up to that point, and one-step-ahead forecasts are generated for the last 20 observations. The RMSFE is then computed as the square root of the average squared forecast errors across these 20 forecast origins.

\subsection{Federated Learning for Electricity and Economic Activity}\label{sec:empirical_design_monthly_electricity}

We apply our federated VAR method to study the linkage between electricity consumption and economic indicators across U.S. states. 
Electricity data are valuable for macroeconomic analysis, serving as timely proxies for local economic activity \citep{arora2016electricity}. However, at finer geographic levels such as counties, cities, or utility service territories, the data are often subject to proprietary restrictions, regulatory limits, or data-sharing constraints, as public release could expose information about individual entities. By contrast, state-level aggregates are more readily available, as they are reported over a large and diverse user base that masks individual-level information. The U.S. Energy Information Administration publishes such state-level statistics on electricity sales, revenue, customer counts, and prices. 
While our empirical analysis uses these publicly accessible aggregates as a reproducible proxy, the modeling rationale directly mirrors the privacy-sensitive settings for which our method is designed.	 
%At such disaggregated levels, public release is more likely to reveal information about an individual utility, a large customer, or a small group of customers. As a result, these data are more often subject to proprietary restrictions, regulatory limits, or data-sharing constraints. 
%electricity statistics for larger geographic units such as states are more readily released because they are reported in aggregated form over a much broader set of users, firms, or utilities, making it less likely that any individual entity can be identified. In particular, the U.S. Energy Information Administration publishes state-level aggregates on electricity sales, revenue, customer counts, and prices. 
%For empirical illustration, we therefore use these publicly available state-level electricity data as a reproducible proxy, and conduct the economic analysis at this aggregate level, while preserving the same privacy-motivated modeling rationale as in the intended small-region deployment.

The five states in our study, Indiana, Ohio, Wisconsin, Michigan, and Illinois are all located in the Great Lakes region. As documented by \citet{ahking2014economies}, the business cycles of these states are substantially synchronized, though Michigan and Illinois exhibit distinct patterns from the others. This combination of shared dynamics and state-specific characteristics makes them well-suited for our personalized federated learning framework, which is designed to leverage cross-state information while preserving local heterogeneity.
%\cite{ahking2014economies} shows that the business cycles of these states are substantially synchronized, while also indicating that Michigan and Illinois are not fully homogeneous with the others. This makes these five states a suitable set of clients in our personalized federated learning framework, which is designed to exploit shared information across states to improve state-level forecasting performance while preserving cross-state heterogeneity.
For each state, we construct a 12-dimensional monthly time series that includes electricity sales, revenue, prices, and key economic indicators such as employment, industrial production, and retail sales; see Table~\ref{tab:data_preprocessing_1} in the Supplementary Material for details. Using a VAR(4) specification, the ridge-type ratio rank selection criterion in~\eqref{eq:Ridge-type estimator} selects $\widehat{r}=3$ for the common component $\bm A_0$ in Stage~I, with all other tuning parameters chosen as described at the beginning of this section. 

Table~\ref{tab:5_state_by_method_rmsfe_transposed} reports the forecasting performance of our federated method and the single-client benchmarks across the five states. We have the following findings. 
First, the non-private federated estimator (NDP) achieves the lowest RMSFE for four of the five states, outperforming all single-client benchmarks. This suggests that borrowing information across economically related states is beneficial for state-level forecasting if the states are sufficiently similar and local sample sizes are moderate.
Second, under moderate privacy noise, the federated estimator remains highly competitive and still outperforms the single-client methods in most cases. Notably, for Michigan, the private federated estimator under $(\varepsilon,\delta)=(0.2,0.05)$ even slightly outperforms the non-private version. This suggests that appropriately calibrated Gaussian noise can act as an implicit regularizer, mitigating overfitting and improving out-of-sample generalization \citep{boursier2024utility}.
Third, stronger privacy protection degrades forecasting performance. Increasing the noise level from $(\varepsilon,\delta)=(0.2,0.05)$ to $(0.1,0.01)$ leads to uniformly higher RMSFE across all five states. This pattern reflects the standard privacy-utility trade-off: while modest noise may occasionally improve generalization through implicit regularization, excessive noise obscures the shared signal in Stage~I and deteriorates forecast accuracy.

\begin{table}[H]
\centering
\caption{RMSFEs of our federated method (Fed-VAR) under three $(\varepsilon,\delta)$-DP settings and four single-client benchmarks (Single-VAR) across five states. Within each state, the smallest value is shown in bold and the second smallest is underlined.}
\label{tab:5_state_by_method_rmsfe_transposed}

\setlength{\tabcolsep}{12.2pt}
\renewcommand{\arraystretch}{0.90}
\footnotesize

\makebox[\textwidth][c]{%
\begin{tabular}{l c c c c @{\hspace{12pt}} c c c c}
\toprule
& & \multicolumn{3}{c}{\textbf{Fed-VAR}} & \multicolumn{4}{c}{\textbf{Single-VAR}} \\
\cmidrule(lr){3-5}\cmidrule(lr){6-9}
\textbf{State}
& \textbf{$T_k$}
& \textbf{NDP}
& $\mathbf{(0.2,0.05)}$
& $\mathbf{(0.1,0.01)}$
& \textbf{Nuc + $\ell_1$}
& \textbf{Nuc}
& \textbf{$\ell_1$}
& \textbf{LS} \\
\midrule

Illinois  & 44 & \textbf{0.751} & \underline{0.780} & 0.813 & 0.819 & 0.881 & 0.798 & 3.437 \\
Indiana   & 44 & \textbf{1.379} & \underline{1.387} & 1.564 & 1.569 & 1.487 & 1.389 & 4.311 \\
Michigan  & 44 & \underline{0.598} & \textbf{0.592} & 0.617 & 0.606 & 0.607 & 0.613 & 1.816 \\
Ohio      & 44 & \textbf{0.911} & \underline{0.965} & 1.015 & 0.977 & 0.983 & 0.994 & 2.910 \\
Wisconsin & 44 & \textbf{0.725} & \underline{0.735} & 0.784 & 0.838 & 0.916 & 0.917 & 2.874 \\
\midrule
\textbf{Average} & 44
& \textbf{0.914} & \underline{0.933} & 1.013 & 1.015 & 1.016 & 0.977 & 3.176 \\
\bottomrule
\end{tabular}%
}
\end{table}

%To provide a structural interpretation of the forecasting results in Table~\ref{tab:5_state_by_method_rmsfe_transposed},
Figure~\ref{fig:heatmap_state_panel} presents heatmaps of the estimated shared component $\widehat{\bm A}_0$ and two representative state-specific deviations, $\widehat{\bm \Delta}_{\mathrm{Wisconsin}}$ and $\widehat{\bm \Delta}_{\mathrm{Illinois}}$, under $(\varepsilon,\delta)=(0.2,0.05)$. Together with Table~\ref{tab:5_state_by_method_rmsfe_transposed}, this figure provides structural insight into why electricity data are useful for predicting state-level economic conditions. 

The top panel reveals that $\widehat{\bm A}_0$ exhibits a clear dependence pattern across electricity variables and macroeconomic indicators, indicating that lagged electricity-market variables carry predictive information for current economic outcomes. Several features are noteworthy. 
First, the strongest common signals are concentrated in the rows corresponding to unemployment, payroll employment, and the coincident index, indicating that current economic conditions are systematically related to lagged electricity and macroeconomic variables.
Second, the largest coefficients appear primarily in the first one or two lag blocks, suggesting that the predictive linkage between electricity variables and economic indicators is primarily short-run. This pattern aligns with the view that changes in electricity usage and electricity-market activity are closely connected to near-term fluctuations in employment and overall economic activity.

The lower two panels in Figure~\ref{fig:heatmap_state_panel} show that the state-specific deviations are much more localized. Most entries in $\widehat{\bm \Delta}_{\mathrm{Wisconsin}}$ and $\widehat{\bm \Delta}_{\mathrm{Illinois}}$ are close to zero, with only a limited subset exhibiting noticeable departures from the shared structure. Notably, $\widehat{\bm \Delta}_{\mathrm{Illinois}}$ appears more diffuse and contains more visible nonzero entries than $\widehat{\bm \Delta}_{\mathrm{Wisconsin}}$, indicating that Illinois deviates more substantially from the common Great Lakes pattern. This aligns with \citet{ahking2014economies}' finding that Illinois is not fully homogeneous with the other Great Lakes states.

\begin{figure}[H]
    \centering
    \includegraphics[width=1.0\textwidth]{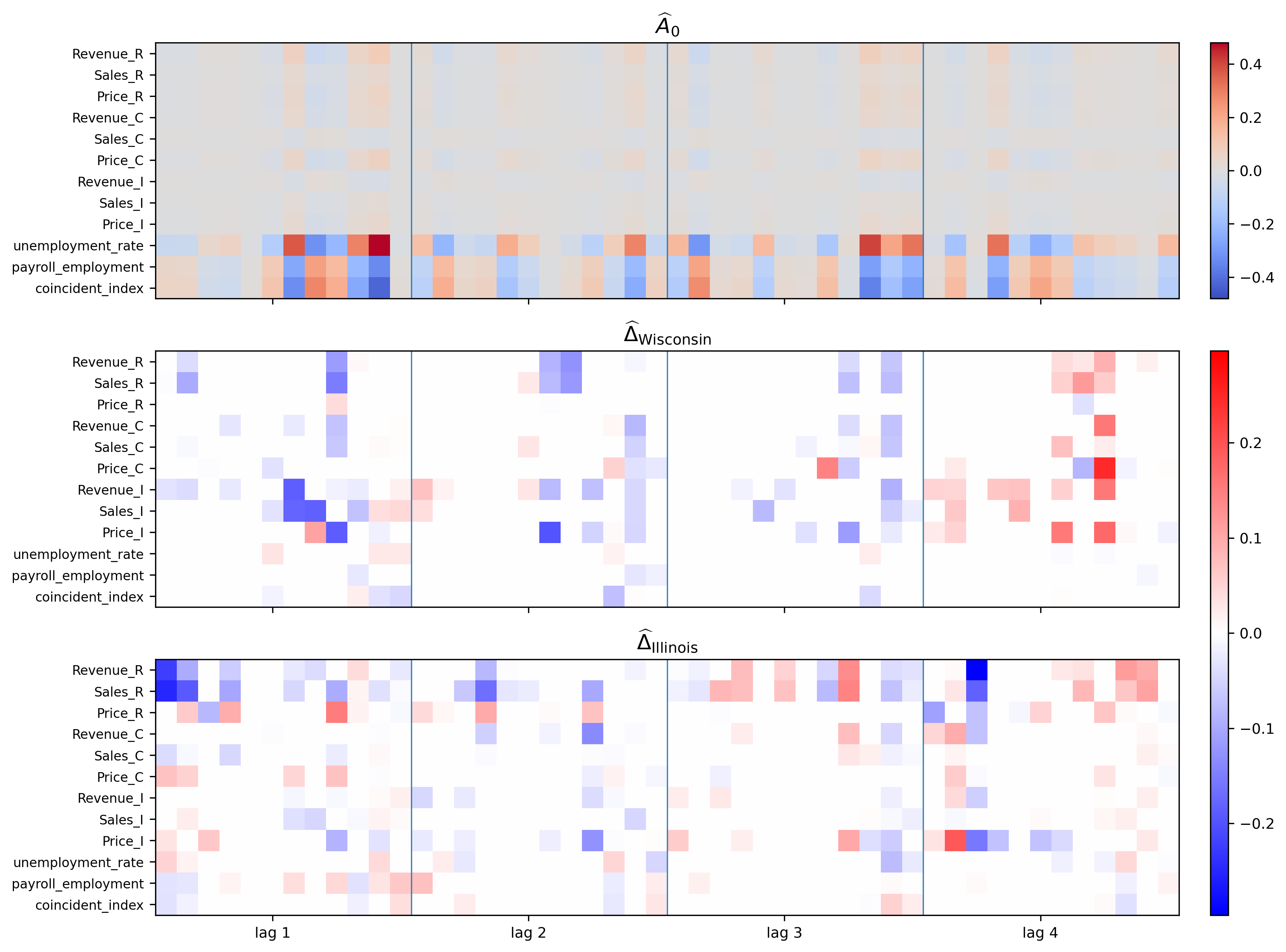}
    \caption{Heatmaps of the estimated shared component $\widehat{\bm A}_0$ and two representative state-specific deviations, $\widehat{\bm \Delta}_{\mathrm{Wisconsin}}$ and $\widehat{\bm \Delta}_{\mathrm{Illinois}}$, from the federated VAR model under $(0.2,0.05)$-DP. Columns are grouped by lag order.}
    \label{fig:heatmap_state_panel}
\end{figure}

\subsection{Macroeconomic Forecasting}\label{sec:macro-forecasting}

We apply our federated VAR method to macroeconomic forecasting across eight countries: the United States, Australia, Canada, Germany, Korea, Norway, Sweden, and Denmark. For each country, we construct a 25-dimensional quarterly macroeconomic panel including real activity, trade, production and sales, housing, interest rates, labour markets, prices, and financial indicators; see  Table~\ref{tab:data_preprocessing_2} in the Supplementary Material for details. The sample size varies across countries, with the United States having the longest series and Denmark the shortest. Under a VAR(4) specification, the ridge-type ratio rank selection criterion in~\eqref{eq:Ridge-type estimator} selects $\widehat{r}=4$ for the common component $\bm A_0$ in Stage~I, with other tuning parameters chosen as described at the beginning of this section.

Table~\ref{tab:8_country_by_method_rmsfe_transposed} reports the forecasting performance of our federated method and the single-client benchmarks across the eight countries. The results exhibit patterns broadly similar to those observed in Section \ref{sec:empirical_design_monthly_electricity}. Specifically, the non-private federated estimator (NDP) achieves the lowest average RMSFE across countries and outperforms the single-client benchmarks in most cases. Under moderate privacy noise with $(\varepsilon,\delta)=(0.2,0.05)$, the federated method remains highly competitive, attaining the best performance for several countries. In contrast, under stronger privacy protection with $(\varepsilon,\delta)=(0.1,0.01)$, the forecasting accuracy deteriorates noticeably across all countries, which again reflects the standard privacy-utility trade-off.

\begin{table}[H]
\centering
\caption{RMSFEs of our federated method (Fed-VAR) under three $(\varepsilon,\delta)$-DP settings and four single-client benchmarks (Single-VAR) across eight countries (ordered by decreasing sample size). Within each country, the smallest value is shown in bold and the second smallest is underlined.}
\label{tab:8_country_by_method_rmsfe_transposed}

\setlength{\tabcolsep}{12.2pt}
\renewcommand{\arraystretch}{0.90}
\footnotesize

\makebox[\textwidth][c]{%
\begin{tabular}{l c c c c @{\hspace{12pt}} c c c c}
\toprule
& & \multicolumn{3}{c}{\textbf{Fed-VAR}} & \multicolumn{4}{c}{\textbf{Single-VAR}} \\
\cmidrule(lr){3-5}\cmidrule(lr){6-9}
\textbf{Country}
& \textbf{$T_k$}
& \textbf{NDP}
& $\mathbf{(0.2,0.05)}$
& $\mathbf{(0.1,0.01)}$
& \textbf{Nuc + $\ell_1$}
& \textbf{Nuc}
& \textbf{$\ell_1$}
& \textbf{LS} \\
\midrule

United States & 212 & \textbf{1.771} & \underline{1.771} & 2.069 & 1.784 & 1.771 & 1.810 & 3.222 \\
Australia     & 152 & \underline{1.385} & 1.389 & 1.592 & 1.396 & \textbf{1.384} & 1.409 & 2.727 \\
Canada        & 128 & \textbf{1.737} & \underline{1.767} & 1.942 & 1.775 & 1.825 & 1.801 & 5.566 \\
Germany       & 124 & \textbf{1.462} & \underline{1.502} & 1.727 & 1.573 & 1.608 & 1.564 & 3.800 \\
Korea         &  88 & \textbf{1.060} & \underline{1.069} & 1.201 & 1.092 & 1.162 & 1.119 & 2.561 \\
Norway        &  88 & \underline{1.268} & \textbf{1.240} & 1.464 & 1.287 & 1.308 & 1.321 & 3.128 \\
Sweden        &  84 & \textbf{1.299} & \underline{1.305} & 1.516 & 1.390 & 1.362 & 1.361 & 2.651 \\
Denmark       &  72 & \textbf{1.220} & \underline{1.229} & 1.429 & 1.318 & 1.272 & 1.258 & 1.739 \\
\midrule
\textbf{Average} & 118
& \textbf{1.418} & \underline{1.429} & 1.639 & 1.469 & 1.479 & 1.474 & 3.347 \\
\bottomrule
\end{tabular}
}
\end{table}

To provide a structural insight into the forecasting results in Table~\ref{tab:8_country_by_method_rmsfe_transposed}, Figure~\ref{fig:heatmap_macro_panel} illustrates heatmaps of the estimated shared component $\widehat{\bm A}_0$ and two representative client-specific deviations, $\widehat{\bm \Delta}_{\mathrm{AUS}}$ and $\widehat{\bm \Delta}_{\mathrm{USA}}$. 
The top panel shows that $\widehat{\bm A}_0$ exhibits a clear and non-negligible dependence pattern across variable groups and lag blocks, providing visual evidence of a common dynamic structure shared across countries. Several features are noteworthy.
First, stronger entries appear around export, import, and industrial production variables, particularly in the first lag block, indicating a common short-run predictive linkage between external demand and production activity across countries. 
Second, the rows and columns associated with inflation and housing-related variables display relatively strong and persistent signals, suggesting these variables form an important shared dynamic block. %in the cross-country system. 
Third, %labor-market variables also exhibit visible common dependence patterns: 
unemployment shows several negative coefficients, while employment, hourly earnings, and unit labor cost display more positive interactions with activity and price variables, reflecting coherent labor-market dynamics.
By contrast, the lower two panels reveal that the client-specific deviations are much more localized. Most entries in $\widehat{\bm \Delta}_{\mathrm{AUS}}$ and $\widehat{\bm \Delta}_{\mathrm{USA}}$ are close to zero, with only a limited subset exhibiting noticeable deviations from the shared structure. This pattern visually supports our personalized federated decomposition.
Moreover, the degree of heterogeneity differs across countries. Australia's deviation matrix contains more visible nonzero entries than that of the United States, indicating that Australia requires a richer client-specific correction, whereas U.S. dynamics are comparatively well captured by the common structure with only localized adjustments.

\begin{figure}[H]
    \centering
    \includegraphics[width=1.0\textwidth]{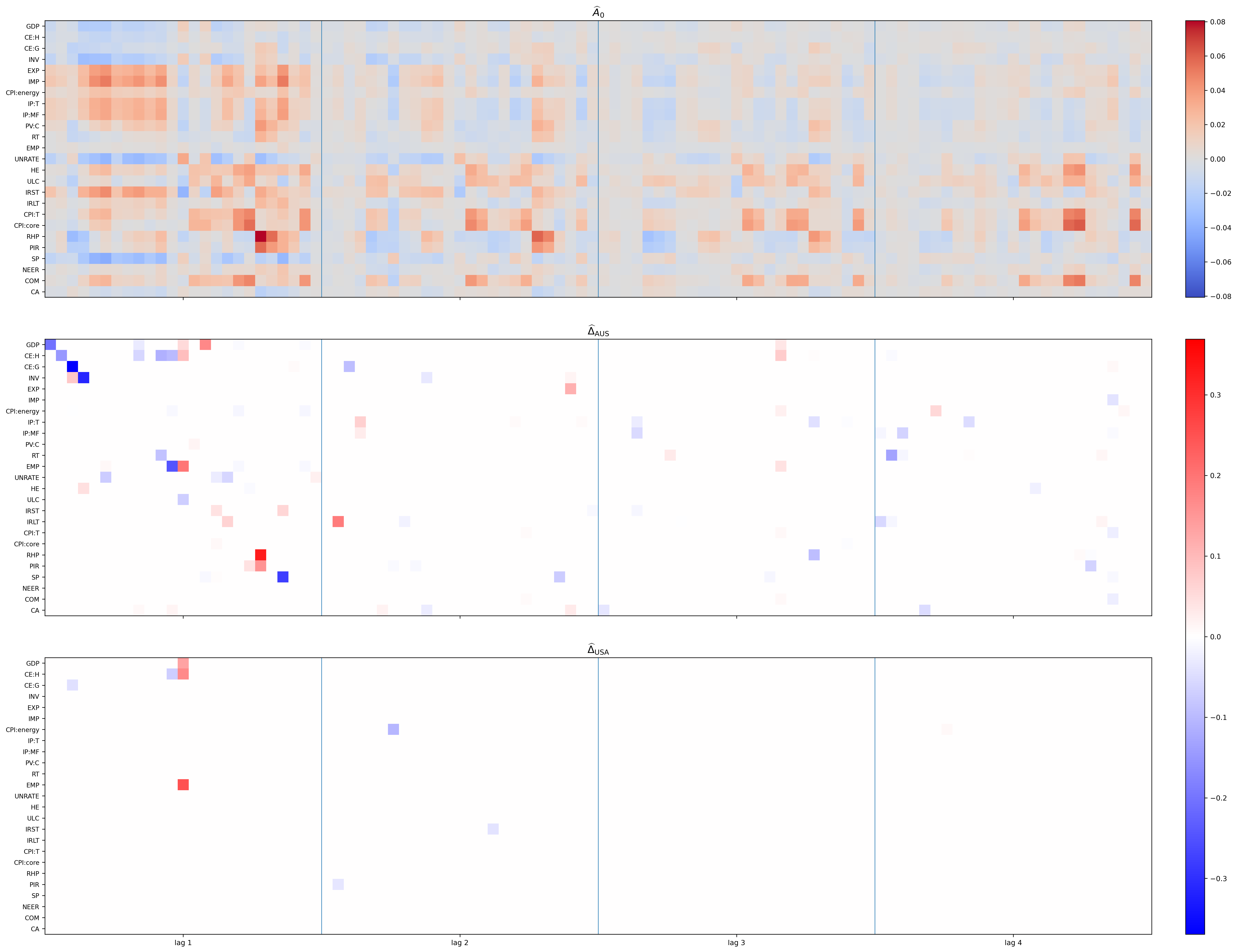}
    \caption{Heatmaps of the estimated shared component $\widehat{\bm A}_0$ and two representative client-specific deviations, $\widehat{\bm \Delta}_{\mathrm{AUS}}$ and $\widehat{\bm \Delta}_{\mathrm{USA}}$, from the federated VAR model under NDP. Columns are grouped by lag order.}
    \label{fig:heatmap_macro_panel}
\end{figure}

\section{Conclusion and Discussion}\label{sec:Conclusions}

This paper develops a privacy-preserving personalized federated learning framework for high-dimensional VAR models. The proposed method leverages a low-rank common structure shared across clients to capture global dynamics, while allowing sparse client-specific deviations to accommodate heterogeneity. The framework operates in two stages: differentially private representation learning for the shared component via noisy gradient aggregation, followed by local personalization for client-specific sparse deviations. We establish non-asymptotic error bounds for both the single-client and federated estimators, and prove consistency of the associated rank selection procedure. A central insight from our theoretical comparison is that the federated estimator outperforms its single-client counterpart when the pooled sample is sufficiently large and client heterogeneity is modest. Simulation results corroborate these findings, showing substantial improvements in estimation accuracy when local sample sizes are limited. 
Empirical applications to state-level electricity–economy linkages and cross-country macroeconomic forecasting further demonstrate the superior forecasting performance of the proposed method, underscoring its practical value in privacy-sensitive and heterogeneous data environments.
%The superior estimation and forecasting performance of the approach is validated through extensive simulations and two empirical applications. 
%Overall, our results provide a principled statistical foundation for privacy-preserving federated forecasting with heterogeneous clients, and suggest that structured representations are a promising route to improving both utility and interpretability under realistic data-sharing constraints.

Several promising directions remain for future research. 
First, the proposed federated learning framework can be extended to more complex time series structures, such as matrix- or tensor-valued time series models, where the shared component may exhibit multi-way low-rank structure. 
Second, given the heavy-tailed nature of many economic and financial time series, integrating robust loss functions or truncation-based methods into the proposed two-stage estimation could improve stability and ensure theoretical guarantees under relaxed moment conditions. 
%one may represent the collection of client- and lag-specific VAR coefficient matrices as a higher-order tensor, and model cross-client commonality and client heterogeneity with a shared low-rank structure and client-specific departures. 
Third, extending the privacy-preserving federated learning framework to time-varying second moments, such as multivariate volatility and covariance dynamics, would enable privacy-aware risk monitoring and portfolio applications using decentralized financial data.
%an important direction is privacy-preserving modeling of time-varying second moments, such as multivariate volatility and covariance dynamics. Extending the selective federated DP framework to volatility models may enable privacy-aware risk monitoring and portfolio applications with decentralized financial data.

\newpage
\appendix
\setcounter{section}{0}
\setcounter{subsection}{0}
\setcounter{subsubsection}{0}

\renewcommand\thesection{S.\arabic{section}}
\renewcommand\thesubsection{S.\arabic{section}.\arabic{subsection}}
\renewcommand\thesubsubsection{S.\arabic{section}.\arabic{subsection}.\arabic{subsubsection}}

\makeatletter
\renewcommand{\theHsection}{appendix.\arabic{section}}
\renewcommand{\theHsubsection}{appendix.\arabic{section}.\arabic{subsection}}
\renewcommand{\theHsubsubsection}{appendix.\arabic{section}.\arabic{subsection}.\arabic{subsubsection}}

\section*{Appendix}
\makeatother
This appendix provides technical details, proofs and additional simulation results supporting the main paper. Section \ref{append:notation and preliminaries} introduces the essential notation and preliminaries. Section \ref{append:Primary lemmas} presents primary lemmas that form the theoretical foundation for our main results. Section \ref{append:proofs of main results} contains proofs of all theorems and propositions. Section \ref{append:proofs of primary lemmas} provides detailed proofs of the primary lemmas from Section \ref{append:Primary lemmas}. Section \ref{append:technical lemmas} collects technical lemmas and provides corresponding proofs for stationary VAR processes that are of independent interest and broadly applicable. Section \ref{sec:ADMM} presents the ADMM algorithm for solving the local estimation problem \eqref{eq:local-estimator}. Finally, Section \ref{sec:add-tables} provides additional data description tables for the empirical studies in Section \ref{sec:RealData}.

	\setcounter{lemma}{0}
	\setcounter{equation}{0}

\section{Notation and Preliminaries}\label{append:notation and preliminaries}
	Throughout the appendix, vectors and matrices are denoted by boldface lowercase and uppercase letters, respectively (e.g., $\bbm a$, $\bbm \xi$, $\bm A$, $\bm \Delta$). For a vector $\bbm{a}$, its Euclidean norm is denoted by $\|\bbm{a}\|_2$. For a matrix $\bm A \in \mathbb R^{m \times n}$, we denote its transpose, Frobenius norm, nuclear norm, and operator norm by $\bm A^\top$, $\|\bm A\|_\F$, $\|\bm A\|_*$, and $\|\bm A\|_{\mathrm{op}}$, respectively. 
	For $q \in [1,\infty)$, the $\ell_q$ norm of $\bm A$ is defined as $\|\bm A\|_q =( \sum_{i=1}^m \sum_{j=1}^n |A_{ij}|^q)^{1/q}$; the same formula for $q \in (0,1)$ defines the entrywise $\ell_q$ quasi-norm. When $q=0$, $\|\bm A\|_0=\#\{(i,j): A_{ij} \neq 0\}$ is the number of nonzero entries of $\bm A$. When $q=\infty$, $\|\bm A\|_\infty = \max_{i,j} |A_{ij}|$ is the entrywise maximum norm.
	Moreover, let $\mathrm{SVD}_r(\bm A) = \sum_{i=1}^r \sigma_i \bbm u_i \bbm v_i^\top$ denote the best rank-$r$ approximation of $\bm A$ obtained from its singular value decomposition, where $\sigma_i$ is the $i$-th singular value, and $\bbm u_i$ and $\bbm v_i$ are the corresponding left and right singular vectors.

	For any subspace $\mathcal S$, let $\overline{\mathcal S}$ denote its closure and $\overline{\mathcal S}^\perp$ its orthogonal complement. For a rank-$r$ matrix $\bm A \in \mathbb R^{n_1 \times n_2}$ with singular value decomposition $\bm A = \bm U \bm\Sigma \bm V^\top$, its tangent space is defined as $\mathcal T_r(\bm A) = \bigl\{\bm \Xi = \bm U \bm R^\top + \bm L \bm V^\top: \bm R \in \mathbb R^{n_2 \times r}, \bm L \in \mathbb R^{n_1 \times r}\bigr\}$, and its orthogonal complement (with respect to the Frobenius inner product) is $\mathcal T_r^\perp(\bm A) = \{ \bm \Xi \in \mathbb R^{n_1 \times n_2} : \bm U^\top \bm \Xi = \bm 0,\ \bm \Xi \bm V = \bm 0 \}$.
	The projection of $\bm B \in \mathbb R^{n_1 \times n_2}$ onto the tangent space $\mathcal T_r(\bm A)$ is denoted as $\mathcal P_{\mathcal T_r(\bm A)}(\bm B) = \bm U \bm U^\top \bm B + \bm B \bm V \bm V^\top - \bm U \bm U^\top \bm B \bm V \bm V^\top$.

	Define the following spaces of square-summable sequences of real vectors:
	$$
		\ell_2 = \left\{(\ldots,\bbm x_{-1}^\top,\bbm x_0^\top,\bbm x_1^\top,\ldots)^\top : \sum_{j\in\mathbb Z} \|\bbm x_j\|_2^2 < \infty\right\},
		\ \text{ and } \ 
		\ell_2^+ = \left\{(\bbm x_0^\top,\bbm x_1^\top,\ldots)^\top : \sum_{j=0}^\infty \|\bbm x_j\|_2^2 < \infty\right\}.
	$$
	Both spaces are equipped with the natural inner product and form Hilbert spaces of square-summable real vector-valued sequences. 
	
	For a random variable $X$ or random vector $\bbm \zeta$, $\|X\|_{\psi_2}$ and $\|\bbm \zeta\|_{\psi_2}$ denote their sub-Gaussian Orlicz norms, respectively. For sequences $\{x_n\}$ and $\{y_n\}$, we write $x_n\gtrsim y_n$ if there exists a constant $C>0$ such that $x_n\geq C y_n$ for all $n$, and $x_n\asymp y_n$ if both $x_n\gtrsim y_n$ and $y_n\gtrsim x_n$ hold.
	Throughout the appendix, $C, C', C_1, c$ and similar symbols denote positive constants whose values may change from line to line unless otherwise specified.

	\section{Primary Lemmas}\label{append:Primary lemmas}

	\setcounter{lemma}{0}
	\setcounter{equation}{0}

	In this section, we present primary lemmas that form the theoretical foundation for our main results. Lemmas~\ref{lem:subg-l2-max-sqrtds} and \ref{lem:gaussian-opnorm-hp-strong} provide high-probability bounds for the $\ell_2$-norm of sub-Gaussian vectors and the spectral norm of Gaussian random matrices, respectively. Lemmas~\ref{lem:normal-second-order} and \ref{lem:tangent-contraction-op} are key technical results that characterize the geometry of low-rank matrix spaces and the contraction properties of projections onto tangent spaces. Lemma~\ref{lem:matrix-pert-lesssim} is a matrix perturbation result that controls the error of low-rank approximations under perturbations. Finally, Lemmas~\ref{lem:var-cov-conc} and \ref{lem:RSC-var} establish concentration inequalities for sample covariance matrices and restricted strong convexity conditions for VAR processes, which are crucial for our error analysis. Detailed proofs of these lemmas are provided in Section \ref{append:proofs of primary lemmas}.

	% ============================================================
	\begin{lemma}[Sub-Gaussian $\ell_2$-norm tail bound]\label{lem:subg-l2-max-sqrtds}
		Let $\bbm\xi_{t}\in\mathbb R^d$ be mean-zero and sub-Gaussian with $\|\bbm\zeta_{t}\|_{\psi_2} \leq K_\zeta$. Then for $T' \leq T$, there exists a constant $C>0$ such that
		\begin{align}\label{eq:max-eps-subg-highprob-sqrtds}
			\mathbb P\left(\max_{1\leq t \leq T'}\|\bbm\zeta_{t}\|_2 \leq C K_\zeta\sqrt{d+\log T}\right) \geq 1- \exp(-C\log T).
		\end{align}
	\end{lemma}

	% ============================================================
	\begin{lemma}[High-probability spectral-norm bound for an $i.i.d.$ Gaussian matrix]\label{lem:gaussian-opnorm-hp-strong}
	Let $\bm Z\in\mathbb R^{d\times pd}$ have $i.i.d.$ entries $Z_{ij}\sim \mathcal N(0,\sigma^2)$.
	Then there exist constants $c,C>0$ such that
	\[
		\mathbb P\left(\|\bm Z\|_{\op}\leq c\sigma\bigl(\sqrt{pd}+\sqrt{\log T}\bigr)\right) \geq 1-\exp(-C\log T).
	\]
	\end{lemma}

	% ============================================================
	\begin{lemma}[Adapted from \cite{wei2016guarantees}, Lemma~4.1]\label{lem:normal-second-order}
		Let $\bm A\in\mathbb R^{d\times pd}$ have rank $r$.
		Let $\bm A^*\in\mathbb R^{d\times pd}$ have rank $r$ with $\sigma_r(\bm A^*)>0$. Then
		\[
		\bigl\|\mathcal P_{\mathcal T_r^\perp(\bm A)}(\bm A-\bm A^*)\bigr\|_{\F} = \bigl\|\mathcal P_{\mathcal T_r^\perp(\bm A)}(\bm A^*)\bigr\|_{\F} \leq \frac{1}{\sigma_r(\bm A^*)}\|\bm A-\bm A^*\|_{\op}\|\bm A-\bm A^*\|_{\F},
		\]
		where $\mathcal T_r^\perp(\bm A)$ denotes the orthogonal complement of the tangent space $\mathcal T_r(\bm A)$ at $\bm A$.
	\end{lemma}

	% ============================================================
	\begin{lemma}\label{lem:tangent-contraction-op}
		Let $\bm S\in\mathbb R^{pd\times pd}$ be symmetric. Fix any rank-$r$ matrix $\bm A\in\mathbb R^{d\times pd}$. Then for every $\bm M\in\mathcal T_r(\bm A)$,
		\begin{align}\label{eq:tangent-contraction-core-op}
			\Big\|\bm M-2\rho\,\mathcal P_{\mathcal T_r(\bm A)}\!\big(\bm M\bm S\big)\Big\|_{\F} \leq \|\bm I_{pd}-2\rho\,\bm S\|_{\op}\,\|\bm M\|_{\F}.
		\end{align}
	\end{lemma}

	% ============================================================
	\begin{lemma}[Matrix perturbation; Adapted from \cite{shen2025compu}, Lemma 14]\label{lem:matrix-pert-lesssim}
		Suppose $\bm M\in\mathbb R^{d_1\times d_2}$ has rank $r$ and admits the singular value decomposition $\bm M=\bm U\bm\Sigma\bm V^\top$, $\bm\Sigma=\mathrm{diag}(\sigma_1,\sigma_2,\ldots,\sigma_r)$, and $\sigma_1\geq \sigma_2\geq \cdots \geq \sigma_r>0$. For any $\bm M'\in\mathbb R^{d_1\times d_2}$ satisfying $\|\bm M'-\bm M\|_{\F}\leq \sigma_r/8$,
		\[
			\bigl\|\mathrm{SVD}_r(\bm M')-\bm M\bigr\|_{\F} \leq \|\bm M'-\bm M\|_{\F} + \frac{40\|\bm M'-\bm M\|_{\op}\|\bm M'-\bm M\|_{\F}}{\sigma_r}.
		\]
	\end{lemma}

	% ============================================================
	\begin{lemma}\label{lem:var-cov-conc}
		Suppose Assumptions~\ref{assump:stationarity} and \ref{assump:subg} hold. Denote $\bm X_k = [\bbm x_{k,T_k}, \bbm x_{k,T_k-1}, \ldots, \bbm x_{k,1}]^\top\in\mathbb R^{T_k\times pd}$ and define $\bbm \Sigma_{x,k} = \mathbb E\bbm x_{k,t}\bbm x_{k,t}^\top$. If $T_k\gtrsim p^2 d$, then there exists a constant $C>0$ such that, with probability at least $1-\exp(-Cpd)$,
		\begin{align}\label{eq:var-cov-op-conc}
			\left\|\frac{1}{T_k}\bm X_k\bm X_k^\top - \bbm \Sigma_{x,k}\right\|_{\op} \lesssim C_{\epsilon,\mathcal A}^{(k)}\sigma^2p\sqrt{\frac{d}{T_k}},
		\end{align}
		where $C_{\epsilon,\mathcal A}^{(k)} = \lambda_{\max}(\bbm{\Sigma}_{\epsilon,k})\mu_{\min}^{-1}(\mathcal{A}_k)$.
	\end{lemma}

	% ============================================================
	\begin{lemma}\label{lem:RSC-var}
		Suppose Assumptions~\ref{assump:stationarity} and \ref{assump:subg} hold. Denote $\bm X_k = [\bbm x_{k,T_k}, \bbm x_{k,T_k-1}, \ldots, \bbm x_{k,1}]^\top\in\mathbb R^{T_k\times pd}$ and $\bm E_k = [\bbm\epsilon_{k,T_k}, \bbm\epsilon_{k,T_k-1}, \ldots, \bbm\epsilon_{k,1}]^\top\in\mathbb R^{T_k\times d}$. If $T_k\gtrsim(C_{\epsilon,\mathcal A}^{(k)}/C_{\mathrm{RSC}}^{(k)} \vee 1)^2p^2d$ with $C_{\epsilon,\mathcal A}^{(k)} = \lambda_{\max}(\bbm{\Sigma}_{\epsilon,k})\mu_{\min}^{-1}(\mathcal{A}_k)$ and $C_{\mathrm{RSC}}^{(k)} = c_{\epsilon,\mathcal A}^{(k)}/2 = \lambda_{\min}(\bbm{\Sigma}_{\epsilon,k})\mu_{\max}^{-1}(\mathcal{A}_k)/2$, then there exists a constant $C>0$ such that, with probability at least $1-\exp(-C pd)$,
		\begin{align}\label{eq:RSC-goal}
			\frac{1}{T_k}\sum_{t=1}^{T_k}\|\bm\Delta\bbm x_{k,t}\|_2^2 \geq C_{\mathrm{RSC}}^{(k)}\|\bm\Delta\|_{\mathrm F}^2, \quad \forall \bm\Delta\in\mathbb R^{d\times pd}.
		\end{align}
	\end{lemma}

	% ============================================================
	\begin{lemma}\label{lem:var-XtE-infty-bound}
		Suppose Assumptions~\ref{assump:stationarity} and \ref{assump:subg} hold. Denote $\bm X_k = [\bbm x_{k,T_k}, \bbm x_{k,T_k-1}, \ldots, \bbm x_{k,1}]^\top\in\mathbb R^{T_k\times pd}$ and $\bm E_k = [\bbm\epsilon_{k,T_k}, \bbm\epsilon_{k,T_k-1}, \ldots, \bbm\epsilon_{k,1}]^\top\in\mathbb R^{T_k\times d}$. If $T_k\gtrsim p\log(pd)$, then there exists a constant $C>0$ such that, with probability at least $1-\exp(-C\log(pd))$,
		\begin{align}\label{eq:var-XtE-infty-bound-simplified}
			\Bigl\|\frac{1}{T_k}\bm X_k^{\top}\bm E_k\Bigr\|_{\infty}\lesssim C_{\epsilon,\mathcal A}^{(k)}\sigma^2\sqrt{\frac{\log(pd)}{T_k}},
		\end{align}
		where $C_{\epsilon,\mathcal A}^{(k)} = \lambda_{\max}(\bbm{\Sigma}_{\epsilon,k})\mu_{\min}^{-1}(\mathcal{A}_k)$.
	\end{lemma}

	% ============================================================
	\begin{lemma}\label{lem:var-XtE-op-bound}
		Suppose Assumptions~\ref{assump:stationarity} and \ref{assump:subg} hold. Denote $\bm X_k = [\bbm x_{k,T_k}, \bbm x_{k,T_k-1}, \ldots, \bbm x_{k,1}]^\top\in\mathbb R^{T_k\times pd}$ and $\bm E_k = [\bbm\epsilon_{k,T_k}, \bbm\epsilon_{k,T_k-1}, \ldots, \bbm\epsilon_{k,1}]^\top\in\mathbb R^{T_k\times d}$. If $T_k\gtrsim p^2 d$, then there exists a constant $C>0$ such that, with probability at least $1-\exp(-Cpd)$,
		\begin{align}\label{eq:var-XtE-op-bound-simplified}
			\left\|\frac{1}{T_k}\bm X_k^{\top}\bm E_k\right\|_{\op} \lesssim C_{\epsilon,\mathcal A}^{(k)}\sigma^2\sqrt{\frac{p d}{T_k}},
		\end{align}
		where $C_{\epsilon,\mathcal A}^{(k)} = \lambda_{\max}(\bbm{\Sigma}_{\epsilon,k})\mu_{\min}^{-1}(\mathcal{A}_k)$.
	\end{lemma}

	% ============================================================
	\begin{lemma}\label{lem:var-cov-conc-unified}
		Suppose Assumptions~\ref{assump:stationarity} and \ref{assump:subg} hold. For each client $k\in[K]$, denote $\bm X_k = [\bbm x_{k,T_k}, \bbm x_{k,T_k-1}, \ldots, \bbm x_{k,1}]^\top\in\mathbb R^{T_k\times pd}$ and $\bm S_k = T_k^{-1}\bm X_k^\top \bm X_k \in \mathbb R^{pd\times pd}$. Define the pooled sample covariance and its expectation as
		\[
			\bm S_{\rm pool} = \sum_{k=1}^K \frac{T_k}{T}\bm S_k = \frac{1}{T}\sum_{k=1}^K \bm X_k^\top \bm X_k,
			\ \text{ and } \ 
			\bbm\Sigma_{x,{\rm pool}} = \sum_{k=1}^K\frac{T_k}{T}\bbm\Sigma_{x,k},
		\]
		where $T=\sum_{k=1}^K T_k$. Let $C_{\epsilon,\mathcal A}^{\max}=\max_{k\in[K]} C_{\epsilon,\mathcal A}^{(k)}$ with $C_{\epsilon,\mathcal A}^{(k)} = \lambda_{\max}(\bbm{\Sigma}_{\epsilon,k})\mu_{\min}^{-1}(\mathcal{A}_k)$. If $T\gtrsim p^2 d$, then there exists a constant $C>0$ such that with probability at least $1-\exp{(-Cpd)}$,
		\[
			\left\|\bm S_{\rm pool}-\bbm\Sigma_{x,{\rm pool}}\right\|_{\op} \lesssim C_{\epsilon,\mathcal A}^{\max}\sigma^2p\sqrt{\frac{d}{T}}.
		\]
	\end{lemma}

	% ============================================================
	\begin{lemma}\label{lem:var-XtE-op-bound-pooled}
		Suppose Assumptions~\ref{assump:stationarity} and \ref{assump:subg} hold. For each client $k\in[K]$, denote $\bm X_k = [\bbm x_{k,T_k}, \bbm x_{k,T_k-1}, \ldots, \bbm x_{k,1}]^\top\in\mathbb R^{T_k\times pd}$, $\bm E_k = [\bbm\epsilon_{k,T_k}, \bbm\epsilon_{k,T_k-1}, \ldots, \bbm\epsilon_{k,1}]^\top\in\mathbb R^{T_k\times d}$ and $\bm R_k = T_k^{-1}\bm X_k^\top \bm E_k$. Define the pooled cross term as $\bm R_{\rm pool} = \sum_{k=1}^K (T_k/T) \bm R_k$.
		Let $C_{\epsilon,\mathcal A}^{\max}=\max_{k\in[K]} C_{\epsilon,\mathcal A}^{(k)}$ with $C_{\epsilon,\mathcal A}^{(k)}=\lambda_{\max}(\bbm{\Sigma}_{\epsilon,k})\mu_{\min}^{-1}(\mathcal{A}_k)$. If $T\gtrsim p^2 d$, then there exists a constant $C>0$ such that with probability at least $1-\exp(-Cpd)$,
		\[
			\left\|\bm R_{\rm pool}\right\|_{\op} \lesssim C_{\epsilon,\mathcal A}^{\max}\sigma^2\sqrt{\frac{p d}{T}}.
		\]
	\end{lemma}

	% ============================================================
	\begin{lemma}\label{lem:bernstein-tail-mgf-equiv}
		Let $X$ be a mean-zero random variable. Fix $\nu>0$ and $b>0$. The following two conditions are equivalent.
		\begin{enumerate}[label=(\roman*)]
			\item \textbf{(Two-regime Bernstein tail)} There exists a constant $c>0$ such that for all $t>0$,
			\begin{align}\label{eq:bern-tail}
				\mathbb P\big(|X|\geq t\big)\leq 2\exp\!\left(-c\min\Big\{\frac{t^2}{\nu^2},\frac{t}{b}\Big\}\right).
			\end{align}
			\item \textbf{(Local sub-exponential mgf)} $X$ is sub-exponential with Bernstein parameters $(\nu,b)$, i.e., there exist constants $c,C>0$ such that for all $|\lambda|\leq c/b$,
			\begin{align}\label{eq:bern-mgf}
				\mathbb E\exp(\lambda X)\leq \exp\!\big(C\lambda^2\nu^2\big).
			\end{align}
		\end{enumerate}
	\end{lemma}

	% ============================================================
	\begin{lemma}[Gaussian mechanism for propagation DP]\label{lem:gaussian-mech-prop}
		Fix a client $k\in[K]$ and consider a deterministic mapping $\bbm g:\mathcal Z^{T_k}\to\mathbb R^{d\times pd}$. Assume that for all neighboring datasets $\mathcal D_k\sim\mathcal D_k'$ with $\|\bbm g(\mathcal D_k)-\bbm g(\mathcal D_k')\|_{\F} \leq \Delta$. Define the randomized mechanism $M:\mathcal Z^{T_k}\to\mathbb R^{d\times pd}$ by $M(\mathcal D_k) = \bbm g(\mathcal D_k) + \bm Z_k$, where $\bm Z_k\in\mathbb R^{d\times pd}$ has independent Gaussian entries $\mathcal N(0,\sigma^2)$ and is independent of $\mathcal D_k$.
		Then $M$ is $(\varepsilon,\delta)$-DP if
		\begin{align}\label{eq:gaussian-sigma-choice}
		\sigma^2 \geq \frac{2\Delta^2 \log(1.25/\delta)}{\varepsilon^2}.
		\end{align}
	\end{lemma}

	\section{Proofs of Main Results}\label{append:proofs of main results}

	\begin{proof}[\textbf{Proof of Proposition \ref{thm:local-model-error-upper-bound}}]
		Fix an arbitrary client $k\in[K]$. By construction, both $(\widetilde{\bm A}_{0,k},\widetilde{\bbm \Delta}_k)$ and $(\bm A_0^*,\bbm \Delta_k^*)$ satisfy the constraint $\|\bbm \Delta_{k}\|_\op\leq \zeta$ and hence are feasible for \eqref{eq:local-estimator}. Therefore, by the optimality of $(\widetilde{\bm A}_{0,k},\widetilde{\bbm \Delta}_k)$, we have
		\begin{align*}
			&\frac{1}{T_k}\sum_{t=1}^{T_k}
			\Bigl\|\bbm y_{k,t}-\bigl(\widetilde{\bm A}_{0,k}+\widetilde{\bbm \Delta}_k\bigr)\bbm x_{k,t}\Bigr\|_{2}^{2}
			+\lambda_k\|\widetilde{\bm A}_{0,k}\|_{*}
			+\omega_k\|\widetilde{\bbm \Delta}_k\|_{1} \notag \\
			&\qquad \qquad\qquad\qquad\qquad\qquad \leq
			\frac{1}{T_k}\sum_{t=1}^{T_k}
			\Bigl\|\bbm y_{k,t}-\bigl(\bm A_0^*+\bbm \Delta_k^*\bigr)\bbm x_{k,t}\Bigr\|_{2}^{2}
			+\lambda_k\|\bm A_0^*\|_{*}
			+\omega_k\|\bbm \Delta_k^*\|_{1}.
		\end{align*}
		Rearranging the terms yields
		\begin{align*}
			&\frac{1}{T_k}\sum_{t=1}^{T_k}\left(
			\Bigl\|\bbm y_{k,t}-\bigl(\widetilde{\bm A}_{0,k}+\widetilde{\bbm \Delta}_k\bigr)\bbm x_{k,t}\Bigr\|_{2}^{2} - \Bigl\|\bbm y_{k,t}-\bigl(\bm A_0^*+\bbm \Delta_k^*\bigr)\bbm x_{k,t}\Bigr\|_{2}^{2}\right) \notag \\
			&\qquad \qquad\qquad\qquad\qquad\qquad \leq
			\lambda_k\left(\|\bm A_0^*\|_{*} - \|\widetilde{\bm A}_{0,k}\|_{*}\right)
			+\omega_k\left(\|\bbm \Delta_k^*\|_{1} - \|\widetilde{\bbm \Delta}_k\|_{1}\right).
		\end{align*}
		Considering $\|\bbm a\|_2^2 - \|\bbm b\|_2^2 = \|\bbm a - \bbm b\|_2^2 + 2\langle\bbm a - \bbm b,\bbm b\rangle$ for any vectors $\bbm a$, $\bbm b$, and note that $\bbm y_{k,t} = (\bm A_0^* + \bbm \Delta_k^*)\bbm x_{k,t} + \bbm \epsilon_{k,t}$, we have
		\begin{align}\label{eq:local-rearrage-loss-real}
			&\frac{1}{T_k}\sum_{t=1}^{T_k}\left(\left\|\Bigl(\bigl(\widetilde{\bm A}_{0,k}+\widetilde{\bbm \Delta}_k\bigr)-\bigl(\bm A_0^* + \bbm \Delta_k^*\bigr)\Bigr)\bbm x_{k,t}\right\|_2^2 - 2 \langle\bigl(\widetilde{\bm A}_{0,k} - \bm A_0^*\bigr)\bbm x_{k,t},\bbm \epsilon_{k,t}\rangle - 2 \langle\bigl(\widetilde{\bbm \Delta}_{k} - \bbm \Delta_k^*\bigr)\bbm x_{k,t},\bbm \epsilon_{k,t}\rangle\right) \notag \\
			&\qquad \qquad\qquad\qquad\qquad\qquad \leq
			\lambda_k\left(\|\bm A_0^*\|_{*} - \|\widetilde{\bm A}_{0,k}\|_{*}\right)
			+\omega_k\left(\|\bbm \Delta_k^*\|_{1} - \|\widetilde{\bbm \Delta}_k\|_{1}\right).
		\end{align}
		For the inner product terms on the left-hand side, we have 
		\[	
			\langle\bigl(\widetilde{\bm A}_{0,k} - \bm A_0^*\bigr)\bbm x_{k,t},\bbm \epsilon_{k,t}\rangle = \tr\left(\bbm \epsilon_{k,t}^\top(\widetilde{\bm A}_{0,k} - \bm A_0^*)\bbm x_{k,t}\right) = \tr\left((\widetilde{\bm A}_{0,k} - \bm A_0^*)\bbm x_{k,t}\bbm \epsilon_{k,t}^\top\right) = \langle\widetilde{\bm A}_{0,k} - \bm A_0^*, \bbm \epsilon_{k,t}\bbm x_{k,t}^\top\rangle,
		\]
		and similarly, $\langle\bigl(\widetilde{\bbm \Delta}_{k} - \bbm \Delta_k^*\bigr)\bbm x_{k,t},\bbm \epsilon_{k,t}\rangle = \langle\widetilde{\bbm \Delta}_{k} - \bbm \Delta_k^*, \bbm \epsilon_{k,t}\bbm x_{k,t}^\top\rangle$. Then, rearranging the terms in \eqref{eq:local-rearrage-loss-real} and applying the H\"older's inequality yields
		\begin{align}\label{eq:local-model-error-basic-inequality}
			&\frac{1}{T_k}\sum_{t=1}^{T_k}\left(\left\|\Bigl(\bigl(\widetilde{\bm A}_{0,k}+\widetilde{\bbm \Delta}_k\bigr)-\bigl(\bm A_0^* + \bbm \Delta_k^*\bigr)\Bigr)\bbm x_{k,t}\right\|_2^2\right) \leq \lambda_k\left(\|\bm A_0^*\|_{*} - \|\widetilde{\bm A}_{0,k}\|_{*}\right)
			+\omega_k\left(\|\bbm \Delta_k^*\|_{1} - \|\widetilde{\bbm \Delta}_k\|_{1}\right) \notag \\
			&\qquad\qquad +
			2\left(\|\widetilde{\bm A}_{0,k} - \bm A_0^*\|_{*}\left\|\frac{1}{T_k}\sum_{t=1}^{T_k}\bbm x_{k,t}\bbm \epsilon_{k,t}^\top\right\|_{\op} + \|\widetilde{\bbm \Delta}_{k} - \bbm \Delta_k^*\|_{1}\left\|\frac{1}{T_k}\sum_{t=1}^{T_k}\bbm x_{k,t}\bbm \epsilon_{k,t}^\top\right\|_{\infty}\right).
		\end{align}
		To bound the right-hand side of \eqref{eq:local-model-error-basic-inequality}, we begin by deriving the constrained space for the estimators $\widetilde{\bm A}_{0,k}$ and $\widetilde{\bbm \Delta}_k$. To this end, we define the coordinate subspace induced by the large entries of $\bbm \Delta_k^*$. Specifically, for a given threshold $\kappa>0$, let
		\begin{align}\label{eq:S-space}
			\mathcal S_\kappa = \Bigl\{\bm S\in\mathbb R^{d\times pd}:S_{ij}=0 \text{ for all }(i,j)\in[d]\times[pd] \text{ such that }|(\bbm \Delta_k^*)_{ij}|<\kappa\Bigr\},
		\end{align}
		and let $\overline{\mathcal S}_\kappa^{\perp}$ denote the closed orthogonal complement of $\mathcal S_\kappa$, namely
		\begin{align}\label{eq:S-complement-space}
			\overline{\mathcal S}_\kappa^{\perp} = \Bigl\{\bm S\in\mathbb R^{d\times pd}:S_{ij}=0 \text{ for all }(i,j)\in[d]\times[pd] \text{ such that }|(\bbm \Delta_k^*)_{ij}|\geq\kappa\Bigr\}.
		\end{align}
		For any $\bm S\in\mathbb R^{d\times pd}$, we write $\bm S_{\mathcal S_\kappa}$ for the restriction of $\bm S$ to the coordinates with $|(\bbm \Delta_k^*)_{ij}|\geq\kappa$, i.e., $(\bm S_{\mathcal S_\kappa})_{ij}=S_{ij}$ if $|(\bbm \Delta_k^*)_{ij}|\geq \kappa$ and $(\bm S_{\mathcal S_\kappa})_{ij}=0$ otherwise, and define $\bm S_{\overline{\mathcal S}_\kappa^{\perp}}$ analogously.

		We also introduce the low-rank model subspace associated with $\bm A_0^*$. Let $\bm A_0^*=\bm U^*\bm \Sigma^*\bm V^{*\top}$ be an SVD and denote by $\bm U_r^*\in\mathbb R^{d\times r}$ and $\bm V_r^*\in\mathbb R^{pd\times r}$ the matrices collecting the first $r$ left and right singular vectors of $\bm A_0^*$, respectively. Define
		\[
			\mathcal L_r = \Bigl\{\bm M\in\mathbb R^{d\times pd}: \mathrm{col}(\bm M)\subseteq \mathrm{col}(\bm U_r^*), \text{ and } \mathrm{col}(\bm M^\top)\subseteq \mathrm{col}(\bm V_r^*)\Bigr\},
		\]
		and let $\overline{\mathcal L}_r^{\perp}$ denote the closed orthogonal complement of $\mathcal L_r$, that is,
		\[
			\overline{\mathcal L}_r^{\perp} = \Bigl\{\bm M\in\mathbb R^{d\times pd}: \mathrm{col}(\bm M)\perp \mathrm{col}(\bm U_r^*), \text{ and } \mathrm{col}(\bm M^\top)\perp \mathrm{col}(\bm V_r^*)\Bigr\}.
		\]
		For any $\bm M\in\mathbb R^{d\times pd}$, we write $\bm M_{\mathcal L_r}$ for the orthogonal projection of $\bm M$ onto $\mathcal L_r$, and write $\bm M_{\overline{\mathcal L}_r^{\perp}}$ for the orthogonal projection onto $\overline{\mathcal L}_r^{\perp}$. Then, by the decomposability of the nuclear norm and the $\ell_1$ norm, for the first two terms on the right-hand side of \eqref{eq:local-model-error-basic-inequality}, we have
		\begin{align}\label{eq:decomposable-norms}
			&\lambda_k\bigl(\|\bm A_0^*\|_{*}-\|\widetilde{\bm A}_{0,k}\|_{*}\bigr)
			+\omega_k\bigl(\|\bbm \Delta_k^*\|_{1}-\|\widetilde{\bbm \Delta}_k\|_{1}\bigr) \notag \\
			&\quad=
			\lambda_k\Bigl(\|\bm A_0^*\|_{*}-\bigl\|(\widetilde{\bm A}_{0,k}-\bm A_0^*)+ \bm A_0^*\bigr\|_{*}\Bigr)
			+\omega_k\Bigl(\|\bbm \Delta_k^*\|_{1}-\bigl\|(\widetilde{\bbm \Delta}_k-\bbm \Delta_k^*)+\bbm \Delta_k^*\bigr\|_{1}\Bigr) \notag \\
			&\quad=
			\lambda_k\Bigl(\bigl\|(\bm A_0^*)_{\mathcal L_r}+(\bm A_0^*)_{\overline{\mathcal L}_r^{\perp}}\bigr\|_{*}
			-\bigl\|(\widetilde{\bm A}_{0,k}-\bm A_0^*)_{\mathcal L_r}
			+(\widetilde{\bm A}_{0,k}-\bm A_0^*)_{\overline{\mathcal L}_r^{\perp}}
			+(\bm A_0^*)_{\mathcal L_r}+(\bm A_0^*)_{\overline{\mathcal L}_r^{\perp}}\bigr\|_{*}
			\Bigr) \notag \\
			&\qquad\quad
			+\omega_k\Bigl(\bigl\|(\bbm \Delta_k^*)_{\mathcal S_\kappa}+(\bbm \Delta_k^*)_{\mathcal S_\kappa^{\perp}}\bigr\|_{1}
			-\bigl\|(\widetilde{\bbm \Delta}_k-\bbm \Delta_k^*)_{\mathcal S_\kappa}
			+(\widetilde{\bbm \Delta}_k-\bbm \Delta_k^*)_{\mathcal S_\kappa^{\perp}}
			+(\bbm \Delta_k^*)_{\mathcal S_\kappa}+(\bbm \Delta_k^*)_{\mathcal S_\kappa^{\perp}}\bigr\|_{1}
			\Bigr) \notag \\
			&\quad\leq
			\lambda_k\Bigl(\bigl\|(\bm A_0^*)_{\mathcal L_r}+(\bm A_0^*)_{\overline{\mathcal L}_r^{\perp}}\bigr\|_{*}
			-\bigl\|(\bm A_0^*)_{\mathcal L_r}
			+(\widetilde{\bm A}_{0,k}-\bm A_0^*)_{\overline{\mathcal L}_r^{\perp}}\bigr\|_{*}
			+\bigl\|(\bm A_0^*)_{\overline{\mathcal L}_r^{\perp}}
			+(\widetilde{\bm A}_{0,k}-\bm A_0^*)_{\mathcal L_r}\bigr\|_{*}
			\Bigr) \notag \\
			&\qquad\quad
			+\omega_k\Bigl(\bigl\|(\bbm \Delta_k^*)_{\mathcal S_\kappa}+(\bbm \Delta_k^*)_{\mathcal S_\kappa^{\perp}}\bigr\|_{1}
			-\bigl\|(\bbm \Delta_k^*)_{\mathcal S_\kappa}
			+(\widetilde{\bbm \Delta}_k-\bbm \Delta_k^*)_{\mathcal S_\kappa^{\perp}}\bigr\|_{1}
			+\bigl\|(\bbm \Delta_k^*)_{\mathcal S_\kappa^{\perp}}
			+(\widetilde{\bbm \Delta}_k-\bbm \Delta_k^*)_{\mathcal S_\kappa}\bigr\|_{1}
			\Bigr) \notag \\
			&\quad=
			\lambda_k\Bigl(
			\bigl\|(\widetilde{\bm A}_{0,k}-\bm A_0^*)_{\mathcal L_r}\bigr\|_{*}
			-\bigl\|(\widetilde{\bm A}_{0,k}-\bm A_0^*)_{\overline{\mathcal L}_r^{\perp}}\bigr\|_{*}
			+2\bigl\|(\bm A_0^*)_{\overline{\mathcal L}_r^{\perp}}\bigr\|_{*}\Bigr) \notag \\
			&\qquad\quad
			+\omega_k\Bigl(
			\bigl\|(\widetilde{\bbm \Delta}_k-\bbm \Delta_k^*)_{\mathcal S_\kappa}\bigr\|_{1}
			-\bigl\|(\widetilde{\bbm \Delta}_k-\bbm \Delta_k^*)_{\mathcal S_\kappa^{\perp}}\bigr\|_{1}
			+2\bigl\|(\bbm \Delta_k^*)_{\mathcal S_\kappa^{\perp}}\bigr\|_{1}\Bigr).
		\end{align}
		The last equality follows from the fact that the nuclear norm and the $\ell_1$ norm are both decomposable with respect to the defined subspaces, i.e., $\|\bm M_{\mathcal L_r} + \bm N_{\overline{\mathcal L}_r^{\perp}}\|_{*} = \|\bm M_{\mathcal L_r}\|_{*} + \|\bm N_{\overline{\mathcal L}_r^{\perp}}\|_{*}$ and $\|\bm M_{\mathcal S_\kappa} + \bm N_{\mathcal S_\kappa^{\perp}}\|_{1} = \|\bm M_{\mathcal S_\kappa}\|_{1} + \|\bm N_{\mathcal S_\kappa^{\perp}}\|_{1}$. 
		
		Next, we bound the deviation terms. By Lemmas~\ref{lem:var-XtE-infty-bound} and \ref{lem:var-XtE-op-bound}, there exist constants $C_{\op},C_{\infty}>0$ such that, with probability at least $1 - C\exp(-C pd) - C\exp(-C\log(pd))$,
		\[
			\left\|\frac{1}{T_k}\sum_{t=1}^{T_k}\bbm x_{k,t}\bbm\epsilon_{k,t}^{\top}\right\|_{\op} \leq C_{\op} C_{\epsilon,\mathcal A}^{(k)}\sigma^2\sqrt{\frac{pd}{T_k}},
			\ \text{ and } \ 
			\left\|\frac{1}{T_k}\sum_{t=1}^{T_k}\bbm x_{k,t}\bbm\epsilon_{k,t}^{\top}\right\|_{\infty} \leq C_{\infty} C_{\epsilon,\mathcal A}^{(k)}\sigma^2\sqrt{\frac{\log(pd)}{T_k}}.
		\]
		Consistent with the prescribed rates of $\lambda_k$ and $\omega_k$, we choose the tuning parameters with sufficiently large constants, namely $\lambda_k \geq 4C_{\op}C_{\epsilon,\mathcal A}^{(k)}\sigma^2\sqrt{pd/T_k}$ and $\omega_k \geq 4C_{\infty}\, C_{\epsilon,\mathcal A}^{(k)}\sigma^2\sqrt{\log(pd)/T_k}$. Then, on the same event,
		\begin{align}\label{eq:deviation-bounds-choice-of-lambda-omega}
			\left\|\frac{1}{T_k}\sum_{t=1}^{T_k}\bbm x_{k,t}\bbm\epsilon_{k,t}^{\top}\right\|_{\op}\leq \frac{\lambda_k}{4},
			\ \text{ and } \ 
			\left\|\frac{1}{T_k}\sum_{t=1}^{T_k}\bbm x_{k,t}\bbm\epsilon_{k,t}^{\top}\right\|_{\infty}\leq \frac{\omega_k}{4}.
		\end{align}
		For the last two terms on the right-hand side of \eqref{eq:local-model-error-basic-inequality}, using these bounds and the decomposability of the norms again, we have, with probability at least $1 - C\exp(-C pd) - C\exp(-C\log(pd))$,
		\begin{align}\label{eq:local-model-error-deviation-bound}
			&2\left(\|\widetilde{\bm A}_{0,k} - \bm A_0^*\|_{*}\left\|\frac{1}{T_k}\sum_{t=1}^{T_k}\bbm x_{k,t}\bbm \epsilon_{k,t}^\top\right\|_{\op}
			+\|\widetilde{\bbm \Delta}_{k} - \bbm \Delta_k^*\|_{1}\left\|\frac{1}{T_k}\sum_{t=1}^{T_k}\bbm x_{k,t}\bbm \epsilon_{k,t}^\top\right\|_{\infty}\right) \\
			&\quad\leq 
			\frac{\lambda_k}{2}\left(\|(\widetilde{\bm A}_{0,k} - \bm A_0^*)_{\mathcal L_r}\|_{*} + \|(\widetilde{\bm A}_{0,k} - \bm A_0^*)_{\overline{\mathcal L}_r^{\perp}}\|_{*}\right) + \frac{\omega_k}{2}\left(\|(\widetilde{\bbm \Delta}_{k} - \bbm \Delta_k^*)_{\mathcal S_\kappa}\|_{1} + \|(\widetilde{\bbm \Delta}_{k} - \bbm \Delta_k^*)_{\mathcal S_\kappa^{\perp}}\|_{1}\right) \notag.
		\end{align}
		Combining \eqref{eq:local-model-error-basic-inequality}, \eqref{eq:decomposable-norms}, and \eqref{eq:local-model-error-deviation-bound} yields
		\begin{align*}
			0 &\leq
			\frac{1}{T_k}\sum_{t=1}^{T_k}\left(\left\|\Bigl(\bigl(\widetilde{\bm A}_{0,k}+\widetilde{\bbm \Delta}_k\bigr)-\bigl(\bm A_0^* + \bbm \Delta_k\bigr)\Bigr)\bbm x_{k,t}\right\|_2^2\right) \\
			&\leq
			\lambda_k\Bigl(\frac{3}{2}\|(\widetilde{\bm A}_{0,k}-\bm A_0^*)_{\mathcal L_r}\|_{*}
			-\frac{1}{2}\|(\widetilde{\bm A}_{0,k}-\bm A_0^*)_{\overline{\mathcal L}_r^{\perp}}\|_{*}
			+2\|(\bm A_0^*)_{\overline{\mathcal L}_r^{\perp}}\|_{*}\Bigr) \\
			&\quad +
			\omega_k\Bigl(\frac{3}{2}\|(\widetilde{\bbm \Delta}_k-\bbm \Delta_k^*)_{\mathcal S_\kappa}\|_{1}
			-\frac{1}{2}\|(\widetilde{\bbm \Delta}_k-\bbm \Delta_k^*)_{\mathcal S_\kappa^{\perp}}\|_{1}
			+2\|(\bbm \Delta_k^*)_{\mathcal S_\kappa^{\perp}}\|_{1}\Bigr).
		\end{align*}
		Rearranging the terms, we obtain the following cone-type inequality
		\begin{align}\label{eq:cone-constraint-Crk-correct}
		&\lambda_k\bigl\|\bigl(\widetilde{\bm A}_{0,k}-\bm A_0^*\bigr)_{\overline{\mathcal L}_r^{\perp}}\bigr\|_{*}
		+\omega_k\bigl\|\bigl(\widetilde{\bbm \Delta}_k-\bbm \Delta_k^*\bigr)_{\mathcal S_\kappa^{\perp}}\bigr\|_{1}\notag\\
		&\leq \lambda_k\left(3\bigl\|\bigl(\widetilde{\bm A}_{0,k}-\bm A_0^*\bigr)_{\mathcal L_r}\bigr\|_{*}
		+4\bigl\|\bigl(\bm A_0^*\bigr)_{\overline{\mathcal L}_r^{\perp}}\bigr\|_{*}\right)
		+\omega_k\left(3\bigl\|\bigl(\widetilde{\bbm \Delta}_k-\bbm \Delta_k^*\bigr)_{\mathcal S_\kappa}\bigr\|_{1}
		+4\bigl\|\bigl(\bbm \Delta_k^*\bigr)_{\mathcal S_\kappa^{\perp}}\bigr\|_{1}\right).
		\end{align}
		Next, we derive the upper bound for the estimation error in this constrained space. By just plugging in \eqref{eq:deviation-bounds-choice-of-lambda-omega} into the right-hand side of \eqref{eq:local-model-error-basic-inequality}, we have
		\begin{align}\label{eq:local-model-error-upper-bound-before-RSC}
			&\frac{1}{T_k}\sum_{t=1}^{T_k}\left(\left\|\Bigl(\bigl(\widetilde{\bm A}_{0,k}+\widetilde{\bbm \Delta}_k\bigr)-\bigl(\bm A_0^* + \bbm \Delta_k^*\bigr)\Bigr)\bbm x_{k,t}\right\|_2^2\right) \notag\\
			&\leq \lambda_k\left(\|\bm A_0^*\|_{*} - \|\widetilde{\bm A}_{0,k}\|_{*}\right) + \omega_k\left(\|\bbm \Delta_k^*\|_{1} - \|\widetilde{\bbm \Delta}_k\|_{1}\right) + \frac{\lambda_k}{2}\left(\|\widetilde{\bm A}_{0,k} - \bm A_0^*\|_{*}\right) + \frac{\omega_k}{2}\left(\|\widetilde{\bbm \Delta}_{k} - \bbm \Delta_k^*\|_{1}\right) \notag\\
			&\leq \frac{3\lambda_k}{2}\left(\|\widetilde{\bm A}_{0,k} - \bm A_0^*\|_{*}\right) + \frac{3\omega_k}{2}\left(\|\widetilde{\bbm \Delta}_{k} - \bbm \Delta_k^*\|_{1}\right).
		\end{align}
		Then, we derive the lower bound for the left-hand side of \eqref{eq:local-model-error-upper-bound-before-RSC}. This can be achieved by two steps. First, by Lemma~\ref{lem:RSC-var}, there exists a constant $C_{\mathrm{RSC}}^{(k)} = c_{\epsilon,\mathcal A}^{(k)}/2>0$ such that, given the sample size condition $T_k\gtrsim(C_{\epsilon,\mathcal A}^{(k)}/C_{\mathrm{RSC}}^{(k)} \vee 1)^2p^2d$, with probability at least $1 - C\exp(-C pd)$,
		\begin{align}\label{eq:RSC-var-lower-bound}
			\frac{1}{T_k}\sum_{t=1}^{T_k}\left\|\Bigl(\bigl(\widetilde{\bm A}_{0,k}+\widetilde{\bbm \Delta}_k\bigr)-\bigl(\bm A_0^* + \bbm \Delta_k^*\bigr)\Bigr)\bbm x_{k,t}\right\|_2^2 \geq C_{\mathrm{RSC}}^{(k)}\left\|\bigl(\widetilde{\bm A}_{0,k}+\widetilde{\bbm \Delta}_k\bigr)-\bigl(\bm A_0^* + \bbm \Delta_k^*\bigr)\right\|_{\mathrm F}^2.
		\end{align}
		Second, noting that $\|\bm M +\bm N\|_{\mathrm F}^2 \geq \|\bm M\|_{\mathrm F}^2 + \|\bm N\|_{\mathrm F}^2 - 2|\langle \bm M,\bm N\rangle|$ for any matrices $\bm M,\bm N$ of the same dimension, we have
		\begin{align}\label{eq:F-norm-decomposition}
			&C_{\mathrm{RSC}}^{(k)}\bigl\|\bigl(\widetilde{\bm A}_{0,k} + \widetilde{\bbm \Delta}_k\bigr) - \bigl(\bm A_0^* + \bbm \Delta_k^*\bigr)\bigr\|_{\mathrm F}^2 \notag\\ 
			&\geq C_{\mathrm{RSC}}^{(k)}\left(\bigl\|\widetilde{\bm A}_{0,k} - \bm A_0^*\bigr\|_{\mathrm F}^2 + \bigl\|\widetilde{\bbm \Delta}_k - \bbm \Delta_k^*\bigr\|_{\mathrm F}^2\right) - 2C_{\mathrm{RSC}}^{(k)}|\langle \widetilde{\bm A}_{0,k} - \bm A_0^*, \widetilde{\bbm \Delta}_k - \bbm \Delta_k^*\rangle|.
		\end{align}
		By the H\"older's inequality, we have 
		\begin{align*}
			2C_{\mathrm{RSC}}^{(k)}|\langle \widetilde{\bm A}_{0,k} - \bm A_0^*, \widetilde{\bbm \Delta}_k - \bbm \Delta_k^*\rangle| &\leq 2C_{\mathrm{RSC}}^{(k)}\bigl\|\widetilde{\bm A}_{0,k} - \bm A_0^*\bigr\|_{*}\bigl\|\widetilde{\bbm \Delta}_k - \bbm \Delta_k^*\bigr\|_{\op} \\
			&\leq 2C_{\mathrm{RSC}}^{(k)}\|\widetilde{\bm A}_{0,k} - \bm A_0^*\|_{*}\left(\bigl\|\widetilde{\bbm \Delta}_k\bigr\|_{\op} + \bigl\|\bbm \Delta_k^*\bigr\|_{\op}\right).
		\end{align*}
		Since $\widetilde{\bbm \Delta}_{k}$ is feasible for \eqref{eq:local-estimator}, we have $\|\widetilde{\bbm \Delta}_{k}\|_{\op}\leq \zeta$. Moreover, $\|\bbm \Delta_k^*\|_{\op}\leq \zeta$. Then, by the prescribed choice of the tuning parameter $\lambda_k \gtrsim C_{\mathrm{RSC}}^{(k)}\zeta$ as in Proposition \ref{thm:local-model-error-upper-bound}, we can pick a constant $C>0$ such that
		\begin{align}\label{eq:inner-product-bound}
			2C_{\mathrm{RSC}}^{(k)}|\langle \widetilde{\bm A}_{0,k} - \bm A_0^*, \widetilde{\bbm \Delta}_k - \bbm \Delta_k^*\rangle| \leq 4C_{\mathrm{RSC}}^{(k)}\zeta \|\widetilde{\bm A}_{0,k} - \bm A_0^*\|_{*} \leq C \lambda_k \|\widetilde{\bm A}_{0,k} - \bm A_0^*\|_{*}.
		\end{align}
		Consequently, combining \eqref{eq:inner-product-bound} with \eqref{eq:F-norm-decomposition} yields
		\begin{align}\label{eq:F-norm-lower-bound}
			&C_{\mathrm{RSC}}^{(k)}\bigl\|\bigl(\widetilde{\bm A}_{0,k}+\widetilde{\bbm \Delta}_k\bigr)-\bigl(\bm A_0^* + \bbm \Delta_k\bigr)\bigr\|_{\mathrm F}^2 \geq C_{\mathrm{RSC}}^{(k)}\left(\bigl\|\widetilde{\bm A}_{0,k} - \bm A_0^*\bigr\|_{\mathrm F}^2 + \bigl\|\widetilde{\bbm \Delta}_k - \bbm \Delta_k^*\bigr\|_{\mathrm F}^2\right) - C \lambda_k \|\widetilde{\bm A}_{0,k} - \bm A_0^*\|_{*} \notag\\
			& \quad \geq C_{\mathrm{RSC}}^{(k)}\left(\bigl\|\widetilde{\bm A}_{0,k} - \bm A_0^*\bigr\|_{\mathrm F}^2 + \bigl\|\widetilde{\bbm \Delta}_k - \bbm \Delta_k^*\bigr\|_{\mathrm F}^2\right) - C \left(\lambda_k\bigl\|\widetilde{\bm A}_{0,k} - \bm A_0^*\bigr\|_{*} + \omega_k\bigl\|\widetilde{\bbm \Delta}_k - \bbm \Delta_k^*\bigr\|_{1}\right).
		\end{align}
		Combining \eqref{eq:RSC-var-lower-bound} and \eqref{eq:F-norm-lower-bound}, we have, with probability at least $1 - C\exp(-C pd)$,
		\begin{align}\label{eq:local-model-error-lower-bound-after-RSC}
			&\frac{1}{T_k}\sum_{t=1}^{T_k}\left(\left\|\Bigl(\bigl(\widetilde{\bm A}_{0,k}+\widetilde{\bbm \Delta}_k\bigr)-\bigl(\bm A_0^* + \bbm \Delta_k^*\bigr)\Bigr)\bbm x_{k,t}\right\|_2^2\right) + C \left(\lambda_k\bigl\|\widetilde{\bm A}_{0,k} - \bm A_0^*\bigr\|_{*} + \omega_k\bigl\|\widetilde{\bbm \Delta}_k - \bbm \Delta_k^*\bigr\|_{1}\right) \notag\\
			&\quad \geq C_{\mathrm{RSC}}^{(k)}\left(\bigl\|\widetilde{\bm A}_{0,k} - \bm A_0^*\bigr\|_{\mathrm F}^2 + \bigl\|\widetilde{\bbm \Delta}_k - \bbm \Delta_k^*\bigr\|_{\mathrm F}^2\right).
		\end{align}
		Finally, combining \eqref{eq:local-model-error-upper-bound-before-RSC} and \eqref{eq:local-model-error-lower-bound-after-RSC} yields
		\begin{align}\label{eq:RSC-applied}
			&C_{\mathrm{RSC}}^{(k)}
			\Bigl(\bigl\|\widetilde{\bm A}_{0,k} - \bm A_0^*\bigr\|_{\F}^2
			+\bigl\|\widetilde{\bbm \Delta}_k - \bbm \Delta_k^*\bigr\|_{\F}^2\Bigr) \leq
			\Bigl(\frac{3}{2}+C\Bigr)
			\Bigl(\lambda_k\bigl\|\widetilde{\bm A}_{0,k} - \bm A_0^*\bigr\|_{*}
			+\omega_k\bigl\|\widetilde{\bbm \Delta}_k - \bbm \Delta_k^*\bigr\|_{1}\Bigr) \notag\\
			&\quad \leq
			C'\Bigl[
			\lambda_k\Bigl(\bigl\|(\widetilde{\bm A}_{0,k}-\bm A_0^*)_{\mathcal L_r}\bigr\|_{*}
			+\bigl\|(\widetilde{\bm A}_{0,k}-\bm A_0^*)_{\overline{\mathcal L}_r^{\perp}}\bigr\|_{*}\Bigr)
			+\omega_k\Bigl(\bigl\|(\widetilde{\bbm \Delta}_k-\bbm \Delta_k^*)_{\mathcal S_\kappa}\bigr\|_{1}
			+\bigl\|(\widetilde{\bbm \Delta}_k-\bbm \Delta_k^*)_{\mathcal S_\kappa^{\perp}}\bigr\|_{1}\Bigr)
			\Bigr] \notag\\
			&\quad \overset{(i)}{\leq}
			4C'\Bigl(
			\lambda_k\bigl\|(\widetilde{\bm A}_{0,k}-\bm A_0^*)_{\mathcal L_r}\bigr\|_{*}
			+\omega_k\bigl\|(\widetilde{\bbm \Delta}_k-\bbm \Delta_k^*)_{\mathcal S_\kappa}\bigr\|_{1}
			+\lambda_k\bigl\|(\bm A_0^*)_{\overline{\mathcal L}_r^{\perp}}\bigr\|_{*}
			+\omega_k\bigl\|(\bbm \Delta_k^*)_{\mathcal S_\kappa^{\perp}}\bigr\|_{1}
			\Bigr) \notag\\
			&\quad \overset{(ii)}{\leq}
			4C'\Bigl(
			\lambda_k\sqrt{2r}\,\bigl\|\widetilde{\bm A}_{0,k}-\bm A_0^*\bigr\|_{\F} 
			+\omega_k\bigl\|(\widetilde{\bbm \Delta}_k-\bbm \Delta_k^*)_{\mathcal S_\kappa}\bigr\|_{1}
			+\omega_k\bigl\|(\bbm \Delta_k^*)_{\mathcal S_\kappa^{\perp}}\bigr\|_{1}\Bigr).
		\end{align}
		Inequality $(i)$ follows from the cone-type bound \eqref{eq:cone-constraint-Crk-correct}. Inequality $(ii)$ follows from the rank condition $\mathrm{rank}((\bm A_0^*)_{\mathcal{L}_r}) = \mathrm{rank}(\bm A_0^*) = r$ and the restriction $\mathrm{rank}((\widetilde{\bm A}_{0,k})_{\mathcal{L}_r})\leq r$, together with the fact that $\|\bm M\|_{*}\leq \sqrt{r}\,\|\bm M\|_{\F}$ for any matrix $\bm M$ with $\mathrm{rank}(\bm M)\leq r$. Moreover, since $\bm A_0^*\in\mathcal L_r$, we have $(\bm A_0^*)_{\overline{\mathcal L}_r^{\perp}}=\bm 0$. Considering the threshold $\kappa$ of the sparse space $\mathcal S_\kappa=\{\bm S\in\mathbb R^{d\times pd}:S_{ij}=0 \text{ for all }(i,j)\in[d]\times[pd] \text{ such that }|(\bbm \Delta_k^*)_{ij}|<\kappa\}$, by the definition of $\mathbb B_q(s_q)$ in Section \ref{sec:theoretical_analysis}, we have $s_q \geq |\mathcal S_\kappa| \kappa^q$, and thus $|\mathcal S_\kappa| \leq s_q \kappa^{-q}$. 
		Then, it follows that $\|(\widetilde{\bbm \Delta}_k-\bbm \Delta_k^*)_{\mathcal S_\kappa}\|_{1} \leq \sqrt{|\mathcal{S}_\kappa|}\|(\widetilde{\bbm \Delta}_k-\bbm \Delta_k^*)_{\mathcal S_\kappa}\|_{\F} \leq \sqrt{s_q}\kappa^{-q/2}\|\widetilde{\bbm \Delta}_k-\bbm \Delta_k^*\|_{\F}$.
		Besides, 
		\[
			\|(\bbm \Delta_k^*)_{\overline{\mathcal S}_\kappa^{\perp}}\|_{1} = \sum_{(i,j)\in\overline{\mathcal S}_\kappa^{\perp}} |(\bbm \Delta_k^*)_{ij}| \leq \sum_{(i,j)\in\overline{\mathcal S}_\kappa^{\perp}} |(\bbm \Delta_k^*)_{ij}|^q \kappa^{1-q} \leq s_q \kappa^{1-q}.
		\]
		To conclude the derivation, we translate the preceding inequality \eqref{eq:RSC-applied} into an explicit error bound. For brevity, let $x_A=\|\widetilde{\bm A}_{0,k}-\bm A_0^*\|_{\F}\geq 0$ and $x_\Delta=\|\widetilde{\bbm\Delta}_k-\bbm\Delta_k^*\|_{\F}\geq 0$. Consequently, \eqref{eq:RSC-applied} further implies
		\begin{align}\label{eq:two-var-start}
			x_A^2+x_\Delta^2 \leq \frac{4C'}{C_{\mathrm{RSC}}^{(k)}}\Bigl(\lambda_k\sqrt{2r}x_A + \omega_k\sqrt{s_q}\kappa^{-q/2}x_\Delta + \omega_k s_q \kappa^{1-q}\Bigr).
		\end{align}
		Using $ab \leq (a^2+b^2)/2$, we further obtain that for any $\eta>0$,
		\[
			\lambda_k\sqrt{2r}x_A  = \sqrt{\eta}x_A \cdot \frac{\lambda_k\sqrt{2r}}{\sqrt{\eta}} \leq \frac{\eta}{2}x_A^2+\frac{\lambda_k^2(2r)}{2\eta},
			\ \text{ and } \
			\omega_k\sqrt{s_q}\kappa^{-q/2}x_\Delta = \sqrt{\eta}x_\Delta \cdot \frac{\omega_k\sqrt{s_q}\kappa^{-q/2}}{\sqrt{\eta}}\leq \frac{\eta}{2}x_\Delta^2+\frac{\omega_k^2 s_q\kappa^{-q}}{2\eta}.
		\]
		Choosing $\eta = C_{\mathrm{RSC}}^{(k)}/(4C')$ and substituting these two inequalities into \eqref{eq:two-var-start} yields
		\[
			x_A^2+x_\Delta^2 \leq \frac12(x_A^2+x_\Delta^2) + \frac 12\Bigl(\frac{4C'}{C_{\mathrm{RSC}}^{(k)}}\Bigr)^2\Bigl(\lambda_k^2(2r) + \omega_k^2 s_q\kappa^{-q}\Bigr) + \frac{4C'}{C_{\mathrm{RSC}}^{(k)}}\omega_k s_q\kappa^{1-q}.
		\]
		Rearranging the terms and absorbing absolute constants into $\lesssim$, we have
		\[
			x_A^2+x_\Delta^2 \lesssim r\left(\frac{\lambda_k}{C_{\mathrm{RSC}}^{(k)}}\right)^2 + s_q\left(\frac{\omega_k}{C_{\mathrm{RSC}}^{(k)}}\right)^2\kappa^{-q} + s_q\frac{\omega_k}{C_{\mathrm{RSC}}^{(k)}}\kappa^{1-q}.
		\]
		Then, choosing $\kappa \asymp \omega_k/C_{\mathrm{RSC}}^{(k)}$ and returning $x_A$ and $x_\Delta$ to their matrix notations, we have
		\[
			\bigl\|\widetilde{\bm A}_{0,k}-\bm A_0^*\bigr\|_{\F}^2 + \bigl\|\widetilde{\bbm\Delta}_k-\bbm\Delta_k^*\bigr\|_{\F}^2 \lesssim r\left(\frac{\lambda_k}{C_{\mathrm{RSC}}^{(k)}}\right)^2 + s_q \left(\frac{\omega_k}{C_{\mathrm{RSC}}^{(k)}}\right)^{2-q}.
		\]
		Furthermore, using the inequalities $(a+b)^2 \leq 2(a^2+b^2)$ and $a^2+b^2\leq (a+b)^2$ for $a,b\geq 0$, we have
		\[
			\bigl\|\widetilde{\bm A}_{0,k}-\bm A_0^*\bigr\|_{\F} + \bigl\|\widetilde{\bbm\Delta}_k-\bbm\Delta_k^*\bigr\|_{\F} \lesssim \sqrt{r}\frac{\lambda_k}{C_{\mathrm{RSC}}^{(k)}} + \sqrt{s_q}\left(\frac{\omega_k}{C_{\mathrm{RSC}}^{(k)}}\right)^{1-q/2}.
		\]
		Combining the above bound with the sample size conditions as well as the probability with respect to the deviation bounds \eqref{eq:deviation-bounds-choice-of-lambda-omega} and the RSC bound \eqref{eq:RSC-var-lower-bound}. This concludes the proof.
	\end{proof}

	\begin{proof}[\textbf{Proof of Theorem \ref{thm:common-A0-error-bounds}}]
		The proof proceeds in three parts. In Part~I, we derive the upper bound for $\psi_k^{(n)} = \bigl\|\mathcal P_{\mathcal T_r(\bm A_0^{(n)})}\bigl(\bm G_k(\bm A_0^{(n)};\mathcal D_k)\bigr)-\mathcal P_{\mathcal T_r(\bm A_0^{(n)})}\bigl(\bm G_k(\bm A_0^{(n)};\mathcal D_k')\bigr)\bigr\|_{\F}$ with respect to $\mathcal D_k = \{\{\bbm y_{k,t}\}_{t=1-p}^{0},\{\bbm\epsilon_{k,t}\}_{t=1}^{T_k}\}$ and $\mathcal D'_k = \{\{\bbm y'_{k,t}\}_{t=1-p}^{0},\{\bbm\epsilon'_{k,t}\}_{t=1}^{T_k}\}$. In Part~II, we prove by induction that both $\psi_k^{(n)}$ and the estimation error $\|\bm A_0^{(n)}-\bm A_0^*\|_{\F}$ remain uniformly bounded with high probability over all $n < N_g$. In Part~III, we select a suitable iteration number $N_g$ to derive the final error bound.

		\noindent
		\textbf{Part I (Sensitivity characterization).}
		Consider $\mathcal D_k = \{\{\bbm y_{k,t}\}_{t=1-p}^{0},\{\bbm\epsilon_{k,t}\}_{t=1}^{T_k}\}$ and the neighboring local data set $\mathcal D_k'$ with respect to the sensitive index set $\mathcal I$ at client $k$, i.e., $\mathcal D_k \overset{\mathcal I}{\sim} \mathcal D_k'$. Fix a sensitive index set $\mathcal I\subset[d]$ and let $\mathcal I^{\mathsf c}:=[d]\setminus\mathcal I$. For any vector $\bbm v\in\mathbb R^{d}$, denote by $\bbm v_{\mathcal I}\in\mathbb R^{|\mathcal I|}$ the subvector obtained by restricting $\bbm v$ to the coordinates indexed by $\mathcal I$, and similarly for $\bbm v_{\mathcal I^{\mathsf c}}$. There exists $t_0\in[T_k]$ such that $\bbm\epsilon_{k,t,\mathcal I}=\bbm\epsilon'_{k,t,\mathcal I}$ for all $t\neq t_0$ and $\bbm\epsilon_{k,t_0,\mathcal I}\neq \bbm\epsilon'_{k,t_0,\mathcal I}$, where $\bbm\epsilon'_{k,t_0,\mathcal I}$ is an $i.i.d.$ copy of $\bbm\epsilon_{k,t_0,\mathcal I}$. Moreover, $\bbm\epsilon_{k,t,\mathcal I^{\mathsf c}}=\bbm\epsilon'_{k,t,\mathcal I^{\mathsf c}}$ for all $t\in[T_k]$. The perturbed trajectory is generated by the same VAR model with the same coefficient $\bm A_k$ and the same future innovations, $\bbm y'_{k,t} = \bm A_k\bbm x'_{k,t} + \bbm\epsilon'_{k,t}$ for $t\geq 1$, using the same initial values as $\{\bbm y_{k,t}\}$ for $t < t_0$.
		
		Let $\bbm y_{k,t}^\Delta=\bbm y'_{k,t}-\bbm y_{k,t}$ and $\bbm x_{k,t}^\Delta=\bbm x'_{k,t}-\bbm x_{k,t}$. Without loss of generality, we assume the first $|\mathcal I|$ coordinates of $\bbm\epsilon_{k,t}$ and $\bbm\epsilon'_{k,t}$ are the sensitive coordinates, and the remaining coordinates are non-sensitive, i.e., $\mathcal I=\{1,\ldots,|\mathcal I|\}$. Then, the primitive discrepancy is the selective innovation difference $\bbm \epsilon_{k,t_0,\mathcal I}^\Delta= \bbm\epsilon'_{k,t_0,\mathcal I}-\bbm\epsilon_{k,t_0,\mathcal I}\in\mathbb R^{|\mathcal I|}$, and thus $\bbm\epsilon_{k,t_0}^\Delta = \bbm\epsilon_{k,t_0}' - \bbm\epsilon_{k,t_0} = (\bbm\epsilon_{k,t_0,\mathcal I}^\Delta,\bm 0_{\mathcal I^{\mathsf c}})$. Consequently, $\bbm\epsilon'_{k,t}-\bbm\epsilon_{k,t}=\bm 0$ for $t\neq t_0$ and $\bbm\epsilon'_{k,t_0}-\bbm\epsilon_{k,t_0}=\bbm\epsilon_{k,t_0}^\Delta$.

		Consider a fixed $k$, under Assumption \ref{assump:stationarity}, there are no roots in the closed unit disk $\{|z|\leq 1\}$, hence every root $z_0$ satisfies $|z_0|>1$. Therefore, the quantity $\varrho_* = \inf\{|z|: \det\mathcal A_k(z)=0\}$ is well-defined, equals $+\infty$ when $\det\mathcal A_k(z)$ has no roots, and otherwise equals the minimum modulus among the finitely many roots; in particular $\varrho_*>1$. Fix any radius $\varrho\in(1,\varrho_*)$. Define the uniform transfer-function bound
		\begin{align}\label{eq:Car-def}
			C_{\mathcal{A}}(\varrho) = \sup_{k\in[K]}\sup_{|z|=\varrho}\bigl\|\mathcal A_k^{-1}(z)\bigr\|_{\mathrm{op}} = \left(\inf_{k\in[K]}\inf_{|z|=\varrho}\lambda_{\min}\big(\mathcal A_k^{\dagger}(z)\mathcal A_k(z)\big)\right)^{-\frac{1}{2}} < \infty,
		\end{align}
		which is finite because $z\mapsto \mathcal A_k^{-1}(z)$ is continuous on the compact circle $\{|z|=\varrho\}$.

		Note that Algorithm~\ref{algorithmTL} only enforces differential privacy on each client’s local dataset through the gradient-transfer step. Therefore, only $\|\mathcal P_{\mathcal T_r(\bm A_0^{(n)})}\!\bigl(\bm G_k(\bm A_0^{(n)};\mathcal D_k)\bigr)-\mathcal P_{\mathcal T_r(\bm A_0^{(n)})}\bigl(\bm G_k(\bm A_0^{(n)};\mathcal D_k')\bigr)\|_{\F}$ are related to sensitivities, where $\bm G_k(\bm A_0^{(n)};\mathcal D_k)$ and $\bm G_k(\bm A_0^{(n)};\mathcal D_k')$ denote the corresponding full gradients computed at client $k$ from $\mathcal D_k$ and $\mathcal D_k'$. Hereafter, we derive an upper bound for $\bigl\|\mathcal P_{\mathcal T_r(\bm A_0^{(n)})}\bigl(\bm G_k(\bm A_0^{(n)};\mathcal D_k)\bigr)-\mathcal P_{\mathcal T_r(\bm A_0^{(n)})}\bigl(\bm G_k(\bm A_0^{(n)};\mathcal D_k')\bigr)\bigr\|_{\F}$ in terms of the estimation error $\|\bm A_0^{(n)}-\bm A_k^*\|_{\F}$. To control the discrepancy induced by the difference of $\mathcal{D}_k$ and $\mathcal{D}_k'$, we use a geometric bound on the impulse-response coefficients associated with $\mathcal A_k^{-1}(z)$ in Lemma~\ref{lem:psi-geom}. Specifically, set $\vartheta=\varrho^{-1}$. For any $|z|<\varrho_*$, the inverse lag polynomial admits the power-series expansion $\mathcal A_k^{-1}(z)=\sum_{h\geq 0}\bm\Psi_{k,h}z^h$, and the corresponding impulse-response coefficients $\{\bm\Psi_{k,h}\}_{h\geq 0}$ satisfy the geometric decay bound:
		\begin{align}\label{eq:Psi-op-bound}
			\|\bm\Psi_{k,h}\|_{\op}\leq C_{\mathcal A}(\varrho)\vartheta^h, \ \text{ for all } \ h\geq 0.
		\end{align}

		Before the proof, we first provide high-probability upper bounds for $\|\bbm \epsilon_{k,t}\|_2$ and $\|\bbm x_{k,t}\|_2$ that will be repeatedly used later. By Assumption \ref{assump:subg} and Lemma~\ref{lemma:subg-vector}, $\bbm \epsilon_{k,t}$ is a mean-zero sub-Gaussian random vector with sub-Gaussian constant $\sigma \sqrt{\lambda_{\max}(\bbm \Sigma_{\epsilon,k})}$. Moreover, by Assumption \ref{assump:subg} and Lemma~\ref{lemma:stacked-subg-vector}, the stacked lag vector $\bbm x_{k,t}=(\bbm y_{k,t-1}^\top,\ldots,\bbm y_{k,t-p}^\top)^\top$ is also sub-Gaussian with sub-Gaussian constant $\sigma\sqrt{C_{\epsilon,\mathcal{A}}^{(k)}}$, where $C_{\epsilon,\mathcal{A}}^{(k)} = \lambda_{\max}(\bbm \Sigma_{\epsilon,k})/\mu_{\min}(\mathcal A_k)$.
		We can now utilize these sub-Gaussian properties to derive uniform-in-time $\ell_2$-norm bounds. Applying Lemma~\ref{lem:subg-l2-max-sqrtds} to $\{\bbm \epsilon_{k,t}\}_{t=1}^{T_k}$ with dimension $d$ and sub-Gaussian constant $K_\zeta=\sigma\sqrt{\lambda_{\max}(\bbm \Sigma_{\epsilon,k})}$, we obtain that for any $T_k\leq T$ there exists a constant $C>0$ such that
		\begin{align}\label{eq:eps-l2-max-hp}
			\mathbb P\Big(\max_{1\leq t\leq T_k}\|\bbm \epsilon_{k,t}\|_2 \leq C \sigma\sqrt{\lambda_{\max}(\bbm \Sigma_{\epsilon,k})}\sqrt{d+\log T}\Big)\geq 1-\exp(-C\log T).
		\end{align}
		Similarly, applying Lemma~\ref{lem:subg-l2-max-sqrtds} to $\{\bbm x_{k,t}\}_{t=1}^{T_k}$ with dimension $pd$ and sub-Gaussian constant $K_\zeta=\sigma\sqrt{C_{\epsilon,\mathcal{A}}^{(k)}}$ yields
		\begin{align}\label{eq:x-l2-max-hp}
			\mathbb P\Big(\max_{1\leq t\leq T_k}\|\bbm x_{k,t}\|_2 \leq C \sigma\sqrt{C_{\epsilon,\mathcal{A}}^{(k)}}\sqrt{pd+\log T}\Big) \geq 1-\exp(-C\log T).
		\end{align}

		For $\bbm\epsilon_{k,t_0}^\Delta=(\bbm\epsilon_{k,t_0,\mathcal I}^\Delta,\bm 0_{\mathcal I^{\mathsf c}})$, we now derive a high-probability upper bound for $\|\bbm\epsilon_{k,t_0}^\Delta\|_2$. Note that $\bbm\epsilon_{k,t_0,\mathcal I}^\Delta = \bbm\epsilon'_{k,t_0,\mathcal I}-\bbm\epsilon_{k,t_0,\mathcal I}$ where $\bbm\epsilon'_{k,t_0,\mathcal I}$ is an $i.i.d.$ copy of $\bbm\epsilon_{k,t_0,\mathcal I}$. Since $\|\bbm\epsilon_{k,t_0}^\Delta\|_2=\|\bbm\epsilon_{k,t_0,\mathcal I}^\Delta\|_2$, it suffices to control the $m$-dimensional vector $\bbm\epsilon_{k,t_0,\mathcal I}^\Delta$.
		For any $\bbm v\in\mathbb S^{|\mathcal I|-1}$ and $\bbm v_{aug} = (\bbm v^\top,\bm 0_{\mathcal I^{\mathsf c}}^\top)^\top\in\mathbb S^{d-1}$, so that $\|\langle \bbm v,\bbm\epsilon_{k,t_0,\mathcal I}\rangle\|_{\psi_2} = \|\langle \bbm v_{aug},\bbm\epsilon_{k,t_0}\rangle\|_{\psi_2} \leq \sigma\sqrt{\lambda_{\max}(\bbm\Sigma_{\epsilon,k})}$. Then, by the sub-additivity of the Orlicz norm,
		\[
			\bigl\|\langle \bbm v,\bbm\epsilon_{k,t_0,\mathcal I}^\Delta\rangle\bigr\|_{\psi_2}
			=
			\bigl\|\langle \bbm v,\bbm\epsilon'_{k,t_0,\mathcal I}\rangle-\langle \bbm v,\bbm\epsilon_{k,t_0,\mathcal I}\rangle\bigr\|_{\psi_2}
			\leq
			\bigl\|\langle \bbm v,\bbm\epsilon'_{k,t_0,\mathcal I}\rangle\bigr\|_{\psi_2}
			+
			\bigl\|\langle \bbm v,\bbm\epsilon_{k,t_0,\mathcal I}\rangle\bigr\|_{\psi_2}
			\leq
			2\sigma\sqrt{\lambda_{\max}(\bbm\Sigma_{\epsilon,k})}.
		\]
		Therefore, $\bbm\epsilon_{k,t_0,\mathcal I}^\Delta$ is a mean-zero sub-Gaussian random vector in $\mathbb R^{|\mathcal I|}$ with vector sub-Gaussian constant $2\sigma\sqrt{\lambda_{\max}(\bbm\Sigma_{\epsilon,k})}$. By the standard $1/2$-net argument (as in the proof of Lemma~\ref{lemma:subg-vector}), there exists a constant $C>0$ such that
		\begin{align}\label{eq:deltaepsI-l2-hp}
			\mathbb P\Big(\|\bbm\epsilon_{k,t_0}^\Delta\|_2 = \|\bbm\epsilon_{k,t_0,\mathcal I}^\Delta\|_2 \leq C\sigma\sqrt{\lambda_{\max}(\bbm\Sigma_{\epsilon,k})}\sqrt{|\mathcal I|+\log T}\Big) \geq 1-\exp(-C\log T).
		\end{align}

		Let $\Lambda_{\epsilon,\max} = \max_{k\in[K]}\lambda_{\max}(\bbm\Sigma_{\epsilon,k})$. For each client $k\in[K]$, define the event
		\begin{align}\label{eq:Ek-def-Gamma}
			\mathcal E_k = \Bigl\{\max_{1\leq t\leq T_k}\|\bbm x_{k,t}\|_2 \leq \Gamma_x,\ \max_{1\leq t\leq T_k}\|\bbm \epsilon_{k,t}\|_2 \leq \Gamma_\epsilon,\ \|\bbm\epsilon^\Delta_{k,t_0}\|_2 \leq \Gamma_{\Delta\epsilon}\Bigr\},
		\end{align}
		where $\Gamma_x = C\sigma\sqrt{C_{\epsilon,\mathcal A}^{\max}}\sqrt{pd+\log T}$, $\Gamma_\epsilon = C\sigma\sqrt{\Lambda_{\epsilon,\max}}\sqrt{d+\log T}$ and $\Gamma_{\Delta\epsilon} = C\sigma\sqrt{\Lambda_{\epsilon,\max}}\sqrt{|\mathcal I| +\log T}$. Then, by a union bound over the three high-probability bounds above, there exists a constant $C>0$ such that
		\begin{align}\label{eq:prob-Ek}
			\mathbb P(\mathcal E_k)\geq 1-\exp(-C\log T).
		\end{align}

		Hereafter, we derive upper bounds for $\psi_k^{(n)}$ conditional on the event $\mathcal E_k$. We begin by deriving the upper bounds for $\|\bbm x_{k,t}^\Delta\|_2$. Recall that $\mathcal D_k$ and $\mathcal D_k'$ differ only on the sensitive coordinates $\mathcal I$ at time $t_0$, i.e., $\bbm\epsilon'_{k,t}=\bbm\epsilon_{k,t}$ for all $t\neq t_0$ and $\bbm\epsilon'_{k,t_0}-\bbm\epsilon_{k,t_0}=\bbm\epsilon_{k,t_0}^\Delta=(\bbm\epsilon_{k,t_0,\mathcal I}^\Delta,\bm 0_{\mathcal I^{\mathsf c}})$. Consequently, the induced trajectory difference $\bbm y_{k,t}^\Delta = \bbm y'_{k,t}-\bbm y_{k,t}$ has the following three cases:
		\begin{enumerate}[label=(\roman*)]
			\item \emph{Pre-change ($t<t_0$).} Since the innovations coincide up to time $t_0-1$, the two trajectories coincide as well,
			hence $\bbm y_{k,t}^\Delta=\bm 0$.
			\item \emph{Change time ($t=t_0$).} The only discrepancy is the replaced innovation, so
			$\bbm y_{k,t_0}^\Delta=\bbm\epsilon_{k,t_0}^\Delta$.
			\item \emph{Post-change ($t>t_0$).} The discrepancy propagates through the VAR($p$) process.
		\end{enumerate}
		To formalize (iii), we write the VAR($p$) model for client $k$ via the lag polynomial $\mathcal A_k(L)=\bm I_d-\sum_{j=1}^p \bm A_{k,j}L^j$:
		\[
			\mathcal A_k(L)\bbm y_{k,t}=\bbm\epsilon_{k,t},
			\ \text{ and } \ 
			\mathcal A_k(L)\bbm y'_{k,t}=\bbm\epsilon'_{k,t}.
		\]
		By the stationarity condition in Assumption \ref{assump:stationarity}, $\mathcal A_k^{-1}(z)=\sum_{h\geq0}\bbm\Psi_{k,h}z^h$ for $|z|<\varrho_*$. Applying the inverse linear filter $\mathcal A_k^{-1}(L)$ yields
		\begin{align}\label{eq:yDelta-Psi-epsDelta}
			\bbm y_{k,t}^\Delta
			= \sum_{h\geq 0}\bbm\Psi_{k,h}\bigl(\bbm\epsilon'_{k,t-h}-\bbm\epsilon_{k,t-h}\bigr)
			= \sum_{h\geq 0}\bbm\Psi_{k,h}\bbm\epsilon_{k,t_0}^\Delta\,\mathds 1\{t-h=t_0\}.
		\end{align}

		Considering that $\mathds 1\{t-h=t_0\}=1$ holds if and only if $h=t-t_0$. Therefore,
		\begin{align}\label{eq:yDelta-Psi-epsDelta-simplified}
			\sum_{h\geq 0}\bbm\Psi_{k,h}\bbm\epsilon_{k,t_0}^\Delta\,\mathds 1\{t-h=t_0\}
			=\sum_{h\geq 0}\bbm\Psi_{k,h}\bbm\epsilon_{k,t_0}^\Delta\,\mathds 1\{h=t-t_0\}
			=
			\begin{cases}
			\bbm\Psi_{k,t-t_0}\bbm\epsilon_{k,t_0}^\Delta, & t\geq t_0,\\
			\bm 0, & t<t_0.
			\end{cases}
		\end{align}
		Then, combining \eqref{eq:yDelta-Psi-epsDelta} and \eqref{eq:yDelta-Psi-epsDelta-simplified} and writing $t=t_0+h$, we obtain for all $h\geq0$, $\bbm y_{k,t_0+h}^\Delta=\bbm\Psi_{k,h}\bbm\epsilon_{k,t_0}^\Delta$. Therefore, using the impulse-response bound \eqref{eq:Psi-op-bound},
		\begin{align}\label{eq:e-prop}
		\|\bbm y_{k,t_0+h}^\Delta\|_2 \leq \|\bbm\Psi_{k,h}\|_{\op}\|\bbm\epsilon_{k,t_0}^\Delta\|_2 \leq C_{\mathcal A}(\varrho)\vartheta^h\|\bbm\epsilon_{k,t_0}^\Delta\|_2.
		\end{align}

		Next, we derive the upper bounds for $\|\bbm x_{k,t_0+h}^\Delta\|_2$. Recall $\bbm x_{k,t_0+h}^\Delta = (\bbm y_{k,t_0+h-1}^{\Delta\top},\ldots,\bbm y_{k,t_0+h-p}^{\Delta\top})^\top$ and note that $\bbm y_{k,t}^\Delta=\bm 0$ for $t<t_0$. These together with \eqref{eq:e-prop} imply that
		\begin{align}\label{eq:xDelta-sum}
			\|\bbm x_{k,t_0+h}^\Delta\|_2^2 &= \sum_{j=1}^{p}\|\bbm y_{k,t_0+h-j}^\Delta\|_2^2 = \sum_{j=1}^{\min(p,h)}\|\bbm y_{k,t_0+h-j}^\Delta\|_2^2 \leq \sum_{j=1}^{\min(p,h)}\Bigl(C_{\mathcal A}(\varrho)\vartheta^{h-j}\|\bbm\epsilon_{k,t_0}^\Delta\|_2\Bigr)^2 \notag \\
			&= C_{\mathcal A}^2(\varrho)\|\bbm\epsilon_{k,t_0}^\Delta\|_2^2\sum_{j=1}^{\min(p,h)}\vartheta^{2(h-j)}.
		\end{align}
		Now we bound the geometric sum in \eqref{eq:xDelta-sum}. For $1\leq h\leq p$, then $\min(p,h)=h$ and
		\[
			\sum_{j=1}^{h}\vartheta^{2(h-j)}=\sum_{\ell=0}^{h-1}\vartheta^{2\ell}\leq \sum_{\ell=0}^{p-1}\vartheta^{2\ell} = \frac{1-\vartheta^{2p}}{1-\vartheta^2}.
		\]
		For $1 \leq p < h$, then $\min(p,h)=p$ and
		\[
			\sum_{j=1}^{p}\vartheta^{2(h-j)}=\vartheta^{2(h-p)}\sum_{\ell=0}^{p-1}\vartheta^{2\ell} \leq \vartheta^{2(h-p)}\frac{1-\vartheta^{2p}}{1-\vartheta^2}.
		\]
		Denote $x_+ = \max(x,0)$ for any $x\in\mathbb R$. Combining the two cases yields
		\begin{align}\label{eq:geom-sum-bound}
			\sum_{j=1}^{\min(p,h)}\vartheta^{2(h-j)}\leq \vartheta^{2(h-p)_+}\frac{1-\vartheta^{2p}}{1-\vartheta^2}.
		\end{align}
		Plugging \eqref{eq:geom-sum-bound} into \eqref{eq:xDelta-sum} and taking square roots gives
		\begin{align}\label{eq:d-prop}
			\|\bbm x_{k,t_0+h}^\Delta\|_2 \leq C_{\mathcal A}(\varrho)\Bigl(\frac{1-\vartheta^{2p}}{1-\vartheta^2}\Bigr)^{1/2}\vartheta^{(h-p)_+}\|\bbm\epsilon_{k,t_0}^\Delta\|_2.
		\end{align}
		For $t\leq t_0$, note that $\bbm x_{k,t}^\Delta=\bm 0$, since $\bbm x_{k,t} = \bbm x_{k,t}'$.

		Recall that for each $t$, the unprojected gradient of the squared loss $1/2\|\bbm y_{k,t}-\bm A\bbm x_{k,t}\|_2^2$ with respect to $\bm A$ is $(\bm A\bbm x_{k,t}-\bbm y_{k,t})\bbm x_{k,t}^\top$. Thus, with $\bm A=\bm A_0^{(n)}$, denote $\bm G_{k,t}^{(n)}=(\bm A_0^{(n)}\bbm x_{k,t}-\bbm y_{k,t})\bbm x_{k,t}^{\top}$ and $\bm G_{k,t}^{(n)'}=(\bm A_0^{(n)}\bbm x'_{k,t}-\bbm y'_{k,t}){\bbm x'}_{k,t}^{\top}$. Then, expanding the product yields
		\begin{align}\label{eq:gt-tildegt-identity}
			\bm G_{k,t}^{(n)}-\bm G_{k,t}^{(n)'} = -(\bm A_0^{(n)}\bbm x_{k,t}-\bbm y_{k,t})\bbm x_{k,t}^{\Delta\top} - (\bm A_0^{(n)}\bbm x_{k,t}^{\Delta}-\bbm y_{k,t}^{\Delta})\bbm x_{k,t}^{\top} - (\bm A_0^{(n)}\bbm x_{k,t}^{\Delta}-\bbm y_{k,t}^{\Delta})\bbm x_{k,t}^{\Delta\top}.
		\end{align}

		By the additivity of the projection operator $\mathcal P_{\mathcal T_r(\bm A_0^{(n)})}(\cdot)$ and the sub-additivity of the Frobenius norm $\|\cdot\|_{\F}$, we have
		\begin{align*}
			\psi_k^{(n)}&=\bigl\|\mathcal P_{\mathcal T_r(\bm A_0^{(n)})}\!\bigl(\bm G_k(\bm A_0^{(n)};\mathcal D_k)\bigr) - \mathcal P_{\mathcal T_r(\bm A_0^{(n)})}\bigl(\bm G_k(\bm A_0^{(n)};\mathcal D_k')\bigr)\bigr\|_{\F} = \Bigl\|\mathcal P_{\mathcal T_r(\bm A_0^{(n)})}\Bigl(\frac{1}{T_k}\sum_{t=1}^{T_k}(\bm G_{k,t}^{(n)}-\bm G_{k,t}^{(n)'})\Bigr)\Bigr\|_{\F} \\
			&\quad\leq \frac{1}{T_k}\sum_{t=1}^{T_k}\bigl\|\mathcal P_{\mathcal T_r(\bm A_0^{(n)})}(\bm G_{k,t}^{(n)}) - \mathcal P_{\mathcal T_r(\bm A_0^{(n)})}(\bm G_{k,t}^{(n)'})\bigr\|_{\F} = \frac{1}{T_k}\sum_{t=t_0}^{T_k}\bigl\|\mathcal P_{\mathcal T_r(\bm A_0^{(n)})}(\bm G_{k,t}^{(n)}) - \mathcal P_{\mathcal T_r(\bm A_0^{(n)})}(\bm G_{k,t}^{(n)'})\bigr\|_{\F},
		\end{align*}
		where the last equality holds because $\bbm y_{k,t}=\bbm y_{k,t}'$ and $\bbm x_{k,t}=\bbm x_{k,t}'$ for all $t<t_0$, and thus $\bm G_{k,t} = \bm G'_{k,t}$ for all $t<t_0$.
		Next, we derive a Frobenius-norm bound for the projected gradient discrepancy $\|\mathcal P_{\mathcal T_r(\bm A_0^{(n)})}(\bm G_{k,t}^{(n)})-\mathcal P_{\mathcal T_r(\bm A_0^{(n)})}(\bm G_{k,t}^{(n)'})\|_{\F}$, which will be summed over $t$ in the subsequent step. Since $\mathcal P_{\mathcal T_r(\bm A_0^{(n)})}(\cdot)$ is a non-expansive orthogonal projection, i.e., $\|\mathcal P_{\mathcal T_r(\bm A_0^{(n)})}(\bm M)\|_{\F}\leq \|\bm M\|_{\F}$ for all $\bm M \in \mathbb{R}^{d\times pd}$. Hence $\|\mathcal P_{\mathcal T_r(\bm A_0^{(n)})}(\bm G_{k,t})-\mathcal P_{\mathcal T_r(\bm A_0^{(n)})}(\bm G'_{k,t})\|_{\F} \leq \|\bm G_{k,t}^{(n)}-\bm G_{k,t}^{(n)'}\|_{\F}$. Therefore, it suffices to bound $\|\bm G_{k,t}^{(n)}-\bm G_{k,t}^{(n)'}\|_{\F}$ for each fixed time $t$. We treat two cases separately: the injection time $t=t_0$, where the innovation is replaced, and the propagation regime $t>t_0$, where the discrepancy evolves through the VAR process. 

		At $t=t_0$, we have $\bbm x_{k,t_0}^\Delta=\bm 0$ and $\bbm y_{k,t_0}^\Delta=\bbm\epsilon_{k,t_0}^\Delta$. Plugging into \eqref{eq:gt-tildegt-identity}, only the second term remains and simplifies to $\bm G_{k,t_0}^{(n)}-\bm G_{k,t_0}^{(n)'} = -(\bm A_0^{(n)}\bm 0-\bbm\epsilon_{k,t_0}^\Delta)\bbm x_{k,t_0}^\top = \bbm\epsilon_{k,t_0}^\Delta\bbm x_{k,t_0}^\top$. Recall that for any vectors $\bbm u,\bbm v$, $\|\bbm u\bbm v^\top\|_\F=\|\bbm u\|_2\|\bbm v\|_2$, then we have
		\begin{align}\label{eq:t0-diff}
			\|\mathcal P_{\mathcal T_r(\bm A_0^{(n)})}(\bm G_{k,t_0}^{(n)})-\mathcal P_{\mathcal T_r(\bm A_0^{(n)})}(\bm G_{k,t_0}^{(n)'})\|_{\F} \leq \|\bm G_{k,t_0}^{(n)}-\bm G_{k,t_0}^{(n)'}\|_{\F} = \|\bbm\epsilon_{k,t_0}^\Delta\|_2\|\bbm x_{k,t_0}\|_2 \leq \Gamma_{\Delta\epsilon}\Gamma_x.
		\end{align}
		The last inequality holds because event $\mathcal E_k$ in \eqref{eq:Ek-def-Gamma} ensures $\|\bbm x_{k,t_0}\|_2\leq \Gamma_x$ and $\|\bbm\epsilon_{k,t_0}^\Delta\|_2 \leq \Gamma_{\Delta\epsilon}$.

		For $t>t_0$, the neighboring datasets differ only at time $t_0$, $\bbm\epsilon_{k,t}=\bbm\epsilon_{k,t}'$. Therefore, $\bbm y_{k,t}^\Delta = \bbm y'_{k,t}-\bbm y_{k,t} = (\bm A_k^*\bbm x'_{k,t}+\bbm\epsilon_{k,t}')-(\bm A_k^*\bbm x_{k,t}+\bbm\epsilon_{k,t}) = \bm A_k^*\bbm x_{k,t}^\Delta$ with $\bbm x_{k,t}^\Delta=\bbm x'_{k,t}-\bbm x_{k,t}$. Consequently, we have $\bm A_0^{(n)}\bbm x_{k,t}-\bbm y_{k,t} = (\bm A_0^{(n)}-\bm A_k^*)\bbm x_{k,t}-\bbm\epsilon_{k,t}$ and $\bm A_0^{(n)}\bbm x_{k,t}^\Delta-\bbm y_{k,t}^\Delta = (\bm A_0^{(n)}-\bm A_k^*)\bbm x_{k,t}^\Delta$. Recall that for any vectors $\bbm u,\bbm v$, $\|\bbm u\bbm v^\top\|_\F=\|\bbm u\|_2\|\bbm v\|_2$ and $\|\bbm u\bbm v^\top\|_{\op}=\|\bbm u\|_2\|\bbm v\|_2$. Moreover, for any matrix $\bm M$ and vector $\bbm v$, $\|\bm M\bbm v\|_2\leq \|\bm M\|_{\op}\|\bbm v\|_2$. Plugging the above identities into \eqref{eq:gt-tildegt-identity}, using the linearity of the projection operator, and then applying the sub-additivity of the Frobenius norm yields
		\begin{align}\label{eq:gt-tildegt-bound}
			&\|\mathcal P_{\mathcal T_r(\bm A_0^{(n)})}(\bm G_{k,t}^{(n)})
			-\mathcal P_{\mathcal T_r(\bm A_0^{(n)})}(\bm G_{k,t}^{(n)'})\|_\F 
			=
			\|\mathcal P_{\mathcal T_r(\bm A_0^{(n)})}
			(\bm G_{k,t}^{(n)}-\bm G_{k,t}^{(n)'})\|_\F \notag\\
			&=
			\Bigl\|
			\mathcal P_{\mathcal T_r(\bm A_0^{(n)})}
			\Bigl(
			\bbm\epsilon_{k,t}\bbm x_{k,t}^{\Delta\top}
			-(\bm A_0^{(n)}-\bm A_k^*)\bbm x_{k,t}\bbm x_{k,t}^{\Delta\top}
			-(\bm A_0^{(n)}-\bm A_k^*)\bbm x_{k,t}^{\Delta}\bbm x_{k,t}^{\top}
			-(\bm A_0^{(n)}-\bm A_k^*)\bbm x_{k,t}^{\Delta}\bbm x_{k,t}^{\Delta\top}
			\Bigr)
			\Bigr\|_\F \notag\\
			&\leq
			\bigl\|
			\mathcal P_{\mathcal T_r(\bm A_0^{(n)})}
			(\bbm\epsilon_{k,t}\bbm x_{k,t}^{\Delta\top})
			\bigr\|_\F 
			+
			\bigl\|
			\mathcal P_{\mathcal T_r(\bm A_0^{(n)})}
			((\bm A_0^{(n)}-\bm A_k^*)\bbm x_{k,t}\bbm x_{k,t}^{\Delta\top})
			\bigr\|_\F \notag\\
			&\quad+
			\bigl\|
			\mathcal P_{\mathcal T_r(\bm A_0^{(n)})}
			((\bm A_0^{(n)}-\bm A_k^*)\bbm x_{k,t}^{\Delta}\bbm x_{k,t}^{\top})
			\bigr\|_\F
			+
			\bigl\|
			\mathcal P_{\mathcal T_r(\bm A_0^{(n)})}
			((\bm A_0^{(n)}-\bm A_k^*)\bbm x_{k,t}^{\Delta}\bbm x_{k,t}^{\Delta\top})
			\bigr\|_\F \notag\\
			&\leq
			\|\bbm\epsilon_{k,t}\bbm x_{k,t}^{\Delta\top}\|_\F
			+
			\sqrt{2r}
			\bigl\|
			(\bm A_0^{(n)}-\bm A_k^*)\bbm x_{k,t}\bbm x_{k,t}^{\Delta\top}
			\bigr\|_{\op} 
			+
			\sqrt{2r}
			\bigl\|
			(\bm A_0^{(n)}-\bm A_k^*)\bbm x_{k,t}^{\Delta}\bbm x_{k,t}^{\top}
			\bigr\|_{\op} \notag\\
			&\quad+
			\sqrt{2r}
			\bigl\|
			(\bm A_0^{(n)}-\bm A_k^*)\bbm x_{k,t}^{\Delta}\bbm x_{k,t}^{\Delta\top}
			\bigr\|_{\op} \notag\\
			&\leq
			\|\bbm\epsilon_{k,t}\|_2\|\bbm x_{k,t}^{\Delta}\|_2
			+
			2\sqrt{2r}
			\|\bm A_0^{(n)}-\bm A_k^*\|_{\op}
			\|\bbm x_{k,t}\|_2
			\|\bbm x_{k,t}^{\Delta}\|_2 
			+
			\sqrt{2r}
			\|\bm A_0^{(n)}-\bm A_k^*\|_{\op}
			\|\bbm x_{k,t}^{\Delta}\|_2^2 \notag\\
			&=
			\Bigl(
			2\sqrt{2r}
			\|\bm A_0^{(n)}-\bm A_k^*\|_{\op}
			\|\bbm x_{k,t}\|_2
			+
			\|\bbm\epsilon_{k,t}\|_2
			\Bigr)
			\|\bbm x_{k,t}^{\Delta}\|_2
			+
			\sqrt{2r}
			\|\bm A_0^{(n)}-\bm A_k^*\|_{\op}
			\|\bbm x_{k,t}^{\Delta}\|_2^2 .
		\end{align}

		Conditional on the event $\mathcal E_k$ in \eqref{eq:Ek-def-Gamma}, we have $\|\bbm x_{k,t}\|_2\leq \Gamma_x$ and $\|\bbm\epsilon_{k,t}\|_2\leq \Gamma_\epsilon$ for all $t\in[T_k]$. Then, substituting these bounds into \eqref{eq:gt-tildegt-bound} and summing over $t$ from $t_0+1$ to $T_k$ yields
		\begin{align}\label{eq:gt-tildegt-bound-Ek}
			&\sum_{t=t_0+1}^{T_k}\bigl\|\mathcal P_{\mathcal T_r(\bm A_0^{(n)})}(\bm G_{k,t}^{(n)}) - \mathcal P_{\mathcal T_r(\bm A_0^{(n)})}(\bm G_{k,t}^{(n)'})\bigr\|_\F \notag\\
			& \leq \Bigl(2\sqrt{2r}\|\bm A_0^{(n)}-\bm A_k^*\|_{\op}\Gamma_x+\Gamma_\epsilon\Bigr)\sum_{t=t_0+1}^{T_k}\|\bbm x_{k,t}^{\Delta}\|_2 + \sqrt{2r}\|\bm A_0^{(n)}-\bm A_k^*\|_{\op}\sum_{t=t_0+1}^{T_k}\|\bbm x_{k,t}^{\Delta}\|_2^2.
		\end{align}
		Thus, it suffices to bound $\sum_{t=t_0+1}^{T_k}\|\bbm x_{k,t}^\Delta\|_2$ and $\sum_{t=t_0+1}^{T_k}\|\bbm x_{k,t}^\Delta\|_2^2$. 
		By \eqref{eq:d-prop} and $\|\bbm\epsilon_{k,t_0}^\Delta\|_2\leq \Gamma_{\Delta\epsilon}$ implied by the event $\mathcal E_k$ in \eqref{eq:Ek-def-Gamma}, for $t=t_0+h$ with $h\geq 1$, we have
		\begin{align*}
			\|\bbm x_{k,t_0+h}^\Delta\|_2 \leq C_{\mathcal A}(\varrho)\Bigl(\frac{1-\vartheta^{2p}}{1-\vartheta^2}\Bigr)^{1/2}\Gamma_{\Delta\epsilon}\vartheta^{(h-p)_+},
		\end{align*}
		where $\vartheta=\varrho^{-1}\in(0,1)$. Consequently,
		\begin{align*}
			&\sum_{t=t_0+1}^{T_k}\|\bbm x_{k,t}^\Delta\|_2 = \sum_{h=1}^{T_k-t_0}\|\bbm x_{k,t_0+h}^\Delta\|_2 \leq C_{\mathcal A}(\varrho)\Bigl(\frac{1-\vartheta^{2p}}{1-\vartheta^2}\Bigr)^{\frac{1}{2}}\Gamma_{\Delta\epsilon} \sum_{h=1}^{T_k-t_0}\vartheta^{(h-p)_+},\\
			&\sum_{t=t_0+1}^{T_k}\|\bbm x_{k,t}^\Delta\|_2^2 = \sum_{h=1}^{T_k-t_0}\|\bbm x_{k,t_0+h}^\Delta\|_2^2 \leq C_{\mathcal A}^2(\varrho)\Bigl(\frac{1-\vartheta^{2p}}{1-\vartheta^2}\Bigr)\Gamma_{\Delta\epsilon}^2 \sum_{h=1}^{T_k-t_0}\vartheta^{2(h-p)_+}.
		\end{align*}
		Using the deterministic geometric-sum bounds, we have
		\[
			\sum_{h=1}^{T_k-t_0}\vartheta^{(h-p)_+} \leq p+\frac{\vartheta}{1-\vartheta} \leq p+\frac{1}{1-\vartheta},
			\ \text{ and } \ 
			\sum_{h=1}^{T_k-t_0}\vartheta^{2(h-p)_+} \leq p+\frac{\vartheta^2}{1-\vartheta^2} \leq p+\frac{1}{1-\vartheta^2}.
		\]
		Plugging these into \eqref{eq:gt-tildegt-bound-Ek} yields, conditional on $\mathcal E_k$ in \eqref{eq:Ek-def-Gamma},
		\begin{align}\label{eq:gt-tildegt-bound-Ek-final}
			&\sum_{t=t_0+1}^{T_k}\bigl\|\mathcal P_{\mathcal T_r(\bm A_0^{(n)})}(\bm G_{k,t}) - \mathcal P_{\mathcal T_r(\bm A_0^{(n)})}(\bm G'_{k,t})\bigr\|_\F \notag\\
			& \leq \Bigl(2\sqrt{2r}\|\bm A_0^{(n)}-\bm A_k^*\|_{\op}\Gamma_x+\Gamma_\epsilon\Bigr)C_{\mathcal A}(\varrho)\Bigl(\frac{1-\vartheta^{2p}}{1-\vartheta^2}\Bigr)^{1/2}\Gamma_{\Delta\epsilon}\Bigl(p + \frac{1}{1-\vartheta}\Bigr) \\
			&\qquad+ \sqrt{2r}\|\bm A_0^{(n)}-\bm A_k^*\|_{\op} C_{\mathcal A}^2(\varrho)\Bigl(\frac{1-\vartheta^{2p}}{1-\vartheta^2}\Bigr)\Gamma_{\Delta\epsilon}^2\Bigl(p+\frac{1}{1-\vartheta^2}\Bigr) \notag.
		\end{align}
		Finally, combining the bounds \eqref{eq:t0-diff} and \eqref{eq:gt-tildegt-bound-Ek-final} yields, conditional on $\mathcal E_k$ in \eqref{eq:Ek-def-Gamma},
		\begin{align}\label{eq:xi-final-explicit}
			\psi_k^{(n)} &\leq \frac{1}{T_k}\Biggl[\Gamma_x\Gamma_{\Delta\epsilon} + \Bigl(2\sqrt{2r}\|\bm A_0^{(n)}-\bm A_k^*\|_{\op}\Gamma_x + \Gamma_\epsilon\Bigr)C_{\mathcal A}(\varrho)\Bigl(\frac{1-\vartheta^{2p}}{1-\vartheta^2}\Bigr)^{1/2}\Gamma_{\Delta\epsilon}\Bigl(p + \frac{1}{1-\vartheta}\Bigr)\notag\\
			&\qquad + \sqrt{2r}\|\bm A_0^{(n)}-\bm A_k^*\|_{\op} C_{\mathcal A}^2(\varrho)\Bigl(\frac{1-\vartheta^{2p}}{1-\vartheta^2}\Bigr)\Gamma_{\Delta\epsilon}^2\Bigl(p+\frac{1}{1-\vartheta^2}\Bigr)\Biggr].
		\end{align}
		For convenience, define the two deterministic factors
		\begin{align*}
			\kappa_1(\varrho)&= C_{\mathcal A}(\varrho)\Bigl(\frac{1-\vartheta^{2p}}{1-\vartheta^2}\Bigr)^{1/2}\Bigl(p+\frac{1}{1-\vartheta}\Bigr), \ \text{ and } \ \kappa_2(\varrho)= C_{\mathcal A}^2(\varrho)\Bigl(\frac{1-\vartheta^{2p}}{1-\vartheta^2}\Bigr)\Bigl(p+\frac{1}{1-\vartheta^2}\Bigr). 
		\end{align*}
		Then, conditional on $\mathcal E_k$ in \eqref{eq:Ek-def-Gamma}, \eqref{eq:xi-final-explicit} can be rewritten as
		\begin{align}\label{eq:xi-final-kappa}
		\psi_k^{(n)} \leq \frac{1}{T_k}\Bigl[\Gamma_x\Gamma_{\Delta\epsilon} + \kappa_1(\varrho)\Gamma_{\Delta\epsilon}\Gamma_\epsilon + \sqrt{2r}\|\bm A_0^{(n)}-\bm A_k^*\|_{\op}\Bigl(2\kappa_1(\varrho)\Gamma_x\Gamma_{\Delta\epsilon} + \kappa_2(\varrho)\Gamma_{\Delta\epsilon}^2\Bigr)\Bigr].
		\end{align}

		\noindent
		\textbf{Part II (High-probability bounds via induction).}
		Next, we establish high-probability control of the estimation error $\|\bm A_0^{(n)}-\bm A_0^*\|_{\F}$ together with the sensitivity $\psi_k^{(n)}$ by mathematical induction. Fix an iteration $n\in\{0,\cdots,N_g-1\}$ and assume that 
		\begin{align}\label{eq:induction-hypothesis}
			\|\bm A_0^{(n)}-\bm A_0^*\|_{\F}\leq R, \ \text{with} \ R=\frac{c_{\epsilon,\mathcal A}^{\min}(3C_{\epsilon,\mathcal A}^{\max}+c_{\epsilon,\mathcal A}^{\min})}{20(6C_{\epsilon,\mathcal A}^{\max}+c_{\epsilon,\mathcal A}^{\min})^2}\sigma_r(\bm A_0^*).
		\end{align}

		Conditional on the high-probability event $\mathcal E = \bigcap_{k=1}^K\mathcal E_k$, we will show that $\|\bm A_0^{(n+1)}-\bm A_0^*\|_{\F}$ satisfies the same upper bound as well. 
		Since this part does not need the neighboring-dataset definition of differential privacy, we abbreviate $\bm G_k(\bm A_0^{(n)};\mathcal D_k)$ as $\bm G_k(\bm A_0^{(n)})$ for notational simplicity. 
		
		Note that
		\begin{align}\label{eq:A0-update-decomp}
			\bm A_0^{(n)} - \rho \sum_{k=1}^{K} \bm Z_k^{(n)} - \bm A_0^* = \bm A_0^{(n)} - \rho \sum_{k=1}^{K}\frac{T_k}{T}\mathcal P_{\mathcal T_r(\bm A_0^{(n)})}\!\Big(\bm G_k(\bm A_0^{(n)}) + \bm W_k^{(n)}\Big) - \bm A_0^*.
		\end{align}
		Moreover, using $\bm Y_k = \bm A_k^* \bm X_k + \bm E_k$, the local gradient admits the decomposition
		\begin{align}\label{eq:Gk-error-decomp}
			\bm G_k(\bm A_0^{(n)}) &= \frac{2}{T_k}(\bm A_0^{(n)}\bm X_k - \bm Y_k)\bm X_k^\top = \frac{2}{T_k}\left((\bm A_0^{(n)} - \bm A_k^*)\bm X_k - \bm E_k\right)\bm X_k^\top \notag\\
			& = \frac{2}{T_k}(\bm A_0^{(n)} - \bm A_0^*)\bm X_k \bm X_k^\top + \frac{2}{T_k}(\bm A_0^* - \bm A_k^*)\bm X_k \bm X_k^\top - \frac{2}{T_k}\bm E_k\bm X_k^\top.
		\end{align}
		Thus, combining \eqref{eq:A0-update-decomp} and \eqref{eq:Gk-error-decomp}, we have
		\begin{align}\label{eq:A0-one-step-expanded}
			&\Bigl\|\bm A_0^{(n)}-\rho\sum_{k=1}^K \bm Z_k^{(n)}-\bm A_0^*\Bigr\|_{\F} = \Bigl\|\bm A_0^{(n)}-\rho\sum_{k=1}^K\frac{T_k}{T}\,\mathcal P_{\mathcal T_r(\bm A_0^{(n)})}\!\Big(\bm G_k(\bm A_0^{(n)})+\bm W_k^{(n)}\Big) - \bm A_0^*\Bigr\|_{\F} \notag\\
			&= \Biggl\|\bm A_0^{(n)}-\bm A_0^* - \rho\sum_{k=1}^K\frac{T_k}{T}\,\mathcal P_{\mathcal T_r(\bm A_0^{(n)})}\!\Biggl(\frac{2}{T_k}(\bm A_0^{(n)}-\bm A_0^*)\bm X_k\bm X_k^\top + \frac{2}{T_k}(\bm A_0^*-\bm A_k^*)\bm X_k\bm X_k^\top - \frac{2}{T_k}\bm E_k\bm X_k^\top + \bm W_k^{(n)}\Biggr)\Biggr\|_{\F} \notag\\
			&\leq \underbrace{\Bigl\|(\bm A_0^{(n)}-\bm A_0^*) - 2\rho\mathcal P_{\mathcal T_r(\bm A_0^{(n)})}\bigl((\bm A_0^{(n)}-\bm A_0^*)\bm S_{\mathrm{pool}}\bigr)\Bigr\|_{\F}}_{\textnormal{Term I}} + 2\rho\underbrace{\Bigl\|\mathcal P_{\mathcal T_r(\bm A_0^{(n)})}\Bigl(\frac{1}{T}\sum_{k=1}^K(\bm A_0^*-\bm A_k^*)\bm X_k\bm X_k^\top\Bigr)\Bigr\|_{\F}}_{\textnormal{Term II}} \notag\\
			&\quad + 2\rho\underbrace{\Bigl\|\mathcal P_{\mathcal T_r(\bm A_0^{(n)})}\bigl(\bm R_{\mathrm{pool}}^{\top}\bigr)\Bigr\|_{\F}}_{\textnormal{Term III}} + \rho\underbrace{\Bigl\|\mathcal P_{\mathcal T_r(\bm A_0^{(n)})}\Bigl(\sum_{k=1}^K\frac{T_k}{T}\bm W_k^{(n)}\Bigr)\Bigr\|_{\F}}_{\textnormal{Term IV}}.
		\end{align}
		where $\bm S_{\mathrm{pool}} = T^{-1}\sum_{k=1}^K \bm X_k\bm X_k^\top$ and $\bm R_{\mathrm{pool}} = T^{-1}\sum_{k=1}^K \bm X_k\bm E_k^\top$. Next, we bound each of the four terms in turn. 
		
		For \textbf{term I}, we further decompose the expression and upper bound each component separately. In particular, we first rewrite the following quantity into a sum of terms:
		\begin{align*}
			\textnormal{Term I}
			&= \Bigl\|(\bm A_0^{(n)}-\bm A_0^*) - 2\rho\mathcal P_{\mathcal T_r(\bm A_0^{(n)})}\bigl((\bm A_0^{(n)} - \bm A_0^*)\bm S_{\mathrm{pool}}\bigr)\Bigr\|_{\F}\\
			&= \Bigl\|\mathcal P_{\mathcal T_r(\bm A_0^{(n)})}(\bm A_0^{(n)}-\bm A_0^*) + \mathcal P_{\mathcal T^\perp_r(\bm A_0^{(n)})}(\bm A_0^{(n)}-\bm A_0^*) - 2\rho\mathcal P_{\mathcal T_r(\bm A_0^{(n)})}\bigl((\bm A_0^{(n)}-\bm A_0^*)\bm S_{\mathrm{pool}}\bigr)\Bigr\|_{\F}\\
			&\leq \Bigl\|\mathcal P_{\mathcal T^\perp_r(\bm A_0^{(n)})}(\bm A_0^{(n)}-\bm A_0^*)\Bigr\|_{\F} + \Bigl\|\mathcal P_{\mathcal T_r(\bm A_0^{(n)})}\bigl((\bm A_0^{(n)}-\bm A_0^*)\bigr) - 2\rho\mathcal P_{\mathcal T_r(\bm A_0^{(n)})}\bigl((\bm A_0^{(n)}-\bm A_0^*)\bm S_{\mathrm{pool}}\bigr)\Bigr\|_{\F}\\
			&\leq \Bigl\|\mathcal P_{\mathcal T^\perp_r(\bm A_0^{(n)})}(\bm A_0^{(n)}-\bm A_0^*)\Bigr\|_{\F} + \Bigl\|\mathcal P_{\mathcal T_r(\bm A_0^{(n)})}(\bm A_0^{(n)}-\bm A_0^*) - 2\rho\mathcal P_{\mathcal T_r(\bm A_0^{(n)})}\bigl(\mathcal P_{\mathcal T_r(\bm A_0^{(n)})}(\bm A_0^{(n)}-\bm A_0^*) \bm S_{\mathrm{pool}}\bigr)\Bigr\|_{\F}\\
			&\qquad + 2\rho\Bigl\|\mathcal P_{\mathcal T_r(\bm A_0^{(n)})}\bigl(\mathcal P_{\mathcal T^\perp_r(\bm A_0^{(n)})}(\bm A_0^{(n)}-\bm A_0^*) \bm S_{\mathrm{pool}}\bigr)\Bigr\|_{\F}\\
			&= \underbrace{\Bigl\|\mathcal P_{\mathcal T^\perp_r(\bm A_0^{(n)})}(\bm A_0^{(n)}-\bm A_0^*)\Bigr\|_{\F}}_{\textnormal{Term I(a)}} + \underbrace{\Bigl\|\mathcal P_{\mathcal T_r(\bm A_0^{(n)})}\bigl(\mathcal P_{\mathcal T_r(\bm A_0^{(n)})}(\bm A_0^{(n)}-\bm A_0^*)(\bm I_{pd} - 2\rho\bm S_{\mathrm{pool}})\bigr)\Bigr\|_{\F}}_{\textnormal{Term I(b)}}\\
			&\qquad + 2\rho\underbrace{\Bigl\|\mathcal P_{\mathcal T_r(\bm A_0^{(n)})}\bigl(\mathcal P_{\mathcal T^\perp_r(\bm A_0^{(n)})}(\bm A_0^{(n)}-\bm A_0^*) \bm S_{\mathrm{pool}}\bigr)\Bigr\|_{\F}}_{\textnormal{Term I(c)}}.
		\end{align*}
		The last equality holds because $\mathcal P_{\mathcal T_r(\bm A_0^{(n)})}(\bm A_0^{(n)}-\bm A_0^*)\in \mathcal T_r(\bm A_0^{(n)})$, and the projection $\mathcal P_{\mathcal T_r(\bm A_0^{(n)})}(\cdot)$ acts as the identity on its range, i.e., for any $\bm M\in \mathcal T_r(\bm A_0^{(n)})$, $\mathcal P_{\mathcal T_r(\bm A_0^{(n)})}(\bm M)=\bm M$.

		For term I(a), by Lemma~\ref{lem:normal-second-order}, we have
		\[
			\textnormal{Term I(a)} = \Bigl\|\mathcal P_{\mathcal T^\perp_r(\bm A_0^{(n)})}(\bm A_0^{(n)}-\bm A_0^*)\Bigr\|_{\F} \leq \frac{1}{\sigma_r(\bm A_0^*)}\|\bm A_0^{(n)} - \bm A_0^*\|_{\op}\|\bm A_0^{(n)} - \bm A_0^*\|_{\F}.
		\]

		For term I(b), note that $\mathcal P_{\mathcal T_r(\bm A_0^{(n)})}(\bm A_0^{(n)}-\bm A_0^*) \in \mathcal T_r(\bm A_0^{(n)})$. Applying Lemma~\ref{lem:tangent-contraction-op} with $\bm A=\bm A_0^{(n)}$, $\bm S=\bm S_{\rm pool}$ and $\bm M=\mathcal P_{\mathcal T_r(\bm A_0^{(n)})}(\bm A_0^{(n)}-\bm A_0^*)$ together with the non-expansiveness of the projection $\mathcal P_{\mathcal T_r(\bm A_0^{(n)})}(\cdot)$ in Frobenius norm, we have
		\begin{align*}
			\textnormal{Term I(b)} &= \Bigl\|\mathcal P_{\mathcal T_r(\bm A_0^{(n)})}\Bigl(\mathcal P_{\mathcal T_r(\bm A_0^{(n)})}(\bm A_0^{(n)}-\bm A_0^*)(\bm I_{pd}-2\rho\bm S_{\rm pool})\Bigr)\Bigr\|_{\F}\\
			& \leq \|\bm I_{pd}-2\rho \bm S_{\rm pool}\|_{\op}\Bigl\|\mathcal P_{\mathcal T_r(\bm A_0^{(n)})}(\bm A_0^{(n)}-\bm A_0^*)\Bigr\|_{\F} \leq \|\bm I_{pd}-2\rho \bm S_{\rm pool}\|_{\op}\|\bm A_0^{(n)}-\bm A_0^*\|_{\F}
		\end{align*}

		For term I(c), by the non-expansiveness of the projection $\mathcal P_{\mathcal T_r(\bm A_0^{(n)})}(\cdot)$ in Frobenius norm, $\|\bm M\bm N \|_\F \leq \|\bm M\|_\F \|\bm N\|_\op$ and Lemma~\ref{lem:normal-second-order}, we have
		\begin{align*}
			\textnormal{Term I(c)} &= \Bigl\|\mathcal P_{\mathcal T_r(\bm A_0^{(n)})}\Bigl(\mathcal P_{\mathcal T_r^\perp(\bm A_0^{(n)})}(\bm A_0^{(n)}-\bm A_0^*)\,\bm S_{\mathrm{pool}}\Bigr)\Bigr\|_{\F} \leq \Bigl\|\mathcal P_{\mathcal T_r^\perp(\bm A_0^{(n)})}(\bm A_0^{(n)}-\bm A_0^*)\bm S_{\mathrm{pool}}\Bigr\|_{\F}\\
			&\leq \Bigl\|\mathcal P_{\mathcal T_r^\perp(\bm A_0^{(n)})}(\bm A_0^{(n)}-\bm A_0^*)\Bigr\|_{\F} \|\bm S_{\mathrm{pool}}\|_{\op} \leq \frac{\|\bm S_{\mathrm{pool}}\|_{\op}}{\sigma_r(\bm A_0^*)} \|\bm A_0^{(n)}-\bm A_0^*\|_{\op} \|\bm A_0^{(n)}-\bm A_0^*\|_{\F}.
		\end{align*}
		\noindent
		Therefore, combining the three components yields
		\begin{align}\label{eq:term1-bound}
			\textnormal{Term I} &\leq \textnormal{Term I(a)}+\textnormal{Term I(b)}+2\rho\textnormal{Term I(c)} \notag\\
			&\leq \|\bm I_{pd}-2\rho\bm S_{\rm pool}\|_{\op}\|\bm A_0^{(n)}-\bm A_0^*\|_{\F} + \frac{1+2\rho\|\bm S_{\rm pool}\|_{\op}}{\sigma_r(\bm A_0^*)} \|\bm A_0^{(n)}-\bm A_0^*\|_{\op}\|\bm A_0^{(n)}-\bm A_0^*\|_{\F}\notag\\
			&= \Bigl(\|\bm I_{pd}-2\rho\bm S_{\rm pool}\|_{\op} + \frac{1+2\rho\|\bm S_{\rm pool}\|_{\op}}{\sigma_r(\bm A_0^*)}\|\bm A_0^{(n)}-\bm A_0^*\|_{\op}\Bigr)\|\bm A_0^{(n)}-\bm A_0^*\|_{\F}. 
		\end{align}

		For \textbf{term~II}, by the sub-additivity of the Frobenius norm, the non-expansiveness of the projection $\mathcal P_{\mathcal T_r(\bm A_0^{(n)})}(\cdot)$ in Frobenius norm and the inequality $\|\bm M\bm N\|_{\F}\leq \|\bm M\|_{\F}\|\bm N\|_{\op}$, we have
		\begin{align}\label{eq:eq:term2-ini-bound}
			\textnormal{Term II}
			&=
			\Bigl\|\mathcal P_{\mathcal T_r(\bm A_0^{(n)})}\Bigl(\frac{1}{T}\sum_{k=1}^K(\bm A_0^*-\bm A_k^*)\bm X_k\bm X_k^\top\Bigr)\Bigr\|_{\F} = \Bigl\|\mathcal P_{\mathcal T_r(\bm A_0^{(n)})}\Bigl(\frac{1}{T}\sum_{k=1}^K\bbm \Delta_k^*\bm X_k\bm X_k^\top\Bigr)\Bigr\|_{\F}\\
			&\leq
			\sqrt{2r}\Bigl\|\frac{1}{T}\sum_{k=1}^K\bbm \Delta_k^*\bm X_k\bm X_k^\top\Bigr\|_{\op} \leq
			\frac{\sqrt{2r}}{T}\sum_{k=1}^K\|\bbm \Delta_k^*\|_{\op}\,\|\bm X_k\bm X_k^\top\|_{\op} = \sqrt{2r}\sum_{k=1}^K\frac{T_k}{T}\|\bbm \Delta_k^*\|_{\op}\,\left\|\bm S_k\right\|_{\op},\notag
		\end{align}
		where $\bbm \Delta_k^* = \bm A_k^* - \bm A_0^*$ and $\bm S_k = T_k^{-1}\bm X_k\bm X_k^\top$. Moreover, note that $\bm S_k\succeq 0$ for each $k$ and $\bm S_{\mathrm{pool}}=T^{-1}\sum_{k=1}^K \bm X_k\bm X_k^\top=\sum_{k=1}^K(T_k/T)\bm S_k$. Therefore, for each fixed $k$ we have $\bm S_{\mathrm{pool}} \succeq (T_k/T)\bm S_k$, which implies
		\begin{align}\label{eq:Sk-op-upper-bound}
			\|\bm S_k\|_{\op}=\lambda_{\max}(\bm S_k) \leq \frac{T}{T_k}\lambda_{\max}(\bm S_{\mathrm{pool}}) = \frac{T}{T_k}\|\bm S_{\mathrm{pool}}\|_{\op}.
		\end{align}
		By Lemma~\ref{lem:var-cov-conc-unified}, with probability at least $1-\exp(-Cpd)$,
		\begin{align}\label{eq:Spool-op-bound}
			\|\bm S_{\mathrm{pool}}\|_{\op} \leq \|\bbm\Sigma_{x,{\rm pool}}\|_{\op} + \|\bm S_{\mathrm{pool}}-\bbm\Sigma_{x,{\rm pool}}\|_{\op} \lesssim \|\bbm\Sigma_{x,{\rm pool}}\|_{\op} + C_{\epsilon,\mathcal A}^{\max}\sigma^2p\sqrt{\frac{d}{T}}.
		\end{align}
		For $\bbm\Sigma_{x,{\rm pool}} \succeq 0$, by Lemma~\ref{lem:Sigma-x-pd}, we have $\lambda_{\max}(\bbm\Sigma_{x,k})\leq C_{\epsilon,\mathcal A}^{(k)}\leq C_{\epsilon,\mathcal A}^{\max}$ for each client $k$. Since $\bbm\Sigma_{x,{\rm pool}}=\sum_{k=1}^K (T_k/T)\bbm\Sigma_{x,k}\succeq 0$, it follows that
		\begin{align}\label{eq:Sigma-pool-op-bound}
			\|\bbm\Sigma_{x,{\rm pool}}\|_{\op} = \lambda_{\max}(\bbm\Sigma_{x,{\rm pool}}) = \sum_{k=1}^K\frac{T_k}{T}\lambda_{\max}(\bbm\Sigma_{x,k}) \leq C_{\epsilon,\mathcal A}^{\max}.
		\end{align}
		Combining \eqref{eq:Sk-op-upper-bound}, \eqref{eq:Spool-op-bound} and \eqref{eq:Sigma-pool-op-bound}, and given $T\gtrsim p^2 d$, we conclude that with probability at least $1-\exp(-Cpd)$,
		\begin{align}\label{eq:Spool-op-bound-absorbed}
			\|\bm S_k\|_{\op} \leq \frac{T}{T_k}\|\bm S_{\mathrm{pool}}\|_{\op} \lesssim \frac{T}{T_k}\left(C_{\epsilon,\mathcal A}^{\max} + C_{\epsilon,\mathcal A}^{\max}\sigma^2p\sqrt{\frac{d}{T}}\right) \lesssim C_{\epsilon,\mathcal A}^{\max}\frac{T}{T_k}.
		\end{align}
		Plugging \eqref{eq:Spool-op-bound-absorbed} into \eqref{eq:eq:term2-ini-bound} and noting that $\|\bbm \Delta_k^*\|_\op \leq \zeta$ for any $k\in[K]$ yields that, with probability at least $1-\exp(-Cpd)$,
		\begin{align}\label{eq:term2-final-bound-hp-absorbed}
			\textnormal{Term II} = \Bigl\|\mathcal P_{\mathcal T_r(\bm A_0^{(n)})}\Bigl(\frac{1}{T}\sum_{k=1}^K(\bm A_0^*-\bm A_k^*)\bm X_k\bm X_k^\top\Bigr)\Bigr\|_{\F} \leq \sqrt{2r}\sum_{k=1}^{K}\frac{T_k}{T}\|\bbm \Delta_k^*\|_\op \|\bm S_k\|_\op \lesssim \sqrt{r}KC_{\epsilon,\mathcal A}^{\max}\zeta.
		\end{align}

		For \textbf{Term III}, note that the tangent space consists of matrices of rank at most $2r$. Then, by the non-expansiveness of the projection $\mathcal P_{\mathcal T_r(\bm A_0^{(n)})}(\cdot)$ in operator norm and the inequality $\|\bm M\|_{\F}\leq \sqrt{\mathrm{rank}(\bm M)}\|\bm M\|_{\op}$, we have
		\begin{align}\label{eq:tangent-proj-F-op}
			\textnormal{Term III} = \bigl\|\mathcal P_{\mathcal T_r(\bm A_0^{(n)})}(\bm R_{\rm pool}^\top)\bigr\|_{\F} \leq \sqrt{2r}\bigl\|\mathcal P_{\mathcal T_r(\bm A_0^{(n)})}(\bm R_{\rm pool}^\top)\bigr\|_{\op} \leq \sqrt{2r}\|\bm R_{\rm pool}\|_{\op},
		\end{align}
		This together with Lemma~\ref{lem:var-XtE-op-bound-pooled}, and given $T\gtrsim p^2 d$, there exists a constant $C>0$ such that, with probability at least $1-\exp(-Cpd)$,
		\begin{align}\label{eq:term3-hp-bound}
			\textnormal{Term III} = \Bigl\|\mathcal P_{\mathcal T_r(\bm A_0^{(n)})}(\bm R_{\rm pool}^\top)\Bigr\|_{\F} \leq \sqrt{2r}\|\bm R_{\rm pool}\|_{\op} \lesssim \sqrt{r}C_{\epsilon,\mathcal A}^{\max}\sigma^2\sqrt{\frac{p d}{T}}.
		\end{align}

		For \textbf{Term IV}, note that $(\bm W_k^{(n)})_{ij} \overset{i.i.d.}{\sim} \mathcal N(0,(\sigma_k^{(n)})^2)$ with $\sigma_k^{(n)} = \psi_k^{(n)}\sqrt{2\log(1.25/\delta)}/\epsilon$. Recall \eqref{eq:xi-final-kappa}, we then derive the upper bound for $\psi_k^{(n)}$.
		By the induction hypothesis in \eqref{eq:induction-hypothesis}, we have $\|\bm A_0^{(n)}-\bm A_0^*\|_{\op} \leq \|\bm A_0^{(n)}-\bm A_0^*\|_{\F}\leq \sigma_r(\bm A_0^*)$. Moreover, note that $\|\bm A_k^*-\bm A_0^*\|_{\op} = \|\bm \Delta_k^*\|_{\op} \leq \zeta$, for all $k\in[K]$. Therefore, by the sub-additiveness of Frobenius norm, for any $k\in[K]$, $\|\bm A_0^{(n)}-\bm A_k^*\|_{\op} \leq \|\bm A_0^{(n)}-\bm A_0^*\|_{\op}+\|\bm A_0^*-\bm A_k^*\|_{\op} \leq \sigma_r(\bm A_0^*)+\zeta$. Then, $\psi_k^{(n)} \leq (\alpha_0 + \sqrt{2r}(\sigma_r(\bm A_0^*)+\zeta)\alpha_1)/T_k$ with $\alpha_0 = \Gamma_x\Gamma_{\Delta\epsilon}+\kappa_1(\varrho)\Gamma_{\Delta\epsilon}\Gamma_\epsilon$ and $\alpha_1 = 2\kappa_1(\varrho)\Gamma_x\Gamma_{\Delta\epsilon}+\kappa_2(\varrho)\Gamma_{\Delta\epsilon}^2$. Recall that $\Gamma_x=C\sigma\sqrt{C_{\epsilon,\mathcal A}^{\max}}\sqrt{pd+\log T}$, $\Gamma_\epsilon=C\sigma\sqrt{\Lambda_{\epsilon,\max}}\sqrt{d+\log T}$, and $\Gamma_{\Delta\epsilon}=C\sigma\sqrt{\Lambda_{\epsilon,\max}}\sqrt{|\mathcal I| +\log T}$. By direct substitution of $(\Gamma_x,\Gamma_\epsilon,\Gamma_{\Delta\epsilon})$ and absorbing constants into $C$, we have
		\begin{align}\label{eq:alpha0-alpha1-sqrtform}
			\alpha_0 &\leq C\sigma^2\Bigl[\sqrt{C_{\epsilon,\mathcal A}^{\max}\Lambda_{\epsilon,\max}}\sqrt{pd+\log T}\sqrt{|\mathcal I| +\log T} + \kappa_1(\varrho)\Lambda_{\epsilon,\max}\sqrt{d+\log T}\sqrt{|\mathcal I| +\log T}\Bigr]\notag\\
			&\leq C\sigma^2\Lambda_{\epsilon,\max}\,p\,\sqrt{C_{\epsilon,\mathcal A}^{\max}}\sqrt{pd+\log T}\sqrt{|\mathcal I| +\log T} =: \bar \alpha,\\[0.3em]
			\alpha_1 &\leq C\sigma^2\Bigl[\kappa_1(\varrho)\sqrt{C_{\epsilon,\mathcal A}^{\max}\Lambda_{\epsilon,\max}}\sqrt{pd+\log T}\sqrt{|\mathcal I| +\log T} + \kappa_2(\varrho)\Lambda_{\epsilon,\max}(|\mathcal I| +\log T)\Bigr]\notag\\
			&\leq C\sigma^2\Lambda_{\epsilon,\max}\,p\,\sqrt{C_{\epsilon,\mathcal A}^{\max}}\sqrt{pd+\log T}\sqrt{|\mathcal I| +\log T} =: \bar \alpha,
		\end{align}
		where we use $\kappa_1(\varrho)\lesssim p$, $\kappa_2(\varrho)\lesssim p$ and $|\mathcal I| \leq d \leq pd$. Let $\xi_k^{(n)} = \sqrt{r}(\sigma_r(\bm A_0^*)+\zeta)\bar \alpha/T_k$. Then,
		\begin{align}\label{eq:gradient-xi-upper-bound}
		\psi_k^{(n)} \leq \frac{C}{T_k}\Bigl(\alpha_0 + \sqrt{r}\bigl(\sigma_r(\bm A_0^*)+\zeta\bigr)\alpha_1\Bigr) \leq \frac{C}{T_k}\sqrt{r}\bigl(\sigma_r(\bm A_0^*)+\zeta\bigr)\,\bar\alpha = C\xi_k^{(n)}.
		\end{align}
		
		By the non-expansiveness of the projection $\mathcal P_{\mathcal T_r(\bm A_0^{(n)})}(\cdot)$ in Frobenius norm and the inequality $\|\bm M\|_{\F}\leq \sqrt{\mathrm{rank}(\bm M)}\|\bm M\|_{\op}$, we have
		\begin{align}\label{eq:term4-step1}
			\textnormal{Term IV} &= \Bigl\|\mathcal P_{\mathcal T_r(\bm A_0^{(n)})}\!\Bigl(\sum_{k=1}^K\frac{T_k}{T}\bm W_k^{(n)}\Bigr)\Bigr\|_{\F} \leq \sqrt{2r}\Bigl\|\sum_{k=1}^K\frac{T_k}{T}\bm W_k^{(n)}\Bigr\|_{\op}.
		\end{align}

		For each client $k$, $\bm W_k^{(n)}$ has $i.i.d.$ entries distributed as $\mathcal N(0,(\sigma_k^{(n)})^2)$.
		To achieve $(\varepsilon/N_g,\delta/N_g)$-DP for each client $k\in[K]$, by Lemma \ref{lem:gaussian-mech-prop}, it suffices to set $\sigma_k^{(n)} = \xi_k^{(n)}\sqrt{2\log(1.25N_g/\delta)}/\varepsilon$. Denote $\widetilde{\bm W}^{(n)} = \sum_{k=1}^K(T_k/T)\bm W_k^{(n)}\in\mathbb R^{d\times pd}$. Since $\{\bm W_k^{(n)}\}_{k=1}^K$ are independent across $k$, for each fixed $(i,j)$, we have $\widetilde {\bm W}_{ij}^{(n)} = \sum_{k=1}^K(T_k/T)(\bm W_k^{(n)})_{ij} \sim \mathcal N(0, (\widetilde\sigma^{(n)})^2)$, with $(\widetilde\sigma^{(n)})^2 = \sum_{k=1}^K(T_k/T)^2(\sigma_k^{(n)})^2$.
		Hence $\widetilde{\bm W}^{(n)}$ has $i.i.d.$ entries following $\mathcal N(0,(\widetilde\sigma^{(n)})^2)$. Applying Lemma~\ref{lem:gaussian-opnorm-hp-strong} to $\widetilde{\bm W}^{(n)}$, there exists a constant $C>0$ such that with probability at least $1-\exp(-C\log T)$,
		\begin{align}\label{eq:term4-op-bound}
			\|\widetilde{\bm W}^{(n)}\|_{\op}\leq C\,\widetilde\sigma^{(n)}\Bigl(\sqrt{pd}+\sqrt{\log T}\Bigr).
		\end{align}
		Note that the aggregated noise scale satisfies
		\begin{align}\label{eq:aggregated-noise-scale}
			(\widetilde\sigma^{(n)})^2 = \sum_{k=1}^K\Bigl(\frac{T_k}{T}\Bigr)^2(\sigma_k^{(n)})^2 = \frac{2N_g^2}{\varepsilon^2}\log\Bigl(\frac{1.25N_g}{\delta}\Bigr)\sum_{k=1}^K\Bigl(\frac{T_k}{T}\Bigr)^2(\xi_k^{(n)})^2.
		\end{align}
		Recall that $\xi_k^{(n)} = \sqrt{r}(\sigma_r(\bm A_0^*)+\zeta)\bar \alpha/T_k$. Then we have
		\begin{align}\label{eq:xi-sum-plugged-beta}
			\Biggl(\sum_{k=1}^K\Bigl(\frac{T_k}{T}\Bigr)^2(\xi_k^{(n)})^2\Biggr)^{1/2}\leq \frac{1}{T}\Biggl(\sum_{k=1}^K\biggl(\sqrt{r}\bigl(\sigma_r(\bm A_0^*) + \zeta\bigr)\bar \alpha\biggr)^2\Biggr)^{1/2} = \frac{\sqrt {rK}}{T}\bigl(\bigl(\sigma_r(\bm A_0^*) + \zeta\bigr)\bar \alpha\bigr).
		\end{align}
		Then, coditional on $\mathcal E = \bigcap_{k=1}^K \mathcal E_k$, combining \eqref{eq:term4-step1}, \eqref{eq:term4-op-bound}, \eqref{eq:aggregated-noise-scale} and \eqref{eq:xi-sum-plugged-beta} yields that, with probability at least $1-\exp(-C\log T)$,
		\begin{align}\label{eq:term4-final-hp-tight-final-beta}
			\textnormal{Term IV} &\leq C2\sqrt{rK}\Bigl(\sqrt{pd}+\sqrt{\log T}\Bigr) \frac{N_g}{\varepsilon}\sqrt{\log\Bigl(\frac{1.25N_g}{\delta}\Bigr)} \cdot \frac{\sqrt{r}}{T}\bigl(\sigma_r(\bm A_0^*) + \zeta\bigr)\bar \alpha \notag\\
			&\leq C\sqrt{C_{\epsilon,\mathcal A}^{\max}}\sigma^2r(\sigma_r(\bm A_0^*)+\zeta)p\frac{N_g}{\varepsilon}\Lambda_{\epsilon,\max}\sqrt{\log\Bigl(\frac{1.25N_g}{\delta}\Bigr)}\frac{\sqrt{K(|\mathcal I| +\log T)}(pd+\log T)}{T} \notag\\
			&=: \mathsf{Error}_{\mathrm{DP}}^{(n)}.
		\end{align}

		Denote $\alpha_n(\rho) = \|\bm I_{pd}-2\rho\bm S_{\rm pool}\|_{\op} + (1+2\rho\|\bm S_{\rm pool}\|_{\op})e_n$ with $e_n=\|\bm A_0^{(n)}-\bm A_0^*\|_{\op}/\sigma_r(\bm A_0^*)$. Thus, combining the bounds for \textbf{Term I}--\textbf{Term IV} in \eqref{eq:term1-bound}, \eqref{eq:term2-final-bound-hp-absorbed}, \eqref{eq:term3-hp-bound}, and \eqref{eq:term4-final-hp-tight-final-beta}, we conclude that on an event of probability at least $1-\exp(-Cpd)-\exp(-C\log T)$,
		\begin{align}\label{eq:four-terms-combined}
			&\Bigl\|\bm A_0^{(n)}-\rho\sum_{k=1}^K \bm Z_k^{(n)}-\bm A_0^*\Bigr\|_{\F} = \textnormal{Term I} + 2\rho\textnormal{Term II} + 2\rho\textnormal{Term III} + \rho\textnormal{Term IV} \notag\\
			&\leq \alpha_n(\rho)\|\bm A_0^{(n)}-\bm A_0^*\|_{\F} + C\rho \sqrt{r}KC_{\epsilon,\mathcal A}^{\max}\zeta + C\rho\sqrt{r}\,C_{\epsilon,\mathcal A}^{\max}\sigma^2\sqrt{\frac{pd}{T}} \notag\\
			&\quad + C\rho \sqrt{C_{\epsilon,\mathcal A}^{\max}}\sigma^2r(\sigma_r(\bm A_0^*)+\zeta)p\frac{N_g}{\varepsilon}\Lambda_{\epsilon,\max}\sqrt{\log\Bigl(\frac{1.25N_g}{\delta}\Bigr)}\frac{\sqrt{K(|\mathcal I| +\log T)}(pd+\log T)}{T}.
		\end{align}
		\noindent
		Set $\rho^*=2/(c_{\epsilon,\mathcal A}^{\min}+3C_{\epsilon,\mathcal A}^{\max})$. Then, we show that $\alpha_n(\rho^*) < 1$ for any $n\in \{1,\cdots,N_g-1\}$. By Lemma~\ref{lem:var-cov-conc-unified}, given $T\gtrsim p^2d$, there exist constants $c_0, C>0$, with probability at least $1-\exp(-Cpd)$,
		\begin{align}\label{eq:S-pool-op-error-bound}
			\bigl\|\bm S_{\rm pool}-\bbm\Sigma_{x,{\rm pool}}\bigr\|_{\op} \leq c_0\,C_{\epsilon,\mathcal A}^{\max}\sigma^2p\sqrt{\frac{d}{T}} =: \delta_T.
		\end{align}

		By Lemma~\ref{lem:Sigma-x-pd}, for each client $k\in[K]$, we have $c_{\epsilon,\mathcal A}^{(k)} \leq \lambda_{\min}(\bbm\Sigma_{x,k}) \leq \lambda_{\max}(\bbm\Sigma_{x,k}) \leq C_{\epsilon,\mathcal A}^{(k)}$. Recall that $\bbm\Sigma_{x,{\rm pool}}=\sum_{k=1}^K (T_k/T)\bbm\Sigma_{x,k}$. Therefore,
		\begin{align}\label{eq:Sigma-pool-eig-bounds}
			c_{\epsilon,\mathcal A}^{\min} \leq \sum_{k=1}^K\frac{T_k}{T}\,c_{\epsilon,\mathcal A}^{(k)} \leq \lambda_{\min}(\bbm\Sigma_{x,{\rm pool}}) \leq \lambda_{\max}(\bbm\Sigma_{x,{\rm pool}}) \leq \sum_{k=1}^K\frac{T_k}{T}\,C_{\epsilon,\mathcal A}^{(k)} \leq C_{\epsilon,\mathcal A}^{\max}.
		\end{align}
		Next, since both $\bm S_{\rm pool}$ and $\bbm\Sigma_{x,{\rm pool}}$ are symmetric, Weyl's inequality \citep{weyl1912asymptotische} together with \eqref{eq:S-pool-op-error-bound} implies that, with probability at least $1-\exp(-Cpd)$,
		\begin{align}\label{eq:Weyl-pool}
			\lambda_{\min}(\bbm\Sigma_{x,{\rm pool}})-\delta_T \leq \lambda_{\min}(\bm S_{\rm pool}) \leq \lambda_{\max}(\bm S_{\rm pool}) \leq \lambda_{\max}(\bbm\Sigma_{x,{\rm pool}})+\delta_T.
		\end{align}
		Combining \eqref{eq:Sigma-pool-eig-bounds} and \eqref{eq:Weyl-pool} yields that, with probability at least $1-\exp(-Cpd)$,
		\begin{align}\label{eq:Spool-eig-bounds-nice}
			c_{\epsilon,\mathcal A}^{\min}-\delta_T \leq \lambda_{\min}(\bm S_{\rm pool}) \leq \lambda_{\max}(\bm S_{\rm pool}) \leq C_{\epsilon,\mathcal A}^{\max}+\delta_T.
		\end{align}
		Finally, taking $T\gtrsim p^2d$ sufficiently large so that $\delta_T\leq \tfrac12 c_{\epsilon,\mathcal A}^{\min}$ and
		$\delta_T\leq \tfrac12 C_{\epsilon,\mathcal A}^{\max}$. Then, \eqref{eq:Spool-eig-bounds-nice} simplifies to
		\begin{align}\label{eq:Spool-op-bound-clean}
			\frac12\,c_{\epsilon,\mathcal A}^{\min} \leq \lambda_{\min}(\bm S_{\rm pool}) \leq \lambda_{\max}(\bm S_{\rm pool}) \leq \frac32\,C_{\epsilon,\mathcal A}^{\max}.
		\end{align}

		Donote $\lambda(\bm S_{\rm pool}) = \{\lambda_1(\bm S_{\rm pool}),\cdots,\lambda_{pd}(\bm S_{\rm pool})\}$ as the set of eigenvalues of $\bm S_{\rm pool}$.
		Hence, for any $\rho>0$, $\|\bm I_{pd}-2\rho\,\bm S_{\rm pool}\|_{\op} = \max_{\lambda\in\lambda(\bm S_{\rm pool})}|1-2\rho\lambda| \leq \max\{|1-\rho c_{\epsilon,\mathcal A}^{\min}|,\ |1-3\rho C_{\epsilon,\mathcal A}^{\max}|\}$, and also $1+2\rho\|\bm S_{\rm pool}\|_{\op}\leq 1+3\rho C_{\epsilon,\mathcal A}^{\max}$. Note that if $\rho^* \leq 1/(3C_{\epsilon,\mathcal A}^{\max})$, then $|1-3\rho^* C_{\epsilon,\mathcal A}^{\max}| = 1-3\rho^* C_{\epsilon,\mathcal A}^{\max} \leq 1-\rho^* c_{\epsilon,\mathcal A}^{\min} = |1-\rho^* c_{\epsilon,\mathcal A}^{\min}|$. If $\rho^* > 1/(3C_{\epsilon,\mathcal A}^{\max})$, then $|1-3\rho^* C_{\epsilon,\mathcal A}^{\max}| = 3\rho^* C_{\epsilon,\mathcal A}^{\max}-1$, and since $\rho^* (c_{\epsilon,\mathcal A}^{\min}+3C_{\epsilon,\mathcal A}^{\max})=2$, we have $1-\rho^* c_{\epsilon,\mathcal A}^{\min} = 3\rho^* C_{\epsilon,\mathcal A}^{\max}-1 >0$. Therefore, in either case, $\|\bm I_{pd}-2\rho^* \bm S_{\rm pool}\|_{\op} \leq 1-\rho^* c_{\epsilon,\mathcal A}^{\min}$.
		With the choice $\rho^*=2/(c_{\epsilon,\mathcal A}^{\min}+3C_{\epsilon,\mathcal A}^{\max})$, we have
		\begin{align}\label{eq:Iminus2rhoSpool-op-exact-rho-part}
			\|\bm I_{pd}-2\rho^* \bm S_{\rm pool}\|_{\op} \leq \frac{3C_{\epsilon,\mathcal A}^{\max} - c_{\epsilon,\mathcal A}^{\min}}{3C_{\epsilon,\mathcal A}^{\max} + c_{\epsilon,\mathcal A}^{\min}}.
		\end{align}
		Moreover, under the induction hypothesis, we have
		\begin{align}\label{eq:e_n}
			e_n = \frac{\|\bm A_0^{(n)}-\bm A_0^*\|_{\op}}{\sigma_r(\bm A_0^*)}
			\leq \frac{\|\bm A_0^{(n)}-\bm A_0^*\|_{\F}}{\sigma_r(\bm A_0^*)}
			\leq \frac{c_{\epsilon,\mathcal A}^{\min}(3C_{\epsilon,\mathcal A}^{\max}+c_{\epsilon,\mathcal A}^{\min})}{20(6C_{\epsilon,\mathcal A}^{\max}+c_{\epsilon,\mathcal A}^{\min})^2}.
		\end{align}
		For $1+2\rho^*\|\bm S_{\rm pool}\|_{\op}$, using $\|\bm S_{\rm pool}\|_{\op}\leq \frac32 C_{\epsilon,\mathcal A}^{\max}$, we have
		\begin{align}\label{eq:one-plus-2rhoSpool-rhostar-rho-part}
			1+2\rho^*\|\bm S_{\rm pool}\|_{\op} \leq 1+3\rho^* C_{\epsilon,\mathcal A}^{\max} = \frac{c_{\epsilon,\mathcal A}^{\min}+9C_{\epsilon,\mathcal A}^{\max}}{c_{\epsilon,\mathcal A}^{\min}+3C_{\epsilon,\mathcal A}^{\max}}.
		\end{align}
		Then, plugging \eqref{eq:Iminus2rhoSpool-op-exact-rho-part}, \eqref{eq:e_n} and \eqref{eq:one-plus-2rhoSpool-rhostar-rho-part} into $\alpha_n(\rho^*)$ yields
		\begin{align}\label{eq:alpha-rhostar-bound-rho-part}
			\alpha_n(\rho^*) &\leq \frac{3C_{\epsilon,\mathcal A}^{\max} - c_{\epsilon,\mathcal A}^{\min}}{3C_{\epsilon,\mathcal A}^{\max} + c_{\epsilon,\mathcal A}^{\min}} + \frac{c_{\epsilon,\mathcal A}^{\min}(3C_{\epsilon,\mathcal A}^{\max}+c_{\epsilon,\mathcal A}^{\min})}{20(6C_{\epsilon,\mathcal A}^{\max}+c_{\epsilon,\mathcal A}^{\min})^2} \cdot \frac{c_{\epsilon,\mathcal A}^{\min}+9C_{\epsilon,\mathcal A}^{\max}}{c_{\epsilon,\mathcal A}^{\min}+3C_{\epsilon,\mathcal A}^{\max}}\notag\\
            &\leq \frac{3C_{\epsilon,\mathcal A}^{\max} - c_{\epsilon,\mathcal A}^{\min}}{3C_{\epsilon,\mathcal A}^{\max} + c_{\epsilon,\mathcal A}^{\min}} + \frac{c_{\epsilon,\mathcal A}^{\min}}{c_{\epsilon,\mathcal A}^{\min}+9C_{\epsilon,\mathcal A}^{\max}} \cdot \frac{c_{\epsilon,\mathcal A}^{\min}+9C_{\epsilon,\mathcal A}^{\max}}{c_{\epsilon,\mathcal A}^{\min}+3C_{\epsilon,\mathcal A}^{\max}} = \frac{3C_{\epsilon,\mathcal A}^{\max}}{3C_{\epsilon,\mathcal A}^{\max} + c_{\epsilon,\mathcal A}^{\min}} < 1,
		\end{align}
		Recall that $\|\bm A_0^{(n)}-\rho^*\sum_{k=1}^K \bm Z_k^{(n)}-\bm A_0^*\|_{\F}\leq \alpha_n(\rho)\|\bm A_0^{(n)}-\bm A_0^*\|_{\F} + \mathsf{Rem}^{(n)}(\rho)$, where $\mathsf{Rem}^{(n)}(\rho) = C\rho K C_{\epsilon,\mathcal A}^{\max}\zeta + C\rho\sqrt{r}C_{\epsilon,\mathcal A}^{\max}\sigma^2\sqrt{pd/T} + \rho\mathsf{Error}_{\mathrm{DP}}^{(n)}$.
		Note that Theorem \ref{thm:common-A0-error-bounds} implies that $C\rho \sqrt{r}K C_{\epsilon,\mathcal A}^{\max}\allowbreak \zeta \leq c_{\epsilon,\mathcal A}^{\min}R/4(3C_{\epsilon,\mathcal A}^{\max} + c_{\epsilon,\mathcal A}^{\min})$. Moreover, provided $T \gtrsim\Bigl(\sqrt{C_{\epsilon,\mathcal A}^{\max}K|\mathcal I|}\sigma^2\sigma_r(\bm A_0^*)p\Lambda_{\epsilon,\max}(\sigma_r(\bm A_0^*)\allowbreak+\phi(s_q)) \vee \bigl(C_{\epsilon,\mathcal A}^{\max}\sigma^2\bigr)^2\Bigr)rpd$, we have $C\rho^*\sqrt{r}C_{\epsilon,\mathcal A}^{\max}\sigma^2\sqrt{pd/T} + \rho^*\mathsf{Error}_{\mathrm{DP}}^{(n)} < c_{\epsilon,\mathcal A}^{\min}R/4(3C_{\epsilon,\mathcal A}^{\max} + c_{\epsilon,\mathcal A}^{\min})$, uniformly for $n\leq N_g-1$.
		Consequently, conditional on $\mathcal E = \bigcap_{k=1}^K \mathcal E_k$ and under the induction hypothesis $\|\bm A_0^{(n)}-\bm A_0^*\|_{\F}\leq R$, by the choice of $\rho^*=2/(c_{\epsilon,\mathcal A}^{\min}+3C_{\epsilon,\mathcal A}^{\max})$, we have
		\begin{align*}
			\Bigl\|\bm A_0^{(n)}-\rho^*\sum_{k=1}^K \bm Z_k^{(n)}-\bm A_0^*\Bigr\|_{\op}
			&\leq \Bigl\|\bm A_0^{(n)}-\rho^*\sum_{k=1}^K \bm Z_k^{(n)}-\bm A_0^*\Bigr\|_{\F} \leq \alpha_n(\rho^*)\|\bm A_0^{(n)}-\bm A_0^*\|_{\F}+\mathsf{Rem}^{(n)}(\rho^*)\\
            &\leq \frac{6C_{\epsilon,\mathcal A}^{\max} + c_{\epsilon,\mathcal A}^{\min}}{2(3C_{\epsilon,\mathcal A}^{\max} + c_{\epsilon,\mathcal A}^{\min})}R < R < \frac{\sigma_r(\bm A_0^*)}{8}.
		\end{align*}
        This together with Lemma~\ref{lem:matrix-pert-lesssim} by setting $\bm M' = \bm A_0^{(n)}-\rho^*\sum_{k=1}^K \bm Z_k^{(n)}$ and $\bm M = \bm A_0^*$, we have
		\begin{align}\label{eq:final-upper-bound}
			&\|\bm A_0^{(n+1)}-\bm A_0^*\|_{\F} =\Bigl\|\mathrm{SVD}_r\Bigl(\bm A_0^{(n)} - \rho^*\sum_{k=1}^K \bm Z_k^{(n)}\Bigr) - \bm A_0^*\Bigr\|_{\F} \notag\\
			&\leq \Bigl\|\bm A_0^{(n)} - \rho^*\sum_{k=1}^K \bm Z_k^{(n)} - \bm A_0^*\Bigr\|_{\F} + \frac{40\Bigl\|\bm A_0^{(n)}-\rho^*\sum_{k=1}^K \bm Z_k^{(n)}-\bm A_0^*\Bigr\|_{\op}\Bigl\|\bm A_0^{(n)} - \rho^*\sum_{k=1}^K \bm Z_k^{(n)} - \bm A_0^*\Bigr\|_{\F}}{\sigma_r(\bm A_0^*)} \notag\\
			&\leq \frac{6C_{\epsilon,\mathcal A}^{\max} + c_{\epsilon,\mathcal A}^{\min}}{2(3C_{\epsilon,\mathcal A}^{\max} + c_{\epsilon,\mathcal A}^{\min})}R + \frac{40}{\sigma_r(\bm A_0^*)}\left(\frac{6C_{\epsilon,\mathcal A}^{\max} + c_{\epsilon,\mathcal A}^{\min}}{2(3C_{\epsilon,\mathcal A}^{\max} + c_{\epsilon,\mathcal A}^{\min})}R\right)^2\\
            &= \frac{6C_{\epsilon,\mathcal A}^{\max} + c_{\epsilon,\mathcal A}^{\min}}{2(3C_{\epsilon,\mathcal A}^{\max} + c_{\epsilon,\mathcal A}^{\min})}R
            + 10\frac{(6C_{\epsilon,\mathcal A}^{\max} + c_{\epsilon,\mathcal A}^{\min})^2}{(3C_{\epsilon,\mathcal A}^{\max} + c_{\epsilon,\mathcal A}^{\min})^2}\cdot\frac{c_{\epsilon,\mathcal A}^{\min}(3C_{\epsilon,\mathcal A}^{\max}+c_{\epsilon,\mathcal A}^{\min})}{20(6C_{\epsilon,\mathcal A}^{\max}+c_{\epsilon,\mathcal A}^{\min})^2}R \notag\\
            &= \frac{6C_{\epsilon,\mathcal A}^{\max} + c_{\epsilon,\mathcal A}^{\min}}{2(3C_{\epsilon,\mathcal A}^{\max} + c_{\epsilon,\mathcal A}^{\min})}R + \frac{c_{\epsilon,\mathcal A}^{\min}}{2(3C_{\epsilon,\mathcal A}^{\max}+c_{\epsilon,\mathcal A}^{\min})}R = R. \notag
            \end{align}
        This completes the induction.

		Note that all high probability events needed in the induction holds uniformly for any iteration $n$. Hence, these events can be fixed throughout the induction, and the probability does not accumulate with the number of iterations. For clarity, we still state the corresponding probabilistic inequalities at each step to indicate where they are invoked.
		
		\noindent
		\textbf{Part III (Privacy guarantee and final error bound).}
		We first translate the sensitivity control into a privacy guarantee for the whole federated Stage~I procedure, and then combine the induction recursion to obtain the final estimation error bound. The preceding sensitivity bound is established on the high-probability event $\mathcal E_k$. Conditional on $\mathcal E_k$, for every realized dataset satisfying this event, the pre-noise gradient-transfer message of client $k$ has sensitivity at most $\xi_k^{(n)}$ at iteration $n$. Therefore, by the Gaussian mechanism, releasing the noisy gradient-transfer message at iteration $n$ is $(\varepsilon/N_g,\delta/N_g)$-DP for client $k\in[K]$ conditional on $\mathcal E_k$. By the basic composition theorem for differential privacy, e.g., \citet[Theorem~3.1.6]{dwork2014algorithmic}, the collection of $N_g$ noisy gradient-transfer messages is $(\varepsilon,\delta)$-DP for client $k\in[K]$ conditional on $\mathcal E_k$. Since the remaining steps of the algorithm are post-processing of these messages, the entire Stage~I algorithm is also conditional $(\varepsilon,\delta)$-DP for client $k\in[K]$ on $\mathcal E_k$. Equivalently, because $\mathbb P(\mathcal E_k^c)\leq \exp(-C\log T)$ by \eqref{eq:prob-Ek}, the realized Stage~I federated algorithm satisfies selective federated conditional $(\varepsilon,\delta)$-DP for client $k\in[K]$, with probability at least $1-\exp(-C\log T)$ over the VAR data-generating process.

		To conclude, we choose an appropriate rate of the iteration number $N_g$ to derive the final error bound. Note that $\alpha_n(\rho^*)<1$ for all $n\in[N_g]$ under the induction hypothesis, and thus $\alpha(\rho^*) = \sup_{n\in[N_g]} \alpha_n(\rho^*)<1$.
		Then, by the non-expansiveness of $\mathrm{SVD}_r(\cdot)$ in Frobenius norm and \eqref{eq:four-terms-combined} under \eqref{eq:final-upper-bound}, we have, with probability at least $1-\exp(-Cpd)-\exp(-C\log T)$,
		\begin{align*}
			&\|\bm A_0^{(N_g)}-\bm A_0^*\|_{\F}
			= \Bigl\|\mathrm{SVD}_r\Bigl(\bm A_0^{(N_g-1)} - \rho^*\sum_{k=1}^K \bm Z_k^{(N_g-1)}\Bigr) - \bm A_0^*\Bigr\|_{\F} \\
			&\lesssim \Bigl\|\bm A_0^{(N_g-1)} - \rho^*\sum_{k=1}^K \bm Z_k^{(N_g-1)} - \bm A_0^*\Bigr\|_{\F} \\
			&\leq \alpha(\rho^*)\|\bm A_0^{(N_g-1)}-\bm A_0^*\|_{\F}
			+ C\rho^*\sqrt{r}C_{\epsilon,\mathcal A}^{\max}K\zeta
			+ C\rho^*\sigma^2C_{\epsilon,\mathcal A}^{\max}\sqrt{r}\sqrt{\frac{pd}{T}} \\
			&\quad + C\rho^* \sigma^2r(\sigma_r(\bm A_0^*)+\zeta)p\sqrt{C_{\epsilon,\mathcal A}^{\max}}\Lambda_{\epsilon,\max}\frac{N_g}{\varepsilon}\sqrt{\log\Bigl(\frac{1.25N_g}{\delta}\Bigr)}\frac{\sqrt{K(|\mathcal I| +\log T)}(pd+\log T)}{T} \\
			&\leq \alpha^{N_g}(\rho^*)\|\bm A_0^{(0)}-\bm A_0^*\|_{\F}
			+ \frac{C\rho^*}{1-\alpha(\rho^*)}\sqrt{r} C_{\epsilon,\mathcal A}^{\max}K\zeta
			+ \frac{C\rho^*}{1-\alpha(\rho^*)}\sigma^2C_{\epsilon,\mathcal A}^{\max}\sqrt{r}\sqrt{\frac{pd}{T}} \\
			&\quad+ \frac{C\rho^*}{1-\alpha(\rho^*)}\sigma^2r(\sigma_r(\bm A_0^*)+\zeta)p\sqrt{C_{\epsilon,\mathcal A}^{\max}}\Lambda_{\epsilon,\max}\frac{N_g}{\varepsilon}\sqrt{\log\Bigl(\frac{1.25N_g}{\delta}\Bigr)}\frac{\sqrt{K(|\mathcal I| +\log T)}(pd+\log T)}{T}.
		\end{align*}

		Here the factor $\rho^*/(1-\alpha(\rho^*))$ comes from the geometric recursion. Since $\rho^*=2/(c_{\epsilon,\mathcal A}^{\min}+3C_{\epsilon,\mathcal A}^{\max})$ and $1-\alpha(\rho^*)$ is of order $c_{\epsilon,\mathcal A}^{\min}/(c_{\epsilon,\mathcal A}^{\min}+3C_{\epsilon,\mathcal A}^{\max})$, we have $\rho^*/(1-\alpha(\rho^*))\lesssim 1/c_{\epsilon,\mathcal A}^{\min}$. Now let $N_g \asymp \log(T)$, we have $\alpha^{N_g}(\rho^*)\|\bm A_0^{(0)}-\bm A_0^*\|_{\F} \asymp 1/T$ which can be absorbed by the last error term, indicating that there is no improvement by running the algorithm after $O(\log T)$ iterations. Thus, given $\sigma_k^{(n)}\asymp \xi_k^{(n)}\log T \sqrt{\log(1.25\log T/\delta)}/\varepsilon$, with 
		\[
			\xi_k^{(n)}\asymp \frac{\sqrt{C_{\epsilon,\mathcal A}^{\max}r}\sigma^2\bigl(\sigma_r(\bm A_0^*)+\zeta\bigr)p\Lambda_{\epsilon,\max}\sqrt{pd+\log T}\sqrt{|\mathcal I|+\log T}}{T_k},
		\]
		Denote $\kappa_{\epsilon,\mathcal A} = C_{\epsilon,\mathcal A}^{\max}/c_{\epsilon,\mathcal A}^{\min}$.
		Then, conditional on the event $\mathcal E = \bigcap_{k=1}^K \mathcal E_k$, with probability at least $1-\exp(-Cpd)-\exp(-C\log T)$, 
		\begin{align*}
			\|\widehat{\bm A}_0-\bm A_0^*\|_{\F}
			\lesssim \kappa_{\epsilon,\mathcal A}\Biggl(
			&\underbrace{\sigma^2\sqrt{r}\sqrt{\frac{pd}{T}}}_{\text{statistical estimation error }(\mathsf{Error}_{\mathrm{stat}})}
			+ \underbrace{\sqrt{r}K\zeta}_{\text{client heterogeneity error }(\mathsf{Error}_{\mathrm{h}})}\\
			& + \underbrace{\frac{\sigma^2r\bigl(\sigma_r(\bm A_0^*)+\zeta\bigr)p\Lambda_{\epsilon,\max}}{\sqrt{C_{\epsilon,\mathcal A}^{\max}}}\frac{\log T}{\varepsilon}\sqrt{\log\Bigl(\frac{1.25\log T}{\delta}\Bigr)} \frac{\sqrt{K(|\mathcal I|+\log T)}(pd+\log T)}{T}}_{\text{DP Gaussian-noise error }(\mathsf{Error}_{\mathrm{DP}})}
			\Biggr).
		\end{align*}
		To make the statistical and privacy-induced stochastic errors negligible, it is enough to impose a pooled sample-size condition that dominates the covariance-concentration requirement and the Gaussian-noise contribution. 
		Specifically, it suffices to require $T\gtrsim \{(C_{\epsilon,\mathcal A}^{\max}\sigma^2)^2 \vee \sigma^2(\sigma_r(\bm A_0^*)+\zeta)p\Lambda_{\epsilon,\max}/(\varepsilon\sqrt{C_{\epsilon,\mathcal A}^{\max}})\sqrt{K|\mathcal I|\log(1.25/\delta)}\}rpd$.
		Moreover, by \eqref{eq:prob-Ek} and the Bonferroni inequality, we have
		\[
			\mathbb P(\mathcal E) = \mathbb P\Bigl(\bigcap_{k=1}^K \mathcal E_k\Bigr) \geq \sum_{k=1}^K \mathbb P(\mathcal E_k) - (K-1) \geq 1-K\exp(-C\log T).
		\]
		Combining the above arguments, we conclude that with probability at least $1-\exp(-Cpd)-K\exp(-C\log T)$,
		\[
			\|\widehat{\bm A}_0-\bm A_0^*\|_{\F} \lesssim \kappa_{\epsilon,\mathcal A}\left(\mathsf{Error}_{\mathrm{stat}} + \mathsf{Error}_{\mathrm{h}} + \mathsf{Error}_{\mathrm{DP}}\right),
		\]
		where the three error terms are defined in the preceding display with the condition-number factor $\kappa_{\epsilon,\mathcal A}=C_{\epsilon,\mathcal A}^{\max}/c_{\epsilon,\mathcal A}^{\min}$.
		This concludes the proof. 
	\end{proof}

	\begin{proof}[\textbf{Proof of Theorem \ref{thm:personalized-error-bounds}}]
		Fix an arbitrary client $k\in[K]$. Note that the total estimation error can be decomposed into an optimization (algorithmic) error and a statistical error:
		\begin{equation*}
			\|\widehat{\bbm \Delta}_k - \bbm \Delta_k^*\|_{\F} \leq \|\widehat{\bbm \Delta}_k - \widehat{\bbm \Delta}_k^{\mathrm{opt}}\|_{\F} + \|\widehat{\bbm \Delta}_k^{\mathrm{opt}} - \bbm \Delta_k^*\|_{\F}.
		\end{equation*}
		Accordingly, we first establish a bound for the statistical error $\|\widehat{\bbm \Delta}_k^{\mathrm{opt}} - \bbm \Delta_k^*\|_{\F}$, and then control the optimization error $\|\widehat{\bbm \Delta}_k - \widehat{\bbm \Delta}_k^{\mathrm{opt}}\|_{\F}$ using the convergence guarantee of FISTA.\\
	\textbf{Part I (Statistical error).} By the optimality of $\widehat{\bbm \Delta}_k^{\mathrm{opt}}$ in \eqref{eq:stage2-opt}, we have
		\begin{align*}
			&\frac{1}{T_k} \sum_{t=1}^{T_k}\left\|\bbm{y}_{k,t} - (\widehat{\bm A}_0 + \widehat{\bm\Delta}_k^{\mathrm{opt}})\bbm{x}_{k,t}\right\|_2^2 + \varpi_k \|\widehat{\bm\Delta}_k^{\mathrm{opt}}\|_1 \leq \frac{1}{T_k} \sum_{t=1}^{T_k}\left\|\bbm{y}_{k,t} - (\widehat{\bm A}_0 + \bbm \Delta_k^*)\bbm{x}_{k,t}\right\|_2^2 + \varpi_k \|\bbm \Delta_k^*\|_1.\\
			\Rightarrow \ &\frac{1}{T_k} \sum_{t=1}^{T_k}\left\|\bbm{y}_{k,t} - (\widehat{\bm A}_0 + \widehat{\bm\Delta}_k^{\mathrm{opt}})\bbm{x}_{k,t}\right\|_2^2 - \frac{1}{T_k} \sum_{t=1}^{T_k}\left\|\bbm{y}_{k,t} - (\widehat{\bm A}_0 + \bbm \Delta_k^*)\bbm{x}_{k,t}\right\|_2^2  \leq \varpi_k \left(\|\bbm \Delta_k^*\|_1 - \|\widehat{\bm\Delta}_k^{\mathrm{opt}}\|_1\right).
		\end{align*}
		Note that $\|\bbm a\|_2^2 - \|\bbm b\|_2^2 = \|\bbm a - \bbm b\|_2^2 + 2\langle\bbm a - \bbm b,\bbm b\rangle$, then we have 
		\begin{align*}
			&\frac{1}{T_k} \sum_{t=1}^{T_k}\left\|\bbm{y}_{k,t} - (\widehat{\bm A}_0 + \widehat{\bm\Delta}_k^{\mathrm{opt}})\bbm{x}_{k,t}\right\|_2^2 - \frac{1}{T_k} \sum_{t=1}^{T_k}\left\|\bbm{y}_{k,t} - (\widehat{\bm A}_0 + \bbm \Delta_k^*)\bbm{x}_{k,t}\right\|_2^2\\
			=\ &\frac{1}{T_k}\sum_{t=1}^{T_k}\|(\widehat{\bbm \Delta}_k^{\mathrm{opt}} - \bbm \Delta_k^*)\bbm x_{k,t}\|_2^2 - \frac{2}{T_k}\sum_{t=1}^{T_k}\langle(\widehat{\bbm \Delta}_k^{\mathrm{opt}} - \bbm \Delta_k^*)\bbm x_{k,t},\bbm y_{k,t} - (\widehat{\bm A}_0+\bbm \Delta_k^*)\bbm x_{k,t}\rangle\\
			=\ &\frac{1}{T_k}\sum_{t=1}^{T_k}\|(\widehat{\bbm \Delta}_k^{\mathrm{opt}} - \bbm \Delta_k^*)\bbm x_{k,t}\|_2^2 - \frac{2}{T_k}\sum_{t=1}^{T_k}\langle(\widehat{\bbm \Delta}_k^{\mathrm{opt}} - \bbm \Delta_k^*)\bbm x_{k,t},\bbm\epsilon_{k,t}\rangle -\frac{2}{T_k}\sum_{t=1}^{T_k}\langle(\widehat{\bbm \Delta}_k^{\mathrm{opt}} - \bbm \Delta_k^*)\bbm x_{k,t}, (\widehat{\bm A}_0 - \bm A_0^*)\bbm x_{k,t}\rangle.
		\end{align*}
		This together with the above optimality inequality yields
		\begin{align}\label{eq:stage2-basic-ineq-pre}
			\frac{1}{T_k}\sum_{t=1}^{T_k}\|(\widehat{\bbm \Delta}_k^{\mathrm{opt}}-\bbm\Delta_k^*)\bbm x_{k,t}\|_2^2
			&\leq \varpi_k\bigl(\|\bbm\Delta_k^*\|_1-\|\widehat{\bbm\Delta}_k^{\mathrm{opt}}\|_1\bigr) + \frac{2}{T_k}\sum_{t=1}^{T_k}\Bigl\langle(\widehat{\bbm \Delta}_k^{\mathrm{opt}}-\bbm \Delta_k^*)\bbm x_{k,t},\bbm\epsilon_{k,t}\Bigr\rangle \notag\\
			&\quad + \frac{2}{T_k}\sum_{t=1}^{T_k}\Bigl\langle(\widehat{\bbm \Delta}_k^{\mathrm{opt}}-\bbm \Delta_k^*)\bbm x_{k,t},(\widehat{\bm A}_0-\bm A_0^*)\bbm x_{k,t}\Bigr\rangle.
		\end{align}
		For the first inner-product term in the right-hand side of \eqref{eq:stage2-basic-ineq-pre}, by H\"older's inequality, we have
		\begin{align*}
			\frac{1}{T_k}\sum_{t=1}^{T_k}\Bigl\langle(\widehat{\bbm \Delta}_k^{\mathrm{opt}}-\bbm \Delta_k^*)\bbm x_{k,t},\bbm\epsilon_{k,t}\Bigr\rangle = \Bigl\langle \widehat{\bbm \Delta}_k^{\mathrm{opt}}-\bbm \Delta_k^*,\frac{1}{T_k}\sum_{t=1}^{T_k}\bbm x_{k,t}\bbm\epsilon_{k,t}^\top\Bigr\rangle \leq \|\widehat{\bbm \Delta}_k^{\mathrm{opt}}-\bbm \Delta_k^*\|_1 \Bigl\|\frac{1}{T_k}\sum_{t=1}^{T_k}\bbm x_{k,t}\bbm\epsilon_{k,t}^\top\Bigr\|_\infty,
		\end{align*}
		For the second inner-product term in the right-hand side of \eqref{eq:stage2-basic-ineq-pre}, by the Young's inequality, i.e., $\langle \bbm a,\bbm b\rangle \leq \frac{1}{4}\|\bbm a\|_2^2 + 4\|\bbm b\|_2^2$, we have
		\begin{align*}
			\frac{1}{T_k}\sum_{t=1}^{T_k}\Bigl\langle(\widehat{\bbm \Delta}_k^{\mathrm{opt}}-\bbm \Delta_k^*)\bbm x_{k,t},(\widehat{\bm A}_0-\bm A_0^*)\bbm x_{k,t}\Bigr\rangle \leq \frac{1}{T_k}\sum_{t=1}^{T_k}\left(\frac{1}{4}\|(\widehat{\bbm \Delta}_k^{\mathrm{opt}}-\bbm \Delta_k^*)\bbm x_{k,t}\|_2^2 + 4\|(\widehat{\bm A}_0-\bm A_0^*)\bbm x_{k,t}\|_2^2\right).
		\end{align*}
		Plugging the above bounds for the two inner-product terms into \eqref{eq:stage2-basic-ineq-pre}, we have
		\begin{align*}
			\frac{1}{T_k}\sum_{t=1}^{T_k}\|(\widehat{\bbm \Delta}_k^{\mathrm{opt}}-\bbm\Delta_k^*)\bbm x_{k,t}\|_2^2 &\leq \varpi_k\bigl(\|\bbm\Delta_k^*\|_1-\|\widehat{\bbm\Delta}_k^{\mathrm{opt}}\|_1\bigr) + 2\|\widehat{\bbm \Delta}_k^{\mathrm{opt}}-\bbm \Delta_k^*\|_1\Bigl\|\frac{1}{T_k}\sum_{t=1}^{T_k}\bbm x_{k,t}\bbm\epsilon_{k,t}^\top\Bigr\|_\infty \\
			&\quad + \frac{2}{T_k}\sum_{t=1}^{T_k}\left(\frac{1}{4}\|(\widehat{\bbm \Delta}_k^{\mathrm{opt}}-\bbm \Delta_k^*)\bbm x_{k,t}\|_2^2 + 4\|(\widehat{\bm A}_0-\bm A_0^*)\bbm x_{k,t}\|_2^2\right).
		\end{align*}
		Rearranging the terms and moving the quadratic term in $\widehat{\bbm \Delta}_k^{\mathrm{opt}}-\bbm\Delta_k^*$ to the left-hand side yields
		\begin{align}\label{eq:stage2-basic-ineq-absorb}
			\frac{1}{2T_k}\sum_{t=1}^{T_k}\|(\widehat{\bbm \Delta}_k^{\mathrm{opt}}-\bbm\Delta_k^*)\bbm x_{k,t}\|_2^2 &\leq \varpi_k\bigl(\|\bbm\Delta_k^*\|_1-\|\widehat{\bbm\Delta}_k^{\mathrm{opt}}\|_1\bigr) + 2\|\widehat{\bbm \Delta}_k^{\mathrm{opt}}-\bbm \Delta_k^*\|_1\Bigl\|\frac{1}{T_k}\sum_{t=1}^{T_k}\bbm x_{k,t}\bbm\epsilon_{k,t}^\top\Bigr\|_\infty \notag\\
			&\quad + \frac{8}{T_k}\sum_{t=1}^{T_k}\|(\widehat{\bm A}_0-\bm A_0^*)\bbm x_{k,t}\|_2^2.
		\end{align}

		Recall the predefined subspace $\mathcal S_\kappa$ and its closed orthogonal complement $\overline{\mathcal S}_\kappa^{\perp}$ in \eqref{eq:S-space} and \eqref{eq:S-complement-space}, i.e.,
		\[
			\mathcal S_\kappa=\{\bm S\in\mathbb R^{d\times pd}: S_{ij}=0 \text{ if } |(\bbm\Delta_k^*)_{ij}|<\kappa\},
			\ \text{and} \ 
			\overline{\mathcal S}_\kappa^{\perp}=\{\bm S\in\mathbb R^{d\times pd}: S_{ij}=0 \text{ if } |(\bbm\Delta_k^*)_{ij}|\geq \kappa\}.
		\]
		Then, for the first term in \eqref{eq:stage2-basic-ineq-absorb}, using the decomposability of the entrywise $\ell_1$ norm with respect to the pair $\bigl(\mathcal S_\kappa,\overline{\mathcal S}_\kappa^{\perp}\bigr)$, we have
		\begin{align}\label{eq:l1-decomp-stage2-nolett}
			&\|\widehat{\bbm\Delta}_k^{\mathrm{opt}}\|_1 - \|\bbm\Delta_k^*\|_1 =\bigl\|(\widehat{\bbm\Delta}_k^{\mathrm{opt}}-\bbm\Delta_k^*) + \bbm\Delta_k^*\bigr\|_1 - \|\bbm\Delta_k^*\|_1 \notag\\
			&=\bigl\|(\widehat{\bbm\Delta}_k^{\mathrm{opt}}-\bbm\Delta_k^*)_{\mathcal S_\kappa} + (\widehat{\bbm\Delta}_k^{\mathrm{opt}}-\bbm\Delta_k^*)_{\overline{\mathcal S}_\kappa^{\perp}} + (\bbm\Delta_k^*)_{\mathcal S_\kappa} + (\bbm\Delta_k^*)_{\overline{\mathcal S}_\kappa^{\perp}}\bigr\|_1 - \bigl\|(\bbm\Delta_k^*)_{\mathcal S_\kappa}+(\bbm\Delta_k^*)_{\overline{\mathcal S}_\kappa^{\perp}}\bigr\|_1 \notag\\
			&\geq\bigl\|(\widehat{\bbm\Delta}_k^{\mathrm{opt}}-\bbm\Delta_k^*)_{\overline{\mathcal S}_\kappa^{\perp}} + (\bbm\Delta_k^*)_{\mathcal S_\kappa}\bigr\|_1 - \bigl\|(\widehat{\bbm\Delta}_k^{\mathrm{opt}}-\bbm\Delta_k^*)_{\mathcal S_\kappa} + (\bbm\Delta_k^*)_{\overline{\mathcal S}_\kappa^{\perp}}\bigr\|_1 - \bigl\|(\bbm\Delta_k^*)_{\mathcal S_\kappa}\bigr\|_1 - \bigl\|(\bbm\Delta_k^*)_{\overline{\mathcal S}_\kappa^{\perp}}\bigr\|_1 \notag\\
			&=\bigl\|(\widehat{\bbm\Delta}_k^{\mathrm{opt}}-\bbm\Delta_k^*)_{\overline{\mathcal S}_\kappa^{\perp}}\bigr\|_1 - 2\bigl\|(\bbm\Delta_k^*)_{\overline{\mathcal S}_\kappa^{\perp}}\bigr\|_1 - \bigl\|(\widehat{\bbm\Delta}_k^{\mathrm{opt}}-\bbm\Delta_k^*)_{\mathcal S_\kappa}\bigr\|_1.
		\end{align}
		By Lemma \ref{lem:var-XtE-infty-bound}, there exists a constant $C_{\infty}>0$ such that, given $T_k\gtrsim p\log(pd)$, with probability at least $1-\exp(-C\log(pd))$, 
		\[
			\Bigl\|\frac{1}{T_k}\sum_{t=1}^{T_k}\bbm x_{k,t}\bbm\epsilon_{k,t}^\top\Bigr\|_\infty \leq C_{\infty}C_{\epsilon,\mathcal A}^{(k)}\sigma^2\sqrt{\frac{\log(pd)}{T_k}}.
		\]
		Given the rate of $\varpi_k$ in Thoerem \ref{thm:personalized-error-bounds}, we choose $\varpi_k \geq 4C_{\infty}C_{\epsilon,\mathcal A}^{(k)}\sigma^2\sqrt{\log(pd)/{T_k}}$. Then, with probability at least $1 - C\exp(-C\log(pd))$, we have
		\begin{align}\label{eq:xe-infty-upper-bound}
			\Bigl\|\frac{1}{T_k}\sum_{t=1}^{T_k}\bbm x_{k,t}\bbm\epsilon_{k,t}^\top\Bigr\|_\infty \leq \frac{\varpi_k}{4}.
		\end{align}
		Then, combining \eqref{eq:stage2-basic-ineq-absorb}, \eqref{eq:l1-decomp-stage2-nolett} and \eqref{eq:xe-infty-upper-bound}, with probability at least $1 - C\exp(-C\log(pd))$,
		\begin{align*}
			0&\leq \frac{1}{2T_k}\sum_{t=1}^{T_k}\|(\widehat{\bbm \Delta}_k^{\mathrm{opt}}-\bbm\Delta_k^*)\bbm x_{k,t}\|_2^2\\
			&\leq \varpi_k\bigl(\|\bbm\Delta_k^*\|_1-\|\widehat{\bbm\Delta}_k^{\mathrm{opt}}\|_1\bigr)
			+ \frac{\varpi_k}{2}\|\widehat{\bbm \Delta}_k^{\mathrm{opt}}-\bbm \Delta_k^*\|_1
			+ \frac{8}{T_k}\sum_{t=1}^{T_k}\|(\widehat{\bm A}_0-\bm A_0^*)\bbm x_{k,t}\|_2^2\\
			&\leq \varpi_k\Bigl(\bigl\|(\widehat{\bbm\Delta}_k^{\mathrm{opt}}-\bbm\Delta_k^*)_{\mathcal S_\kappa}\bigr\|_1
			+ 2\bigl\|(\bbm\Delta_k^*)_{\overline{\mathcal S}_\kappa^{\perp}}\bigr\|_1
			- \bigl\|(\widehat{\bbm\Delta}_k^{\mathrm{opt}}-\bbm\Delta_k^*)_{\overline{\mathcal S}_\kappa^{\perp}}\bigr\|_1\Bigr)\\
			&\quad
			+ \frac{\varpi_k}{2}\Bigl(\bigl\|(\widehat{\bbm\Delta}_k^{\mathrm{opt}}-\bbm\Delta_k^*)_{\mathcal S_\kappa}\bigr\|_1
			+ \bigl\|(\widehat{\bbm\Delta}_k^{\mathrm{opt}}-\bbm\Delta_k^*)_{\overline{\mathcal S}_\kappa^{\perp}}\bigr\|_1\Bigr)
			+ \frac{8}{T_k}\sum_{t=1}^{T_k}\|(\widehat{\bm A}_0-\bm A_0^*)\bbm x_{k,t}\|_2^2\\
			&\leq \frac{\varpi_k}{2}\Bigl(3\bigl\|(\widehat{\bbm\Delta}_k^{\mathrm{opt}}-\bbm\Delta_k^*)_{\mathcal S_\kappa}\bigr\|_1
			+ 4\bigl\|(\bbm\Delta_k^*)_{\overline{\mathcal S}_\kappa^{\perp}}\bigr\|_1
			- \bigl\|(\widehat{\bbm\Delta}_k^{\mathrm{opt}}-\bbm\Delta_k^*)_{\overline{\mathcal S}_\kappa^{\perp}}\bigr\|_1\Bigr)
			+ \frac{8}{T_k}\sum_{t=1}^{T_k}\|(\widehat{\bm A}_0-\bm A_0^*)\bbm x_{k,t}\|_2^2.
		\end{align*}
		Then, we have the cone-type inequality
		\begin{align}\label{eq:cone-stage2}
			\bigl\|(\widehat{\bbm\Delta}_k^{\mathrm{opt}}-\bbm\Delta_k^*)_{\overline{\mathcal S}_\kappa^{\perp}}\bigr\|_1 \leq 3\bigl\|(\widehat{\bbm\Delta}_k^{\mathrm{opt}}-\bbm\Delta_k^*)_{\mathcal S_\kappa}\bigr\|_1 + 4\bigl\|(\bbm\Delta_k^*)_{\overline{\mathcal S}_\kappa^{\perp}}\bigr\|_1 + \frac{2}{\varpi_k}\cdot\frac{8}{T_k}\sum_{t=1}^{T_k}\|(\widehat{\bm A}_0-\bm A_0^*)\bbm x_{k,t}\|_2^2.
		\end{align}
		Consequently, plugging \eqref{eq:xe-infty-upper-bound} and \eqref{eq:cone-stage2} into \eqref{eq:stage2-basic-ineq-absorb} yields
		\begin{align}\label{eq:upper-bound-stage2-comp-before-RSC}
			&\frac{1}{2T_k}\sum_{t=1}^{T_k}\|(\widehat{\bbm \Delta}_k^{\mathrm{opt}}-\bbm\Delta_k^*)\bbm x_{k,t}\|_2^2 \notag\\
			\leq\ &\varpi_k\bigl(\|\bbm\Delta_k^*\|_1-\|\widehat{\bbm\Delta}_k^{\mathrm{opt}}\|_1\bigr) + 2\|\widehat{\bbm \Delta}_k^{\mathrm{opt}}-\bbm \Delta_k^*\|_1\Bigl\|\frac{1}{T_k}\sum_{t=1}^{T_k}\bbm x_{k,t}\bbm\epsilon_{k,t}^\top\Bigr\|_\infty + \frac{8}{T_k}\sum_{t=1}^{T_k}\|(\widehat{\bm A}_0-\bm A_0^*)\bbm x_{k,t}\|_2^2 \notag\\
			\leq\ &\varpi_k\bigl(\|\bbm\Delta_k^*\|_1-\|\widehat{\bbm\Delta}_k^{\mathrm{opt}}\|_1\bigr) + \frac{\varpi_k}{2}\|\widehat{\bbm \Delta}_k^{\mathrm{opt}}-\bbm \Delta_k^*\|_1 + \frac{8}{T_k}\sum_{t=1}^{T_k}\|(\widehat{\bm A}_0-\bm A_0^*)\bbm x_{k,t}\|_2^2 \notag\\
			\leq\ &\frac{3\varpi_k}{2}\|\widehat{\bbm \Delta}_k^{\mathrm{opt}}-\bbm \Delta_k^*\|_1 + \frac{8}{T_k}\sum_{t=1}^{T_k}\|(\widehat{\bm A}_0-\bm A_0^*)\bbm x_{k,t}\|_2^2 \notag\\
			=\ &\frac{3\varpi_k}{2}\Bigl(\bigl\|(\widehat{\bbm\Delta}_k^{\mathrm{opt}}-\bbm\Delta_k^*)_{\mathcal S_\kappa}\bigr\|_1 + \bigl\|(\widehat{\bbm\Delta}_k^{\mathrm{opt}}-\bbm\Delta_k^*)_{\overline{\mathcal S}_\kappa^{\perp}}\bigr\|_1\Bigr) + \frac{8}{T_k}\sum_{t=1}^{T_k}\|(\widehat{\bm A}_0-\bm A_0^*)\bbm x_{k,t}\|_2^2 \notag\\
			\leq\ &\frac{3\varpi_k}{2}\Bigl(\bigl\|(\widehat{\bbm\Delta}_k^{\mathrm{opt}}-\bbm\Delta_k^*)_{\mathcal S_\kappa}\bigr\|_1 + 3\bigl\|(\widehat{\bbm\Delta}_k^{\mathrm{opt}}-\bbm\Delta_k^*)_{\mathcal S_\kappa}\bigr\|_1 + 4\bigl\|(\bbm\Delta_k^*)_{\overline{\mathcal S}_\kappa^{\perp}}\bigr\|_1 + \frac{2}{\varpi_k}\cdot\frac{8}{T_k}\sum_{t=1}^{T_k}\|(\widehat{\bm A}_0-\bm A_0^*)\bbm x_{k,t}\|_2^2\Bigr) \notag\\
			&\quad+ \frac{8}{T_k}\sum_{t=1}^{T_k}\|(\widehat{\bm A}_0-\bm A_0^*)\bbm x_{k,t}\|_2^2 \notag\\
			=\ &6\varpi_k\Bigl(\bigl\|(\widehat{\bbm\Delta}_k^{\mathrm{opt}}-\bbm\Delta_k^*)_{\mathcal S_\kappa}\bigr\|_1 + \bigl\|(\bbm\Delta_k^*)_{\overline{\mathcal S}_\kappa^{\perp}}\bigr\|_1\Bigr) + \frac{32}{T_k}\sum_{t=1}^{T_k}\|(\widehat{\bm A}_0-\bm A_0^*)\bbm x_{k,t}\|_2^2.
		\end{align}

		Then, we derive the lower bound for $(2T_k)^{-1}\|(\widehat{\bbm \Delta}_k^{\mathrm{opt}}-\bbm\Delta_k^*)\bbm x_{k,t}\|_2^2$. Given the sample size condition $T_k\gtrsim(C_{\epsilon,\mathcal A}^{(k)}/C_{\mathrm{RSC}}^{(k)} \vee 1)^2p^2d$, by Lemma~\ref{lem:RSC-var}, with probability at least $1 - C\exp(-C pd)$, the following RSC condition holds:
		\begin{align}\label{eq:upper-bound-stage2-comp}
			\frac{1}{2T_k}\sum_{t=1}^{T_k}\|(\widehat{\bbm \Delta}_k^{\mathrm{opt}}-\bbm\Delta_k^*)\bbm x_{k,t}\|_2^2 \geq \frac{C_{\mathrm{RSC}}^{(k)}}{2}\|\widehat{\bbm \Delta}_k^{\mathrm{opt}}-\bbm\Delta_k^*\|_{\mathrm F}^2.
		\end{align}
		Then, combining \eqref{eq:upper-bound-stage2-comp-before-RSC} and \eqref{eq:upper-bound-stage2-comp} yields
		\begin{align}\label{eq:upper-bound-stage2-comp-after-RSC}
			&\frac{C_{\mathrm{RSC}}^{(k)}}{2}\|\widehat{\bbm \Delta}_k^{\mathrm{opt}}-\bbm\Delta_k^*\|_{\mathrm F}^2 \leq 6\varpi_k\Bigl(\bigl\|(\widehat{\bbm\Delta}_k^{\mathrm{opt}}-\bbm\Delta_k^*)_{\mathcal S_\kappa}\bigr\|_1 + \bigl\|(\bbm\Delta_k^*)_{\overline{\mathcal S}_\kappa^{\perp}}\bigr\|_1\Bigr) + \frac{32}{T_k}\sum_{t=1}^{T_k}\|(\widehat{\bm A}_0-\bm A_0^*)\bbm x_{k,t}\|_2^2.
		\end{align}
		Considering the threshold $\kappa$ of the sparse space $\mathcal S_\kappa=\{\bm S\in\mathbb R^{d\times pd}:S_{ij}=0 \text{ for all }(i,j)\in[d]\times[pd] \text{ such that }|(\bbm \Delta_k^*)_{ij}|<\kappa\}$, by the definition of $\mathbb B_q(s_q)$, we have $s_q \geq |\mathcal S_\kappa| \kappa^q$, and thus $|\mathcal S_\kappa| \leq s_q \kappa^{-q}$. Then, it follows that $\|(\widehat{\bbm \Delta}_k^{\mathrm{opt}}-\bbm \Delta_k^*)_{\mathcal S_\kappa}\|_{1} \leq \sqrt{|\mathcal{S}_\kappa|}\|(\widehat{\bbm \Delta}_k^{\mathrm{opt}}-\bbm \Delta_k^*)_{\mathcal S_\kappa}\|_{\F} \leq \sqrt{s_q}\kappa^{-q/2}\|\widehat{\bbm \Delta}_k^{\mathrm{opt}}-\bbm \Delta_k^*\|_{\F}$.
		Besides, 
		\[
			\|(\bbm \Delta_k^*)_{\overline{\mathcal S}_\kappa^{\perp}}\|_{1} = \sum_{(i,j)\in\overline{\mathcal S}_\kappa^{\perp}} |(\bbm \Delta_k^*)_{ij}| \leq \sum_{(i,j)\in\overline{\mathcal S}_\kappa^{\perp}} |(\bbm \Delta_k^*)_{ij}|^q \kappa^{1-q} \leq s_q \kappa^{1-q}.
		\]
		Then, it follows from \eqref{eq:upper-bound-stage2-comp-after-RSC} that
		\begin{align}\label{eq:upper-bound-before-yongs}
			\|\widehat{\bbm \Delta}_k^{\mathrm{opt}}-\bbm\Delta_k^*\|_{\mathrm F}^2 \leq \frac{12\varpi_k}{C_{\mathrm{RSC}}^{(k)}}\Bigl(\sqrt{s_q}\kappa^{-q/2}\|\widehat{\bbm \Delta}_k^{\mathrm{opt}} - \bbm \Delta_k^*\|_{\F} + s_q\kappa^{1-q}\Bigr) + \frac{64}{C_{\mathrm{RSC}}^{(k)}T_k}\sum_{t=1}^{T_k}\|(\widehat{\bm A}_0-\bm A_0^*)\bbm x_{k,t}\|_2^2.
		\end{align}
		Note that $ab \leq (a^2 + b^2)/2$ for any $a,b\geq 0$, we have
		\begin{align}\label{eq:yongs-delta-inequality}
			12\varpi_k\sqrt{s_q}\kappa^{-q/2}\|\widehat{\bbm \Delta}_k^{\mathrm{opt}}-\bbm\Delta_k^*\|_{\mathrm F} = \sqrt{C_{\mathrm{RSC}}^{(k)}}\|\widehat{\bbm \Delta}_k^{\mathrm{opt}}-\bbm\Delta_k^*\|_{\mathrm F} \cdot \frac{12\varpi_k\sqrt{s_q}\kappa^{-q/2}}{\sqrt{C_{\mathrm{RSC}}^{(k)}}}\notag\\
			\leq \frac{C_{\mathrm{RSC}}^{(k)}}{2}\|\widehat{\bbm \Delta}_k^{\mathrm{opt}}-\bbm\Delta_k^*\|_{\mathrm F}^2 + 72\ \frac{\varpi_k^2 s_q\kappa^{-q}}{C_{\mathrm{RSC}}^{(k)}}.
		\end{align}
		Then, substituting \eqref{eq:yongs-delta-inequality} into \eqref{eq:upper-bound-before-yongs} and rearranging terms, we have
		\begin{align*}
			\|\widehat{\bbm \Delta}_k^{\mathrm{opt}}-\bbm\Delta_k^*\|_{\mathrm F}^2 \leq 144\ \frac{\varpi_k^2s_q\kappa^{-q}}{(C_{\mathrm{RSC}}^{(k)})^2} + 24\ \frac{\varpi_ks_q\kappa^{1-q}}{C_{\mathrm{RSC}}^{(k)}} + \frac{128}{C_{\mathrm{RSC}}^{(k)}T_k}\sum_{t=1}^{T_k}\|(\widehat{\bm A}_0-\bm A_0^*)\bbm x_{k,t}\|_2^2.
		\end{align*}
		Then, choosing $\kappa \asymp \varpi_k/C_{\mathrm{RSC}}^{(k)}$ yields
		\begin{align}\label{eq:upper-bound-before-ao-bound}
			\|\widehat{\bbm \Delta}_k^{\mathrm{opt}}-\bbm\Delta_k^*\|_{\F}^2 \leq C s_q\left(\frac{\varpi_k}{C_{\mathrm{RSC}}^{(k)}}\right)^{2-q} + \frac{128}{C_{\mathrm{RSC}}^{(k)}T_k}\sum_{t=1}^{T_k}\|(\widehat{\bm A}_0-\bm A_0^*)\bbm x_{k,t}\|_2^2.
		\end{align}
		It remains to control the last term in \eqref{eq:upper-bound-before-ao-bound} involving $\widehat{\bm A}_0-\bm A_0^*$. Note that
		\begin{align*}
			&\frac{1}{T_k}\sum_{t=1}^{T_k}\|(\widehat{\bm A}_0-\bm A_0^*)\bbm x_{k,t}\|_2^2 = \frac{1}{T_k}\sum_{t=1}^{T_k}\bbm x_{k,t}^\top(\widehat{\bm A}_0-\bm A_0^*)^\top(\widehat{\bm A}_0-\bm A_0^*)\bbm x_{k,t}\\
			=\ &\tr\left((\widehat{\bm A}_0-\bm A_0^*)\Bigl(\frac{1}{T_k}\sum_{t=1}^{T_k}\bbm x_{k,t}\bbm x_{k,t}^\top\Bigr)(\widehat{\bm A}_0-\bm A_0^*)^\top\right) \leq \|\widehat{\bm A}_0-\bm A_0^*\|_{\F}^{2}\left\|\frac{1}{T_k}\sum_{t=1}^{T_k}\bbm x_{k,t}\bbm x_{k,t}^\top\right\|_{\op},
		\end{align*}
		where the inequality follows from $\tr(\bm U\bm M\bm U^\top)\leq \|\bm U\|_{\F}^2\|\bm M\|_{\op}$. Moreover, combining $\|\bbm\Sigma_{x,k}\|_{\op}\leq C_{\epsilon,\mathcal A}^{(k)}$ implied by Lemma~\ref{lem:Sigma-x-pd} given $T_k \gtrsim p^2d$ and Lemma \ref{lem:normal-second-order}, with probability at least $1-\exp(-Cpd)$, we have
		\[
			\left\|\frac{1}{T_k}\sum_{t=1}^{T_k}\bbm x_{k,t}\bbm x_{k,t}^\top\right\|_{\op} \lesssim \|\bbm\Sigma_{x,k}\|_{\op} + C_{\epsilon,\mathcal A}^{(k)}\sigma^2p\sqrt{\frac{d}{T_k}} \leq C' C_{\epsilon,\mathcal A}^{(k)}.
		\]

		Consequently, combining the above inequalities, $T_k^{-1}\sum_{t=1}^{T_k}\|(\widehat{\bm A}_0-\bm A_0^*)\bbm x_{k,t}\|_2^2 \leq C' C_{\epsilon,\mathcal A}^{(k)}\|\widehat{\bm A}_0-\bm A_0^*\|_{\F}^2$. Substituting this bound into \eqref{eq:upper-bound-before-ao-bound} and collecting the sample-size conditions and probability bounds from the preceding steps, provided $T_k \gtrsim (C_{\epsilon,\mathcal A}^{(k)}/C_{\mathrm{RSC}}^{(k)}\vee 1)^2p^2d \vee p\log(pd)$, with probability at least $1-C\exp(-C\log(pd))-C\exp(-Cpd)$, we have
		\[
			\|\widehat{\bbm \Delta}_k^{\mathrm{opt}}-\bbm\Delta_k^*\|_{\F}^2 \lesssim s_q\left(\frac{\varpi_k}{C_{\mathrm{RSC}}^{(k)}}\right)^{2-q} + \frac{C_{\epsilon,\mathcal A}^{(k)}}{C_{\mathrm{RSC}}^{(k)}} \|\widehat{\bm A}_0-\bm A_0^*\|_{\F}^{2}.
		\]

		Taking square roots and using $\sqrt{a+b}\leq \sqrt a+\sqrt b$, we further obtain
		\begin{align}\label{eq:final-delta-bound}
			\|\widehat{\bbm \Delta}_k^{\mathrm{opt}}-\bbm\Delta_k^*\|_{\F} \lesssim \left(\frac{C_{\epsilon,\mathcal A}^{(k)}}{C_{\mathrm{RSC}}^{(k)}}\right)^{1/2}\|\widehat{\bm A}_0-\bm A_0^*\|_{\F} + \sqrt{s_q}\left(\frac{\varpi_k}{C_{\mathrm{RSC}}^{(k)}}\right)^{1-q/2}=: \mathsf{Error}_{\Delta}^{(k)}.
		\end{align}

		\noindent
		\textbf{Part II (Optimization error).}
		Now, we turn to control the optimization error $\|\widehat{\bbm \Delta}_k-\widehat{\bbm \Delta}_k^{\mathrm{opt}}\|_{\F}$ based on the convergence guarantee of FISTA. To this end, we rewrite \eqref{eq:stage2-opt} in a compact matrix form. Denote $\bm Y_k = [\bbm y_{k,1},\ldots,\bbm y_{k,T_k}]^\top \in \mathbb R^{T_k\times d}$ and $\bm X_k = [\bbm x_{k,1},\ldots,\bbm x_{k,T_k}]^\top \in \mathbb R^{T_k\times pd}$. Then \eqref{eq:stage2-opt} is equivalent to
		\begin{align}\label{eq:stage2-opt-matrix}
			\widehat{\bm\Delta}_k^{\mathrm{opt}} \in \argmin_{\bm\Delta\in\mathbb R^{d\times pd}}\left\{\frac{1}{T_k}\bigl\|\bm Y_k-\bm X_k(\widehat{\bm A}_0+\bm\Delta)^\top\bigr\|_{\F}^{2} + \varpi_k\|\bm\Delta\|_{1}\right\}.
		\end{align}

		Denote $f_k(\bm\Delta) = T_k^{-1}\|\bm Y_k-\bm X_k(\widehat{\bm A}_0+\bm\Delta)^\top\|_{\F}^{2}$, and $g_k(\bm\Delta) = \varpi_k\|\bm\Delta\|_1$. 
		We first verify the Lipschitz continuity of the gradient of $f_k(\bm\Delta)$. A direct calculation yields
		\[
			\nabla f_k(\bm\Delta) = \frac{2}{T_k}\bigl(\bm\Delta+\widehat{\bm A}_0\bigr)\bm X_k^\top\bm X_k - \frac{2}{T_k}\bm Y_k^\top\bm X_k,
		\]
		and hence, for any $\bm\Delta_1,\bm\Delta_2\in\mathbb R^{d\times pd}$,
		\[
			\|\nabla f_k(\bm\Delta_1)-\nabla f_k(\bm\Delta_2)\|_{\F}
			=\frac{2}{T_k}\bigl\|(\bm\Delta_1-\bm\Delta_2)\bm X_k^\top\bm X_k\bigr\|_{\F}
			\leq \frac{2}{T_k}\|\bm X_k^\top\bm X_k\|_{\op}\,\|\bm\Delta_1-\bm\Delta_2\|_{\F}.
		\]
		Therefore, $\nabla f_k$ is Lipschitz continuous with constant $L_k=2T_k^{-1}\|\bm X_k^\top\bm X_k\|_{\op}$ conditional on $\mathcal{F}_{T_k}$.
		
		Although the analysis in \citet{beck2009fast} is presented for vector-valued variables, Remark~2.1 therein notes that the key inequalities and lemmas underlying the FISTA rate extend verbatim to any Hilbert space equipped with an inner product. Since $\mathbb R^{d\times pd}$ endowed with the Frobenius inner product is a finite-dimensional Hilbert space, the same convergence analysis applies to our matrix-valued optimization problem. Then it suffices to verify the conditions for applying the FISTA convergence guarantee.

		Since $g_k(\bbm\Delta)=\varpi_k\|\bbm\Delta\|_1$ is convex but nonsmooth and $f_k(\bbm\Delta)$ is continuously differentiable with an $L_k$-Lipschitz continuous gradient, the objective $F_k(\bbm\Delta)=f_k(\bbm\Delta)+g_k(\bbm\Delta)$ fits the standard FISTA framework. Note that $\{\bm\Delta_k^{(n)}\}_{n\geq0}$ are the iterates generated by Algorithm~\ref{algorithmTL-stage2} with stepsize $\eta\leq 1/L_k$. Then, by \citet[Theorem~4.4]{beck2009fast} together with Remark~2.1 therein, for any $n\geq1$,
		\begin{align*}
			F_k(\bm\Delta_k^{(n)})-F_k(\widehat{\bm\Delta}_k^{\mathrm{opt}}) \leq \frac{2L_k\,\|\bm\Delta_k^{(0)}-\widehat{\bm\Delta}_k^{\mathrm{opt}}\|_{\F}^{2}}{(n+1)^2}.
		\end{align*}
		In particular, under the initialization $\bm\Delta_k^{(0)}=\bm 0$ in Algorithm~\ref{algorithmTL-stage2} and $\widehat{\bbm\Delta}_k=\bm\Delta_k^{(N_l)}$, we have
		\begin{align}\label{eq:fista-rate-stage2-Nl}
			F_k(\widehat{\bbm\Delta}_k)-F_k(\widehat{\bm\Delta}_k^{\mathrm{opt}}) \leq \frac{2L_k\,\|\widehat{\bm\Delta}_k^{\mathrm{opt}}\|_{\F}^{2}}{(N_l+1)^2}.
		\end{align}

		Next, we relate $F_k(\widehat{\bm\Delta}_k)-F_k(\widehat{\bm\Delta}_k^{\mathrm{opt}})$ to $\|\widehat{\bm\Delta}_k-\widehat{\bm\Delta}_k^{\mathrm{opt}}\|_{\F}$. For any $\bbm\Delta\in\mathbb R^{d\times pd}$, a direct expansion of $f_k(\bm\Delta)$ at $\widehat{\bbm\Delta}_k^{\mathrm{opt}}$ gives
		\begin{align}\label{eq:f-exact-expansion}
			f_k(\bm\Delta) = f_k(\widehat{\bbm\Delta}_k^{\mathrm{opt}}) + \bigl\langle\nabla f_k(\widehat{\bbm\Delta}_k^{\mathrm{opt}}),\,\bbm\Delta - \widehat{\bbm\Delta}_k^{\mathrm{opt}}\bigr\rangle + \frac{1}{T_k} \sum_{t=1}^{T_k} \bigl\|(\bbm\Delta - \widehat{\bbm\Delta}_k^{\mathrm{opt}})\bbm x_{k,t}\bigr\|_2^2.
		\end{align}
		Moreover, $g_k(\bbm\Delta)=\varpi_k\|\bbm\Delta\|_1$ is convex but non-smooth. By the subgradient inequality of $\ell_1$ norm, for any $\bbm\Xi \in \mathbb{R}^{d\times pd}$ and any $\bm Z\in\partial\|\bbm\Xi\|_1$, $\|\bbm\Delta\|_1 \geq \|\bbm\Xi\|_1 + \langle \bm Z,\,\bbm\Delta-\bbm\Xi\rangle$ for any $\bbm\Delta\in \mathbb{R}^{d\times pd}$. Applying this inequality with $\bbm\Xi=\widehat{\bbm\Delta}_k^{\mathrm{opt}}$ and $\bm Z=\bm Z_k^{\mathrm{opt}}\in\partial\|\widehat{\bbm\Delta}_k^{\mathrm{opt}}\|_1$ yields
		\begin{align}\label{eq:g-subgrad-ineq}
			g_k(\bbm\Delta) \geq g_k(\widehat{\bbm\Delta}_k^{\mathrm{opt}}) + \bigl\langle\varpi_k \bm Z_k^{\mathrm{opt}},\bbm\Delta - \widehat{\bbm\Delta}_k^{\mathrm{opt}}\bigr\rangle, \quad \text{for any } \bbm\Delta\in\mathbb{R}^{d\times pd}.
		\end{align}
		Combining \eqref{eq:f-exact-expansion} and \eqref{eq:g-subgrad-ineq}, we obtain the following lower bound for $F_k=f_k+g_k$:
		\begin{align}\label{eq:F-combined-lower}
			F_k(\bm\Delta) \geq F_k(\widehat{\bbm\Delta}_k^{\mathrm{opt}}) + \bigl\langle \nabla f_k(\widehat{\bbm\Delta}_k^{\mathrm{opt}}) + \varpi_k \bm Z_k^{\mathrm{opt}}, \bbm\Delta - \widehat{\bbm\Delta}_k^{\mathrm{opt}} \bigr\rangle + \frac{1}{T_k} \sum_{t=1}^{T_k} \bigl\| (\bbm\Delta - \widehat{\bbm\Delta}_k^{\mathrm{opt}})\bbm x_{k,t} \bigr\|_2^2.
		\end{align}
		Since $\widehat{\bbm\Delta}_k^{\mathrm{opt}}$ minimizes the convex function $F_k$, the first-order optimality condition ensures that there exists $\bm Z_k^{\mathrm{opt}}\in\partial\|\widehat{\bbm\Delta}_k^{\mathrm{opt}}\|_1$ such that
		\begin{align}\label{eq:opt-condition-main}
			\nabla F_k(\widehat{\bbm\Delta}_k^{\mathrm{opt}}) = \nabla f_k(\widehat{\bbm\Delta}_k^{\mathrm{opt}}) + \varpi_k \bm Z_k^{\mathrm{opt}} = \bm 0.
		\end{align}
		Plugging \eqref{eq:opt-condition-main} into \eqref{eq:F-combined-lower} leads to
		\begin{align}\label{eq:quadratic-growth-final}
			\frac{1}{T_k} \sum_{t=1}^{T_k} \bigl\| (\bbm\Delta - \widehat{\bbm\Delta}_k^{\mathrm{opt}})\bbm x_{k,t} \bigr\|_2^2 \leq F_k(\bm\Delta) - F_k(\widehat{\bbm\Delta}_k^{\mathrm{opt}}),  \quad \text{for any } \bbm\Delta\in\mathbb{R}^{d\times pd}.
		\end{align}
		Finally, by taking $\bm\Delta=\widehat{\bm\Delta}_k$, the RSC condition in Lemma~\ref{lem:RSC-var} together with \eqref{eq:fista-rate-stage2-Nl} yields that, with probability at least $1-\exp(-Cpd)$,
		\begin{align}\label{eq:opt-quadratic-lower-bound-restate}
			C_{\mathrm{RSC}}^{(k)}\bigl\|\widehat{\bm\Delta}_k - \widehat{\bm\Delta}_k^{\mathrm{opt}}\bigr\|_{\F}^2 \leq \frac{1}{T_k} \sum_{t=1}^{T_k} \bigl\| (\widehat{\bm\Delta}_k - \widehat{\bm\Delta}_k^{\mathrm{opt}})\bbm x_{k,t} \bigr\|_2^2 \leq F_k(\widehat{\bm\Delta}_k) - F_k(\widehat{\bm\Delta}_k^{\mathrm{opt}}) \leq \frac{2L_k\|\widehat{\bm\Delta}_k^{\mathrm{opt}}\|_{\F}^{2}}{(N_l+1)^2}.
		\end{align}
		Denote the client-specific condition-number factor by $\kappa_{\epsilon,\mathcal A}^{(k)}=C_{\epsilon,\mathcal A}^{(k)}/c_{\epsilon,\mathcal A}^{(k)}$. Moreover, as implied by Lemma~\ref{lem:var-cov-conc} and Lemma~\ref{lem:Sigma-x-pd}, given $T_k\gtrsim p^2 d$, there exists a constant $C'>0$ such that, with probability at least $1-\exp(-Cpd)$,
		\[
			L_k = 2\left\|\frac{1}{T_k}\sum_{t=1}^{T_k}\bbm x_{k,t}\bbm x_{k,t}^{\top}\right\|_{\op} \leq 2C' C_{\epsilon,\mathcal A}^{(k)}.
		\]
		Then, substituting this bound into \eqref{eq:opt-quadratic-lower-bound-restate}, and noting that $\|\widehat{\bbm \Delta}_k^{\mathrm{opt}}\|_{\F}^2 \leq 2\|\widehat{\bbm \Delta}_k^{\mathrm{opt}}-\bbm\Delta_k^*\|_{\F}^2 + 2\|\bbm\Delta_k^*\|_{\F}^2$, $\|\bm\Delta_k^*\|_{\F}\leq \phi(s_q)$, $C_{\mathrm{RSC}}^{(k)}=1/2c_{\epsilon,\mathcal A}^{(k)} \leq C_{\epsilon,\mathcal A}^{(k)}$, we have, with probability at least $1-\exp(-Cpd)$,
		\begin{align}\label{eq:delta-gap-fista}
			\bigl\|\widehat{\bm\Delta}_k - \widehat{\bm\Delta}_k^{\mathrm{opt}}\bigr\|_{\F}^2 \lesssim \kappa_{\epsilon,\mathcal A}^{(k)}\cdot\frac{\|\widehat{\bm\Delta}_k^{\mathrm{opt}}-\bm\Delta_k^*\|_{\F}^2+\phi^2(s_q)}{(N_l+1)^2},
		\end{align}
		Recall \eqref{eq:final-delta-bound}, we have $\|\widehat{\bm\Delta}_k^{\mathrm{opt}}-\bm\Delta_k^*\|_{\F}^2 \lesssim (\mathsf{Error}_{\Delta}^{(k)})^2$, and hence \eqref{eq:delta-gap-fista} implies
		\[
			\|\widehat{\bm\Delta}_k-\widehat{\bm\Delta}_k^{\mathrm{opt}}\|_{\F}^2 \lesssim \kappa_{\epsilon,\mathcal A}^{(k)}\cdot\frac{(\mathsf{Error}_{\Delta}^{(k)})^2 + \phi^2(s_q)}{(N_l+1)^2}.
		\]

		Consequently, we have
		\[
			\|\widehat{\bm\Delta}_k-\bm\Delta_k^*\|_{\F}^2 \leq 2\|\widehat{\bm\Delta}_k-\widehat{\bm\Delta}_k^{\mathrm{opt}}\|_{\F}^2 + 2\|\widehat{\bm\Delta}_k^{\mathrm{opt}}-\bm\Delta_k^*\|_{\F}^2 \lesssim \kappa_{\epsilon,\mathcal A}^{(k)}\cdot\frac{(\mathsf{Error}_{\Delta}^{(k)})^2 + \phi^2(s_q)}{(N_l+1)^2} + (\mathsf{Error}_{\Delta}^{(k)})^2.
		\]
		Finally, taking $N_l \asymp \sqrt{\kappa_{\epsilon,\mathcal A}^{(k)}}\phi(s_q)/\mathsf{Error}_{\Delta}^{(k)}$,
		the optimization error is negligible relative to the statistical error, i.e.,
		\[
			\kappa_{\epsilon,\mathcal A}^{(k)}\cdot\frac{(\mathsf{Error}_{\Delta}^{(k)})^2+\phi^2(s_q)}{(N_l+1)^2} \lesssim (\mathsf{Error}_{\Delta}^{(k)})^2.
		\]
		By $a^2+b^2\leq(a+b)^2$ for any $a,b\geq0$, we have
		\[
			\|\widehat{\bm\Delta}_k-\bm\Delta_k^*\|_{\F} \lesssim \mathsf{Error}_{\Delta}^{(k)} \lesssim \sqrt{\kappa_{\epsilon,\mathcal A}^{(k)}}\|\widehat{\bm A}_0-\bm A_0^*\|_{\F} + \sqrt{s_q}\left(\frac{\varpi_k}{C_{\mathrm{RSC}}^{(k)}}\right)^{1-q/2}.
		\]
		To conclude the proof, we further control \textbf{Term~II} appearing in the bound for $\|\widehat{\bm A}_0-\bm A_0^*\|_{\F}$ in the proof of Theorem~\ref{thm:common-A0-error-bounds}. Given $T_k \gtrsim p^2 d$ for all $k\in[K]$, Lemma~\ref{lem:var-cov-conc} and Lemma~\ref{lem:Sigma-x-pd} imply that $\|\bm S_k\|_{\op} \leq \|\bm S_k-\mathbb E\bm S_k\|_{\op}+\|\mathbb E\bm S_k\|_{\op} \leq C' C_{\epsilon,\mathcal A}^{\max}$ with probability at least $1-\exp(-Cpd)$. Consequently, by applying this bound to \eqref{eq:term2-final-bound-hp-absorbed}, we have that, with probability at least $1-K\exp(-Cpd)$, Term~II in the proof of Theorem~\ref{thm:common-A0-error-bounds} admits the following tighter upper bound:
		\begin{align}\label{eq:term2-final-bound-hp-absorbed-Tk}
			\textnormal{Term II} = \Bigl\|\mathcal P_{\mathcal T_r(\bm A_0^{(n)})}\Bigl(\frac{1}{T}\sum_{k=1}^K(\bm A_0^*-\bm A_k^*)\bm X_k\bm X_k^\top\Bigr)\Bigr\|_{\F} \leq \sqrt{2r}\sum_{k=1}^{K}\frac{T_k}{T}\|\bbm \Delta_k^*\|_\op \|\bm S_k\|_\op \lesssim C_{\epsilon,\mathcal A}^{\max}\sqrt{r}\zeta.
		\end{align}
		Therefore, under the the event in \eqref{eq:term2-final-bound-hp-absorbed-Tk} and conditions of Theorem~\ref{thm:common-A0-error-bounds}, with probability at least $1-K\exp(-Cpd)-K\exp(-C\log T)$,  we have
		\[
			\|\widehat{\bm A}_0-\bm A_0^*\|_{\F} \leq \kappa_{\epsilon,\mathcal A}\left(\mathsf{Error}_{\mathrm{stat}} + \mathsf{Error}_{\mathrm{DP}} + \mathsf{Error}_{\mathrm{h}'}\right),
		\]
		where $\mathsf{Error}_{\mathrm{h}'} = \sqrt{r}\zeta$. Combining all high probability events required in the proof yields that, with probability at least $1-K\exp(-Cpd)-K\exp(-C\log T)$,
		\[
			\|\widehat{\bm A}_0-\bm A_0^*\|_{\F} + \|\widehat{\bm\Delta}_k-\bm\Delta_k^*\|_{\F}
			\lesssim
			\sqrt{\kappa_{\epsilon,\mathcal A}^{(k)}}\kappa_{\epsilon,\mathcal A}\left(\mathsf{Error}_{\mathrm{stat}} + \mathsf{Error}_{\mathrm{DP}} + \underbrace{C_{\epsilon,\mathcal A}^{\max}\sqrt{r}\zeta}_{\mathsf{Error}_{\mathrm{h}}'}\right)
			+ \underbrace{\sqrt{s_q}\left(\frac{\varpi_k}{C_{\mathrm{RSC}}^{(k)}}\right)^{1-q/2}}_{\mathsf{Error}_{\mathrm{p}}^{(k)}}.
		\]
		This completes the proof.
	\end{proof}

	\begin{proof}[\textbf{Proof of Theorem~\ref{thm:rank-selection-consistency}}]
		Fix an arbitrary client $k\in[K]$. Note that $r^* \leq \bar r$. By Proposition~\ref{thm:local-model-error-upper-bound}, under the stated conditions, we have
		\[
			\|\widetilde{\bm A}_{0,k}-\bm A_0^*\|_{\F} \lesssim \sqrt{\overline{r}}\frac{\lambda_k}{C_{\mathrm{RSC}}^{(k)}} + \sqrt{s_q}\Bigl(\frac{\omega_k}{C_{\mathrm{RSC}}^{(k)}}\Bigr)^{1-q/2} = o\bigl(c(d,T_k)\bigr)
		\]
		with probability tending to $1$ as $d,T_k\to\infty$. Then by Weyl's inequality \cite{weyl1912asymptotische}, we have
		\[
			\max_{1\leq j\leq pd}\bigl|\widetilde{\sigma}_{k,j}-\sigma_j(\bm A_0^*)\bigr| \leq \|\widetilde{\bm A}_{0,k}-\bm A_0^*\|_{\op} \leq \|\widetilde{\bm A}_{0,k}-\bm A_0^*\|_{\F} = o_p\bigl(c(d,T_k)\bigr).
		\]

		Note that $\sigma_{r^*}(\bm A_0^*)/\sigma_j(\bm A_0^*)\leq 1$. Then, under conditions in Theorem~\ref{thm:rank-selection-consistency}, we have
		\[
			c(d,T_k) \ll \min_{1 \leq j \leq r^* - 1} \frac{\sigma_{r^*}(\bm A_0^*) \sigma_{j+1}(\bm A_0^*)}{\sigma_{j}(\bm A_0^*)} \leq \min_{1 \leq j \leq r^* - 1} \sigma_{j+1}(\bm A_0^*) \leq \sigma_{r}(\bm A_0^*), \ \text{for all} \ r \leq r^*-1.
		\]
		Hence, $c(d,T_k) = o(\sigma_{r}(\bm A_0^*))$ for all $r \leq r^*-1$.

		Note that $\widetilde{\sigma}_{k,r} + c(d,T_k) = \sigma_{r}(\bm A_0^*) + \big(\widetilde{\sigma}_{k,r} - \sigma_{r}(\bm A_0^*)\big) + c(d,T_k)$. Then, for each $r \in [\bar r-1]$, we have the following cases:
		\begin{itemize}
			\item [(a)] For $r > r^*$, since $\sigma_r(\bm A_0^*) = 0$ and $\widetilde{\sigma}_{k,r} - \sigma_{r}(\bm A_0^*) = o_p(c(d,T_k))$, the term $c(d,T_k)$ dominates in $\widetilde{\sigma}_{k,r} + c(d,T_k)$, i.e., $\widetilde{\sigma}_{k,r} + c(d,T_k) = O_p(c(d,T_k))$;
			\item [(b)] For $r \leq r^*$, since $\widetilde{\sigma}_{k,r} - \sigma_{r}(\bm A_0^*) = o_p(c(d,T_k))$ and $c(d,T_k) = o(\sigma_r(\bm A_0^*))$, the term $\sigma_r(\bm A_0^*)$ dominates in $\widetilde{\sigma}_{k,r} + c(d,T_k)$, i.e., $\widetilde{\sigma}_{k,r} + c(d,T_k) = O_p(\sigma_r(\bm A_0^*))$.
		\end{itemize}

		\noindent
		Recall the rank estimator
		\[
			\widehat{r}_k = \argmin_{1 \leq r \leq \bar{r}-1}\widehat{R}_k(r),
			\ \text{with} \
			\widehat{R}_k(r) = \frac{\widetilde{\sigma}_{k,r+1} + c(d,T_k)}{\widetilde{\sigma}_{k,r} + c(d,T_k)},
		\]
		The above analysis in (a) and (b) implies the following limits in probability as $d,T_k \to \infty$:

		\noindent
		\textbf{Part I ($r>r^*$).} Both $\widetilde{\sigma}_{k,r+1}+c(d,T_k)$ and $\widetilde{\sigma}_{k,r}+c(d,T_k)$ are dominated by $c(d,T_k)$, and hence
		\[
			\widehat{R}_k(r)\xrightarrow{p} 1.
		\]

		\noindent
		\textbf{Part II ($r<r^*$).} Both $\widetilde{\sigma}_{k,r+1}+c(d,T_k)$ and $\widetilde{\sigma}_{k,r}+c(d,T_k)$ are dominated by $\sigma_{r+1}(\bm A_0^*)$ and $\sigma_r(\bm A_0^*)$, respectively, so that
		\[
			\widehat{R}_k(r)\xrightarrow{p} \frac{\sigma_{r+1}(\bm A_0^*)}{\sigma_r(\bm A_0^*)}.
		\]

		\noindent
		\textbf{Part III ($r=r^*$).} The numerator $\widetilde{\sigma}_{k,r^*+1}+c(d,T_k)$ is dominated by $c(d,T_k)$, while the denominator $\widetilde{\sigma}_{k,r^*}+c(d,T_k)$ is dominated by $\sigma_{r^*}(\bm A_0^*)$, and hence
		\[
			\widehat{R}_k(r^*)\xrightarrow{p} \frac{c(d,T_k)}{\sigma_{r^*}(\bm A_0^*)}.
		\]
		Moreover, the condition
		\[
			\frac{c(d,T_k)}{\sigma_{r^*}(\bm A_0^*)} \ll \min_{1\leq r\leq r^*-1}\frac{\sigma_{r+1}(\bm A_0^*)}{\sigma_r(\bm A_0^*)},
		\]
		in Theorem \ref{thm:rank-selection-consistency} implies $c(d,T_k)/\sigma_{r^*}(\bm A_0^*)\ll 1$. Therefore, $\widehat{R}_k(r^*)=o_p(1)$. Combining parts (I)–(III), it follows that, with probability tending to $1$, $\widehat{R}_k(r^*)$ is the unique minimizer of $\{\widehat{R}_k(r)\}_{r=1}^{\bar r-1}$, which yields $\widehat r_k=r^*$. This proves that $\mathbb P(\widehat r_k=r^*)\to 1$ as $d,T_k\to\infty$.
	\end{proof}

	\section{Proofs of Primary Lemmas}\label{append:proofs of primary lemmas}

	\begin{proof}[\textbf{Proof of Lemma \ref{lem:subg-l2-max-sqrtds}}]
		Following a standard property of the sub-Gaussian Orlicz norm: if a scalar random variable $X$ satisfies $\|X\|_{\psi_2}\leq K$, then there exists a constant $c_1>0$ such that for $u\geq 0$,
		\[
			\mathbb P(|X|\geq u)\leq 2\exp\Big(-c_1\frac{u^2}{K^2}\Big).
		\]
		Fix any $\bbm u\in\mathbb S^{d-1}$. Since $\|\bbm\zeta_{t}\|_{\psi_2} \leq K_\zeta$, we have $\|\langle \bbm u,\bbm\zeta_{t}\rangle\|_{\psi_2}\leq K_\zeta$. Hence, for any $u\geq 0$,
		\begin{align}\label{eq:fixed-direction-tail-psi2}
			\mathbb P\left(|\langle \bbm u,\bbm\zeta_{t}\rangle|\geq u\right) \leq 2\exp\Big(-c_1\frac{u^2}{K_\zeta^2}\Big).
		\end{align}
		Let $\mathcal N$ be a $1/2$-net of $\mathbb S^{d-1}$ under $\ell_2$-norm. It is well known that such a net exists with $|\mathcal N|\leq 5^d$ and that for any $\bbm x\in\mathbb R^d$, $\|\bbm x\|_2 \leq 2\max_{\bbm v\in\mathcal N} |\langle \bbm v,\bbm x\rangle|$. Therefore,
		\begin{align}\label{eq:l2-to-net}
			\mathbb P\big(\|\bbm\zeta_{t}\|_2\geq 2u\big)
			&\leq \mathbb P\Big(\max_{\bbm v\in\mathcal N}|\langle \bbm v,\bbm\zeta_{t}\rangle|\geq u\Big) \leq \sum_{\bbm v\in\mathcal N}\mathbb P\big(|\langle \bbm v,\bbm\zeta_{t}\rangle|\geq u\big)\notag\\
			&\leq 2|\mathcal N|\exp\!\Big(-c_1\frac{u^2}{K_\zeta^2}\Big) \leq 2\exp\Big(d\log 5 - c_1\frac{u^2}{K_\zeta^2}\Big).
		\end{align}
		Now set $u = K_\zeta\sqrt{(d\log 5 + c_2\log T)/c_1}$. Then, \eqref{eq:l2-to-net} implies that for any $t\in \mathbb{N}^+$,
		\[
			\mathbb P\big(\|\bbm\zeta_{t}\|_2\geq 2u\big) \leq 2\exp(-c_2\log T).
		\]
		Since $\sqrt{d\log 5+c_2\log T}\asymp \sqrt{d+\log T}$, absorbing constants into $C>0$ yields
		\begin{align}\label{eq:single-time-subg-l2}
			\mathbb P\big(\|\bbm\zeta_{t}\|_2\geq C K_\zeta\sqrt{d+\log T}\big)\leq 2\exp(-C\log T), \quad \text{for any } t\in \mathbb{N}^+.
		\end{align}
		Finally, by a union bound over $t=1,\ldots,T'$, we have
		\begin{align*}
			\mathbb P\Big(\max_{1\leq t\leq T'}\|\bbm\zeta_t\|_2\geq C K_\zeta\sqrt{d+\log T}\Big) &\leq \sum_{t=1}^{T'} \mathbb P\big(\|\bbm\zeta_t\|_2\geq C K_\zeta\sqrt{d+\log T}\big)\\
			&\leq 2T'\exp(-C\log T) \leq 2T\exp(-C\log T).
		\end{align*}
		Adjusting constants so that $2T\exp(-C\log T)\leq \exp(-C\log T)$, which proves \eqref{eq:max-eps-subg-highprob-sqrtds}.
	\end{proof}

	\begin{proof}[\textbf{Proof of Lemma \ref{lem:gaussian-opnorm-hp-strong}}]
		Write $\bm Z=\sigma\bm G$, where $\bm G$ has $i.i.d.$ $\mathcal N(0,1)$ entries.
		Then $\|\bm Z\|_{\op}=\sigma\|\bm G\|_{\op}$.
		Since each $G_{ij}$ is sub-Gaussian with $\|G_{ij}\|_{\psi_2}\leq C_0$ for a constant $c_0$,
		we may apply the bound for matrices with independent mean-zero sub-Gaussian entries, i.e., Theorem~4.4.5 of \cite{vershynin2018high}, which states that for $t\geq 0$,
		\[
			\mathbb P\left(\|\bm G\|_{\op} > c_0 \bigl(\sqrt d+\sqrt{pd}+t\bigr)\right)\leq 2\exp(-C_0 t^2).
		\]
		Absorbing $\sqrt d$ into $\sqrt{pd}$, we have
		\[
			\mathbb P\left(\|\bm G\|_{\op} > c(\sqrt{pd}+t)\right)\leq 2\exp(-C t^2).
		\]
		Multiplying the above inequality by $\sigma$ and setting $t=\sqrt{\log T}$ yields the desired bound.
	\end{proof}

	\begin{proof}[\textbf{Proof of Lemma \ref{lem:tangent-contraction-op}}]
		Since $\mathcal P_{\mathcal T_r(\bm A)}(\cdot)$ is the orthogonal projector onto $\mathcal T_r(\bm A)$, we have the orthogonal decomposition $\bm N=\mathcal P_{\mathcal T_r(\bm A)}(\bm N)+\mathcal P_{\mathcal T_r^\perp(\bm A)}(\bm N)$ for any $\bm M\in\mathcal T_r(\bm A)$ and any $\bm N\in\mathbb R^{d\times pd}$. Therefore,
		\begin{align}\label{eq:proj-innerprod}
			\langle \bm M,\bm N\rangle
			=\big\langle \bm M,\mathcal P_{\mathcal T_r(\bm A)}(\bm N)\big\rangle
			+\big\langle \bm M,\mathcal P_{\mathcal T_r^\perp(\bm A)}(\bm N)\big\rangle
			=\big\langle \bm M,\mathcal P_{\mathcal T_r(\bm A)}(\bm N)\big\rangle,
		\end{align}
		where the last equality holds because $\bm M\in\mathcal T_r(\bm A)$ is orthogonal to $\mathcal P_{\mathcal T_r^\perp(\bm A)}(\bm N)\in\mathcal T_r^\perp(\bm A)$. Since $\bm S$ is symmetric, for any $\bm M_1,\bm M_2\in\mathcal T_r(\bm A)$ we have
		\begin{align}\label{eq:self-adjoint-identity}
			\big\langle \bm M_1,\mathcal P_{\mathcal T_r(\bm A)}(\bm M_2\bm S)\big\rangle
			=\langle \bm M_1,\bm M_2\bm S\rangle
			=\langle \bm M_1\bm S,\bm M_2\rangle
			=\big\langle \mathcal P_{\mathcal T_r(\bm A)}(\bm M_1\bm S),\bm M_2\big\rangle,
		\end{align}
		where the first and last equalities follow from \eqref{eq:proj-innerprod}, while the middle equality uses the symmetry of $\bm S$. Consequently, the map $\bm M\mapsto \mathcal P_{\mathcal T_r(\bm A)}(\bm M\bm S)$ is self-adjoint on $(\mathcal T_r(\bm A),\langle\cdot,\cdot\rangle)$. Hence it admits an orthonormal eigenbasis $\{\bm U_j\}\subset\mathcal T_r(\bm A)$ with real eigenvalues $\{\lambda_j\}$ satisfying $\mathcal P_{\mathcal T_r(\bm A)}(\bm U_j\bm S)=\lambda_j\bm U_j$. For any $\bm M\in\mathcal T_r(\bm A)$, write $\bm M=\sum_j a_j\bm U_j$. Since $\{\bm U_j\}$ is orthonormal under $\langle\cdot,\cdot\rangle$, we have
		\begin{align}\label{eq:M-expansion-norm}
			\|\bm M\|_{\F}^2
			=\Big\langle \sum_j a_j\bm U_j,\sum_{\ell} a_{\ell}\bm U_{\ell}\Big\rangle
			=\sum_{j,\ell} a_j a_{\ell}\langle \bm U_j,\bm U_{\ell}\rangle
			=\sum_j a_j^2.
		\end{align}
		Moreover, we have
		\[
			\mathcal P_{\mathcal T_r(\bm A)}(\bm M\bm S)
			=\mathcal P_{\mathcal T_r(\bm A)}\!\Big(\Big(\sum_j a_j\bm U_j\Big)\bm S\Big)
			=\sum_j a_j\,\mathcal P_{\mathcal T_r(\bm A)}(\bm U_j\bm S)
			=\sum_j a_j\lambda_j\bm U_j,
		\]
		and therefore
		\begin{align*}
			\bm M-2\rho\,\mathcal P_{\mathcal T_r(\bm A)}(\bm M\bm S) = \sum_j a_j\bm U_j-2\rho\sum_j a_j\lambda_j\bm U_j = \sum_j a_j(1-2\rho\lambda_j)\bm U_j.
		\end{align*}
		Taking the Frobenius norm and expanding using orthonormality gives
		\begin{align}\label{eq:spectral-sum-bound}
			\Big\|\bm M-2\rho\,\mathcal P_{\mathcal T_r(\bm A)}(\bm M\bm S)\Big\|_{\F}^2
			&=\Big\langle \sum_j a_j(1-2\rho\lambda_j)\bm U_j,\ \sum_{\ell} a_{\ell}(1-2\rho\lambda_{\ell})\bm U_{\ell}\Big\rangle\notag\\
			&=\sum_{j,\ell} a_j a_{\ell}(1-2\rho\lambda_j)(1-2\rho\lambda_{\ell})\langle \bm U_j,\bm U_{\ell}\rangle\notag\\
			&=\sum_j a_j^2(1-2\rho\lambda_j)^2
			\leq \Big(\max_j|1-2\rho\lambda_j|\Big)^2\sum_j a_j^2.
		\end{align}
		By \eqref{eq:M-expansion-norm}, the right-hand side equals $\big(\max_j|1-2\rho\lambda_j|\big)^2\|\bm M\|_{\F}^2$.

		It remains to upper bound $\max_j|1-2\rho\lambda_j|$ by $\|\bm I_{pd}-2\rho\bm S\|_{\op}$. For each eigenpair $(\lambda_j,\bm U_j)$ with $\|\bm U_j\|_{\F}=1$, we have
		\begin{align}\label{eq:lambda-op-control}
			|1-2\rho\lambda_j|
			&=\bigl\|(1-2\rho\lambda_j)\bm U_j\bigr\|_{\F}
			=\bigl\|\bm U_j-2\rho\,\mathcal P_{\mathcal T_r(\bm A)}(\bm U_j\bm S)\bigr\|_{\F}\notag\\
			&=\bigl\|\mathcal P_{\mathcal T_r(\bm A)}\bigl(\bm U_j(\bm I_{pd}-2\rho\bm S)\bigr)\bigr\|_{\F}
			\leq \bigl\|\bm U_j(\bm I_{pd}-2\rho\bm S)\bigr\|_{\F}\notag\\
			&\leq \|\bm I_{pd}-2\rho\bm S\|_{\op}\,\|\bm U_j\|_{\F}
			=\|\bm I_{pd}-2\rho\bm S\|_{\op},
		\end{align}
		where we used the non-expansiveness of orthogonal projections in Frobenius norm and $\|\bm U_j\bm B\|_{\F}\leq \|\bm U_j\|_{\F}\|\bm B\|_{\op}$ for any conformable matrix $\bm B$. Thus, we have $\max_j|1-2\rho\lambda_j|\leq \|\bm I_{pd}-2\rho\bm S\|_{\op}$. Plugging the eigen bound into \eqref{eq:spectral-sum-bound} and using \eqref{eq:M-expansion-norm} yields
		\[
			\Big\|\bm M-2\rho\,\mathcal P_{\mathcal T_r(\bm A)}(\bm M\bm S)\Big\|_{\F}^2
			\leq \|\bm I_{pd}-2\rho\bm S\|_{\op}^2\,\|\bm M\|_{\F}^2,
		\]
		which concludes the proof.
	\end{proof}

	\begin{proof}[\textbf{Proof of Lemma \ref{lem:var-cov-conc}}]
		Fix any client $k\in[K]$. Let $\mathbb S^{pd-1}$ be the Euclidean unit sphere in $\mathbb R^{pd}$. Fix any $\bbm v\in\mathbb S^{pd-1}$ and choose $\bbm u\in\mathcal N$ such that $\|\bbm v-\bbm u\|_2\leq \varepsilon$. Then, for any symmetric matrix $\bm M$,
		\[
			\bbm v^{\top}\bm M\bbm v-\bbm u^{\top}\bm M\bbm u = (\bbm v-\bbm u)^{\top}\bm M(\bbm v-\bbm u)+2(\bbm v-\bbm u)^{\top}\bm M\bbm u.
		\]
		Therefore, for any $\bbm v\in\mathbb S^{pd-1}$ and the corresponding $\bbm u\in\mathcal N$,
		\begin{align*}
			\big|\bbm v^{\top}\bm M\bbm v\big| &\leq \big|\bbm u^{\top}\bm M\bbm u\big| + \big|(\bbm v-\bbm u)^{\top}\bm M(\bbm v-\bbm u)\big|+2\big|(\bbm v-\bbm u)^{\top}\bm M\bbm u\big| \\
			&\leq \big|\bbm u^{\top}\bm M\bbm u\big| + \|\bm M\|_{\op}\|\bbm v-\bbm u\|_2^2 + 2\|\bm M\|_{\op}\|\bbm v-\bbm u\|_2\|\bbm u\|_2 \\
			&\leq \big|\bbm u^{\top}\bm M\bbm u\big|+(2\varepsilon+\varepsilon^2)\|\bm M\|_{\op},
		\end{align*}
		where we used $\|\bbm u\|_2=1$ and $\|\bbm v-\bbm u\|_2\leq \varepsilon$. Taking the supremum over $\bbm v\in\mathbb S^{pd-1}$ yields
		\[
			\|\bm M\|_{\op} = \sup_{\bbm v\in\mathbb S^{pd-1}}\big|\bbm v^{\top}\bm M\bbm v\big| \leq \max_{\bbm u\in\mathcal N}\big|\bbm u^{\top}\bm M\bbm u\big|+(2\varepsilon+\varepsilon^2)\|\bm M\|_{\op}.
		\]
		Thus, after rearranging the terms, we have
		\[
			\|\bm M\|_{\op}\leq \frac{1}{1-2\varepsilon-\varepsilon^2}\max_{\bbm u\in\mathcal N}\big|\bbm u^{\top}\bm M\bbm u\big|.
		\]
		Then, take $\varepsilon=1/4$ and let $\mathcal N\subset\mathbb S^{pd-1}$ be a corresponding $1/4$-net with
		$|\mathcal N|\leq (1+2/\varepsilon)^{pd}=9^{pd}$. With $\varepsilon=1/4$, we have $1-2\varepsilon-\varepsilon^2=7/16$, so that $\|\bm M\|_{\op}\leq 16/7\max_{\bbm u\in\mathcal N}|\bbm u^{\top}\bm M\bbm u|$.
		Applying this bound to the symmetric matrix $1/T_k \bm X_k\bm X_k^\top - \bbm \Sigma_{x,k}$ gives
		\[
			\left\|\frac{1}{T_k}\bm X_k\bm X_k^\top - \bbm \Sigma_{x,k}\right\|_{\op} \leq \frac{16}{7}\max_{\bbm u\in\mathcal N}\left|\bbm u^{\top}\left(\frac{1}{T_k}\bm X_k\bm X_k^\top - \bbm \Sigma_{x,k}\right)\bbm u\right|.
		\]
		Fix any $\bbm u\in\mathcal N$ with $\|\bbm u\|_2=1$. Applying Lemma~\ref{lemma:HW-uXTXu} yields that for every $t>0$,
		\[
			\mathbb{P}\Big(\big|\bbm{u}^\top\bm{X}_k^\top\bm{X}_k\bbm{u}-\mathbb{E}[\bbm{u}^\top\bm{X}_k^\top\bm{X}_k\bbm{u}]\big|\geq t\Big)
			\leq 2\exp\Bigg(-c\min\Bigg\{\frac{t^2}{(C_{\epsilon,\mathcal A}^{(k)})^2\sigma^4pT_k},\frac{t}{C_{\epsilon,\mathcal A}^{(k)}\sigma^2p}\Bigg\}\Bigg).
		\]
		Note that
		\begin{align*}
			\mathbb P\left(\left|\bbm u^{\top}\left(\frac{1}{T_k}\bm X_k\bm X_k^\top - \bbm \Sigma_{x,k}\right)\bbm u\right|\geq t\right) &=
			\mathbb P\left(\left|\bbm u^\top\bm X_k^\top\bm X_k\bbm u-\mathbb E[\bbm u^\top\bm X_k^\top\bm X_k\bbm u]\right|\geq T_k t\right)\\
			&\leq 2\exp\Bigg(-c\min\Bigg\{\frac{T_kt^2}{(C_{\epsilon,\mathcal A}^{(k)})^2\sigma^4p},\frac{T_kt}{C_{\epsilon,\mathcal A}^{(k)}\sigma^2p}
			\Bigg\}\Bigg).
		\end{align*}
		We now choose $t = c_0C_{\epsilon,\mathcal A}^{(k)}\sigma^2p\sqrt{d/T_k}$ and hence $T_k t=c_0C_{\epsilon,\mathcal A}^{(k)}\sigma^2p\sqrt{dT_k}$ and $T_k t^2=c_0^2(C_{\epsilon,\mathcal A}^{(k)})^2\sigma^4p^2d$. With this choice, the two terms in the Bernstein minimum become
		\[
			\frac{T_kt^2}{(C_{\epsilon,\mathcal A}^{(k)})^2\sigma^4p} = \frac{c_0^2(C_{\epsilon,\mathcal A}^{(k)})^2\sigma^4p^2 d}{(C_{\epsilon,\mathcal A}^{(k)})^2\sigma^4p} = c_0^2pd,
			\ \text{ and } \ 
			\frac{T_kt}{C_{\epsilon,\mathcal A}^{(k)}\sigma^2p} = \frac{c_0C_{\epsilon,\mathcal A}^{(k)}\sigma^2p\sqrt{dT_k}}{C_{\epsilon,\mathcal A}^{(k)}\sigma^2p} = c_0\sqrt{dT_k}.
		\]
		Hence, given $T_k\geq c_1 p^2 d$, we have
		\[
			\mathbb P\left(\left|\bbm u^{\top}\Big(\frac{1}{T_k}\bm X_k\bm X_k^\top - \bbm \Sigma_{x,k}\Big)\bbm u\right| \geq c_0C_{\epsilon,\mathcal A}^{(k)}\sigma^2p\sqrt{\frac{d}{T_k}}\right) \leq 2\exp(-cpd),
		\]
		after adjusting absolute constants. Consequently, taking a union bound over $\bbm u\in\mathcal N$ with $|\mathcal N|\leq 9^{pd}$ yields
		\begin{align*}
			\mathbb P\left(\left\|\frac{1}{T_k}\bm X_k\bm X_k^\top - \bbm \Sigma_{x,k}\right\|_{\op} \geq \frac{16}{7}c_0C_{\epsilon,\mathcal A}^{(k)}\sigma^2p\sqrt{\frac{d}{T_k}}\right) &\leq \mathbb P\left(\max_{\bbm u\in\mathcal N}\left|\bbm u^{\top}\Big(\frac{1}{T_k}\bm X_k\bm X_k^\top - \bbm \Sigma_{x,k}\Big)\bbm u\right| \geq c_0C_{\epsilon,\mathcal A}^{(k)}\sigma^2p\sqrt{\frac{d}{T_k}}\right)\\
			&\leq 2|\mathcal N|\exp(-cpd) \leq 2\cdot 9^{pd}\exp(-cpd) \leq \exp(-Cpd),
		\end{align*}
		for some constant $C>0$. Absorbing the prefactor $16c_0/7$ into the absolute constant in the operator-norm bound concludes the proof.
	\end{proof}

	\begin{proof}[\textbf{Proof of Lemma \ref{lem:RSC-var}}]
		Fix any client $k\in[K]$. Note that
		\begin{align}
			\label{eq:RSC-identity}
			\frac{1}{T_k}\sum_{t=1}^{T_k}\|\bm\Delta\bbm x_{k,t}\|_2^2
			&= \frac{1}{T_k}\sum_{t=1}^{T_k}\bbm x_{k,t}^\top \bbm \Delta^\top \bbm \Delta \bbm x_{k,t}
			= \frac{1}{T_k}\sum_{t=1}^{T_k}\tr\left(\bbm x_{k,t}^\top\bm\Delta^{\top}\bm\Delta\bbm x_{k,t}\right) \notag \\
			&= \frac{1}{T_k}\sum_{t=1}^{T_k}\tr\left(\bm\Delta\bbm x_{k,t}\bbm x_{k,t}^\top\bm\Delta^{\top}\right)
			= \tr\left(\bm\Delta\left(\frac{1}{T_k}\bm X_k\bm X_k^\top\right)\bm\Delta^{\top}\right) \\
			&\geq \lambda_{\min}\left(\frac{1}{T_k}\bm X_k\bm X_k^\top\right)\|\bm\Delta\|_{\mathrm F}^2.\notag
		\end{align}
		Under the sample size condition in Lemma~\ref{lem:RSC-var}, i.e., $T_k\gtrsim (C_{\epsilon,\mathcal A}^{(k)}/C_{\mathrm{RSC}}^{(k)}\vee 1)^2p^2d$, we have in particular $T_k\gtrsim (C_{\epsilon,\mathcal A}^{(k)}/C_{\mathrm{RSC}}^{(k)})^2p^2d$, and hence there exists a constant $c_0>0$ such that
		\[
			c_0C_{\epsilon,\mathcal A}^{(k)}\sigma^2p\sqrt{\frac{d}{T_k}}\leq C_{\mathrm{RSC}}^{(k)}.
		\]
		Therefore, by Lemma~\ref{lem:var-cov-conc}, since $T_k\gtrsim (C_{\epsilon,\mathcal A}^{(k)}/C_{\mathrm{RSC}}^{(k)}\vee 1)^2p^2d$ and hence $T_k\gtrsim p^2d$, there exist a constant $C>0$ such that, with probability at least $1-\exp(-C pd)$,
		\begin{align}\label{eq:RSC-dev-bound}
			\left\|\frac{1}{T_k}\bm X_k\bm X_k^\top-\bbm \Sigma_{x,k}\right\|_{\op}
			\leq c_0\,C_{\epsilon,\mathcal A}^{(k)}\sigma^2p\sqrt{\frac{d}{T_k}}
			\leq C_{\mathrm{RSC}}^{(k)}.
		\end{align}
		Moreover, by Lemma~\ref{lem:Sigma-x-pd}, we have $\lambda_{\min}(\bbm \Sigma_{x,k})\geq c_{\epsilon,\mathcal A}^{(k)}=2C_{RSC}^{(k)}$. Combining this with \eqref{eq:RSC-dev-bound} in Weyl's inequality yields
		\[
			\lambda_{\min}\left(\frac{1}{T_k}\bm X_k\bm X_k^\top\right)
			\geq \lambda_{\min}(\bbm \Sigma_{x,k})-\left\|\frac{1}{T_k}\bm X_k\bm X_k^\top-\bbm \Sigma_{x,k}\right\|_{\op}
			\geq 2C_{\mathrm{RSC}}^{(k)}-C_{\mathrm{RSC}}^{(k)}=C_{\mathrm{RSC}}^{(k)}.
		\]
		Plugging this bound into \eqref{eq:RSC-identity} concludes the proof.
	\end{proof}

	\begin{proof}[\textbf{Proof of Lemma \ref{lem:var-XtE-infty-bound}}]
		Fix any client $k\in[K]$ and any entry $(i,j)$ of $\bm X_k^{\top}\bm E_k$. Taking $\bbm u=\bbm e_i^{(pd)}\in\mathbb R^{pd}$ and $\bbm v=\bbm e_j^{(d)}\in\mathbb R^{d}$ in Lemma~\ref{lemma:HW-uXEv} yields, for any $t>0$,
		\[
			\mathbb P\bigl(|(\bm X_k^{\top}\bm E_k)_{ij}|\geq t\bigr) = \mathbb P\left(\left|(\bbm e_i^{(pd)})^\top\bm X_k^{\top}\bm E_k\,\bbm e_j^{(d)}\right|\geq t\right) \leq 2\exp\left(-c\min\left\{\frac{t^2}{(C_{\epsilon,\mathcal A}^{(k)})^2\sigma^4T_k},\frac{t}{C_{\epsilon,\mathcal A}^{(k)}\sigma^2\sqrt{p}}\right\}\right).
		\]
		Then, setting $t=c_0T_kC_{\epsilon,\mathcal A}^{(k)}\sigma^2\sqrt{\log(pd)/T_k}$ yields
		\begin{align}\label{eq:entry-tail-XtEij-expd}
			\mathbb P\left(\left|\left(\frac{1}{T_k}\bm X_k^{\top}\bm E_k\right)_{ij}\right|\geq c_0C_{\epsilon,\mathcal A}^{(k)}\sigma^2\sqrt{\frac{\log(pd)}{T_k}}\right) \leq 2\exp\left(-c\min\left\{c_0^2\log(pd), c_0\sqrt{\frac{T_k\log(pd)}{p}}\right\}\right).
		\end{align}
		Hence, given $T_k\geq c_1p\log(pd)$, absorbing the constants $\log2$, $c_0$ and $c_0\sqrt{c_1}$ into $c$ yields
		\[
			\mathbb P\left(\left|\left(\frac{1}{T_k}\bm X_k^{\top}\bm E_k\right)_{ij}\right|\geq c_0C_{\epsilon,\mathcal A}^{(k)}\sigma^2\sqrt{\frac{\log(pd)}{T_k}}\right) \leq 2\exp(-c\log(pd)).
		\]
		Take the union bound over all $p d^2$ entries yields
		\begin{align*}
			& \mathbb P\left(\left\|\frac{1}{T_k}\bm X_k^{\top}\bm E_k\right\|_{\infty} \geq c_0C_{\epsilon,\mathcal A}^{(k)}\sigma^2\sqrt{\frac{\log(pd)}{T_k}}\right) \leq  \sum_{i=1}^{pd}\sum_{j=1}^{d}\mathbb P\left(\left|\left(\frac{1}{T_k}\bm X_k^{\top}\bm E_k\right)_{ij}\right| \geq c_0C_{\epsilon,\mathcal A}^{(k)}\sigma^2\sqrt{\frac{\log(pd)}{T_k}}\right)  \\
			&\ \ \leq 2p d^2 \exp(-c\log(pd)) = \exp\Bigl(\log 2 + \log d + \log (pd) - c\log(pd)\Bigr) \leq \exp(-C\log(pd)).
		\end{align*}
		for some constant $C>0$, which proves \eqref{eq:var-XtE-infty-bound-simplified}.
	\end{proof}

	\begin{proof}[\textbf{Proof of Lemma \ref{lem:var-XtE-op-bound}}]
		Fix any client $k\in[K]$. Let $\mathbb S^{pd-1}$ and $\mathbb S^{d-1}$ be the Euclidean unit spheres. Fix $\varepsilon=1/4$, and let $\mathcal N_u\subset\mathbb S^{pd-1}$ and $\mathcal N_v\subset\mathbb S^{d-1}$ be $\varepsilon$-nets with $|\mathcal N_u|\leq (1+2/\varepsilon)^{pd} = 9^{pd}$ and $|\mathcal N_v|\leq (1+2/\varepsilon)^{d} = 9^{d}$. A standard net argument implies
		\[
			\left\|\frac{1}{T_k}\bm X_k^{\top}\bm E_k\right\|_{\op}
			\leq \frac{1}{(1-\varepsilon)^2}\max_{\bbm u\in\mathcal N_u,\,\bbm v\in\mathcal N_v}
			\left|\bbm u^{\top}\left(\frac{1}{T_k}\bm X_k^{\top}\bm E_k\right)\bbm v\right|
			\leq 4\max_{\bbm u\in\mathcal N_u,\,\bbm v\in\mathcal N_v}
			\left|\bbm u^{\top}\left(\frac{1}{T_k}\bm X_k^{\top}\bm E_k\right)\bbm v\right|.
		\]
		Fix $(\bbm u,\bbm v)\in\mathcal N_u\times\mathcal N_v$. Applying Lemma~\ref{lemma:HW-uXEv} with $\|\bbm u\|_2=\|\bbm v\|_2=1$ yields, for any $t>0$,
		\[
			\mathbb P\left(\left|\bbm u^{\top}\bm X_k^{\top}\bm E_k\bbm v\right|\geq t\right) \leq 2\exp\Bigg(-c\min\Big\{\frac{t^2}{(C_{\epsilon,\mathcal A}^{(k)})^2\sigma^4T_k},\frac{t}{C_{\epsilon,\mathcal A}^{(k)}\sigma^2\sqrt{p}}\Big\}\Bigg).
		\]
		Then, setting $t=c_0T_kC_{\epsilon,\mathcal A}^{(k)}\sigma^2\sqrt{pd/T_k}$ yields
		\begin{align}\label{eq:net-pair-tail-op-expd}
			\mathbb P\left(\left|\bbm u^{\top}\left(\frac{1}{T_k}\bm X_k^{\top}\bm E_k\right)\bbm v\right|\geq c_0C_{\epsilon,\mathcal A}^{(k)}\sigma^2\sqrt{\frac{pd}{T_k}}\right) \leq 2\exp\Bigg(-c\min\Big\{c_0^2 p d,\ c_0\sqrt{T_k d}\Big\}\Bigg).
		\end{align}
		Hence, given $T_k\geq c_1p^2d$, absorbing the constants $\log2$, $c_0^2$ and $c_0\sqrt{c_1}$ into $c$ yields
		\[
			\mathbb P\left(\left|\bbm u^{\top}\left(\frac{1}{T_k}\bm X_k^{\top}\bm E_k\right)\bbm v\right|\geq c_0C_{\epsilon,\mathcal A}^{(k)}\sigma^2\sqrt{\frac{p d}{T_k}}\right) \leq 2\exp(-c pd).
		\]
		Consequently, taking a union bound over $\mathcal N_u \times\mathcal N_v$ with $|\mathcal N_u||\mathcal N_v|\leq 9^{pd+d}$ yields
		\begin{align*}
			&\mathbb P\left(\left\|\frac{1}{T_k}\bm X_k^{\top}\bm E_k\right\|_{\op}\geq 4c_0C_{\epsilon,\mathcal A}^{(k)}\sigma^2\sqrt{\frac{pd}{T_k}}\right) \leq \sum_{\bbm u\in\mathcal N_u}\sum_{\bbm v\in\mathcal N_v} \mathbb P\left(\left|\bbm u^{\top}\left(\frac{1}{T_k}\bm X_k^{\top}\bm E_k\right)\bbm v\right|\geq c_0C_{\epsilon,\mathcal A}^{(k)}\sigma^2\sqrt{\frac{pd}{T_k}}\right) \\
			&\ \ \leq 2\cdot 9^{pd+d}\exp(-c pd) = \exp\Bigl(\log 2 + (pd+d)\log 9 - c pd\Bigr) \leq \exp(-C pd).
		\end{align*}
		for some constant $C>0$, which proves \eqref{eq:var-XtE-op-bound-simplified}.
	\end{proof}

	\begin{proof}[\textbf{Proof of Lemma \ref{lem:var-cov-conc-unified}}]
		Fix any client $k\in[K]$. Let $\mathbb S^{pd-1}$ and $\mathbb S^{d-1}$ be the Euclidean unit spheres.
		Fix $\varepsilon=1/4$, and let $\mathcal N\subset \mathbb S^{pd-1}$ be an $\varepsilon$-net with $|\mathcal N|\leq (1+2/\varepsilon)^{pd}=9^{pd}$. Note that $\bm S_{\rm pool}-\bbm\Sigma_{x,{\rm pool}}$ is symmetric. Then, following the standard net argument as in the proof of Lemma \ref{lem:var-cov-conc} yields
		\begin{align}\label{eq:net-bound-pooled}
			\left\|\bm S_{\rm pool}-\bbm\Sigma_{x,{\rm pool}}\right\|_{\op}\leq \frac{16}{7}\max_{\bbm u\in\mathcal N}\left|\bbm u^\top\left(\bm S_{\rm pool}-\bbm\Sigma_{x,{\rm pool}}\right)\bbm u\right|.
		\end{align}
		Fix any $\bbm u\in\mathcal N$ with $\|\bbm u\|_2=1$ and define $Z_k(\bbm u)=\bbm u^\top\bm X_k^\top\bm X_k\bbm u-\mathbb E[\bbm u^\top\bm X_k^\top\bm X_k\bbm u]$ for each $k\in[K]$.
		Consequently, we have $\bbm u^\top(\bm S_{\rm pool}-\bbm\Sigma_{x,{\rm pool}})\bbm u=1/T\sum_{k=1}^K Z_k(\bbm u)$. Under Assumption~\ref{assump:subg}, the innovations $\{\bbm \epsilon_{k,t}\}$ are independent across clients $k$, hence $\{Z_k(\bbm u)\}_{k=1}^K$ are independent for each fixed $\bbm u$. Moreover, denote $C_{\epsilon,\mathcal A}^{\max} = \max_{k\in[K]} C_{\epsilon,\mathcal A}^{(k)}$, $\nu_k^2 = (C_{\epsilon,\mathcal A}^{\max})^2\sigma^4pT_k$ and $b_k = C_{\epsilon,\mathcal A}^{\max}\sigma^2p$. Then, by Lemma \ref{lemma:HW-uXTXu}, we have, for $t>0$ and any $k\in[K]$,
		\begin{align*}\label{eq:Zk-tail-unified}
			\mathbb P\big(|Z_k(\bbm u)|\geq t\big) &\leq 2\exp\Bigg(-c\min\Bigg\{\frac{t^2}{(C_{\epsilon,\mathcal A}^{(k)})^2\sigma^4pT_k},\frac{t}{C_{\epsilon,\mathcal A}^{(k)}\sigma^2p}\Bigg\}\Bigg)\\
			& \leq 2\exp\Bigg(-c\min\Bigg\{\frac{t^2}{(C_{\epsilon,\mathcal A}^{\max})^2\sigma^4pT_k},\frac{t}{C_{\epsilon,\mathcal A}^{\max}\sigma^2p}\Bigg\}\Bigg) = 2\exp\!\left(-c\min\Big\{\frac{t^2}{\nu_k^2},\frac{t}{b_k}\Big\}\right).
		\end{align*}
		Hence by Lemma \ref{lem:bernstein-tail-mgf-equiv}, $Z_k(\bbm u)$ is sub-exponential with Bernstein parameters $(\nu_k,b_k)$. Consequently, Bernstein’s inequality for sums of independent sub-exponential variables, e.g., Eq.~(2.18) in \cite{wainwright2019high}, yields that, for any $x>0$,
		\begin{equation}\label{eq:bern-sumZ}
			\mathbb P\left(\Big|\sum_{k=1}^K Z_k(\bbm u)\Big|\geq x\right) \leq 2\exp\Bigg(-c\min\Bigg\{\frac{x^2}{\sum_{k=1}^K \nu_k^2},\frac{x}{\max_{k\in[K]} b_k}\Bigg\}\Bigg).
		\end{equation}
		Noting that $\sum_{k=1}^K \nu_k^2=(C_{\epsilon,\mathcal A}^{\max})^2\sigma^4p\sum_{k=1}^K T_k = (C_{\epsilon,\mathcal A}^{\max})^2\sigma^4pT$ and $\max_k b_k=C_{\epsilon,\mathcal A}^{\max}\sigma^2p$, we choose $x=c_0C_{\epsilon,\mathcal A}^{\max}\sigma^2p\sqrt{dT}$. Then, the two terms in the Bernstein minimum become
		\[
			\frac{x^2}{\sum_{k=1}^K \nu_k^2} = \frac{c_0^2 (C_{\epsilon,\mathcal A}^{\max})^2\sigma^4 (p^2 d T)}{(C_{\epsilon,\mathcal A}^{\max})^2\sigma^4 (p T)} = c_0^2pd
			\ \text{ and } \ 
			\frac{x}{\max_k b_k} = \frac{c_0C_{\epsilon,\mathcal A}^{\max}\sigma^2p\sqrt{d T}}{C_{\epsilon,\mathcal A}^{\max}\sigma^2 p} = c_0\sqrt{dT}.
		\]
		Hence, given $T\geq c_1p^2d$, we have
		\[
			\mathbb P\left(\Big|\sum_{k=1}^K Z_k(\bbm u)\Big|\geq c_0C_{\epsilon,\mathcal A}^{\max}\sigma^2p\sqrt{dT}\right)\leq 2\exp(-cpd).
		\]
		Recall that $\bbm u^\top(\bm S_{\rm pool}-\bbm\Sigma_{x,{\rm pool}})\bbm u= T^{-1}\sum_{k=1}^K Z_k(\bbm u)$, then for any $\bbm u\in\mathcal N$, we have
		\[
			\mathbb P\left(\Big|\bbm u^\top(\bm S_{\rm pool}-\bbm\Sigma_{x,{\rm pool}})\bbm u\Big|\geq c_0C_{\epsilon,\mathcal A}^{\max}\sigma^2p\sqrt{\frac{d}{T}}\right) \leq 2\exp(-cpd).
		\]
		Consequently, taking a union bound over $\bbm u\in\mathcal N$ with $|\mathcal N|\leq 9^{pd}$ yields
		\[
			\mathbb P\left(\max_{\bbm u\in\mathcal N}\Big|\bbm u^\top(\bm S_{\rm pool}-\bbm\Sigma_{x,{\rm pool}})\bbm u\Big| \geq c_0C_{\epsilon,\mathcal A}^{\max}\sigma^2p\sqrt{\frac{d}{T}}\right) \leq 2|\mathcal N|\exp(-cpd) \leq 2\cdot 9^{pd}\exp(-cpd) \leq \exp(-Cpd).
		\]
		Combining this with the $\varepsilon$-net reduction \eqref{eq:net-bound-pooled} and absorbing the factor $16/7$ into the leading constant concludes the proof.
	\end{proof}

	\begin{proof}[\textbf{Proof of Lemma \ref{lem:var-XtE-op-bound-pooled}}]
		Fix $\varepsilon=1/4$, and let $\mathcal N_u\subset\mathbb S^{pd-1}$ and $\mathcal N_v\subset\mathbb S^{d-1}$ be $\varepsilon$-nets with $|\mathcal N_u|\leq (1+2/\varepsilon)^{pd} = 9^{pd}$ and $|\mathcal N_v|\leq (1+2/\varepsilon)^{d} = 9^{d}$. Then, following the standard net argument as in the proof of Lemma~\ref{lem:var-XtE-op-bound} yields
		\begin{align}\label{eq:net-bound-pooled-R}
			\left\|\bm R_{\rm pool}\right\|_{\op} \leq \frac{1}{(1-\varepsilon)^2} \max_{\bbm u\in\mathcal N_u, \bbm v\in\mathcal N_v}\left|\bbm u^\top \bm R_{\rm pool}\bbm v\right| \leq 4\max_{\bbm u\in\mathcal N_u, \bbm v\in\mathcal N_v}\left|\bbm u^\top \bm R_{\rm pool}\bbm v\right|.
		\end{align}
		Fix any $(\bbm u,\bbm v)\in\mathcal N_u\times\mathcal N_v$ with $\|\bbm u\|_2=\|\bbm v\|_2=1$ and define $Z_k(\bbm u,\bbm v)=\bbm u^\top\bm X_k^\top\bm E_k\bbm v$ for each $k\in[K]$. Consequently, we have $\bbm u^\top \bm R_{\rm pool}\bbm v=T^{-1}\sum_{k=1}^K Z_k(\bbm u,\bbm v)$. Moreover, denote $C_{\epsilon,\mathcal A}^{\max} = \max_{k\in[K]} C_{\epsilon,\mathcal A}^{(k)}$, $\nu_k^2 = (C_{\epsilon,\mathcal A}^{\max})^2\sigma^4T_k$ and $b_k = C_{\epsilon,\mathcal A}^{\max}\sigma^2\sqrt p$. Then, by Lemma~\ref{lemma:HW-uXEv}, we have, for $x>0$ and any $k\in[K]$,
		\begin{align*}\label{eq:Zk-tail-pooled-R}
			\mathbb P\big(|Z_k(\bbm u,\bbm v)|\geq x\big) &\leq 2\exp\Bigg(-c\min\Bigg\{\frac{x^2}{(C_{\epsilon,\mathcal A}^{(k)})^2\sigma^4T_k},\frac{x}{C_{\epsilon,\mathcal A}^{(k)}\sigma^2\sqrt p}\Bigg\}\Bigg)\\
			&\leq 2\exp\Bigg(-c\min\Bigg\{\frac{x^2}{(C_{\epsilon,\mathcal A}^{\max})^2\sigma^4T_k},\frac{x}{C_{\epsilon,\mathcal A}^{\max}\sigma^2\sqrt p}\Bigg\}\Bigg) = 2\exp\!\left(-c\min\Big\{\frac{x^2}{\nu_k^2},\frac{x}{b_k}\Big\}\right).
		\end{align*}

		Hence, by Lemma~\ref{lem:bernstein-tail-mgf-equiv}, $Z_k(\bbm u,\bbm v)$ is sub-exponential with Bernstein parameters $(\nu_k,b_k)$. 
		Note that under Assumption~\ref{assump:subg}, the innovations $\{\bbm \epsilon_{k,t}\}$ are independent across clients $k$, hence $\{Z_k(\bbm u,\bbm v)\}_{k=1}^K$ are independent for each fixed $(\bbm u,\bbm v)$.
		Consequently, Bernstein’s inequality for sums of independent sub-exponential variables, e.g., Eq.~(2.18) in \cite{wainwright2019high}, yields that, for any $x>0$,
		\begin{equation}\label{eq:bern-sumZ-pooled-R}
			\mathbb P\left(\Big|\sum_{k=1}^K Z_k(\bbm u,\bbm v)\Big|\geq x\right) \leq 2\exp\Bigg(-c\min\Bigg\{\frac{x^2}{\sum_{k=1}^K \nu_k^2},\frac{x}{\max_{k\in[K]} b_k}\Bigg\}\Bigg).
		\end{equation}
		Noting that $\sum_{k=1}^K \nu_k^2=(C_{\epsilon,\mathcal A}^{\max})^2\sigma^4\sum_{k=1}^K T_k = (C_{\epsilon,\mathcal A}^{\max})^2\sigma^4T$ and $\max_k b_k=C_{\epsilon,\mathcal A}^{\max}\sigma^2\sqrt p$, we choose $x=c_0C_{\epsilon,\mathcal A}^{\max}\sigma^2\sqrt{pdT}$. Then, the two terms in the Bernstein minimum become
		\[
			\frac{x^2}{\sum_{k=1}^K \nu_k^2} = \frac{c_0^2 (C_{\epsilon,\mathcal A}^{\max})^2\sigma^4 (p d T)}{(C_{\epsilon,\mathcal A}^{\max})^2\sigma^4 T} = c_0^2pd,
			\ \text{ and } \ 
			\frac{x}{\max_k b_k} = \frac{c_0C_{\epsilon,\mathcal A}^{\max}\sigma^2\sqrt{p d T}}{C_{\epsilon,\mathcal A}^{\max}\sigma^2\sqrt p} = c_0\sqrt{dT}.
		\]
		Hence, given $T\geq c_1p^2d$, we have
		\[
			\mathbb P\left(\Big|\sum_{k=1}^K Z_k(\bbm u,\bbm v)\Big| \geq c_0C_{\epsilon,\mathcal A}^{\max}\sigma^2\sqrt{pdT}\right)\leq 2\exp(-cpd).
		\]
		Recall that $\bbm u^\top \bm R_{\rm pool}\bbm v = T^{-1}\sum_{k=1}^K Z_k(\bbm u,\bbm v)$, then for any $(\bbm u,\bbm v)\in\mathcal N_u\times\mathcal N_v$, we have
		\[
			\mathbb P\left(\Big|\bbm u^\top \bm R_{\rm pool}\bbm v\Big| \geq c_0C_{\epsilon,\mathcal A}^{\max}\sigma^2\sqrt{\frac{pd}{T}}\right)\leq 2\exp(-cpd).
		\]
		Consequently, taking a union bound over $(\bbm u,\bbm v)\in\mathcal N_u\times\mathcal N_v$ with $|\mathcal N_u||\mathcal N_v|\leq 9^{pd+d}$ yields
		\[
			\mathbb P\left(\max_{\bbm u\in\mathcal N_u, \bbm v\in\mathcal N_v}\Big|\bbm u^\top \bm R_{\rm pool}\bbm v\Big| \geq c_0C_{\epsilon,\mathcal A}^{\max}\sigma^2\sqrt{\frac{pd}{T}}\right) \leq 2|\mathcal N_u||\mathcal N_v|\exp(-cpd) \leq 2\cdot 9^{pd+d}\exp(-cpd) \leq \exp(-Cpd),
		\]
		Combining this with the $\varepsilon$-net reduction \eqref{eq:net-bound-pooled-R} and absorbing the factor $4$ into the leading constant concludes the proof.
	\end{proof}

	\begin{proof}[\textbf{Proof of Lemma \ref{lem:bernstein-tail-mgf-equiv}}]
		We begin by \textbf{the proof of \eqref{eq:bern-tail} $\Rightarrow$ \eqref{eq:bern-mgf}.} Assume \eqref{eq:bern-tail}. To prove \eqref{eq:bern-mgf}, we first bound moments using the two-regime tail. For any integer $m\geq 1$, the tail integral formula gives
		$\mathbb E|X|^m=m\int_0^\infty t^{m-1}\mathbb P(|X|\geq t)dt$. Let $t_0 = \nu^2/b$ so that $t_0^2/\nu^2=t_0/b$. Splitting at $t_0$ and using $\min\{t^2/\nu^2,t/b\}=t^2/\nu^2$ for $t\leq t_0$ and $=t/b$ for $t\geq t_0$, we have
		\[
			\mathbb E|X|^m \leq 2m\int_0^{t_0} t^{m-1} e^{-c t^2/\nu^2}dt + 2m\int_{t_0}^{\infty} t^{m-1} e^{-c t/b}dt \leq 2m\int_0^{\infty} t^{m-1} e^{-c t^2/\nu^2}dt + 2m\int_{0}^{\infty} t^{m-1} e^{-c t/b}dt.
		\]
		We now evaluate the two integrals in closed form up to Gamma factors. For the first term, set $t=\nu r$. Then,
		\[
			2m\int_0^{\infty} t^{m-1}e^{-c t^2/\nu^2}dt = 2m\nu^m\int_0^\infty r^{m-1}e^{-c r^2}dr = m(\nu/\sqrt{c})^m\int_0^\infty s^{m/2-1}e^{-s}ds = m(\nu/\sqrt{c})^m\Gamma\Big(\frac{m}{2}\Big).
		\]
		The penultimate equality follows from the change of variables $s=c r^2$, under which $r=\sqrt{s/c}$ and $dr=\frac12 c^{-1/2}s^{-1/2}ds$. Similarly, set $t=b r$ for the second term. Then,
		\[
			2m\int_0^{\infty} t^{m-1}e^{-c t/b}dt = 2mb^m\int_0^\infty r^{m-1}e^{-c r}dr = 2m(b/c)^{m}\int_0^\infty s^{m-1}e^{-s}ds = 2m(b/c)^{m}\Gamma(m).
		\]
		The penultimate equality follows from the change of variables $s = cr$, under which $dr=ds/c$. Putting the two parts together,
		\begin{align}\label{eq:moment-bound-Gamma}
			\mathbb E|X|^m\leq m(\nu/\sqrt{c})^{m}\Gamma(m/2) + 2m(b/c)^{m}\Gamma(m).
		\end{align}
			
		It remains to upper bound the Gamma factors in the right-hand side of \eqref{eq:moment-bound-Gamma}. Note that by Stirling's formula for the Gamma function (see equation~(1.1) in \cite{nemes2010more}), we have
		\[
			\Gamma(x)=\sqrt{2\pi}x^{x-\frac12}e^{-x}\bigl(1+o(1)\bigr), \ \text{ as } \ b x\to\infty,
		\]
		and thus there exist constants $c_1,c_2>0$ such that for all $x\geq 1$, $c_1x^{x-\frac12}e^{-x} \leq \Gamma(x)\leq c_2x^{x-\frac12}e^{-x}$. Dropping the factors $x^{-1/2}\leq 1$, $e^{-x}\leq 1$ and absorbing $c_2$ into $C^x$, we obtain that there exists
		a constant $C>0$ such that for all $x\geq 1$, $\Gamma(x)\leq C^xx^{x}$. Applying this bound with $x=m$ and $x=m/2$ gives the desired bounds for $\Gamma(m)$ and $\Gamma(m/2)$. Hence
		\[
			mc^{-m/2}\Gamma\Big(\frac{m}{2}\Big) \leq mc^{-m/2}C_0\Big(\frac{m}{2}\Big)^{m/2} \leq (C_1\sqrt m)^m,
		\]
		for $C_1>0$. Similarly,
		\[
			2mc^{-m}\Gamma(m)\leq 2mc^{-m}(Cm)^m\leq (C_2 m)^m
		\]
		for $C_2>0$. Combining these bounds yields that there exists a constant $C>0$ such that for all integers $m\geq2$,
		\[
			\mathbb E|X|^m \leq C^m\Big[(\nu\sqrt m)^m+(b m)^m\Big].
		\]
		Now fix $|\lambda|\leq c_0/b$ with $c_0>0$ sufficiently small. Since $\mathbb E X=0$, the Taylor expansion yields
		\[
		\mathbb E e^{\lambda X} = 1+\sum_{m=2}^\infty \frac{\lambda^m\mathbb E[X^m]}{m!} \leq 1+\sum_{m=2}^\infty \frac{|\lambda|^m\mathbb E|X|^m}{m!} \leq 1+S_{\nu}+S_{b},
		\]
		where $S_{\nu}=\sum_{m=2}^\infty |\lambda|^m (C\nu\sqrt m)^m/m!$ and $S_{b}=\sum_{m=2}^\infty |\lambda|^m (Cb m)^m/m!$.

		To control $S_{b}$,the Stirling bound $m!\geq (m/e)^m$ yields $|\lambda|^m (Cb m)^m/m!\leq (eC|\lambda|b)^m$. Writing $u=eC|\lambda|b$, under $|\lambda|\leq c_0/b$ and $c_0\leq (2eC)^{-1}$, we have $u\leq 1/2$ and thus $S_{b}\leq \sum_{m=2}^\infty u^m=u^2/(1-u)\leq 2u^2\leq 2(eC)^2\lambda^2b^2$. Consequently, $1+S_{b}\leq \exp(S_{b})\leq \exp\big(2(eC)^2\lambda^2b^2\big)$.

		Next, for $S_{\nu}$ we use $m!\geq (m/e)^m$ again to obtain $|\lambda|^m (C\nu\sqrt m)^m/m!\leq (eC|\lambda|\nu)^m m^{-m/2}$. Let $s=eC|\lambda|\nu$ and define $S=\sum_{m=2}^\infty s^m m^{-m/2}$, so that $S_{\nu}\leq S$. Decompose $S$ into even and odd terms:
		\[
			S=\sum_{\ell\geq 1} s^{2\ell}(2\ell)^{-\ell} + \sum_{\ell\geq 1} s^{2\ell+1}(2\ell+1)^{-(2\ell+1)/2}.
		\]
		For the even terms, the Stirling's formula implies $\ell^{-\ell}\leq e^\ell/\ell!$, we have $s^{2\ell}(2\ell)^{-\ell}\leq (es^2/2)^\ell/\ell!$, hence $\sum_{\ell\geq 1} s^{2\ell}(2\ell)^{-\ell}\leq \exp(es^2/2)-1$. For the odd terms, $(2\ell+1)^{-(2\ell+1)/2}\leq (2\ell)^{-\ell}$ gives $\sum_{\ell\geq 1} s^{2\ell+1}(2\ell+1)^{-(2\ell+1)/2}\leq s\big(\exp(es^2/2)-1\big)$. Therefore $S\leq (1+s)\big(\exp(es^2/2)-1\big)$.

		We now show that there exists a constant $c_1>0$ such that
		\begin{align}\label{eq:1ps-exp-bound}
			(1+s)\left(\exp\left(\frac{e}{2}s^2\right)-1\right)\leq \exp(c_1 s^2)-1
		\end{align}
		for all $s\geq 0$. Indeed, for $0\leq s\leq 1$, we have $1+s\leq 2$ and $\exp(es^2/2)-1\leq e\exp(e/2)s^2/2$, so the left side is at most $e\exp{(e/2)} s^2$, while $\exp(c_1 s^2)-1\geq c_1 s^2$. Choosing $c_1\geq e\exp{(e/2)}$ yields $(1+s)(\exp(es^2/2)-1)\leq \exp(c_1 s^2)-1$.

		For $s>1$, note that $1+s\leq 2s\leq 2\exp{(s^2)}$ and $\exp(es^2/2)-1\leq \exp(es^2/2)$, so $(1+s)\big(\exp(es^2/2)-1\big)\leq 2\exp((e/2+1)s^2)$. If we take $c_1\geq e/2+1+\log 4$, then $2\exp{((e/2+1)s^2)}\leq \exp{(c_1 s^2)}/2$ for all $s\geq 1$. Moreover, since $c_1 s^2\geq c_1\geq \log 2$, we have $1-\exp{(-c_1 s^2)}\geq 1-\exp{(-c_1)}\geq 1/2$, and thus $\exp{(c_1 s^2)}-1=\exp{(c_1 s^2)}(1-\exp{(-c_1 s^2)})\geq \exp{(c_1 s^2)}/2$. Combining the last two displays yields $(1+s)(\exp{(es^2/2)}-1)\leq \exp{(c_1 s^2)}-1$. Taking $c_1\geq \max\{e e^{e/2},e/2+1+\log 4\}$ makes \eqref{eq:1ps-exp-bound} hold for all $s\geq 0$. Consequently $1+S_{\nu}\leq 1+S\leq \exp(c_1 s^2)=\exp(c_1 (e C)^2 \lambda^2\nu^2)$.

		Putting the bounds together, we have, for all $|\lambda|\leq c_0/b$,
		\[
			\mathbb E e^{\lambda X} \leq 1+S_{\nu}+S_{b} \leq (1+S_{\nu})(1+S_{b}) \leq \exp\big(c_1 (e C)^2 \lambda^2\nu^2\big)\exp\big(2(e C)^2\lambda^2b^2\big).
		\]
		Since $|\lambda|\leq c_0/b$, the term $\exp\big(2(e C)^2\lambda^2b^2\big)$ is absorbed into $\exp(C\lambda^2\nu^2)$ by adjusting absolute constants. Hence there exists a constant $C>0$ such that for all $|\lambda|\leq c/b$,
		\[
			\mathbb E \exp(\lambda X)\leq \exp\!\big(C\lambda^2\nu^2\big),
		\]
		which establishes \eqref{eq:bern-mgf}.

		\noindent \textbf{The proof of \eqref{eq:bern-mgf} $\Rightarrow$ \eqref{eq:bern-tail}} is standard; see, e.g.,
		\citet[Section~2.1.3]{wainwright2019high}.
	\end{proof}

	\begin{proof}[\textbf{Proof of Lemma \ref{lem:gaussian-mech-prop}}]
		Denote
		\[
			M(\mathcal D_k) = \bbm g(\mathcal D_k) + \bm Z_k
			\ \ \text{and} \ \
			M(\mathcal D_k') = \bbm g(\mathcal D_k') + \bm Z_k',
		\]
		where $\bm Z_k'$ is an independent copy of $\bm Z_k$.
		Note that both $M(\mathcal D_k)$ and $M(\mathcal D_k')$ are Gaussian random matrices in the sense that their vectorized versions, $\operatorname{vec}(M(\mathcal D_k))$ and $\operatorname{vec}(M(\mathcal D_k'))$, are multivariate Gaussian. Specifically, $\operatorname{vec}(M(\mathcal D_k)) \sim \mathcal{N}(\operatorname{vec}(\bbm g(\mathcal D_k)),\, \sigma^2 \bm I_m)$ and $\operatorname{vec}(M(\mathcal D_k')) \sim \mathcal{N}(\operatorname{vec}(\bbm g(\mathcal D_k')),\, \sigma^2 \bm I_m)$, where $\bm I_m$ is the $m\times m$ identity matrix with $m = pd^2$. Consequently, we have the following closed-form expressions for the densities of $\operatorname{vec}(M(\mathcal D_k))$ and $\operatorname{vec}(M(\mathcal D_k'))$:
		\[
			f_k(\bm X) = (2\pi\sigma^2)^{-\frac{m}{2}} \exp\Bigl(-\frac{1}{2\sigma^2}\bigl\|\operatorname{vec}(\bm X)- \operatorname{vec}(\bbm g(\mathcal D_k))\bigr\|_2^2\Bigr) = (2\pi\sigma^2)^{-\frac{m}{2}} \exp\Bigl(-\frac{1}{2\sigma^2}\bigl\|\bm X- \bbm g(\mathcal D_k)\bigr\|_\F^2\Bigr),
		\]
		\[
			f_k'(\bm X) = (2\pi\sigma^2)^{-\frac{m}{2}}\exp\Bigl(-\frac{1}{2\sigma^2}\bigl\|\operatorname{vec}(\bm X)- \operatorname{vec}(\bbm g(\mathcal D_k'))\bigr\|_2^2\Bigr) = (2\pi\sigma^2)^{-\frac{m}{2}}\exp\Bigl(-\frac{1}{2\sigma^2}\bigl\|\bm X- \bbm g(\mathcal D_k')\bigr\|_\F^2\Bigr),
		\]
		Then, the log-likelihood ratio at $\bm X$ is
		\begin{align*}
			\log\frac{f_k(\bm X)}{f_k'(\bm X)} &= -\frac{1}{2\sigma^2}\Bigl(\bigl\|\bm X-\bbm g(\mathcal D_k)\bigr\|_\F^2 - \bigl\|\bm X-\bbm g(\mathcal D_k')\bigr\|_\F^2\Bigr) \\
			&= -\frac{1}{2\sigma^2}\Bigl(\bigl\|\bbm g(\mathcal D_k)-\bbm g(\mathcal D_k')\bigr\|_\F^2 + 2\bigl\langle \bm X-\bbm g(\mathcal D_k'),\bbm g(\mathcal D_k')-\bbm g(\mathcal D_k)\bigr\rangle\Bigr).
		\end{align*}
		By the Cauchy-Schwarz inequality and the condition that $\|\bbm g(\mathcal D_k)-\bbm g(\mathcal D_k')\|_\F \leq \Delta$, we have
		\[
			\left|\log\frac{f_k(\bm X)}{f_k'(\bm X)}\right| \leq \frac{\Delta}{\sigma^2}\bigl\|\bm X-\bbm g(\mathcal D_k')\bigr\|_\F + \frac{\Delta^2}{2\sigma^2}.
		\]
		Next, fix any $\varepsilon>0$ and define the good region $G_\varepsilon = \{\bm X\in\mathbb R^{d\times pd} : \log (f_k(\bm X)/f_k'(\bm X)) \leq \varepsilon \}$.
		For any measurable set $S\subset \mathbb R^{d\times pd}$, decomposing $S = (S\cap G_\varepsilon) \cup (S\cap G_\varepsilon^{\mathsf c})$ and integrating with respect to the densities yields
		\begin{align}\label{eq:ini-DP}
			\mathbb P(M(\mathcal D_k)\in S) &= \int_{S} f_k(\bm X)\mathrm d\bm X = \int_{S\cap G_\varepsilon} f_k(\bm X)\,\mathrm d\bm X + \int_{S\cap G_\varepsilon^{\mathsf c}} f_k(\bm X)\mathrm d\bm X \notag\\
			&\leq \int_{S\cap G_\varepsilon} \mathrm e^{\varepsilon} f_k'(\bm X)\mathrm d\bm X + \mathbb P(M(\mathcal D_k)\in G_\varepsilon^{\mathsf c})\\
			&\leq \mathrm e^{\varepsilon} \mathbb P(M(\mathcal D_k')\in S) + \mathbb P(M(\mathcal D_k)\in G_\varepsilon^{\mathsf c}), \notag
		\end{align}
		where we used that conditional on $G_\varepsilon$, we have $f_k(\bm X) \leq \mathrm e^{\varepsilon} f_k'(\bm X)$. Hence, the desired $(\varepsilon,\delta)$-DP inequality will follow, once we control the probability of the bad event $G_\varepsilon^{\mathsf c}$. To this end, define the privacy-loss random variable
		\begin{align*}
			L(\bm X) &= \log \frac{f_k(\bm X)}{f_k'(\bm X)}
			= \frac{1}{2\sigma^2}\Bigl(\|\bm X - \bbm g(\mathcal D_k')\|_\F^2-\|\bm X-\bbm g(\mathcal D_k)\|_\F^2\Bigr)\\
			&= \frac{1}{\sigma^2}\Bigl\langle \bm X - \bbm g(\mathcal D_k'),\; \bbm g(\mathcal D_k)-\bbm g(\mathcal D_k')\Bigr\rangle
			- \frac{1}{2\sigma^2}\|\bbm g(\mathcal D_k) - \bbm g(\mathcal D_k')\|_\F^2.
		\end{align*}
		Then, substituting $\bm X=M(\mathcal D_k')$ into the explicit expression for the log-likelihood ratio yields
		\[
			L(M(\mathcal D_k')) = \frac{\langle M(\mathcal D_k')-\bbm g(\mathcal D_k'),\bbm g(\mathcal D_k)-\bbm g(\mathcal D_k')\rangle}{\sigma^2} - \frac{\|\bbm g(\mathcal D_k)-\bbm g(\mathcal D_k')\|_\F^2}{2\sigma^2}.
		\]
		Note that $\langle M(\mathcal D_k')-\bbm g(\mathcal D_k'),\bbm g(\mathcal D_k)-\bbm g(\mathcal D_k')\rangle$ is Gaussian with mean $0$ and variance $\sigma^2\|\bbm g(\mathcal D_k)-\bbm g(\mathcal D_k')\|_\F^2$, while the random variable $L(M(\mathcal D_k'))$ is Gaussian with mean $-\|\bbm g(\mathcal D_k)-\bbm g(\mathcal D_k')\|_\F^2/(2\sigma^2)$ and variance $\|\bbm g(\mathcal D_k)-\bbm g(\mathcal D_k')\|_\F^2/\sigma^2$. Let $\Delta_k = \|\bbm g(\mathcal D_k)-\bbm g(\mathcal D_k')\|_\F$. For the Gaussian mechanism, $L(M(\mathcal D_k)) \sim \mathcal N(-\Delta_k^2/2\sigma^2,\Delta_k^2/\sigma^2)$.
		Moreover, the function $\Delta_k\mapsto \varepsilon\sigma/\Delta_k+\Delta_k/2\sigma$ attains its minimum over $(0,\Delta]$ at $\Delta_k=\Delta$ whenever $\Delta\leq \sqrt{2\varepsilon}\sigma$. In particular, this condition is implied by \eqref{eq:gaussian-sigma-choice}. Therefore, the worst case over neighboring pairs is attained at $\Delta_k=\Delta$, and
		\[
			\sup_{\mathcal D_k\sim\mathcal D_k'}\mathbb P\!\left(L(M(\mathcal D_k))>\varepsilon\right) = \mathbb P\left( X > \frac{\varepsilon\sigma}{\Delta} + \frac{\Delta}{2\sigma} \right), \ \text{ with } \ X\sim\mathcal N(0,1).
		\]
		Applying the Gaussian tail bound $\mathbb P(X>t)\leq \frac12 e^{-t^2/2}$ for $t\geq0$ yields
		\[
			\sup_{\mathcal D_k\sim\mathcal D_k'}\mathbb P\!\left(L(M(\mathcal D_k))>\varepsilon\right) \leq \frac12 \exp\left(-\frac12\Bigl(\frac{\varepsilon\sigma}{\Delta} + \frac{\Delta}{2\sigma}\Bigr)^2\right).
		\]
		Consequently, to ensure $\mathbb P(L(M(\mathcal D_k))>\varepsilon)\leq \delta$, it suffices to bound the right-hand side by $\delta$, which is equivalent to $\varepsilon\sigma/\Delta+\Delta/2\sigma \geq \sqrt{2\log (1/2\delta)}$. A convenient sufficient condition is
		\begin{align}\label{eq:suff-cond}
			\frac{\varepsilon\sigma}{\Delta} \geq \sqrt{2\log \left(\frac{1.25}{\delta}\right)}.
		\end{align}
		Rearranging \eqref{eq:suff-cond} gives the exact requirement of $\sigma$ in \eqref{eq:gaussian-sigma-choice}. Therefore, under \eqref{eq:gaussian-sigma-choice},
		\begin{align}\label{eq:delta-DP}
			\mathbb P\bigl(L(M(\mathcal D_k))>\varepsilon\bigr) = \mathbb P\bigl(M(\mathcal D_k)\in G_\varepsilon^{\mathsf c}\bigr) \leq \delta.
		\end{align}
		Substituting this bound \eqref{eq:delta-DP} into \eqref{eq:ini-DP} yields, for any measurable $S\subseteq\mathbb R^{d\times pd}$, $\mathbb P(M(\mathcal D_k)\in S) \leq \mathrm e^\varepsilon \mathbb P\bigl(M(\mathcal D_k')\in S\bigr) + \delta$. Therefore, if \eqref{eq:gaussian-sigma-choice} holds, then
		\begin{align*}
			\mathbb P\bigl(M(\mathcal D_k)\in S\bigr) \leq \mathrm e^\varepsilon \mathbb P\bigl(M(\mathcal D_k')\in S\bigr) + \delta, \ \text{for any measurable } \ S\subset\mathbb R^{d\times pd}.
		\end{align*}
		This proves that the mechanism $M$ is $(\varepsilon,\delta)$-DP.
	\end{proof}

	\section{Technical Lemmas for Stationary VAR Processes}\label{append:technical lemmas}

	This section collects several technical lemmas for stationary VAR($p$) processes that are repeatedly used in the concentration analysis of the main results. Lemma~\ref{lemma:ar-polynomial} links the spectral behavior of $\bm\Psi(z)$ to that of $\mathcal A(z)$. Lemma~\ref{lemma:spectral-bounds} then transfers the frequency-domain bounds in Lemma~\ref{lemma:ar-polynomial} to the Toeplitz matrix generated by the VAR coefficients. Lemmas~\ref{lemma:subg-vector} and~\ref{lemma:stacked-subg-vector} establish sub-Gaussianity for the innovation vector and the stacked lag vector, respectively, which provides the probabilistic foundation for subsequent tail bounds. Lemmas~\ref{lemma:equivalence-representation-of-quadratic-form} and~\ref{lemma:equivalence-representation-of-uXEv} derive quadratic and bilinear innovation representations for $\bbm u^\top \bm X^\top \bm X \bbm u$ and $\bbm u^\top \bm X^\top \bm E\bbm v$. These representations allow us to invoke the infinite-dimensional Hanson-Wright inequality in Lemma~\ref{lemma:infinite-HW}, which in turn yields the high-probability quadratic-form bound in Lemma~\ref{lemma:HW-uXTXu} and the corresponding cross-term concentration result in Lemma~\ref{lemma:HW-uXEv}. Finally, Lemma~\ref{lem:Sigma-x-pd} establishes the positive definiteness and eigenvalue bounds of the stacked lag covariance matrix $\bbm\Sigma_x$, which is used to verify the restricted strong convexity-type conditions needed in the estimation analysis.

	\begin{lemma}\label{lemma:ar-polynomial}
		Let the VAR($p$) process be stationary with AR polynomial $\mathcal{A}(z) = \mathbf{I}_p - \sum_{k=1}^p \mathbf{A}_k z^k$ and its determinant $\det \mathcal{A}(z) \neq 0$ for all $|z| \leq 1$.
		Define
		\begin{align*}
			\bm{\Psi}(z) = \mathcal{A}(z)^{-1} = \sum_{j \geq 0} \bm{\Psi}_j z^j, \,
			\mu_{\min}(\mathcal{A}) = \min_{|z|=1} \lambda_{\min}\!\big(\mathcal{A}^\dagger(z) \mathcal{A}(z)\big), \,
			\mu_{\max}(\mathcal{A}) = \max_{|z|=1} \lambda_{\max}\!\big(\mathcal{A}^\dagger(z) \mathcal{A}(z)\big).
		\end{align*}
		Then, we have
		\[
			\operatorname*{sup}_{|z|=1} \lambda_{\max}\!\big(\bm{\Psi}^\dagger(z) \bm{\Psi}(z)\big)
			= \frac{1}{\mu_{\min}(\mathcal{A})}, \quad
			\operatorname*{inf}_{|z|=1} \lambda_{\min}\!\big(\bm{\Psi}^\dagger(z) \bm{\Psi}(z)\big)
			= \frac{1}{\mu_{\max}(\mathcal{A})}.
		\]
	\end{lemma}
	\begin{proof}[\textbf{Proof of Lemma \ref{lemma:ar-polynomial}}]
		For $|z|=1$, $\mathcal{A}(z)$ is invertible by stationarity. Using the identity $(\mathbf{B}^{-1})^\dagger=(\mathbf{B}^\dagger)^{-1}$ for any invertible matrix $\mathbf{B}$, we have
		\[
			\bm{\Psi}^\dagger(z)\bm{\Psi}(z)
			= \big(\mathcal{A}(z)^{-1}\big)^\dagger \mathcal{A}(z)^{-1}
			= \big(\mathcal{A}^\dagger(z)\big)^{-1}\mathcal{A}(z)^{-1}
			= \big(\mathcal{A}(z)\mathcal{A}^\dagger(z)\big)^{-1}.
		\]
		Note that $\big(\mathcal{A}(z)\mathcal{A}^\dagger(z)\big)^{-1}$ and $\big(\mathcal{A}^\dagger(z)\mathcal{A}(z)\big)^{-1}$ have identical eigenvalues, i.e., 
		\[
			\lambda_{\max}\big((\mathcal{A}(z)\mathcal{A}^\dagger(z))^{-1}\big)
			= \lambda_{\max}\big((\mathcal{A}^\dagger(z)\mathcal{A}(z))^{-1}\big), \quad
			\lambda_{\min}\big((\mathcal{A}(z)\mathcal{A}^\dagger(z))^{-1}\big)
			= \lambda_{\min}\big((\mathcal{A}^\dagger(z)\mathcal{A}(z))^{-1}\big).
		\]
		Therefore, for any fixed $|z|=1$, we have
		\begin{align*}
			\lambda_{\max}\big(\bm{\Psi}^\dagger(z)\bm{\Psi}(z)\big)
			= \lambda_{\max}\big((\mathcal{A}(z)\mathcal{A}^\dagger(z))^{-1}\big)
			= \lambda_{\max}\big((\mathcal{A}^\dagger(z)\mathcal{A}(z))^{-1}\big)
			= \frac{1}{\lambda_{\min}\big(\mathcal{A}^\dagger(z)\mathcal{A}(z)\big)}
		\end{align*}
		and
		\begin{align*}
			\lambda_{\min}\big(\bm{\Psi}^\dagger(z)\bm{\Psi}(z)\big)
			= \lambda_{\min}\big((\mathcal{A}(z)\mathcal{A}^\dagger(z))^{-1}\big)
			= \lambda_{\min}\big((\mathcal{A}^\dagger(z)\mathcal{A}(z))^{-1}\big)
			= \frac{1}{\lambda_{\max}\big(\mathcal{A}^\dagger(z)\mathcal{A}(z)\big)}.
		\end{align*}
		The last equalities in the above two hold owing to $\lambda_{\max}(\mathbf{B}^{-1}) = 1/\lambda_{\min}(\mathbf{B})$ and $\lambda_{\min}(\mathbf{B}^{-1}) = 1/\lambda_{\max}(\mathbf{B})$ for any invertible matrix $\mathbf{B}$. Taking $\sup_{|z|=1}$ and $\inf_{|z|=1}$ on both sides yields the stated identities with $\mu_{\min}(\mathcal{A})$ and $\mu_{\max}(\mathcal{A})$.
	\end{proof}

	\begin{lemma}\label{lemma:spectral-bounds}
		Let $\bm{\Psi}(z) = \sum_{j \geq 0} \bm{\Psi}_j z^j$ and define the semi-infinite Toeplitz matrix $\bm{P}_T$ as
		\begin{align*}
			\bm{P}_T=
			\begin{bmatrix}
				\bm{\Psi}_\mathbf{0} & \bm{\Psi}_1 & \bm{\Psi}_2 & \bm{\Psi}_3 & \ldots & \bm{\Psi}_{T-1}&\ldots\\
				\bm{O} & \bm{\Psi}_\mathbf{0} & \bm{\Psi}_1 & \bm{\Psi}_2 & \ldots & \bm{\Psi}_{T-2}&\ldots\\
				\bm{O} & \bm{O} & \bm{\Psi}_0 & \bm{\Psi}_1 & \ldots & \bm{\Psi}_{T-3}&\ldots\\
				\vdots&\vdots&\vdots&\vdots&\ddots&\vdots&\ldots\\
				\bm{O}&\bm{O}&\bm{O}&\bm{O}&\ldots&\bm{\Psi}_0&\ldots
			\end{bmatrix}\in \mathbb{R}^{dT \times \infty}.
		\end{align*}
		Then, for any $T \geq 1$, the following two side spectral bounds hold:
		\begin{align*}
			\inf_{|z|=1}\lambda_{\min}\big(\bm{\Psi}^\dagger(z)\bm{\Psi}(z)\big) \leq \lambda_{\min}(\bm{P}_T \bm{P}_T^\top) \leq \lambda_{\max}(\bm{P}_T \bm{P}_T^\top) \leq \sup_{|z|=1}\lambda_{\max}\big(\bm{\Psi}^\dagger(z)\bm{\Psi}(z)\big).
		\end{align*}
	\end{lemma}

	\begin{proof}[\textbf{Proof of Lemma \ref{lemma:spectral-bounds}}]
		Take any $\bbm{u} = (\bbm{u}_0^\top, \bbm{u}_1^\top, \ldots, \bbm{u}_{T-1}^\top) \in \mathbb{R}^{dT}$ and define 
		\[
			\bm{U}_T(\theta) = \sum_{t=0}^{T-1} \bbm{u}_t e^{-it\theta} \in \mathbb{C}^d.
		\]
		Set $\bbm{c} = \bm{P}_T^\top \bbm{u} \in \ell_+^2$, i.e., $\bbm{c} = (\bbm{c}_0^\top, \bbm{c}_1^\top, \ldots)^\top$. Then, by the Toeplitz lower triangular structure of $\bm{P}_T$, for any $m \in \mathbb N$, we have $\bbm{c}_m = \sum_{t=0}^{T-1} \bm{\Psi}_{m-t} \bbm{u}_t$, where $\bm{\Psi}_j = \bm{O}$ for $j<0$ is assumed here for notational convenience. Set $\bm{C}(\theta)=\sum_{m\geq0}\bbm c_me^{-im\theta}$. Since the inner summation for $t$ is finite, it is legal to interchange the order of summations, yielding
		\begin{align*}
			\bm C(\theta) &=\sum_{m\geq 0}\bbm c_m e^{-im\theta} = \sum_{m\geq 0}\sum_{t=1}^{T-1}\bbm\Psi_{m-k}^\top \bbm u_k e^{-im\theta}= \sum_{t=1}^{T-1}\sum_{r\geq 0}\bbm\Psi_{j}^\top \bbm u_k e^{-i(r+t)\theta}\\
			&= \left(\sum_{r\geq 0}\bbm\Psi_r^\top e^{-ir\theta}\right)\left(\sum_{t=1}^{T-1}\bbm u_t e^{-it\theta}\right) =\bbm\Psi^\top(e^{-i\theta})\bm U_T(\theta) \in \mathbb{C}^d.
		\end{align*}
		Using the Parseval's identity for the elements of each $\bbm{c}_m$, we have
		\begin{align*}
			\|\bm{P}_T^\top \bbm{u}\|_2^2 = \|\bbm{c}\|_2^2 = \sum_{m\geq0} \|\bbm{c}_m\|_2^2 = \frac{1}{2\pi} \int_{-\pi}^{\pi} \|\bm{C}(\theta)\|_2^2 d\theta = \frac{1}{2\pi} \int_{-\pi}^{\pi} \bm{U}_T^\dagger(\theta) (\bm{\Psi}^\top(e^{-i\theta}))^\dagger \bm{\Psi}^\top(e^{-i\theta}) \bm{U}_T(\theta) d\theta.
		\end{align*}
		For the symmetric matrix $(\bm{\Psi}^\top(e^{-i\theta}))^\dagger \bm{\Psi}^\top(e^{-i\theta})$, applying the Rayleigh-Ritz theorem yields
		\begin{align*}
			\lambda_{\min}\big((\bm{\Psi}^\top(e^{-i\theta}))^\dagger \bm{\Psi}^\top(e^{-i\theta})\big) \|\bm{U}_T(\theta)\|_2^2
			&\leq \bm{U}_T^\dagger(\theta) (\bm{\Psi}^\top(e^{-i\theta}))^\dagger \bm{\Psi}^\top(e^{-i\theta}) \bm{U}_T(\theta)\\
			&\leq \lambda_{\max}\big((\bm{\Psi}^\top(e^{-i\theta}))^\dagger \bm{\Psi}^\top(e^{-i\theta})\big) \|\bm{U}_T(\theta)\|_2^2.
		\end{align*}
		Note that the eigenvalues of $(\bm{\Psi}^\top(e^{-i\theta}))^\dagger \bm{\Psi}^\top(e^{-i\theta})$ are identical to those of $\bm{\Psi}^\dagger(e^{-i\theta}) \bm{\Psi}(e^{-i\theta})$. Thus, integrating the above inequality with respect to $\theta$ and taking the supremum and infimum over $|z|=1$ of $\lambda_{\max}\big(\bm{\Psi}^\dagger(z)\bm{\Psi}(z)\big)$ and $\lambda_{\min}\big(\bm{\Psi}^\dagger(z)\bm{\Psi}(z)\big)$ respectively yields
		\begin{align*}
			\inf_{|z|=1}\lambda_{\min}\big(\bm{\Psi}^\dagger(z)\bm{\Psi}(z)\big) \frac{1}{2\pi} \int_{-\pi}^{\pi} \|\bm{U}_T(\theta)\|_2^2 d\theta &\leq \frac{1}{2\pi} \int_{-\pi}^{\pi} \bm{U}_T^\dagger(\theta) (\bm{\Psi}^\top(e^{-i\theta}))^\dagger \bm{\Psi}^\top(e^{-i\theta}) \bm{U}_T(\theta) d\theta \\
			&\leq \sup_{|z|=1}\lambda_{\max}\big(\bm{\Psi}^\dagger(z)\bm{\Psi}(z)\big) \frac{1}{2\pi} \int_{-\pi}^{\pi} \|\bm{U}_T(\theta)\|_2^2 d\theta.
		\end{align*}
		For the integral term in the above inequality, by the Parseval's identity again, we have
		\begin{align*}
			\frac{1}{2\pi} \int_{-\pi}^{\pi} \|\bm{U}_T(\theta)\|_2^2 d\theta = \sum_{t=0}^{T-1} \|\bbm{u}_t\|_2^2 = \|\bbm{u}\|_2^2.
		\end{align*}
		Combining the above results, we have
		\begin{align*}
			\inf_{|z|=1}\lambda_{\min}\big(\bm{\Psi}^\dagger(z)\bm{\Psi}(z)\big) \|\bbm{u}\|_2^2 &\leq \bbm{u}^\top \bm{P}_T \bm{P}_T^\top \bbm{u} \leq \sup_{|z|=1}\lambda_{\max}\big(\bm{\Psi}^\dagger(z)\bm{\Psi}(z)\big) \|\bbm{u}\|_2^2.
		\end{align*}
		Since the above inequality holds for every $\bbm u\in\mathbb R^{pT}$, it yields a uniform lower and upper bound for the quadratic form associated with $\bm P_T\bm P_T^\top$. By the variational characterization of eigenvalues, the desired spectral bounds follow immediately.
	\end{proof}

	\begin{lemma}\label{lemma:subg-vector}
		Suppose $\{\xi_{it}\}_{1\leq i \leq d}$ are mutually independent and $\mathbb{E}e^{\mu\xi_{it}} \leq e^{\mu^2 \sigma^2/2}$ for all $\mu \in \mathbb{R}$. Then, $\bbm{\xi}_t$ is a sub-Gaussian random vector with 
		\[
			K_\xi = \sup_{\boldsymbol{u}\neq\mathbf{0}}\frac{\|\langle\bbm{u},\bbm{\xi}_t\rangle\|_{\psi_2}}{\|\bbm{u}\|_2} \leq C \sigma
		\]
		for some constant $C>0$. In addition, for $\bbm{\epsilon}_t=\bbm{\Sigma}_\epsilon^{1/2}\bbm{\xi}_t$, 
		\[
			K_\epsilon =\sup_{\boldsymbol{u}\neq\mathbf{0}}\frac{\|\langle\bbm{u},\bbm{\epsilon}_t\rangle\|_{\psi_2}}{\|\bbm{u}\|_2} \leq C \sigma \sqrt{\lambda_{\max}(\bbm{\Sigma}_\epsilon)}.
		\]
	\end{lemma}

	\begin{proof}[\textbf{Proof of Lemma \ref{lemma:subg-vector}}]
		Fix any $\bbm{u} \in \mathbb{R}^d$ and denote the scalar projection $S=\langle\bbm{u}, \bbm{\xi}_t\rangle = \sum_{i=1}^d u_i \xi_{it}$. Then, by mutual independence of $\{\xi_{it}\}_{1\leq i \leq d}$, for any $\mu \in \mathbb{R}$, we have
		\begin{align*}
			\mathbb{E}e^{\mu S} = \mathbb{E}e^{\mu \sum_{i=1}^d u_i \xi_{it}} = \prod_{i=1}^d \mathbb{E}e^{\mu u_i \xi_{it}} \leq \prod_{i=1}^d e^{\mu^2 u_i^2 \sigma^2/2} = e^{\frac{\mu^2 \sigma^2}{2} \|\boldsymbol{u}\|_2^2}.
		\end{align*}
		Hence $S$ is sub-Gaussian with variance proxy $\sigma^2 \|\bbm{u}\|_2^2$. By the equivalence between the variance proxy of a sub-Gaussian random variable and the Orlicz $\psi_2$ norm, e.g., \cite{vershynin2018high}, Section 2.5.2, there exists a constant $C>0$ such that $\|\langle\bbm{u}, \bbm{\xi}_t\rangle\|_{\psi_2} \leq C \sigma \|\bbm{u}\|_2$. Taking the supremum over $\bbm{u} \neq \mathbf{0}$ yields the vector sub-Gaussian constant 
		\[
			K_\xi = \sup_{\boldsymbol{u}\neq\mathbf{0}}\frac{\|\langle\bbm{u},\bbm{\xi}_t\rangle\|_{\psi_2}}{\|\bbm{u}\|_2} \leq C \sigma.
		\] 
		Finally, note that for $\bbm{\epsilon}_t=\bbm{\Sigma}_\epsilon^{1/2}\bbm{\xi}_t$ and any $\bbm{u}$,
		\begin{align*}
			\|\langle \bbm{u},\bbm{\epsilon}_t\rangle\|_{\psi_2}=\|\langle \bbm{\Sigma}_\epsilon^{1/2}\bbm{u},\bbm{\xi}_t\rangle\|_{\psi_2} \leq C\sigma\|\bbm{\Sigma}_\epsilon^{1/2}\bbm{u}\|_2 \leq C\sigma\sqrt{\lambda_{\max}(\bbm{\Sigma}_\epsilon)}\|\bbm{u}\|_2,
		\end{align*}
		Consequently, we have
		\begin{equation*}
			K_\epsilon =\sup_{\bbm{u}\neq\mathbf{0}}\frac{\|\langle\bbm{u},\bbm{\epsilon}_t\rangle\|_{\psi_2}}{\|\bbm{u}\|_2} \leq C \sigma \sqrt{\lambda_{\max}(\bbm{\Sigma}_\epsilon)}. \qedhere
		\end{equation*}
	\end{proof}

	\begin{lemma}\label{lemma:stacked-subg-vector}
		For the stationary VAR($p$) process with its vector moving average representation $\bbm{y}_t = \sum_{j=0}^\infty \bbm{\Psi}_j \bbm{\epsilon}_{t-j}$, suppose that $\{\bbm{\epsilon}_t\}$ are $i.i.d.$ sub-Gaussian random vectors with the vector sub-Gaussian constant $K_\epsilon$, i.e., $K_\epsilon =\sup_{\boldsymbol{u}\neq\mathbf{0}}\|\langle\bbm{u},\bbm{\epsilon}_t\rangle\|_{\psi_2}/\|\bbm{u}\|_2$. Then, the stacked vector $\bbm x_t = (\bbm{y}_{t-1}^\top, \bbm{y}_{t-2}^\top, \ldots, \bbm{y}_{t-p}^\top)^\top$ is also sub-Gaussian with
		\begin{align*}
			\|\langle\bbm{u},\bbm x_t\rangle\|_{\psi_2} \leq \frac{CK_\epsilon}{\sqrt{\mu_{\min}(\mathcal{A})}} \|\bbm{u}\|_2
		\end{align*}
		for any $\bbm{u} \in \mathbb{R}^{pd}$, where $\mu_{\min}(\mathcal{A})$ is defined in Lemma \ref{lemma:ar-polynomial}. Moreover, if $\bbm{\epsilon}_t=\bbm{\Sigma}_\epsilon^{1/2}\bbm{\xi}_t$ with $\bbm{\xi}_t$ defined in Lemma \ref{lemma:subg-vector}, then
		\begin{align*}
			\|\langle\bbm{u},\bbm x_t\rangle\|_{\psi_2} \leq \frac{C \sigma \sqrt{\lambda_{\max}(\bbm{\Sigma}_\epsilon)}}{\sqrt{\mu_{\min}(\mathcal{A})}} \|\bbm{u}\|_2.
		\end{align*}
	\end{lemma}

	\begin{proof}[\textbf{Proof of Lemma \ref{lemma:stacked-subg-vector}}]
		For any $\bbm{u} = (\bbm{u}_1^\top, \bbm{u}_2^\top, \ldots, \bbm{u}_{p}^\top)^\top \in \mathbb{R}^{pd}$ with each $\bbm{u}_k \in \mathbb{R}^d$, we have
		\begin{align*}
			\langle\bbm{u},\bbm x_t\rangle &= \sum_{k=1}^{p} \langle\bbm{u}_k, \bbm{y}_{t-k}\rangle = \sum_{k=1}^{p}\sum_{j\geq0}\langle\bbm{u}_k, \bbm{\Psi}_j \bbm{\epsilon}_{t-k-j}\rangle = \sum_{k=1}^{p}\sum_{j\geq0}\langle\bbm{\Psi}_j^\top\bbm{u}_k, \bbm{\epsilon}_{t-k-j}\rangle\\
			&= \sum_{k=1}^{p} \sum_{m\geq k} \langle\bbm{\Psi}_{m-k}^\top\bbm{u}_k, \bbm{\epsilon}_{t-m}\rangle = \sum_{k=1}^{p} \sum_{m\geq 0} \langle\bbm{\Psi}_{m-k}^\top\bbm{u}_k, \bbm{\epsilon}_{t-m}\rangle.
		\end{align*}
		The second to last equality holds by changing the index $m=k+j$ and the last equality holds since $\bbm{\Psi}_{r}=\mathbf{0}$ for $r < 0$. For each fixed $m \in \mathbb N$, define $\bbm{c}_m = \sum_{k=1}^{p} \bbm{\Psi}_{m-k}^\top \bbm{u}_k$.
		Then, we have $\langle\bbm{u},\bbm x_t\rangle = \sum_{m\geq0} \langle\bbm{c}_m, \bbm{\epsilon}_{t-m}\rangle$ with $\langle\bbm{c}_m, \bbm{\epsilon}_{t-m}\rangle$ being mutually independent across $m$ and $\|\langle\bbm{c}_m, \bbm{\epsilon}_{t-m}\rangle\|_{\psi_2} \leq K_\epsilon \|\bbm{c}_m\|_2$. By the union bound for sums of independent sub-Gaussian random variables, e.g., \cite{vershynin2018high}, Proposition 2.6.1, we have
		\begin{align}\label{eq:Yt-psi-bound}
			\|\langle\bbm{u},\bbm x_t\rangle\|_{\psi_2} \leq C K_\epsilon \left(\sum_{m\geq0} \|\bbm{c}_m\|_2^2\right)^{1/2}.
		\end{align}
		Note that Proposition 2.6.1 in \cite{vershynin2018high} concerns finite sums, but the same result holds for infinite sums as long as the right-hand side is finite which is verified as follows. First, define the polynomials
		\begin{align*}
			\bm{C}(\theta) = \sum_{m\geq0} \bbm{c}_m e^{-im\theta}\in \mathbb{C}^d,
			\ \ \text{and} \ \
			\bm{U}_p(\theta) = \sum_{k=1}^{p}\bbm{u}_k e^{-ik\theta} \in \mathbb{C}^d.
		\end{align*}
		Recalling that $\bbm{\Psi}_r=\bm{O}$ for $r<0$, we have
		\begin{align*}
			\bm C(\theta)&=\sum_{m\geq 0}\bbm c_m e^{-im\theta} = \sum_{m\geq 0}\sum_{k=1}^{p}\bbm\Psi_{m-k}^\top \bbm u_k e^{-im\theta}=\sum_{k=1}^{p}\sum_{j\geq 0}\bbm\Psi_{j}^\top \bbm u_k e^{-i(j+k)\theta}\\
			&=\left(\sum_{j\geq 0}\bbm\Psi_j^\top e^{-ij\theta}\right)\left(\sum_{k=1}^{p}\bbm u_k e^{-ik\theta}\right) =\bbm\Psi^\top(e^{-i\theta})\bm U_p(\theta) \in \mathbb{C}^d.
		\end{align*}
		Using the Parseval's identity for the elements of each $\bbm{c}_m$, we have
		\begin{align*}
			\sum_{m\geq0} \|\bbm{c}_m\|_2^2 = \frac{1}{2\pi} \int_{-\pi}^{\pi} \|\bm{C}(\theta)\|_2^2 d\theta = \frac{1}{2\pi} \int_{-\pi}^{\pi} \bm{U}_p^\dagger(\theta) (\bbm{\Psi}^\top(e^{-i\theta}))^\dagger \bbm{\Psi}^\top(e^{-i\theta}) \bm{U}_p(\theta) d\theta.
		\end{align*}
		For the symmetric matrix $(\bbm{\Psi}^\top(e^{-i\theta}))^\dagger \bbm{\Psi}^\top(e^{-i\theta})$, applying the Rayleigh-Ritz theorem yields
		\begin{align*}
			\bm{U}_p^\dagger(\theta) (\bbm{\Psi}^\top(e^{-i\theta}))^\dagger \bbm{\Psi}^\top(e^{-i\theta}) \bm{U}_p(\theta) \leq \lambda_{\max}\big((\bbm{\Psi}^\top(e^{-i\theta}))^\dagger \bbm{\Psi}^\top(e^{-i\theta})\big) \|\bm{U}_p(\theta)\|_2^2.
		\end{align*}
		Note that the eigenvalues of $(\bbm{\Psi}^\top(e^{-i\theta}))^\dagger \bbm{\Psi}^\top(e^{-i\theta})$ are identical to those of $\bbm{\Psi}^\dagger(e^{-i\theta}) \bbm{\Psi}(e^{-i\theta})$.
		Thus, integrating the above inequality with respect to $\theta$ and taking the supremum over $|z|=1$ of $\lambda_{\max}\big(\bbm{\Psi}^\dagger(z)\bbm{\Psi}(z)\big)$ yields
		\begin{align*}
			\frac{1}{2\pi} \int_{-\pi}^{\pi} \bm{U}_p^\dagger(\theta) (\bbm{\Psi}^\top(e^{-i\theta}))^\dagger \bbm{\Psi}^\top(e^{-i\theta}) \bm{U}_p(\theta) d\theta \leq \sup_{|z|=1}\lambda_{\max}\big(\bbm{\Psi}^\dagger(z)\bbm{\Psi}(z)\big)\,\frac{1}{2\pi} \int_{-\pi}^{\pi} \|\bm{U}_p(\theta)\|_2^2 d\theta.
		\end{align*}
		For the integral term in the above inequality, by the Parseval's identity again, we have
		\begin{align*}
			\frac{1}{2\pi} \int_{-\pi}^{\pi} \|\bm{U}_p(\theta)\|_2^2 d\theta = \sum_{k=1}^{p} \|\bbm{u}_k\|_2^2 = \|\bbm{u}\|_2^2.
		\end{align*}
		Combining the above results and noting that $\sup_{|z|=1}\lambda_{\max}\big(\bbm{\Psi}^\dagger(z)\bbm{\Psi}(z)\big) = 1/\mu_{\min}(\mathcal{A})$ from Lemma \ref{lemma:ar-polynomial}, we have
		\begin{align}\label{eq:upper-bound-for-cm}
			\left(\sum_{m\geq0} \|\bbm{c}_m\|_2^2\right)^{1/2} \leq \frac{1}{\sqrt{\mu_{\min}(\mathcal{A})}} \|\bbm{u}\|_2.
		\end{align}
		Thus, substituting the above inequality into the earlier bound \eqref{eq:Yt-psi-bound} for $\|\langle\bbm{u},\bbm x_t\rangle\|_{\psi_2}$, we have
		\begin{align*}
			\|\langle\bbm{u},\bbm x_t\rangle\|_{\psi_2} \leq \frac{C K_\epsilon}{\sqrt{\mu_{\min}(\mathcal{A})}} \|\bbm{u}\|_2.
		\end{align*}
		The last claim follows by substituting the upper bound for $K_\epsilon$ in Lemma \ref{lemma:subg-vector}.
	\end{proof}

	\begin{lemma}\label{lemma:equivalence-representation-of-quadratic-form}
		Let $\{\bbm{y}_t\}_{t\in\mathbb Z}$ be a stationary VAR($p$) process with VMA($\infty$) representation $\bbm{y}_t = \sum_{j\geq0} \bbm{\Psi}_j \bbm{\epsilon}_{t-j}$, where $\{\bbm{\Psi}_j\}_{j\geq0}$ is absolutely summable. Suppose $\bbm{\epsilon}_t = \bbm{\Sigma}_\epsilon^{1/2}\bbm{\xi}_t$, where $\bbm{\xi}_t=(\xi_{1t},\ldots,\xi_{dt})^\top$ has independent coordinates, is $i.i.d.$ across $t$, and satisfies
		$\mathbb{E}\exp(\mu\xi_{it}) \leq \exp(\mu^2\sigma^2/2)$ for all $\mu\in\mathbb{R}$ and all $i$.
		Fix an integer $p\geq1$ and define the stacked lag vector $\bbm{x}_t = (\bbm{y}_{t-1}^\top,\ldots,\bbm{y}_{t-p}^\top)^\top\in\mathbb{R}^{pd}$. Let $\bm{X}\in\mathbb{R}^{T\times pd}$ be the data matrix with $t$-th row $\bbm{x}_{T-t+1}^\top$. For any fixed $\bbm{u}=(\bbm{u}_1^\top,\ldots,\bbm{u}_p^\top)^\top\in\mathbb{R}^{pd}$, set $\bm{Y}=\bm{X}\bbm{u}\in\mathbb{R}^T$.
		Then there exist a semi-infinite innovation vector $\bm{Z}_T = (\bbm{\xi}_T^\top,\bbm{\xi}_{T-1}^\top,\bbm{\xi}_{T-2}^\top,\ldots,\bbm{\xi}_1^\top,\bbm{\xi}_0^\top,\bbm{\xi}_{-1}^\top,\ldots)^\top$ and a deterministic positive semidefinite matrix $\bm{Q}= (\bm{I}_{\infty}\otimes\bbm{\Sigma}_\epsilon^{1/2})^\top \bm{H}_u^\top \bm{H}_u (\bm{I}_{\infty}\otimes\bbm{\Sigma}_\epsilon^{1/2})$, such that
		\begin{align*}
			\bbm{u}^\top \bm{X}^\top \bm{X} \bbm{u} = \bm{Y}^\top \bm{Y} = \bm{Z}_T^\top \bm{Q} \bm{Z}_T.
		\end{align*}
		Here, $\bm{H}_u$ is the semi-infinite block Toeplitz matrix
		\begin{align}\label{eq:Hu}
			\bm{H}_u =
				\begin{bmatrix}
				\bbm{c}_0^\top & \bbm{c}_1^\top & \bbm{c}_2^\top & \ldots & \bbm{c}_{T-1}^\top & \ldots\\
				\bm{0}        & \bbm{c}_0^\top & \bbm{c}_1^\top & \ldots & \bbm{c}_{T-2}^\top & \ldots\\
				\bm{0}        & \bm{0}        & \bbm{c}_0^\top & \ldots & \bbm{c}_{T-3}^\top & \ldots\\
				\vdots        & \vdots        & \vdots        & \ddots & \vdots           & \ldots\\
				\bm{0}        & \bm{0}        & \bm{0}        & \ldots & \bbm{c}_0^\top    & \ldots
				\end{bmatrix} \in \mathbb{R}^{T\times\infty},
		\end{align}
		with $\bbm{c}_m = \sum_{k=1}^p \bbm{\Psi}_{m-k}^\top \bbm{u}_k$.
		Moreover, $\bm{Q}$ satisfies the norm bounds
		\[
		\|\bm{Q}\|_{\mathrm{op}}\leq \lambda_{\max}(\bbm{\Sigma}_\epsilon)\frac{p}{\mu_{\min}(\mathcal{A})}\|\bbm{u}\|_2^2
		\ \ \text{and} \ \ 
		\|\bm{Q}\|_{\mathrm{F}}\leq \lambda_{\max}(\bbm{\Sigma}_\epsilon)\frac{\sqrt{pT}}{\mu_{\min}(\mathcal{A})}\|\bbm{u}\|_2^2,
		\]
		where $\mu_{\min}(\mathcal{A})$ is defined in Lemma~\ref{lemma:ar-polynomial}.
	\end{lemma}

	\begin{proof}[\textbf{Proof of Lemma \ref{lemma:equivalence-representation-of-quadratic-form}}]
		We begin by deriving the equivalence representation of the quadratic form $\bbm{u}^\top \bm{X}^\top \bm{X} \bbm{u}$. Recall that $\bbm{y}_t = \sum_{j\geq0} \bbm{\Psi}_j \bbm{\epsilon}_{t-j}$ and $\bbm{\Psi}_r = \bm{O}$ for $r<0$.
		Then, for $t=1,\ldots,T$,

		\begin{align}\label{eq:ux-t}
		\langle\bbm{u},\bbm x_{T-t+1}\rangle
		&= \sum_{k=1}^{p} \langle\bbm{u}_k, \bbm{y}_{T-t+1-k}\rangle
		= \sum_{k=1}^{p}\sum_{j\geq0}\langle\bbm{u}_k, \bbm{\Psi}_j \bbm{\epsilon}_{T-t+1-k-j}\rangle
		= \sum_{k=1}^{p}\sum_{j\geq0}\langle\bbm{\Psi}_j^\top\bbm{u}_k, \bbm{\epsilon}_{T-t+1-k-j}\rangle \notag\\
		&= \sum_{k=1}^{p} \sum_{m\geq k} \langle\bbm{\Psi}_{m-k}^\top\bbm{u}_k, \bbm{\epsilon}_{T-t+1-m}\rangle
		= \sum_{k=1}^{p} \sum_{m\geq 0} \langle\bbm{\Psi}_{m-k}^\top\bbm{u}_k, \bbm{\epsilon}_{T-t+1-m}\rangle = \sum_{m\geq0} \langle \bbm{c}_m, \bbm{\epsilon}_{T-t+1-m}\rangle.
		\end{align}
		where $\bbm{c}_m = \sum_{k=1}^p \bbm{\Psi}_{m-k}^\top \bbm{u}_k$.
		Thus, the data vector $\bm{Y} = \bm{X}\bbm{u} \in \mathbb{R}^T$ can be expressed as $\bm{Y} = \bm{H}_u \bm{E}_T$, where $\bm{E}_T = (\bbm{\epsilon}_T^\top,\bbm{\epsilon}_{T-1}^\top,\bbm{\epsilon}_{T-2}^\top,\ldots,\bbm{\epsilon}_1^\top,\bbm{\epsilon}_0^\top,\bbm{\epsilon}_{-1}^\top,\ldots)^\top$ and $\bm{H}_u$ is defined in \eqref{eq:Hu} so that $(\bm{H}_u \bm{E}_T)_t = \sum_{m\geq0} \langle\bbm{c}_m, \bbm{\epsilon}_{T-t+1-m}\rangle$ for $t=1,\ldots,T$. Note that $\bbm{\epsilon}_t = \bbm{\Sigma}_\epsilon^{1/2}\bbm{\xi}_t$ and define the semi-infinite innovation vector $\bm{Z}_T = (\bbm{\xi}_T^\top,\bbm{\xi}_{T-1}^\top,\bbm{\xi}_{T-2}^\top,\ldots,\bbm{\xi}_1^\top,\bbm{\xi}_0^\top,\bbm{\xi}_{-1}^\top,\ldots)^\top$. Then $\bm{E}_T = (\bm{I}_{\infty}\otimes\bbm{\Sigma}_\epsilon^{1/2}) \bm{Z}_T$, and thus
		\begin{align*}
			\bbm{u}^\top \bm{X}^\top \bm{X} \bbm{u}
			= \bm{Y}^\top \bm{Y}
			= \bm{E}_T^\top \bm{H}_u^\top \bm{H}_u \bm{E}_T
			= \bm{Z}_T^\top (\bm{I}_{\infty}\otimes\bbm{\Sigma}_\epsilon^{1/2})^\top \bm{H}_u^\top \bm{H}_u (\bm{I}_{\infty}\otimes\bbm{\Sigma}_\epsilon^{1/2}) \bm{Z}_T.
		\end{align*}
		The positive semidefinite matrix $(\bm{I}_{\infty}\otimes\bbm{\Sigma}_\epsilon^{1/2})^\top \bm{H}_u^\top \bm{H}_u (\bm{I}_{\infty}\otimes\bbm{\Sigma}_\epsilon^{1/2})$ is denoted as $\bm{Q}$ for notation convenience. Next, we derive the operator and Frobenius norm bounds for $\bm{Q}$. By the sub-multiplicative property of matrix norms, we have
		\begin{align}
			\label{eq:frobenius-bound-Q}
			&\|\bm{Q}\|_{\mathrm{op}} \leq \|(\bm{I}_{\infty}\otimes\bbm{\Sigma}_\epsilon^{1/2})^\top\|_{\mathrm{op}} \|\bm{H}_u^\top \bm{H}_u\|_{\mathrm{op}} \|\bm{I}_{\infty}\otimes\bbm{\Sigma}_\epsilon^{1/2}\|_{\mathrm{op}} = \lambda_{\max}(\bbm{\Sigma}_\epsilon) \|\bm{H}_u\|_{\mathrm{op}}^2,\\
			\label{eq:op-bound-Q}
			&\|\bm{Q}\|_{\mathrm{F}} \leq \|(\bm{I}_{\infty}\otimes\bbm{\Sigma}_\epsilon^{1/2})^\top\|_{\mathrm{op}} \|\bm{H}_u^\top \bm{H}_u\|_{\mathrm{F}} \|\bm{I}_{\infty}\otimes\bbm{\Sigma}_\epsilon^{1/2}\|_{\mathrm{op}} \leq \lambda_{\max}(\bbm{\Sigma}_\epsilon) \|\bm{H}_u\|_{\mathrm{op}}\|\bm{H}_u\|_{\mathrm{F}}.
		\end{align}
		Hereafter, we focus on deriving upper bounds for $\|\bm{H}_u\|_{\mathrm{op}}$ and $\|\bm{H}_u\|_{\mathrm{F}}$. Define the bi-infinite block Toeplitz operator $\mathcal{H}_u: \ell_2(\mathbb{Z})\to\ell_2(\mathbb{Z})$ by
		\begin{align*}
			[\mathcal H_u]_{r,s} =
				\begin{cases}
				\bbm c_{\,r-s}^\top, & r\geq s,\\[2pt]
				\bm 0^\top, & r<s,
				\end{cases}
				\qquad r,s\in\mathbb{Z},
		\end{align*}
		with the convention $\bbm{c}_m=\bm{0}$ for $m<0$. For any $\bbm{e}=(\ldots,\bbm{e}_{-1}^\top,\bbm{e}_0^\top,\bbm{e}_1^\top,\ldots)^\top\in\ell_2$, the $r$-th coordinate of $\mathcal{H}_u\bbm{e}$ is $(\mathcal{H}_u \bbm{e})_r = \sum_{m\geq 0} \langle \bbm{c}_m,\bbm{e}_{r-m}\rangle$. By the absolute summability of $\{\bbm\Psi_j\}_{j\geq0}$ and finite $p$, we have $\sum_{m\geq0}\|\bbm{c}_m\|_2<\infty$. Consequently, the series $(\mathcal{H}_u\bbm{e})_r$ converges absolutely for each $r$, and $\mathcal{H}_u$ is a bounded linear operator on $\ell_2$.

		Thus $\mathcal{H}_u$ is a infinite-by-infinite block Toeplitz convolution operator on $\ell_2(\mathbb{Z})$. The finite matrix $\bm{H}_u$ is the one-sided compression of $\mathcal{H}_u$ that keeps only the first $T$ output rows while restricting the columns to the past indices that match the stacking order in $\bm{E}_T$; concretely, $\bm{H}_u = \mathcal{P}_T\mathcal{H}_u\mathcal{P}_{\leq T}$, where $\mathcal{P}_T$ extracts rows $1,\ldots,T$ and $\mathcal{P}_{\leq T}$ keeps the column index set $\{T,T-1,T-2,\ldots\}$. Then, by the definition of operator norm, we have $\|\bm{H}_u\|_{\mathrm{op}} \leq \|\mathcal{H}_u\|_{\mathrm{op}}$.
		
		Define $\bm{C}(\theta)=\sum_{m\geq0} \bbm{c}_m e^{-im\theta}$ and $\bm{E}_\infty(\theta)=\sum_{t\in\mathbb{Z}} \bbm{e}_t e^{-it\theta}$. Then $\bm{F}(\theta) = \sum_{r\in\mathbb{Z}} (\mathcal{H}_u \bm{e})_r e^{-ir\theta} = \langle \bm{C}(\theta), \bm{E}_\infty(\theta)\rangle$. By Parseval's identity applied to $(\mathcal{H}_u \bbm{e})_r$ and its Fourier series $\bm{F}(\theta)$, we have
		\begin{align*}
			\|\mathcal{H}_u \bbm{e}\|_2^2 &= \sum_{r\in\mathbb{Z}} |(\mathcal{H}_u \bbm{e})_r|^2 = \frac{1}{2\pi}\int_{-\pi}^{\pi} |\bm{F}(\theta)|^2  d\theta = \frac{1}{2\pi}\int_{-\pi}^{\pi} \big|\langle \bm{C}(\theta),\bm{E}_\infty(\theta)\rangle\big|^2  d\theta\\
			& \leq \frac{1}{2\pi}\int_{-\pi}^{\pi} \|\bm{C}(\theta)\|_2^2 \|\bm{E}_\infty(\theta)\|_2^2 d\theta \leq \sup_{\theta \in[-\pi,\pi]}\|\bm{C}(\theta)\|_2^2 \cdot \frac{1}{2\pi}\int_{-\pi}^{\pi}\|\bm{E}_\infty(\theta)\|_2^2 d\theta.
		\end{align*}
		Again, by Parseval's identity, we have $1/2\pi\int_{-\pi}^{\pi}\|\bm{E}_\infty(\theta)\|_2^2 d\theta = \sum_{t\in\mathbb Z}\|\bbm{e}_t\|_2^2 = \|\bbm{e}\|_2^2$. Hence, for every such $\bbm{e}$, $\|\mathcal{H}_u \bbm{e}\|_2^2 \leq \sup_{\theta \in[-\pi,\pi]}\|\bm{C}(\theta)\|_2^2 \|\bbm{e}\|_2^2$,
		which implies $\|\mathcal{H}_u\|_{\mathrm{op}} \leq \sup_{\theta \in[-\pi,\pi]}\|\bm{C}(\theta)\|_2$. Combining this with $\|\bm{H}_u\|_{\mathrm{op}}\leq\|\mathcal{H}_u\|_{\mathrm{op}}$ yields that 
		\begin{align}\label{eq:Hu-op-upper-bound}
			\|\bm{H}_u\|_{\mathrm{op}} \leq \sup_{\theta \in[-\pi,\pi]}\|\bm{C}(\theta)\|_2.
		\end{align}
	
		Next, we derive a uniform upper bound for $\sup_{\theta \in[-\pi,\pi]}\|\bm C(\theta)\|_2$. Recall that $\bbm c_m=\sum_{k=1}^{p}\bbm\Psi_{m-k}^\top \allowbreak\bbm u_k$, for any $m\in \mathbb N$, and adopt the convention $\bbm\Psi_j=\bm O$ for all $j<0$; in particular, this implies $\bbm c_0=\bm 0$. Then, we have
		\begin{align*}
		\bm C(\theta)
		&=\sum_{m\geq 0}\bbm c_m e^{-im\theta}
		=\sum_{m\geq 0}\sum_{k=1}^{p}\bbm\Psi_{m-k}^\top \bbm u_k e^{-im\theta}=\sum_{k=1}^{p}\sum_{j\geq 0}\bbm\Psi_{j}^\top \bbm u_k e^{-i(j+k)\theta}\\
		&=\left(\sum_{j\geq 0}\bbm\Psi_j^\top e^{-ij\theta}\right)\left(\sum_{k=1}^{p}\bbm u_k e^{-ik\theta}\right) =\bbm\Psi^\top(e^{-i\theta})\bm U_p(\theta),
		\end{align*}
		where $\bm U_p(\theta)=\sum_{k=1}^{p}\bbm u_k e^{-ik\theta}$. Note that $\|\bm C(e^{-i\theta})\|_2 \leq \|\bbm\Psi^\top(e^{-i\theta})\|_{\op}\|\bm U_p(\theta)\|_2$. Then, using the Rayleigh-Ritz theorem together with Lemma~\ref{lemma:ar-polynomial}, we have
		\begin{align}\label{eq:phi-op-upper-bound}
		\|\bbm\Psi^\top(e^{-i\theta})\|_{\op}^2
		&=\lambda_{\max}\!\Bigl((\bbm\Psi^\top(e^{-i\theta}))^\dagger\bbm\Psi^\top(e^{-i\theta})\Bigr)
		=\lambda_{\max}\!\Bigl(\bbm\Psi^\dagger(e^{-i\theta})\bbm\Psi(e^{-i\theta})\Bigr)
		\leq \frac{1}{\mu_{\min}(\mathcal A)}.
		\end{align}
		Moreover, since $|e^{-ik\theta}|=1$, applying the subadditivity of the $\ell_2$-norm and Cauchy-Schwarz inequality, we have
		\begin{align}\label{eq:Up-norm-bound}
		\|\bm U_p(\theta)\|_2
		=\Bigl\|\sum_{k=1}^{p}\bbm u_k e^{-ik\theta}\Bigr\|_2 \leq \sum_{k=1}^{p}\|\bbm u_k e^{-ik\theta}\|_2
		=\sum_{k=1}^{p}\|\bbm u_k\|_2 \leq \sqrt{p}\Bigl(\sum_{k=1}^{p}\|\bbm u_k\|_2^2\Bigr)^{1/2}
		= \sqrt{p}\,\|\bbm u\|_2,
		\end{align}
		Therefore, combining \eqref{eq:phi-op-upper-bound} and \eqref{eq:Up-norm-bound}, we have, for any $\theta\in[-\pi,\pi]$,
		\[
		\|\bm C(\theta)\|_2 \leq \|\bbm\Psi^\top(e^{-i\theta})\|_{\op}\|\bm U_p(\theta)\|_2 \leq \frac{\sqrt{p}}{\sqrt{\mu_{\min}(\mathcal A)}}\|\bbm u\|_2,
		\]
		and hence, taking the supremum over $\theta\in[-\pi,\pi]$ and combining this with \eqref{eq:Hu-op-upper-bound} yields
		\begin{align}\label{eq:Hu-op-bound}
		\|\bm H_u\|_{\op}\leq \sup_{\theta \in[-\pi,\pi]}\|\bm C(\theta)\|_2 \leq \frac{\sqrt{p}}{\sqrt{\mu_{\min}(\mathcal A)}}\|\bbm u\|_2.
		\end{align}

		Finally, we derive the Frobenius norm bound for $\bm{H}_u$. By the definition of $\bm{H}_u$ in \eqref{eq:Hu}, we have $\|\bm{H}_u\|_{\mathrm{F}}^2 = \sum_{t=1}^T \sum_{m\geq0} \|\bbm{c}_m\|_2^2 = T \sum_{m\geq0} \|\bbm{c}_m\|_2^2$. Then, using the earlier bound for $\sum_{m\geq0} \|\bbm{c}_m\|_2^2$ derived in the proof of Lemma~\ref{lemma:stacked-subg-vector}, see \eqref{eq:upper-bound-for-cm}, we have
		\begin{align}\label{eq:Hu-Fro-bound}
		\|\bm H_u\|_{\F}\leq \frac{\sqrt{T}}{\sqrt{\mu_{\min}(\mathcal A)}}\,\|\bbm u\|_2.
		\end{align}
		Plugging \eqref{eq:Hu-op-bound} and \eqref{eq:Hu-Fro-bound} into \eqref{eq:op-bound-Q} and \eqref{eq:frobenius-bound-Q} yields the desired results.
	\end{proof}

	\begin{lemma}\label{lemma:equivalence-representation-of-uXEv}
		Let $\{\bbm{y}_t\}_{t\in\mathbb Z}$ be a stationary VAR($p$) process with VMA($\infty$) representation $\bbm{y}_t = \sum_{j\geq0} \bbm{\Psi}_j \bbm{\epsilon}_{t-j}$, where $\{\bbm{\Psi}_j\}_{j\geq0}$ is absolutely summable. Suppose $\bbm{\epsilon}_t = \bbm{\Sigma}_\epsilon^{1/2}\bbm{\xi}_t$, where $\bbm{\xi}_t=(\xi_{1t},\ldots,\xi_{dt})^\top$ has independent coordinates, is $i.i.d.$ across $t$, and satisfies $\mathbb{E}\exp(\mu\xi_{it}) \leq \exp(\mu^2\sigma^2/2)$ for all $\mu\in\mathbb{R}$ and all $i$. Fix an integer $p\geq1$ and define the stacked lag vector $\bbm{x}_t = (\bbm{y}_{t-1}^\top,\ldots,\bbm{y}_{t-p}^\top)^\top \in \mathbb{R}^{pd}$. Let $\bm{X}\in\mathbb{R}^{T\times pd}$ be the data matrix with $t$-th row $\bbm{x}_{T-t+1}^\top$ and $\bm{E}\in\mathbb{R}^{T\times d}$ be the innovation matrix with $t$-th row $\bbm{\epsilon}_{T-t+1}^\top$. Then for any fixed $\bbm{u}=(\bbm{u}_1^\top,\ldots,\bbm{u}_p^\top)^\top\in\mathbb{R}^{pd}$ and $\bbm{v}\in\mathbb{R}^d$, there exist a semi-infinite innovation vector $\bm{Z}_T = (\bbm{\xi}_T^\top,\bbm{\xi}_{T-1}^\top,\bbm{\xi}_{T-2}^\top,\ldots,\bbm{\xi}_1^\top,\bbm{\xi}_0^\top,\bbm{\xi}_{-1}^\top,\ldots)^\top$ and a deterministic matrix $\bm{Q}_{u,v} = (\bm{I}_{\infty}\otimes\bbm{\Sigma}_\epsilon^{1/2})^\top \bm{H}_u^\top \bm{D}_v (\bm{I}_{\infty}\otimes\bbm{\Sigma}_\epsilon^{1/2})$, such that
		\begin{align*}
			\bbm{u}^\top\bm{X}^\top\bm{E}\bbm{v} = \bm{Z}_T^\top \bm{Q}_{u,v} \bm{Z}_T.
		\end{align*}
		Here $\bm{H}_u$ is the semi-infinite block lower-triangular Toeplitz matrix from Lemma \ref{lemma:equivalence-representation-of-quadratic-form} and $\bm{D}_v$ is the block-diagonal matrix
		\begin{align*}
			\bm{D}_v =
			\begin{bmatrix}
			\bbm{v}^\top & \bm{0}      & \bm{0}      & \cdots 	& \bm{0}	& \cdots \\
			\bm{0}      & \bbm{v}^\top & \bm{0}      & \cdots	& \bm{0}	& \cdots  \\
			\bm{0}      & \bm{0}      & \bbm{v}^\top & \cdots 	& \bm{0}	& \cdots \\
			\vdots      & \vdots      & \vdots      & \ddots 	& \vdots	& \cdots \\
			\bm{0}      & \bm{0}      & \bm{0}      & \cdots	& \bbm{v}^\top	&\cdots
			\end{bmatrix} \in \mathbb{R}^{T\times\infty},
		\end{align*}
		Moreover, $\bm{Q}_{u,v}$ satisfies the norm bounds
		\begin{align*}
			\|\bm{Q}_{u,v}\|_{\mathrm{op}} \leq \lambda_{\max}(\bbm{\Sigma}_\epsilon)\frac{\sqrt{p}}{\sqrt{\mu_{\min}(\mathcal{A})}}\|\bbm{u}\|_2\|\bbm{v}\|_2.
			\ \ \text{and} \ \ 
			\|\bm{Q}_{u,v}\|_{\mathrm{F}} \leq \lambda_{\max}(\bbm{\Sigma}_\epsilon)\frac{\sqrt{T}}{\sqrt{\mu_{\min}(\mathcal{A})}}\|\bbm{u}\|_2\|\bbm{v}\|_2
		\end{align*}
	\end{lemma}

	\begin{proof}[\textbf{Proof of Lemma \ref{lemma:equivalence-representation-of-uXEv}}]
		We begin by deriving the equivalence representation of $\bbm{u}^\top\bm{X}^\top\bm{E}\bbm{v}$. Define the semi-infinite innovation vector $\bm{E}_T = (\bbm{\epsilon}_T^\top,\bbm{\epsilon}_{T-1}^\top,\bbm{\epsilon}_{T-2}^\top,\ldots,\bbm{\epsilon}_1^\top,\bbm{\epsilon}_0^\top,\bbm{\epsilon}_{-1}^\top,\ldots)^\top$.
		As the equation \eqref{eq:ux-t} in the proof of Lemma \ref{lemma:equivalence-representation-of-quadratic-form}, there exists a semi-infinite block Toeplitz matrix $\bm{H}_u$, generated by $\{\bbm{c}_m\}_{m\geq0}$ with $\bbm{c}_m = \sum_{k=1}^p \bbm{\Psi}_{m-k}^\top \bbm{u}_k$, such that for each $t=1,\ldots,T$,
		\begin{align*}
			\langle\bbm{u},\bbm{x}_{T-t+1}\rangle  = \sum_{m\geq0} \langle\bbm{c}_m, \bbm{\epsilon}_{T-t+1-m}\rangle = (\bm{H}_u \bm{E}_T)_t.
		\end{align*}
		Besides, $(\bm{D}_v \bm{E}_T)_t = \langle\bbm{v},\bbm{\epsilon}_{T-t+1}\rangle$ under the definition of $\bm{D}_v$. Consequently, we have 
		\begin{align*}
			\bbm{u}^\top\bm{X}^\top\bm{E}\bbm{v} = \sum_{t=1}^{T}\langle\bbm{u},\bbm{x}_{T-t+1}\rangle\langle\bbm{\epsilon}_{T-t+1},\bbm{v}\rangle = (\bm{H}_u \bm{E}_T)^\top (\bm{D}_v \bm{E}_T) = \bm{E}_T^\top \bm{H}_u^\top \bm{D}_v \bm{E}_T.
		\end{align*}
		Note that $\bbm{\epsilon}_t = \bbm{\Sigma}_\epsilon^{1/2}\bbm{\xi}_t$ and defining $\bm{Z}_T = (\bbm{\xi}_T^\top,\bbm{\xi}_{T-1}^\top,\bbm{\xi}_{T-2}^\top,\ldots,\bbm{\xi}_1^\top,\bbm{\xi}_0^\top,\bbm{\xi}_{-1}^\top,\ldots)^\top$, we have $\bm{E}_T = (\bm{I}_\infty\otimes\bbm{\Sigma}_\epsilon^{1/2})\bm{Z}_T$, and hence
		\begin{align*}
			\bbm{u}^\top\bm{X}^\top\bm{E}\bbm{v} = \bm{Z}_T^\top (\bm{I}_\infty\otimes\bbm{\Sigma}_\epsilon^{1/2})^\top \bm{H}_u^\top \bm{D}_v (\bm{I}_\infty\otimes\bbm{\Sigma}_\epsilon^{1/2}) \bm{Z}_T = \bm{Z}_T^\top \bm{Q}_{u,v} \bm{Z}_T,
		\end{align*}
		with $(\bm{I}_\infty\otimes\bbm{\Sigma}_\epsilon^{1/2})^\top \bm{H}_u^\top \bm{D}_v (\bm{I}_\infty\otimes\bbm{\Sigma}_\epsilon^{1/2})$ denoted as $\bm{Q}_{u,v}$ for notation convenience. Next, we derive the operator and Frobenius norm bounds for $\bm{Q}_{u,v}$. By the sub-multiplicative property of matrix norms, we have
		\begin{align}
			\label{eq:frobenius-bound-Quv}
			&\|\bm{Q}_{u,v}\|_{\mathrm{op}} \leq \|(\bm{I}_{\infty}\otimes\bbm{\Sigma}_\epsilon^{1/2})^\top\|_{\mathrm{op}} \|\bm{H}_u^\top \bm{D}_v\|_{\mathrm{op}} \|\bm{I}_{\infty}\otimes\bbm{\Sigma}_\epsilon^{1/2}\|_{\mathrm{op}} \leq \lambda_{\max}(\bbm{\Sigma}_\epsilon) \|\bm{H}_u\|_{\mathrm{op}}\|\bm{D}_v\|_{\mathrm{op}},\\
			\label{eq:op-bound-Quv}
			&\|\bm{Q}_{u,v}\|_{\mathrm{F}} \leq \|(\bm{I}_{\infty}\otimes\bbm{\Sigma}_\epsilon^{1/2})^\top\|_{\mathrm{op}} \|\bm{H}_u^\top \bm{D}_v\|_{\mathrm{F}} \|\bm{I}_{\infty}\otimes\bbm{\Sigma}_\epsilon^{1/2}\|_{\mathrm{op}} \leq \lambda_{\max}(\bbm{\Sigma}_\epsilon) \|\bm{H}_u\|_{\mathrm{F}}\|\bm{D}_v\|_{\mathrm{op}}.
		\end{align}
		Note that $\|\bm{D}_v\|_{\mathrm{op}}=\|\bbm{v}\|_2$, so $\|\bm{D}_v\|_{\mathrm{op}}$ in \eqref{eq:frobenius-bound-Quv} and \eqref{eq:op-bound-Quv} can be replaced by $\|\bbm{v}\|_2$.
		We now derive upper bounds for $\|\bm{H}_u\|_{\mathrm{op}}$ and $\|\bm{H}_u\|_{\mathrm{F}}$. For the operator norm, \eqref{eq:Hu-op-bound} yields $\|\bm{H}_u\|_{\mathrm{op}} \leq \sqrt{p/\mu_{\min}(\mathcal{A})}\|\bbm{u}\|_2$, 
		while for the Frobenius norm, \eqref{eq:Hu-Fro-bound} gives $\|\bm{H}_u\|_{\mathrm{F}} \leq \sqrt{T/\mu_{\min}(\mathcal{A})}\,\|\bbm{u}\|_2$. 
		Substituting these bounds into \eqref{eq:frobenius-bound-Quv} and \eqref{eq:op-bound-Quv} yields
		\begin{align*}
			\|\bm{Q}_{u,v}\|_{\mathrm{op}} \leq \lambda_{\max}(\bbm{\Sigma}_\epsilon)\frac{\sqrt{p}}{\sqrt{\mu_{\min}(\mathcal{A})}}\|\bbm{u}\|_2\|\bbm{v}\|_2,
			\ \ \text{and} \ \
			\|\bm{Q}_{u,v}\|_{\mathrm{F}} \leq \lambda_{\max}(\bbm{\Sigma}_\epsilon)\frac{\sqrt{T}}{\sqrt{\mu_{\min}(\mathcal{A})}}\|\bbm{u}\|_2\|\bbm{v}\|_2,
		\end{align*}
		which are exactly the claimed bounds.
	\end{proof}

	\begin{lemma}[Infinite-dimensional quadratic Hanson-Wright inequality]\label{lemma:infinite-HW}
		Let $I$ be a countable index set and let $\{Z_i\}_{i \in I}$ be independent random variables such that there exists a constant $K > 0$ with $\|Z_i\|_{\psi_2}\leq K$ for all $i \in I$. Let $\bm{Q}=(q_{ij})_{i,j\in I}$ be a symmetric linear operator on $\ell_2(I)$ satisfying $\|\bm{Q}\|_{\mathrm{op}} < \infty$ and $\|\bm{Q}\|_{\mathrm{F}}^2 = \sum_{i,j \in I} q_{ij}^2 < \infty$.
		For any finite subset $J \subset I$, let $\bm{Z}^{(J)} = (Z_i)_{i\in J}$ denote the truncated vector and let $\bm{Q}^{(J)} = (q_{ij})_{i,j\in J}$ denote the corresponding truncated matrix. Set $S_J = (\bm{Z}^{(J)})^\top \bm{Q}^{(J)} \bm{Z}^{(J)} = \sum_{i,j\in J} q_{ij} Z_i Z_j$. Then $\{S_J\}_{J}$ is a Cauchy sequence in $L^2$, and hence there exists a limit $S = \lim_{J\uparrow I} S_J$ in $L^2$, which admits the representation $S = \sum_{i,j\in I} q_{ij} Z_i Z_j = \bm{Z}^\top \bm{Q} \bm{Z}$ in the sense of $L^2$ convergence. Moreover, for any $t>0$,
		\begin{align*}
			\mathbb{P}\left(|S-\mathbb{E} S|\geq t\right) \leq 2\exp\left(-c\min\left\{\frac{t^2}{K^4\|\bm{Q}\|_{\mathrm{F}}^2},\frac{t}{K^2\|\bm{Q}\|_{\mathrm{op}}}\right\}\right).
		\end{align*}
	\end{lemma}

	\begin{proof}[\textbf{Proof of Lemma \ref{lemma:infinite-HW}}]
		The proof of this lemma consists of three parts: first we establish the well-definedness of $S$, then we derive a finite-dimensional Hanson-Wright inequality for the truncations $S_J$, and finally we pass to the limit $J\uparrow I$ to obtain the desired bound for $S$.\\
		\textbf{Step 1 (Well-definedness of $S$).} For any finite subset $J\subset I$, define $S_J = \sum_{i,j\in J} q_{ij} Z_i Z_j$. We first show that $\{S_J\}_J$ is a Cauchy net in $L^2$. For finite sets $J\subset J'\subset I$, we have
		\begin{align*}
			S_{J'}-S_J = \sum_{\substack{i\in J',j\in J'}} q_{ij}Z_iZ_j - \sum_{\substack{i\in J,j\in J}} q_{ij}Z_iZ_j = \sum_{\substack{i\in J',j\in J'\\ \text{at least one of }i,j\notin J}} q_{ij}Z_iZ_j.
		\end{align*}
		For notational convenience, set $A(J,J') = \big\{(i,j)\in J'\times J' : i\notin J \text{ or } j\notin J\big\} = (J'\times J')\setminus (J\times J)$. Then, we have 
		\begin{align*}
			\mathbb{E}\big[(S_{J'}-S_J)^2\big] = \sum_{(i,j),(k,\ell)\in A(J,J')} q_{ij}q_{k\ell}\,\mathbb{E}[Z_iZ_jZ_kZ_\ell].
		\end{align*}

		Since the $\{Z_i\}_{i\in I}$ are independent with mean zero, we have $\mathbb{E}[Z_iZ_jZ_kZ_\ell]=0$ unless the indices $(i,j,k,\ell)$ come in pairs. In all nonzero cases we may bound the fourth moments crudely using the sub-Gaussian assumption: there exists a constant $C_0>0$ such that $\mathbb{E}[Z_i^2 Z_j^2]\leq C_0 K^4$ for $i,j\in I$. Consequently,
		\begin{align*}
			\mathbb{E}\big[(S_{J'}-S_J)^2\big]\leq C_1 K^4 \sum_{(i,j)\in A(J,J')} q_{ij}^2,
		\end{align*}
		for some constant $C_1>0$. Since $\sum_{i,j\in I} q_{ij}^2 = \|\bm{Q}\|_{\mathrm{F}}^2 < \infty$, the tail sum $\sum_{(i,j)\in A(J,J')} q_{ij}^2$ tends to $0$ as $J,J'\uparrow I$. Hence $\mathbb{E}\big[(S_{J'}-S_J)^2\big]\xrightarrow[J,J'\uparrow I]{}0$, so $\{S_J\}_J$ is a Cauchy net in $L^2$. Therefore, there exists $S\in L^2$ such that $S_J\to S$ in $L^2$ as $J\uparrow I$. By construction,
		\begin{align*}
			S = \lim_{J \uparrow I} S_J = \lim_{J \uparrow I} \sum_{i,j\in J} q_{ij} Z_i Z_j = \sum_{i,j\in I} q_{ij} Z_i Z_j = \bm{Z}^\top \bm{Q}\bm{Z}.
		\end{align*}
		Consequently, the infinite quadratic form is well defined in $L^2$.\\
		\textbf{Step 2 (Finite-dimensional Hanson-Wright inequality).}
		For each finite $J\subset I$, the vector $\bm{Z}^{(J)}=(Z_i)_{i\in J}$ is a $|J|$-dimensional vector with independent sub-Gaussian entries, and $\bm{Q}^{(J)}=(q_{ij})_{i,j\in J}$ is a symmetric matrix with $\|\bm{Q}^{(J)}\|_{\mathrm{op}}\leq\|\bm{Q}\|_{\mathrm{op}}$ and $\|\bm{Q}^{(J)}\|_{\mathrm{F}}\leq\|\bm{Q}\|_{\mathrm{F}}$.
		By the classical Hanson-Wright inequality, e.g., \cite{vershynin2018high}, Theorem 6.2.1, there exists a constant $c>0$ such that, for any $t>0$,
		\begin{align*}
			\mathbb{P}\left(|S_J-\mathbb{E} S_J|\geq t\right) \leq 2\exp\left(-c\min\left\{\frac{t^2}{K^4\|\bm{Q}^{(J)}\|_{\mathrm{F}}^2},\frac{t}{K^2\|\bm{Q}^{(J)}\|_{\mathrm{op}}}\right\}\right).
		\end{align*}
		Using the bounds $\|\bm{Q}^{(J)}\|_{\mathrm{F}}\leq\|\bm{Q}\|_{\mathrm{F}}$ and $\|\bm{Q}^{(J)}\|_{\mathrm{op}}\leq\|\bm{Q}\|_{\mathrm{op}}$, we have 
		\begin{align}\label{eq:HW-finite-J}
		\mathbb{P}\left(|S_J-\mathbb{E} S_J|\geq t\right) \leq 2\exp\left(-c\min\left\{\frac{t^2}{K^4\|\bm{Q}\|_{\mathrm{F}}^2},\frac{t}{K^2\|\bm{Q}\|_{\mathrm{op}}}\right\}\right),
		\end{align}
		which holds uniformly over all finite $J \subset I$ and $t>0$.\\
		\textbf{Step 3 (Passing to $S$).}
		From Step 1 we know that $S_J\to S$ in $L^2$, and hence $S_J\to S$ in probability and thus $\mathbb{E} S_J\to \mathbb{E} S$ as $J\uparrow I$. We now pass from the finite-dimensional inequality \eqref{eq:HW-finite-J} to the infinite-dimensional bound. Fix $t>0$ and choose any $\eta\in(0,1/2)$. Since $S_J\to S$ in $L^2$, there exists a finite subset $J_0\subset I$ such that for all $J\supset J_0$,
		\begin{align}\label{eq:L2-small}
		\mathbb{E}\big[(S_J-S)^2\big] \leq \frac{\eta^3 t^2}{4}.
		\end{align}
		By the Cauchy-Schwarz inequality, we have
		\begin{align*}
			|\mathbb{E} S_J-\mathbb{E} S| \leq \mathbb{E}|S_J-S| \leq \big(\mathbb{E}[(S_J-S)^2]\big)^{1/2} \leq \frac{\eta^{3/2}t}{2} < \frac{\eta t}{2},
		\end{align*}
		for all such $J \supset J_0$. Define $\Delta_J = S-S_J$. Then we can decompose $S-\mathbb{E} S=(S_J-\mathbb{E} S_J)+(\Delta_J-\mathbb{E}\Delta_J)$. Therefore, for every $t>0$,
		\begin{align}\label{eq:split}
		\mathbb{P}\big(|S-\mathbb{E} S|\geq t\big) \leq \mathbb{P}\big(|S_J-\mathbb{E} S_J|\geq (1-\eta)t\big) + \mathbb{P}\big(|\Delta_J-\mathbb{E}\Delta_J|\geq \eta t\big).
		\end{align}
		Next, we bound the second probability. By the triangle inequality, we have $|\Delta_J-\mathbb{E}\Delta_J| = |(S-S_J)-(\mathbb{E} S-\mathbb{E} S_J)| \leq |S-S_J|+|\mathbb{E} S-\mathbb{E} S_J|$.
		For all sufficiently large $J$ we have $|\mathbb{E} S-\mathbb{E} S_J|<\eta t/2$, and hence $\{|\Delta_J-\mathbb{E}\Delta_J|\geq \eta t\} \subset \{|S-S_J|\geq \eta t/2\}$.
		Therefore, for such $J$,
		\begin{align*}
			\mathbb{P}\big(|\Delta_J-\mathbb{E}\Delta_J|\geq \eta t\big) \leq \mathbb{P}\left(|S-S_J|\geq \frac{\eta t}{2}\right) \leq \frac{\mathbb{E}[(S-S_J)^2]}{(\eta t/2)^2} \leq \eta,
		\end{align*}
		where we used Chebyshev's inequality and \eqref{eq:L2-small}. Substituting this bound into \eqref{eq:split}, we have, for any $J \supset J_0$,
		\begin{align*}
			\mathbb{P}\big(|S-\mathbb{E} S|\geq t\big) \leq \mathbb{P}\big(|S_J-\mathbb{E} S_J|\geq (1-\eta)t\big)+\eta.
		\end{align*}
		Then, applying \eqref{eq:HW-finite-J} with $t$ replaced by $(1-\eta)t$ yields
		\begin{align*}
			\mathbb{P}\big(|S-\mathbb{E} S|\geq t\big) \leq 2\exp\left(-c\min\left\{\frac{(1-\eta)^2 t^2}{K^4\|\bm{Q}\|_{\mathrm{F}}^2},\frac{(1-\eta)t}{K^2\|\bm{Q}\|_{\mathrm{op}}}\right\}\right)+\eta.
		\end{align*}
		Introduce $\Lambda(t)=\min\{t^2/K^4\|\bm{Q}\|_{\mathrm{F}}^2, t/K^2\|\bm{Q}\|_{\mathrm{op}}\}$.
		Then, the previous inequality can be written as $\mathbb{P}(|S-\mathbb{E} S|\geq t) \leq 2\exp(-c(1-\eta)^2\Lambda(t))+\eta$. Now choose $\eta=\exp(-c\Lambda(t)/2) \in (0,1/2)$. For this choice, we have $2\exp(-c(1-\eta)^2\Lambda(t)) \leq 2\exp(-c\Lambda(t)/2)$. Consequently, we have $\mathbb{P}\big(|S-\mathbb{E} S|\geq t\big) \leq 3\exp(-c\Lambda(t)/2)$. Then, we conclude that there exists a constant $c>0$ such that for $t>0$,
		\begin{align*}
			\mathbb{P}\big(|S-\mathbb{E} S|\geq t\big) \leq 2\exp\left(-c\min\left\{\frac{t^2}{K^4\|\bm{Q}\|_{\mathrm{F}}^2},\frac{t}{K^2\|\bm{Q}\|_{\mathrm{op}}}\right\}\right).
		\end{align*}
		This is exactly the desired bound, and completes the proof.
	\end{proof}

	\begin{lemma}[Hanson-Wright type bound for VAR quadratic forms]\label{lemma:HW-uXTXu}
		Let $\{\bbm{y}_t\}_{t\in\mathbb Z}$ be a stationary VAR($p$) process with VMA($\infty$) representation $\bbm{y}_t = \sum_{j\geq0} \bbm{\Psi}_j \bbm{\epsilon}_{t-j}$, where $\{\bbm{\Psi}_j\}_{j\geq0}$ is absolutely summable. Suppose $\bbm{\epsilon}_t = \bbm{\Sigma}_\epsilon^{1/2}\bbm{\xi}_t$, where $\bbm{\xi}_t=(\xi_{1t},\ldots,\xi_{dt})^\top$ has independent coordinates, is $i.i.d.$ across $t$, and satisfies $\mathbb{E}\exp(\mu\xi_{it}) \leq \exp(\mu^2\sigma^2/2)$ for all $\mu\in\mathbb{R}$ and all $i$. Define the stacked lag vector $\bbm{x}_t=(\bbm{y}_{t-1}^\top,\ldots,\bbm{y}_{t-p}^\top)^\top\in\mathbb{R}^{pd}$, and let $\bm{X}\in\mathbb{R}^{T\times pd}$ be the data matrix with $t$-th row $\bbm{x}_{T-t+1}^\top$. Then, for any fixed $\bbm{u}=(\bbm{u}_1^\top,\ldots,\bbm{u}_p^\top)^\top\in\mathbb{R}^{pd}$, there exists a constant $c>0$ such that, for $t>0$,
		\begin{align*}
			\mathbb{P}\Big(\big|\bbm{u}^\top\bm{X}^\top\bm{X}\bbm{u}-\mathbb{E}[\bbm{u}^\top\bm{X}^\top\bm{X}\bbm{u}]\big|\geq t\Big)\leq 2\exp\Bigg(-c\min\Bigg\{\frac{t^2}{C_{\epsilon,\mathcal A}^2\sigma^4pT\|\bbm{u}\|_2^4},\frac{t}{C_{\epsilon,\mathcal A}\sigma^2p\|\bbm{u}\|_2^2}\Bigg\}\Bigg),
		\end{align*}
		where $C_{\epsilon,\mathcal A} = \lambda_{\max}(\bbm{\Sigma}_\epsilon)\mu_{\min}^{-1}(\mathcal{A})$ and $\mu_{\min}(\mathcal{A})$ is defined in Lemma~\ref{lemma:ar-polynomial}.
	\end{lemma}

	\begin{proof}[\textbf{Proof of Lemma \ref{lemma:HW-uXTXu}}]
		Under the stated assumptions, Lemma~\ref{lemma:equivalence-representation-of-quadratic-form} yields a semi-infinite innovation vector $\bm{Z}_T = (\bbm{\xi}_T^\top,\bbm{\xi}_{T-1}^\top,\bbm{\xi}_{T-2}^\top,\ldots,\bbm{\xi}_1^\top,\bbm{\xi}_0^\top,\bbm{\xi}_{-1}^\top,\ldots)^\top$ and a deterministic positive semidefinite matrix $\bm{Q}= (\bm{I}_{\infty}\otimes\bbm{\Sigma}_\epsilon^{1/2})^\top \bm{H}_u^\top \bm{H}_u (\bm{I}_{\infty}\otimes\bbm{\Sigma}_\epsilon^{1/2})$ such that
		\[
			\bbm{u}^\top \bm{X}^\top \bm{X} \bbm{u}=\bm{Z}_T^\top \bm{Q} \bm{Z}_T.
		\]
		Since $\{\xi_{it}\}$ are independent and sub-Gaussian with $\|\xi_{it}\|_{\psi_2}\lesssim \sigma$, the entries of $\bm{Z}_T$ are independent and satisfy the same $\psi_2$ bound.

		Next, Lemma~\ref{lemma:equivalence-representation-of-quadratic-form} also provides the norm bounds $\|\bm{Q}\|_{\mathrm{op}}\leq C_{\epsilon,\mathcal A}p\|\bbm{u}\|_2^2$ and $\|\bm{Q}\|_{\mathrm{F}}\leq C_{\epsilon,\mathcal A}\sqrt{pT}\|\bbm{u}\|_2^2$. We now apply Lemma~\ref{lemma:infinite-HW} with $S=\bm{Z}_T^\top \bm{Q}\bm{Z}_T$ and $\bm{Q}$ as above. Lemma~\ref{lemma:infinite-HW} guarantees the existence of a constant $c>0$ such that, for $t>0$,
		\begin{align*}
			\mathbb{P}\Big(\big|\bbm{u}^\top\bm{X}^\top\bm{X}\bbm{u}-\mathbb{E}[\bbm{u}^\top\bm{X}^\top\bm{X}\bbm{u}]\big|\geq t\Big) \leq 2\exp\left(-c\min\left\{\frac{t^2}{\sigma^4\|\bm{Q}\|_{\mathrm{F}}^2},\frac{t}{\sigma^2\|\bm{Q}\|_{\mathrm{op}}}\right\}\right).
		\end{align*}
		Substituting the bounds on $\|\bm{Q}\|_{\mathrm{F}}$ and $\|\bm{Q}\|_{\mathrm{op}}$ yields
		\begin{align*}
			\mathbb{P}\Big(\big|\bbm{u}^\top\bm{X}^\top\bm{X}\bbm{u}-\mathbb{E}[\bbm{u}^\top\bm{X}^\top\bm{X}\bbm{u}]\big|\geq t\Big) \leq 2\exp\Bigg(-c\min\Bigg\{\frac{t^2}{C_{\epsilon,\mathcal A}^2\sigma^4pT\|\bbm{u}\|_2^4},\frac{t}{C_{\epsilon,\mathcal A}\sigma^2p\|\bbm{u}\|_2^2}\Bigg\}\Bigg).
		\end{align*}
		This completes the proof.
	\end{proof}

	\begin{lemma}[Hanson-Wright type bound for VAR cross terms]\label{lemma:HW-uXEv}
		Let $\{\bbm{y}_t\}_{t\in\mathbb{Z}}$ be a stationary VAR($p$) process with VMA($\infty$) representation $\bbm{y}_t = \sum_{j\geq 0}\bm{\Psi}_j\bbm{\epsilon}_{t-j}$, where $\{\bbm{\Psi}_j\}_{j\geq 0}$ is absolutely summable. Suppose $\bbm{\epsilon}_t = \bbm{\Sigma}_\epsilon^{1/2}\bbm{\xi}_t$, where $\bbm{\xi}_t=(\xi_{1t},\ldots,\xi_{dt})^\top$ has independent coordinates, is $i.i.d.$ across $t$, and satisfies $\mathbb{E}\exp(\mu\xi_{it})\leq \exp(\mu^2\sigma^2/2)$ for all $\mu\in\mathbb{R}$ and all $i$. Define the stacked lag vector $\bbm{x}_t=(\bbm{y}_{t-1}^\top,\ldots,\bbm{y}_{t-p}^\top)^\top\in\mathbb{R}^{pd}$, and let $\bm{X}\in\mathbb{R}^{T\times pd}$ be the data matrix with $t$-th row $\bbm{x}_{T-t+1}^\top$. Let $\bm{E}\in\mathbb{R}^{T\times d}$ be the innovation matrix with $t$-th row $\bbm{\epsilon}_{T-t+1}^\top$. Then, for any fixed $\bbm{u}=(\bbm{u}_1^\top,\ldots,\bbm{u}_p^\top)^\top\in\mathbb{R}^{pd}$ and $\bbm{v}\in\mathbb{R}^d$, there exists a constant $c>0$ such that, for every $t>0$,
		\begin{align*}
			\mathbb{P}\big(|\bbm{u}^\top\bm{X}^\top\bm{E}\bbm{v}|\geq t\big) \leq 2\exp\Bigg(-c\min\Bigg\{\frac{t^2}{C_{\epsilon,\mathcal A}^2 \sigma^4T\|\bbm{u}\|_2^2\|\bbm{v}\|_2^2}, \frac{t}{C_{\epsilon,\mathcal A}\sigma^2\sqrt{p}\|\bbm{u}\|_2\|\bbm{v}\|_2}\Bigg\}\Bigg),
		\end{align*}
		where $C_{\epsilon,\mathcal A} = \lambda_{\max}(\bbm{\Sigma}_\epsilon)\mu_{\min}^{-1}(\mathcal{A})$ and $\mu_{\min}(\mathcal{A})$ is defined in Lemma~\ref{lemma:ar-polynomial}.
	\end{lemma}
	
	\begin{proof}[\textbf{Proof of Lemma \ref{lemma:HW-uXEv}}]
		Under the stated assumptions, Lemma~\ref{lemma:equivalence-representation-of-uXEv} yields a semi-infinite innovation vector $\bm{Z}_T = (\bbm{\xi}_T^\top,\bbm{\xi}_{T-1}^\top,\bbm{\xi}_{T-2}^\top,\ldots,\bbm{\xi}_1^\top,\bbm{\xi}_0^\top,\bbm{\xi}_{-1}^\top,\ldots)^\top$ and a deterministic matrix $\bm{Q}_{u,v}$ such that $\bbm{u}^\top\bm{X}^\top\bm{E}\bbm{v} = \bm{Z}_T^\top \bm{Q}_{u,v}\bm{Z}_T$.

		Moreover, note that for each $t$ the stacked lag vector $\bbm{x}_t$ is measurable with respect to the $\sigma$-field generated by $\{\bbm{\epsilon}_s : s \leq t-1\}$, whereas $\bbm{\epsilon}_t$ is independent of the past with $\mathbb{E}[\bbm{\epsilon}_t]=\bm{0}$. Hence, for any fixed $\bbm{u},\bbm{v}$, we have
		\[
			\mathbb{E}\big[\langle\bbm{u},\bbm{x}_t\rangle \langle\bbm{\epsilon}_t,\bbm{v}\rangle\big] = \mathbb{E}\Big[\langle\bbm{u},\bbm{x}_t\rangle \,\mathbb{E}\big[\langle\bbm{\epsilon}_t,\bbm{v}\rangle \mid \mathcal{F}_{t-1}\big]\Big] = 0,
		\]
		where $\mathcal{F}_{t-1}$ denotes the $\sigma$-field generated by $\{\bbm{\epsilon}_s:s\leq t-1\}$. Summing over $t$ gives
		\[
			\mathbb{E}[\bbm{u}^\top\bm{X}^\top\bm{E}\bbm{v}] = \sum_{t=1}^T \mathbb{E}\big[\langle\bbm{u},\bbm{x}_t\rangle \langle\bbm{\epsilon}_t,\bbm{v}\rangle\big] = 0.
		\]

		Define $\widetilde{\bm{Q}}_{u,v} = 1/2(\bm{Q}_{u,v} + \bm{Q}_{u,v}^\top)$. Then, we have $\bbm{u}^\top\bm{X}^\top\bm{E}\bbm{v} = \bm{Z}_T^\top \bm{Q}_{u,v}\bm{Z}_T = \bm{Z}_T^\top \widetilde{\bm{Q}}_{u,v}\bm{Z}_T$. Since the coordinates $\{\xi_{it}\}$ are independent and sub-Gaussian with $\|\xi_{it}\|_{\psi_2}\lesssim\sigma$, the entries of $\bm{Z}_T$ are independent and satisfy the same $\psi_2$ bound.
		Moreover, Lemma~\ref{lemma:equivalence-representation-of-uXEv} provides the operator and Frobenius norm bounds for $\bm{Q}_{u,v}$, and taking the symmetric part does not increase either norm. Consequently, $\|\widetilde{\bm{Q}}_{u,v}\|_{\mathrm{op}} \leq \|\bm{Q}_{u,v}\|_{\mathrm{op}} \leq C_{\epsilon,\mathcal A}\sqrt{p}\|\bbm{u}\|_2\|\bbm{v}\|_2$ and $\|\widetilde{\bm{Q}}_{u,v}\|_{\mathrm{F}} \leq \|\bm{Q}_{u,v}\|_{\mathrm{F}} \leq C_{\epsilon,\mathcal A}\sqrt{T}\|\bbm{u}\|_2\|\bbm{v}\|_2$.
		Applying Lemma~\ref{lemma:infinite-HW} with $S = \bm{Z}_T^\top\widetilde{\bm{Q}}_{u,v}\bm{Z}_T$ and $\bm{Q}=\widetilde{\bm{Q}}_{u,v}$ yields, for $t>0$,
		\begin{align*}
			\mathbb{P}\big(|\bbm{u}^\top\bm{X}^\top\bm{E}\bbm{v}|\geq t\big) \leq 2\exp\left(-c\min\left\{\frac{t^2}{\sigma^4\|\widetilde{\bm{Q}}_{u,v}\|_{\mathrm{F}}^2},\frac{t}{\sigma^2\|\widetilde{\bm{Q}}_{u,v}\|_{\mathrm{op}}}\right\}\right).
		\end{align*}
		Substituting the bounds on $\|\widetilde{\bm{Q}}_{u,v}\|_{\mathrm{F}}$ and $\|\widetilde{\bm{Q}}_{u,v}\|_{\mathrm{op}}$ yields
		\begin{align*}
			\mathbb{P}\big(|\bbm{u}^\top\bm{X}^\top\bm{E}\bbm{v}|\geq t\big)\leq 2\exp\Bigg(-c\min\Bigg\{\frac{t^2}{C_{\epsilon,\mathcal A}^2 \sigma^4T\|\bbm{u}\|_2^2\|\bbm{v}\|_2^2},\frac{t}{C_{\epsilon,\mathcal A}\sigma^2\sqrt{p}\|\bbm{u}\|_2\|\bbm{v}\|_2}\Bigg\}\Bigg),
		\end{align*}
		This completes the proof.
	\end{proof}

	\begin{lemma}[Positive definiteness and eigenvalue bounds of the stacked lag covariance]\label{lem:Sigma-x-pd}
		Let $\{\bbm{y}_t\}_{t\in\mathbb Z}$ be a stationary VAR($p$) process with AR polynomial $\mathcal{A}(z) = \bm{I}_d - \sum_{k=1}^p \bm{A}_k z^k$, satisfying $\det\mathcal{A}(z)\neq 0$ for all $|z|\leq 1$. Let the corresponding VMA($\infty$) representation be $\bbm{y}_t = \sum_{j\geq 0} \bbm{\Psi}_j \bbm{\epsilon}_{t-j}$, $t\in\mathbb Z$, where $\{\bbm{\Psi}_j\}_{j\geq0}$ is absolutely summable. Assume $\{\bbm{\epsilon}_t\}$ is $i.i.d.$ with $\mathbb{E}[\bbm{\epsilon}_t]=\bm{0}$ and $\mathbb{E}[\bbm{\epsilon}_t\bbm{\epsilon}_t^\top]=\bbm{\Sigma}_\epsilon\succ\bm{0}$. For a fixed integer $p\geq 1$, define the stacked lag vector $\bbm{x}_t = (\bbm{y}_{t-1}^\top,\ldots,\bbm{y}_{t-p}^\top)^\top\in\mathbb{R}^{pd}$, and its covariance matrix $\bbm{\Sigma}_x = \mathbb{E}[\bbm{x}_t\bbm{x}_t^\top]\in\mathbb{R}^{p d\times pd}$. Then $\bbm{\Sigma}_x\succ\bm{0}$. Moreover, $c_{\epsilon,\mathcal A} \leq \lambda_{\min}(\bbm{\Sigma}_x) \leq \lambda_{\max}(\bbm{\Sigma}_x) \leq C_{\epsilon,\mathcal A}$ with $c_{\epsilon,\mathcal A} = \lambda_{\min}(\bbm{\Sigma}_\epsilon)\mu_{\max}^{-1}(\mathcal{A})$ and $C_{\epsilon,\mathcal A} = \lambda_{\max}(\bbm{\Sigma}_\epsilon)/\mu_{\min}^{-1}(\mathcal{A})$ with $\mu_{\min}(\mathcal{A})$ and $\mu_{\max}(\mathcal{A})$ are defined in Lemma~\ref{lemma:ar-polynomial}.
	\end{lemma}

	\begin{proof}[\textbf{Proof of Lemma \ref{lem:Sigma-x-pd}}]
		Fix any nonzero $\bbm{u}\in\mathbb{R}^{pd}$ and write $\bbm{u}=(\bbm{u}_1^\top,\ldots,\bbm{u}_p^\top)^\top$ with $\bbm{u}_k\in\mathbb{R}^d$. Considering $\bbm{y}_{t-k} = \sum_{j\geq0}\bbm{\Psi}_j\bbm{\epsilon}_{t-k-j}$ and the convention $\bbm{\Psi}_r = \bm{O}$ for $r<0$,
		\begin{align*}
			\langle\bbm{u},\bbm x_t\rangle &= \sum_{k=1}^{p} \langle\bbm{u}_k, \bbm{y}_{t-k}\rangle = \sum_{k=1}^{p}\sum_{j\geq0}\langle\bbm{u}_k, \bbm{\Psi}_j \bbm{\epsilon}_{t-k-j}\rangle = \sum_{k=1}^{p}\sum_{j\geq0}\langle\bbm{\Psi}_j^\top\bbm{u}_k, \bbm{\epsilon}_{t-k-j}\rangle\\
			&= \sum_{k=1}^{p} \sum_{m\geq k} \langle\bbm{\Psi}_{m-k}^\top\bbm{u}_k, \bbm{\epsilon}_{t-m}\rangle = \sum_{k=1}^{p} \sum_{m\geq 0} \langle\bbm{\Psi}_{m-k}^\top\bbm{u}_k, \bbm{\epsilon}_{t-m}\rangle = \sum_{m\geq0} \langle\bbm{c}_m, \bbm{\epsilon}_{t-m}\rangle,
		\end{align*}
		where $\bbm{c}_m = \sum_{k=1}^p \bbm{\Psi}_{m-k}^\top \bbm{u}_k$ with the convention $\bbm \Psi_r = \bm O$ for $r < 0$. Since $\{\bbm{\epsilon}_t\}$ are independent with covariance $\bbm{\Sigma}_\epsilon$, we have
		\begin{align}\label{eq:Var-St-cm}
			\mathrm{Var}(\langle\bbm{u}, \bbm{x}_t\rangle) = \mathrm{Var}\left(\sum_{m\geq0} \langle\bbm{c}_m, \bbm{\epsilon}_{t-m}\rangle\right) = \sum_{m\geq0} \bbm{c}_m^\top \bbm{\Sigma}_\epsilon \bbm{c}_m.
		\end{align}
		Note that $\bbm{\Sigma}_\epsilon\succ \bm 0$, then we have
		\begin{align}\label{eq:Var-St-sandwich-1}
			\lambda_{\min}(\bbm{\Sigma}_\epsilon)\sum_{m\geq 0}\|\bbm c_m\|_2^2 \leq \mathrm{Var}(\langle\bbm{u}, \bbm{x}_t\rangle) \leq \lambda_{\max}(\bbm{\Sigma}_\epsilon)\sum_{m\geq 0}\|\bbm c_m\|_2^2,
		\end{align}
		where $\lambda_{\min}(\bbm{\Sigma}_\epsilon)$ and $\lambda_{\max}(\bbm{\Sigma}_\epsilon)$ are the minimum and maximum eigenvalues of $\bbm{\Sigma}_\epsilon$, respectively.

		Next, recall the previously defined $\bbm{\Psi}(e^{-i\theta})=\sum_{j\geq 0}\bbm{\Psi}_j e^{-ij\theta}$ for $\theta\in[-\pi,\pi]$ and $\bm{U}_p(\theta)=\sum_{k=1}^p \bbm{u}_k e^{-ik\theta}$. Then, we have
		\begin{align*}
			\bm C(\theta) &= \sum_{m\geq 0}\bbm c_m e^{-im\theta} = \sum_{m\geq 0}\sum_{k=1}^{p}\bbm\Psi_{m-k}^\top \bbm u_k e^{-im\theta}=\sum_{k=1}^{p}\sum_{j\geq 0}\bbm\Psi_{j}^\top \bbm u_k e^{-i(j+k)\theta}\\
			&=\left(\sum_{j\geq 0}\bbm\Psi_j^\top e^{-ij\theta}\right)\left(\sum_{k=1}^{p}\bbm u_k e^{-ik\theta}\right) =\bbm\Psi^\top(e^{-i\theta})\bm U_p(\theta).
		\end{align*}
		By Parseval's identity for vector-valued sequences, we have
		\begin{align}\label{eq:Parseval-C}
			\sum_{m\geq 0}\|\bbm c_m\|_2^2 = \frac{1}{2\pi}\int_{-\pi}^{\pi}\|\bm{C}(\theta)\|_2^2\,d\theta.
		\end{align}
		Moreover, since $\bm{C}(\theta)=\bbm{\Psi}^\top(e^{-i\theta})\,\bm U_p(\theta)$, it follows that
		\[
			\|\bm{C}(\theta)\|_2^2 = \bm{C}(\theta)^\dagger \bm{C}(\theta) = \bm{U}_p^\dagger(e^{-i\theta})\bigl(\bbm{\Psi}^\dagger(e^{-i\theta})\bbm{\Psi}(e^{-i\theta})\bigr)\bm{U}_p(e^{-i\theta}).
		\]
		By Lemma~\ref{lemma:ar-polynomial}, we have
		\[
			\frac{1}{\mu_{\max}(\mathcal{A})} \leq \lambda_{\min}\big(\bbm{\Psi}^\dagger(z)\bbm{\Psi}(z)\big) \leq \lambda_{\max}\big(\bbm{\Psi}^\dagger(z)\bbm{\Psi}(z)\big) \leq \frac{1}{\mu_{\min}(\mathcal{A})}, \ \text{for any}\ |z|=1.
		\]
		Therefore,
		\[
			\frac{1}{\mu_{\max}(\mathcal{A})}\|\bm{U}_p(e^{-i\theta})\|_2^2 \leq \|\bm{C}(\theta)\|_2^2 \leq \frac{1}{\mu_{\min}(\mathcal{A})}\|\bm{U}_p(e^{-i\theta})\|_2^2.
		\]
		Integrating and using Parseval again for $\bm{U}_p(e^{-i\theta})$, we have
		\[
			\frac{1}{2\pi}\int_{-\pi}^{\pi}\|\bm{U}_p(e^{-i\theta})\|_2^2\,d\theta = \sum_{k=1}^p \|\bbm u_k\|_2^2 = \|\bbm u\|_2^2.
		\]
		Consequently, we have
		\begin{align}\label{eq:cm-norm-sandwich}
			\frac{1}{\mu_{\max}(\mathcal{A})}\|\bbm u\|_2^2 \leq \sum_{m\geq 0}\|\bbm c_m\|_2^2 \leq \frac{1}{\mu_{\min}(\mathcal{A})}\|\bbm u\|_2^2.
		\end{align}
		Combining \eqref{eq:Var-St-sandwich-1} and \eqref{eq:cm-norm-sandwich} yields
		\[
			\frac{\lambda_{\min}(\bbm{\Sigma}_\epsilon)}{\mu_{\max}(\mathcal{A})}\|\bbm u\|_2^2 \leq \bbm u^\top\bbm\Sigma_x\bbm u \leq \frac{\lambda_{\max}(\bbm{\Sigma}_\epsilon)}{\mu_{\min}(\mathcal{A})}\|\bbm u\|_2^2.
		\]
		Taking the infimum and supremum over $\|\bbm u\|_2=1$ completes the proof.
	\end{proof}

	\begin{lemma}\label{lem:psi-geom}
		Let $\varrho\in(1,\varrho_*)$ and define $C_{\mathcal A}(\varrho)$ by \eqref{eq:Car-def}. With $\vartheta=\varrho^{-1}$, the power-series coefficients $\{\bm \Psi_{k,h}\}_{h \geq 0}$ in the expansion $\mathcal A_k^{-1}(z)=\sum_{h \geq 0}\bm \Psi_{k,h} z^h$, for $|z|<\varrho_*$, satisfy the bound
		\begin{align}\label{eq:Psi-bound}
			\|\bm \Psi_{k,h}\|_{\mathrm{op}}\leq C_{\mathcal A}(\varrho)\vartheta^h,\ \text{for any }\ k\in[K],\ h \geq 0.
		\end{align}
	\end{lemma}

	\begin{proof}[\textbf{Proof of Lemma \ref{lem:psi-geom}}]
		Since $\mathcal A_k^{-1}(z)$ is analytic on $\{|z|<\varrho_*\}$, the Cauchy coefficient formula yields
		\[
			\bm \Psi_{k,h} = \frac{1}{2\pi i}\oint_{|z|=\varrho} \mathcal A_k^{-1}(z)\,z^{-h-1}dz,\quad h \geq 0.
		\]
		Parametrize the circle by $z(t)=\varrho e^{it}$ for $t\in[0,2\pi]$, so that $dz = i \varrho e^{it}\,dt$ and $|dz| = \varrho\,dt$.
		For any unit vector $\bbm x\in\mathbb R^{d}$, we have
		\[
			\Bigl(\oint_{|z|=\varrho}\bm B(z)\,dz\Bigr)\bbm x = \int_{-\pi}^{\pi}\bm B\bigl(\varrho e^{it}\bigr)\bbm x i \varrho e^{it}\,dt.
		\]
		Hence, by thes sub-additivity of Bochner integrals and $\|\bm B(z)\bbm x\|_2\leq \|\bm B(z)\|_{\op}\|\bbm x\|_2$, we have
		\begin{align*}
			&\Bigl\|\Bigl(\oint_{|z|=\varrho}\bm B(z)\,dz\Bigr)\bbm x\Bigr\|_2 = \Bigl\|\int_{-\pi}^{\pi}\bm B(\varrho e^{it})\,\bbm x\, i\varrho e^{it}\,dt\Bigr\|_2 \leq \int_{-\pi}^{\pi}\bigl\|\bm B(\varrho e^{it})\,\bbm x\bigr\|_2\,\bigl|i\varrho e^{it}\bigr|\,dt \\
			&= \int_{-\pi}^{\pi}\bigl\|\bm B(\varrho e^{it})\,\bbm x\bigr\|_2\,\varrho\,dt \leq \|\bbm x\|_2\int_{-\pi}^{\pi}\|\bm B(\varrho e^{it})\|_{\op}\,\varrho\,dt = \|\bbm x\|_2\oint_{|z|=\varrho}\|\bm B(z)\|_{\op}\,|dz|.
		\end{align*}
		Taking the supremum over all $\bbm x$ with $\|\bbm x\|_2=1$ yields
		\[
			\Bigl\|\oint_{|z|=\varrho}\bm B(z)\,dz\Bigr\|_{\op} = \sup_{\|\bbm x\|_2=1}\Bigl\|\Bigl(\oint_{|z|=\varrho}\bm B(z)\,dz\Bigr)\bbm x\Bigr\|_2 \leq \oint_{|z|=\varrho}\|\bm B(z)\|_{\op}\,|dz|.
		\]
		Combining this with $\bm B(z)=\mathcal A_k^{-1}(z)z^{-h-1}$, and using $|z|=\varrho$ and $\oint_{|z|=\varrho}|dz|=2\pi\varrho$, we have
		\[
			\|\bm \Psi_{k,h}\|_{\mathrm{op}} \leq \frac{1}{2\pi}\oint_{|z|=\varrho}\|\mathcal A_k^{-1}(z)\|_{\mathrm{op}}\,|z|^{-h-1}\,|dz| \leq C_{\mathcal A}(\varrho)\,\varrho^{-h} = C_{\mathcal A}(\varrho)\,\vartheta^h,
		\]
		where we used the definition $C_{\mathcal A}(\varrho)=\sup_{k\in[K]}\sup_{|z|=\varrho}\|\mathcal A_k^{-1}(z)\|_{\mathrm{op}}$ and $\vartheta=\varrho^{-1}$.
	\end{proof}

\section{ADMM for Single Client Estimation}\label{sec:ADMM}

		In this section, we present the ADMM algorithm for solving the single-client estimation problem in \eqref{eq:local-estimator} and derive the corresponding closed-form updates. For notational simplicity, we suppress the client index $k$ and write $\bm X\in\mathbb R^{T\times pd}$ and $\bm Y\in\mathbb R^{T\times d}$ for the design matrix and response matrix of a fixed client, respectively. Recall that the single-client estimator solves
		\begin{align}
			(\widetilde{\bm A}_{0}, \widetilde{\bbm \Delta})
			\in \argmin_{\bm A_{0}, \bbm \Delta}
			\frac{1}{T}\|\bm Y-\bm X(\bm A_{0}+\bbm \Delta)\|_{\F}^{2}
			+ \lambda \|\bm A_{0}\|_{*}
			+ \omega \|\bbm \Delta\|_{1},
			\label{eq:local-estimator-admm}
		\end{align}
		subject to $\|\bbm \Delta\|_{\op}\leq \zeta$.

		To obtain separable subproblems, we introduce auxiliary variables $\bm B$, $\bm B_0$, $\bm D$, and $\bm Z$, where $\bm B$ represents the full coefficient matrix, $\bm B_0$ represents the low-rank component, $\bm D$ represents the sparse deviation, and $\bm Z$ is an additional copy of $\bm D$ used to impose the operator-norm constraint. The single-client estimator can be equivalently written as
		\begin{align*}
			\min_{\bm B,\bm B_0,\bm D,\bm Z}\ 
			&\ \frac{1}{T}\|\bm Y-\bm X\bm B\|_{\F}^{2}
			+ \lambda \|\bm B_0\|_{*}
			+ \omega \|\bm D\|_{1}
			+ \iota_{\op,\zeta}(\bm Z) \quad
			\text{s.t.}\  \bm B=\bm B_0+\bm D,\ \bm D=\bm Z,
		\end{align*}
		where the indicator function of the operator-norm ball is defined by
		\[
			\iota_{\op,\zeta}(\bm Z)
			=
			\begin{cases}
			0, & \|\bm Z\|_{\op}\leq \zeta,\\
			+\infty, & \text{otherwise}.
			\end{cases}
		\]

		Let $\bm U$ and $\bm V$ be the scaled dual variables associated with the constraints $\bm B=\bm B_0+\bm D$ and $\bm D=\bm Z$, respectively. For penalty parameters $\rho>0$ and $\eta>0$, the scaled augmented Lagrangian is
		\begin{align}\label{eq:augLag-scaled-rewrite}
			\mathcal L_{\rho,\eta}(\bm B,\bm B_0,\bm D,\bm Z,\bm U,\bm V)
			=&\ 
			\frac{1}{T}\|\bm Y-\bm X\bm B\|_{\F}^{2}
			+ \lambda \|\bm B_0\|_{*}
			+ \omega \|\bm D\|_{1}
			+ \iota_{\op,\zeta}(\bm Z) \\
			&\quad
			+ \frac{\rho}{2}\|\bm B-\bm B_0-\bm D+\bm U\|_{\F}^{2}
			+ \frac{\eta}{2}\|\bm D-\bm Z+\bm V\|_{\F}^{2}.
		\end{align}

		We now derive the ADMM updates by minimizing the scaled augmented Lagrangian in \eqref{eq:augLag-scaled-rewrite} with respect to each primal variable in turn while keeping the remaining variables fixed. This leads to the following simple subproblems. The resulting procedure is summarized in Algorithm~\ref{alg:single-client-admm}.

		\noindent
		\textbf{(I) $\bm B$-update.}
		The $\bm B$-subproblem is
		\[
			\bm B^{(n+1)}
			=
			\argmin_{\bm B}\ 
			\frac{1}{T}\|\bm Y-\bm X\bm B\|_{\F}^{2}
			+ \frac{\rho}{2}\|\bm B-\bm B_0^{(n)}-\bm D^{(n)}+\bm U^{(n)}\|_{\F}^{2}.
		\]
		This is a strictly convex quadratic problem. The first-order optimality condition gives
		\[
			\left(\frac{2}{T}\bm X^{\top}\bm X+\rho\bm I_{pd}\right)\bm B^{(n+1)}
			=
			\frac{2}{T}\bm X^{\top}\bm Y
			+ \rho\bigl(\bm B_0^{(n)}+\bm D^{(n)}-\bm U^{(n)}\bigr),
		\]
		and hence
		\[
			\bm B^{(n+1)}
			=
			\left(\frac{2}{T}\bm X^{\top}\bm X+\rho\bm I_{pd}\right)^{-1}
			\left(\frac{2}{T}\bm X^{\top}\bm Y+\rho(\bm B_0^{(n)}+\bm D^{(n)}-\bm U^{(n)})\right).
		\]

		\noindent
		\textbf{(II) $\bm B_0$-update.}
		The $\bm B_0$-subproblem is
		\[
			\bm B_0^{(n+1)}
			=
			\argmin_{\bm B_0}\ 
			\lambda\|\bm B_0\|_{*}
			+ \frac{\rho}{2}\|\bm B^{(n+1)}-\bm B_0-\bm D^{(n)}+\bm U^{(n)}\|_{\F}^{2}.
		\]
		This is the proximal mapping of the nuclear norm. Therefore,
		\[
			\bm B_0^{(n+1)}
			=
			\operatorname{SVT}_{\lambda/\rho}\bigl(\bm B^{(n+1)}-\bm D^{(n)}+\bm U^{(n)}\bigr).
		\]

		\noindent
		\textbf{(III) $\bm D$-update.}
		The $\bm D$-subproblem is
		\begin{align*}
			\bm D^{(n+1)}
			=
			\argmin_{\bm D}\ 
			&\ \omega\|\bm D\|_{1}
			+ \frac{\rho}{2}\|\bm B^{(n+1)}-\bm B_0^{(n+1)}-\bm D+\bm U^{(n)}\|_{\F}^{2} \\
			&\quad
			+ \frac{\eta}{2}\|\bm D-\bm Z^{(n)}+\bm V^{(n)}\|_{\F}^{2}.
		\end{align*}
		Let
		\[
			\bm M_D^{(n)}
			=
			\frac{\rho(\bm B^{(n+1)}-\bm B_0^{(n+1)}+\bm U^{(n)})+\eta(\bm Z^{(n)}-\bm V^{(n)})}{\rho+\eta}.
		\]
		Then the $\bm D$-subproblem is the proximal mapping of the entrywise $\ell_1$ norm, and hence
		\[
			\bm D^{(n+1)}
			=
			\operatorname{Soft}_{\omega/(\rho+\eta)}\bigl(\bm M_D^{(n)}\bigr).
		\]

		\noindent
		\textbf{(IV) $\bm Z$-update.}
		The $\bm Z$-subproblem is
		\[
			\bm Z^{(n+1)}
			=
			\argmin_{\bm Z}\ 
			\iota_{\op,\zeta}(\bm Z)
			+ \frac{\eta}{2}\|\bm D^{(n+1)}-\bm Z+\bm V^{(n)}\|_{\F}^{2}.
		\]
		Therefore,
		\[
			\bm Z^{(n+1)}
			=
			\Pi_{\|\cdot\|_{\op}\leq\zeta}\bigl(\bm D^{(n+1)}+\bm V^{(n)}\bigr).
		\]
		If $\bm M=\bm P\operatorname{diag}(\sigma_1,\ldots,\sigma_r)\bm Q^\top$ is the singular value decomposition of $\bm M$, then this projection is given by clipping singular values:
		\[
			\Pi_{\|\cdot\|_{\op}\leq\zeta}(\bm M)
			=
			\bm P\operatorname{diag}\bigl(\min\{\sigma_1,\zeta\},\ldots,\min\{\sigma_r,\zeta\}\bigr)\bm Q^\top.
		\]

		\noindent
		\textbf{(V) Dual updates.}
		The scaled dual variables are updated by
		\[
			\bm U^{(n+1)}
			=
			\bm U^{(n)}
			+ \bm B^{(n+1)}-\bm B_0^{(n+1)}-\bm D^{(n+1)},
		\]
		and
		\[
			\bm V^{(n+1)}
			=
			\bm V^{(n)}+\bm D^{(n+1)}-\bm Z^{(n+1)}.
		\]

		\begin{breakablealgorithm}
			\caption{Single-client ADMM for low-rank + sparse VAR coefficient estimation}
			\label{alg:single-client-admm}
			\begin{algorithmic}[1]
			\Require $\bm Y\in\mathbb R^{T\times d}$, $\bm X\in\mathbb R^{T\times pd}$, tuning parameters $\lambda>0$, $\omega>0$, $\zeta>0$, penalty parameters $\rho>0$ and $\eta>0$, tolerances $\varepsilon_{\rm pri},\varepsilon_{\rm dual}>0$, and maximum iteration number $N_{\max}$.
			\State Initialize $(\bm B^{(0)},\bm B_0^{(0)},\bm D^{(0)},\bm Z^{(0)},\bm U^{(0)},\bm V^{(0)})$.
			\State Precompute
			\[
			\bm H \gets \frac{2}{T}\bm X^\top\bm X+\rho\bm I_{pd},
			\quad
			\bm G \gets \frac{2}{T}\bm X^\top\bm Y.
			\]
			\For{$n=0,1,2,\ldots,N_{\max}-1$}
			\State $\bm B^{(n+1)} \gets \bm H^{-1}\Bigl(\bm G+\rho(\bm B_0^{(n)}+\bm D^{(n)}-\bm U^{(n)})\Bigr)$.
			\State $\bm B_0^{(n+1)} \gets \operatorname{SVT}_{\lambda/\rho}\bigl(\bm B^{(n+1)}-\bm D^{(n)}+\bm U^{(n)}\bigr)$.
			\State $\bm M_D^{(n)} \gets \{\rho(\bm B^{(n+1)}-\bm B_0^{(n+1)}+\bm U^{(n)})+\eta(\bm Z^{(n)}-\bm V^{(n)})\}/(\rho+\eta)$.
			\State $\bm D^{(n+1)} \gets \operatorname{Soft}_{\omega/(\rho+\eta)}\bigl(\bm M_D^{(n)}\bigr)$.
			\State $\bm Z^{(n+1)} \gets \Pi_{\|\cdot\|_{\op}\leq\zeta}\bigl(\bm D^{(n+1)}+\bm V^{(n)}\bigr)$.
			\State $\bm U^{(n+1)} \gets \bm U^{(n)}+\bm B^{(n+1)}-\bm B_0^{(n+1)}-\bm D^{(n+1)}$.
			\State $\bm V^{(n+1)} \gets \bm V^{(n)}+\bm D^{(n+1)}-\bm Z^{(n+1)}$.
			\State $\bm R_1^{(n+1)} \gets \bm B^{(n+1)}-\bm B_0^{(n+1)}-\bm D^{(n+1)}$.
			\State $\bm R_2^{(n+1)} \gets \bm D^{(n+1)}-\bm Z^{(n+1)}$.
			\State $\bm S^{(n+1)} \gets \rho\bigl(\bm B_0^{(n+1)}-\bm B_0^{(n)}+\bm D^{(n+1)}-\bm D^{(n)}\bigr)+\eta\bigl(\bm Z^{(n+1)}-\bm Z^{(n)}\bigr)$.
			\State \textbf{Stop} if $\|\bm R_1^{(n+1)}\|_{\F}+\|\bm R_2^{(n+1)}\|_{\F}\leq \varepsilon_{\rm pri}\ \text{ and }\ \|\bm S^{(n+1)}\|_{\F}\leq \varepsilon_{\rm dual}$.
			\EndFor
			\Ensure $(\widetilde{\bm A}_0,\widetilde{\bbm \Delta})\gets \bigl((\bm B_0^{(n)})^\top,(\bm Z^{(n)})^\top\bigr)$ and $\widetilde{\bm A}\gets \widetilde{\bm A}_0+\widetilde{\bbm \Delta}$.
			\end{algorithmic}
		\end{breakablealgorithm}

		Here, $\operatorname{Soft}_{\tau}(\cdot)$ denotes the entrywise soft-thresholding operator, i.e.,
		$$
			\operatorname{Soft}_{\tau}(x)=\operatorname{sign}(x)(|x|-\tau)_+,
			\qquad
			\bigl(\operatorname{Soft}_{\tau}(\bm M)\bigr)_{ij}
			=
			\operatorname{Soft}_{\tau}(M_{ij}),
		$$
		where $(u)_+=\max\{u,0\}$.
		$\operatorname{SVT}_{\tau}(\cdot)$ denotes the singular-value soft-thresholding operator: for a matrix $\bm M$ with singular value decomposition $\bm M=\bm P\,\mathrm{diag}(\sigma_1,\ldots,\sigma_r)\,\bm Q^\top$,
		$$
			\operatorname{SVT}_{\tau}(\bm M)
			=
			\bm P\,\mathrm{diag}\bigl((\sigma_1-\tau)_+,\ldots,(\sigma_r-\tau)_+\bigr)\,\bm Q^\top.
		$$
		Finally, $\Pi_{\|\cdot\|_{\op}\leq \zeta}$ denotes the projection onto the operator-norm ball, which is obtained by clipping the singular values at $\zeta$.

	\section{Additional Tables}\label{sec:add-tables}
	\renewcommand{\thetable}{S.\arabic{table}}
	\renewcommand{\thefigure}{S.\arabic{figure}}

	\setcounter{table}{0}
	\setcounter{figure}{0}

	In this section, we summarize the data description for the two empirical applications reported in the main paper. Table~\ref{tab:sample_range_all} lists the sampling frequency and sample period for each dataset. The state-level application uses monthly series for five Great Lakes states (Illinois, Indiana, Michigan, Ohio, and Wisconsin) over 2020--01 to 2023--12. The cross-country application uses quarterly national-level series for a set of advanced economies, with sample periods varying by country depending on data availability.

		\begin{table}[H]
		\centering
		\caption{Sample periods and data frequencies used in the empirical applications.}
		\label{tab:sample_range_all}
		\footnotesize
		\renewcommand{\arraystretch}{0.9}
		\setlength{\tabcolsep}{8pt}

		\begin{tabular}{l c c @{\hspace{18pt}} l c c}
		\toprule
		\multicolumn{3}{c}{\textbf{State-level dataset}} 
		& \multicolumn{3}{c}{\textbf{National dataset}} \\
		\cmidrule(lr){1-3}\cmidrule(lr){4-6}
		\textbf{State} & \textbf{Frequency} & \textbf{Sample period} 
		& \textbf{Nation} & \textbf{Frequency} & \textbf{Sample period} \\
		\midrule
		Illinois  & Monthly & 2020-01 -- 2023-12 & United States  & Quarterly & 1970-Q1 -- 2023-Q4 \\
		Indiana   & Monthly & 2020-01 -- 2023-12 & Australia      & Quarterly & 1985-Q1 -- 2023-Q4 \\
		Michigan  & Monthly & 2020-01 -- 2023-12 & Canada         & Quarterly & 1991-Q1 -- 2023-Q4 \\
		Ohio      & Monthly & 2020-01 -- 2023-12 & Germany        & Quarterly & 1992-Q1 -- 2023-Q4 \\
		Wisconsin & Monthly & 2020-01 -- 2023-12 & Korea          & Quarterly & 2001-Q1 -- 2023-Q4 \\
				&         &                    & Norway         & Quarterly & 2001-Q1 -- 2023-Q4 \\
				&         &                    & Sweden         & Quarterly & 2002-Q1 -- 2023-Q4 \\
				&         &                    & Denmark        & Quarterly & 2005-Q1 -- 2023-Q4 \\
				&         &                    & Japan          & Quarterly & 2003-Q1 -- 2019-Q4 \\
				&         &                    & United Kingdom & Quarterly & 2011-Q1 -- 2023-Q4 \\
		\bottomrule
		\end{tabular}
		\end{table}

		\newpage
		Tables~\ref{tab:data_preprocessing_1} and~\ref{tab:data_preprocessing_2} summarize the variables and preprocessing rules used in our two empirical applications. Table~\ref{tab:data_preprocessing_1} reports the monthly state-level electricity and macroeconomic variables, together with their category labels and preprocessing rules. In particular, electricity-market variables are treated as privacy-sensitive, whereas the macroeconomic controls are publicly available and treated as non-sensitive. Table~\ref{tab:data_preprocessing_2} reports the quarterly cross-country macroeconomic variables, grouped into standard economic categories. For both datasets, we record whether a series is seasonally adjusted and specify the transformation applied prior to estimation, with the transformation rule chosen consistently within each dataset to ensure comparability across units.

		\begin{table}[H]
		\centering
		\footnotesize
		\setlength{\tabcolsep}{8pt}
		\renewcommand{\arraystretch}{0.9}

		\caption{Monthly electricity and macroeconomic variables. Category code (C): 1 = electricity variables, which are treated as privacy-sensitive in the intended deployment setting; 2 = macroeconomic variables, which are publicly available and treated as non-sensitive.
		Seasonal-adjustment indicator (S): 1 = deseasonalized by X-13ARIMA-SEATS; 0 = no seasonal adjustment.
		Transformation code (T): 1 = first difference, 2 = first difference of log-transformed series.
		For each variable, the transformation code is selected globally and then applied uniformly across all states.}
		\label{tab:data_preprocessing_1}

		\begin{tabular}{l c c c p{10.6cm}}
		\toprule
		\textbf{Abbreviation} & \textbf{C} & \textbf{S} & \textbf{T} & \textbf{Description} \\
		\midrule

		% = C=1: electricity variables =
		Revenue\_R & 1 & 1 & 2 & Residential electricity revenue \\
		Sales\_R   & 1 & 1 & 2 & Residential electricity sales \\
		Price\_R   & 1 & 1 & 2 & Residential average retail electricity price \\
		Revenue\_C & 1 & 1 & 2 & Commercial electricity revenue \\
		Sales\_C   & 1 & 1 & 2 & Commercial electricity sales \\
		Price\_C   & 1 & 1 & 2 & Commercial average retail electricity price \\
		Revenue\_I & 1 & 1 & 2 & Industrial electricity revenue \\
		Sales\_I   & 1 & 1 & 2 & Industrial electricity sales \\
		Price\_I   & 1 & 1 & 1 & Industrial average retail electricity price \\

		% = C=2: macroeconomic variables =
		unemployment\_rate  & 2 & 0 & 1 & State unemployment rate \\
		payroll\_employment & 2 & 0 & 2 & Total nonfarm payroll employment \\
		coincident\_index   & 2 & 0 & 1 & Coincident economic activity index \\
		\bottomrule
		\end{tabular}

		\end{table}

		\begin{table}[H]
		\centering
		\footnotesize
		\setlength{\tabcolsep}{8pt}
		\renewcommand{\arraystretch}{0.9}

		\caption{Quarterly macroeconomic variables. Category code (C): 1 = real activity and national accounts; 2 = external sector / trade \& balance; 3 = production and sales; 4 = housing; 5 = money/interest rates; 6 = labour market; 7 = prices and wages; 8 = financial/other.
		Seasonal-adjustment indicator (S): 1 = deseasonalized by X-13ARIMA-SEATS; 0 = no seasonal adjustment.
		Transformation code (T): 1 = first difference and 2 = first difference of log-transformed series.
		For each variable, the transformation code is selected globally and then applied uniformly across all states.}
		\label{tab:data_preprocessing_2}

		\begin{tabular}{l c c c p{10.8cm}}
		\toprule
		\textbf{Abbreviation} & \textbf{C} & \textbf{S} & \textbf{T} & \textbf{Description} \\
		\midrule

		% = C=1: real activity and national accounts =
		GDP        & 1 & 1 & 2 & Real Gross Domestic Product \\
		CE:H       & 1 & 1 & 2 & Household final consumption expenditure \\
		CE:G       & 1 & 1 & 2 & Government final consumption expenditure \\
		INV        & 1 & 1 & 2 & Gross fixed capital formation (investment) \\

		% = C=2: external sector / trade & balance =
		EXP        & 2 & 1 & 2 & Exports of goods and services \\
		IMP        & 2 & 1 & 2 & Imports of goods and services \\
		NEER       & 2 & 0 & 2 & Nominal effective exchange rate \\
		CA         & 2 & 0 & 1 & Current account balance (as \% of GDP) \\

		% = C=3: production and sales =
		IP:T       & 3 & 1 & 2 & Industrial production index: total \\
		IP:MF      & 3 & 1 & 2 & Industrial production index: manufacturing \\
		PV:C       & 3 & 1 & 2 & Production volume index: construction \\
		RT         & 3 & 1 & 2 & Retail trade volume index \\

		% = C=4: housing =
		RHP        & 4 & 0 & 2 & Real house price index \\
		PIR        & 4 & 0 & 2 & House price-to-income ratio \\

		% = C=5: money/interest rates =
		IRST       & 5 & 0 & 1 & Short-term interest rate \\
		IRLT       & 5 & 0 & 1 & Long-term interest rate \\

		% = C=6: labour market =
		EMP        & 6 & 1 & 2 & Employment rate \\
		UNRATE     & 6 & 1 & 1 & Unemployment rate \\

		% = C=7: prices and wages =
		HE         & 7 & 1 & 2 & Hourly earnings \\
		ULC        & 7 & 1 & 2 & Unit labor cost \\
		CPI:T      & 7 & 0 & 2 & Consumer Price Index: total \\
		CPI:core   & 7 & 0 & 2 & Core Consumer Price Index: excluding food and energy \\
		CPI:energy & 7 & 0 & 2 & Consumer Price Index: energy component \\

		% = C=8: financial/other =
		SP         & 8 & 0 & 2 & Share price index \\
		COM        & 8 & 0 & 2 & Commodity price index \\

		\bottomrule
		\end{tabular}
		\end{table}

\linespread{1.54}
\selectfont{}

\setlength{\bibsep}{1pt}
\bibliography{mybib}

\end{document}